\documentclass[11pt]{article}
\pdfoutput=1	
\usepackage{fullpage}
\usepackage{amsfonts, amssymb, amsmath, amsthm}
\usepackage{latexsym}
\usepackage[tracking=smallcaps]{microtype}	
\usepackage{url}
\usepackage{color}
\definecolor{DarkGray}{rgb}{0.1,0.1,0.5}
\usepackage[colorlinks=true,breaklinks, linkcolor= DarkGray,citecolor= DarkGray,urlcolor= DarkGray]{hyperref}	

\usepackage{graphicx}
\usepackage[tight, TABBOTCAP]{subfigure}

\newcommand{\bra}[1]{{\langle#1|}}
\newcommand{\ket}[1]{{|#1\rangle}}
\newcommand{\braket}[2]{{\langle#1|#2\rangle}}
\newcommand{\ketbra}[2]{{\ket{#1}\!\bra{#2}}}

\newcommand{\abs}[1]{{\lvert #1\rvert}}	
\newcommand{\bigabs}[1]{{\big\lvert #1\big\rvert}}
\newcommand{\Bigabs}[1]{{\Big\lvert #1\Big\rvert}}

\newcommand{\norm}[1]{{\| #1 \|}}
\newcommand{\bignorm}[1]{{\big\| #1 \big\|}}
\newcommand{\Bignorm}[1]{{\Big\| #1 \Big\|}}
\newcommand{\Biggnorm}[1]{{\Bigg\| #1 \Bigg\|}}
\newcommand{\trnorm}[1]{{\| #1 \|_{\mathrm{tr}}}}
\newcommand{\bigtrnorm}[1]{{\bigl\| #1 \bigr\|_{\mathrm{tr}}}}	
\newcommand{\Bigtrnorm}[1]{{\Bigl\| #1 \Bigr\|_{\mathrm{tr}}}}

\newcommand{\binomial}[2]{\ensuremath{\left(\begin{smallmatrix}#1 \\ #2 \end{smallmatrix}\right)}}
\newcommand{\smatrx}[1]{\ensuremath{\left(\begin{smallmatrix}#1\end{smallmatrix}\right)}}

\DeclareMathOperator{\Ex}{\operatorname{E}}
\DeclareMathOperator{\Tr}{\operatorname{Tr}}

\def\tensor {\otimes}

\def\A {{\mathcal A}}
\def\B {{\mathcal B}}
\def\C {{\bf C}}

\def\E {{\mathcal E}}
\def\F {{\mathcal F}}
\def\G {{\mathcal G}}
\def\H {{\mathcal H}}
\let\Lstroke\L	\def\L {{\mathcal L}}		

\def\cP {{\mathcal P}}
\def\R {{\bf R}}
\def\S {{\mathcal S}}
\def\U {{\mathcal U}}
\def\V {{\mathcal V}}

\renewcommand{\P}{\ensuremath{\mathsf{P}}}
\newcommand{\NP}{\ensuremath{\mathsf{NP}}}
\newcommand{\IP}{\ensuremath{\mathsf{IP}}}
\newcommand{\PSPACE}{\ensuremath{\mathsf{PSPACE}}}
\newcommand{\BQP}{\ensuremath{\mathsf{BQP}}}

\newcommand{\NEXP}{\ensuremath{\mathsf{NEXP}}}
\newcommand{\QIP}{\ensuremath{\mathsf{QIP}}}
\newcommand{\QMIP}{\ensuremath{\mathsf{QMIP}}}
\newcommand{\MIP}{\ensuremath{\mathsf{MIP}}}

\DeclareMathOperator{\Range}{\operatorname{Range}}

\DeclareMathOperator{\poly}{\operatorname{poly}}
\DeclareMathOperator{\qpoly}{\operatorname{qpoly}}

\newcommand{\identity}{\ensuremath{\boldsymbol{1}}} 

\def\phasegate {G}		
\def\xzdeterminedset {{\mathcal Q}}
\def\stabilizerset {{\mathcal R}}
\newcommand{\kappaEPR}{\kappa_*}

\def\Aad {\A_{\text{ad}}}				
\def\Aadhat {\hat{\A}_{\text{ad}}}		
\def\Bad {\B_{\text{ad}}}				
\def\Badhat {\hat{\B}_{\text{ad}}}		

\def\device{D}	

\def\h #1{h_{#1}}
\def\hA #1{h_{#1}^{\smash{A}}}
\def\hB #1{h_{#1}^{\smash{B}}}
\def\hX #1{h_{#1}^{\smash{\device}}}
\def\hdec #1#2{#1{h}_{#2}}
\def\hAdec #1#2{#1{h}_{#2}^{\smash{A}}}
\def\hBdec #1#2{#1{h}_{#2}^{\smash{B}}}
\def\hXdec #1#2{#1{h}_{#2}^{\smash{\device}}}

\def\RAjah #1#2#3{R^A_{#2}({#3})}
\def\RAdecjah #1#2#3#4{#1{R}^A_{#3}({#4})}
\def\RBjah #1#2#3{R^B_{#2}({#3})}
\def\RAj #1{R^A_{#1}}	
\def\RBj #1{R^B_{#1}}
\def\RAdecj #1#2{#1{R}^A_{#2}}	
\def\RBdecj #1#2{#1{R}^B_{#2}}
\def\RAa #1{R^A_{#1}}
\def\RBa #1{R^B_{#1}}
\def\RXa #1{R^{\device}_{#1}}
\def\RAdeca #1#2{#1{R}^A_{#2}}
\def\RBdeca #1#2{#1{R}^B_{#2}}
\def\RXdeca #1#2{#1{R}^{\device}_{#2}}
\def\RXjah #1#2#3{R^{\device}_{#2}({#3})}
\def\RXdecjah #1#2#3#4{#1{R}^{\device}_{#3}({#4})}
\def\PAjh #1#2{P^A_{#1}({#2})}		
\def\PBjh #1#2{P^B_{#1}({#2})}

\def\PXjh #1#2{P^{\device}_{#1}({#2})}
\def\PABjh #1#2{P^{AB}_{#1}({#2})}
\def\PABj #1{P^{AB}_{#1}}
\def\PABdecj #1#2{#1{P}^{AB}_{#2}}

\def\EAj #1{\E^A_{#1}}
\def\EAdecj #1#2{#1{\E}^A_{#2}}
\def\EAjh #1#2{\E^{A \vert \smash{#2}}_{#1}}
\def\EBj #1{\E^B_{#1}}
\def\EBdecj #1#2{#1{\E}^B_{#2}}
\def\EXj #1{\E^{\device}_{#1}}
\def\EXdecj #1#2{#1{\E}^{\device}_{#2}}
\def\EBjh #1#2{\E^{B \vert \smash{#2}}_{#1}}
\def\EXjh #1#2{\E^{{\device} \vert \smash{#2}}_{#1}}
\def\EABj #1{\E^{AB}_{#1}}
\def\EABdecj #1#2{#1{\E}^{AB}_{#2}}
\def\EABjh #1#2{\E^{AB \vert \smash{#2}}_{#1}}

\def\Bguessj #1{\G^{B}_{#1}}
\def\Bguessdecj #1#2{#1{\G}^{B}_{#2}}
\def\ABguessj #1{\G^{AB}_{#1}}
\def\ABguessdecj #1#2{#1{\G}^{AB}_{#2}}

\def\PAmeasBjh #1#2{F^A_{#1}({#2})}
\def\PAmeasBj #1{F^A_{#1}}
\def\AmeasBj #1{\F^A_{#1}}
\def\AmeasBjh #1#2{\F^{A \vert \smash{#2}}_{#1}}
\def\AmeasBBj #1{\F^{AB}_{#1}}
\def\AmeasBBjh #1#2{\F^{AB \vert \smash{#2}}_{#1}}	

\def\UXj #1{U^{\device}_{#1}}
\def\UAjh #1#2{U^A_{#1}({#2})}
\def\UBjh #1#2{U^B_{#1}({#2})}
\def\UXjh #1#2{U^{\device}_{#1}({#2})}

\def\UsingleAj #1{U^A_{#1}}
\def\UsingleAjh #1#2{U^A_{#1}({#2})}
\def\UsingleXjh #1#2{U^{\device}_{#1}({#2})}

\def\UmultiAj #1{M^A_{#1}}
\def\UmultiA {M^A}
\def\UmultiB {M^B}
\def\UmultiX {M^{\device}}
\def\UmultiXj #1{M^{\device}_{#1}}
\def\UmultiAjh #1#2{M^A_{#1}({#2})}
\def\UmultiXjh #1#2{M^{\device}_{#1}({#2})}

\def\UidealA {\mathcal{I}^A}
\def\UidealB {\mathcal{I}^B}
\def\UidealX {\mathcal{I}^{\device}}
\def\UidealXdec #1{#1{\mathcal{I}}^{\device}}

\def\ABunitary {\Lambda}	
\def\ABunitaryBj #1{\ABunitary^B_{#1}}
\def\AAunitary {V}	
\def\AAunitaryAj #1{\AAunitary^A_{#1}}
\def\AAunitaryBj #1{\AAunitary^B_{#1}}
\def\AAunitaryABj #1{\AAunitary^{AB}_{#1}}
\def\AAunitarysupABj #1{{\cal V}^{AB}_{#1}}
\def\AAunitaryXj #1{\AAunitary^{\device}_{#1}}
\def\BBunitary {W}	
\def\BBunitaryAj #1{\BBunitary^A_{#1}}
\def\BBunitaryBj #1{\BBunitary^B_{#1}}
\def\BBunitarysupABj #1{{\cal W}^{AB}_{#1}}
\def\BBunitaryABj #1{\BBunitary^{AB}_{#1}}
\def\BBunitaryXj #1{\BBunitary^{\device}_{#1}}

\def\psione {\psi}
\def\psidecone #1{#1{\psi}}
\def\psij #1{\psi_{#1}}
\def\psiAh #1{\psi({#1})}

\def\psiAdech #1#2{#1{\psi}({#2})}
\def\psijh #1#2{\psi({#2})}		
\def\psidecjh #1#2#3{#1{\psi}({#3})}		
\def\rhoone{\rho_1}
\def\rhodecone #1{#1{\rho}_1}
\def\rhoj #1{\rho_{#1}}
\def\rhodecj #1#2{#1{\rho}_{#2}}
\def\rhojh #1#2{\rho_{#1}({#2})}		
\def\rhoh #1{\rho({#1})}
\def\rhoAh #1{\rho({#1})}
\def\rhoBh #1{\rho({#1})}

\def\rhoAdech #1#2{#1{\rho}({#2})}
\def\density {\varrho}

\def\XA {\mathcal{X}^A}
\def\XB {\mathcal{X}^B}
\def\XBj #1{\mathcal{X}^{B_{#1}}}
\def\XX {\mathcal{X}^{\device}}
\def\XAB {\mathcal{X}^{AB}}	

\def\XmultiA {\mathcal{Y}^A}
\def\XmultiX {\mathcal{Y}^{\device}}

\newcounter{sprows}

\newlength{\spheight}
\newlength{\spraise}

\newcommand{\comment}[1]{\emph{\color{blue}Comment:\color{black} #1}} 
\newlength{\commentslength}
\newcommand{\comments}[1]{
\hspace{-2\parindent}
\addtolength{\commentslength}{-\commentslength}
\addtolength{\commentslength}{\linewidth}
\addtolength{\commentslength}{-\parindent}
\fcolorbox{blue}{white}{\smallskip\begin{minipage}[c]{\commentslength}
\emph{Comments:}\begin{itemize}#1\end{itemize}\end{minipage}}\bigskip
}
\renewcommand{\comment}[1]{}\renewcommand{\comments}[1]{}
\newcommand{\rem}[1]{}

\numberwithin{equation}{section} 

\newtheorem{theorem}{Theorem}[section]
\newtheorem{lemma}[theorem]{Lemma}
\newtheorem{corollary}[theorem]{Corollary}
\newtheorem{claim}[theorem]{Claim}
\newtheorem{fact}[theorem]{Fact}
\newtheorem{proposition}[theorem]{Proposition}

\newtheorem{definition}[theorem]{Definition}

\newfont{\subsubsecfnt}{ptmri8t at 11pt}
\renewcommand{\subparagraph}[1]{\smallskip{\subsubsecfnt #1.}}

\newcommand{\eqnref}[1]{\hyperref[#1]{{(\ref*{#1})}}}
\newcommand{\thmref}[1]{\hyperref[#1]{{Theorem~\ref*{#1}}}}
\newcommand{\lemref}[1]{\hyperref[#1]{{Lemma~\ref*{#1}}}}
\newcommand{\corref}[1]{\hyperref[#1]{{Corollary~\ref*{#1}}}}
\newcommand{\defref}[1]{\hyperref[#1]{{Definition~\ref*{#1}}}}
\newcommand{\secref}[1]{\hyperref[#1]{{Section~\ref*{#1}}}}
\newcommand{\figref}[1]{\hyperref[#1]{{Figure~\ref*{#1}}}}
\newcommand{\tabref}[1]{\hyperref[#1]{{Table~\ref*{#1}}}}
\newcommand{\remref}[1]{\hyperref[#1]{{Remark~\ref*{#1}}}}
\newcommand{\appref}[1]{\hyperref[#1]{{Appendix~\ref*{#1}}}}
\newcommand{\claimref}[1]{\hyperref[#1]{{Claim~\ref*{#1}}}}
\newcommand{\factref}[1]{\hyperref[#1]{{Fact~\ref*{#1}}}}
\newcommand{\propref}[1]{\hyperref[#1]{{Proposition~\ref*{#1}}}}
\newcommand{\exampleref}[1]{\hyperref[#1]{{Example~\ref*{#1}}}}
\newcommand{\conjref}[1]{\hyperref[#1]{{Conjecture~\ref*{#1}}}}

\allowdisplaybreaks[1]

\begin{document}
\def\compilefullpaper{}
\renewcommand{\comment}[1]{}\renewcommand{\comments}[1]{}

\ifx\compilefullpaper\undefined  
\documentclass[11pt]{article}

\begin{document}
\fi

\title{A classical leash for a quantum system: \\ Command of quantum systems via rigidity of CHSH games}

\author{Ben W.~Reichardt \\ {\small University of Southern California} \and Falk Unger \\ {\small Knight Capital Group} \and Umesh Vazirani \\ {\small UC Berkeley}}
\date{}

\maketitle

\begin{abstract}\normalsize 
Can a classical system command a general adversarial quantum system to realize arbitrary quantum dynamics?  If so, then we could realize the dream of device-independent quantum cryptography: using untrusted quantum devices to establish a shared random key, with security based on the correctness of quantum mechanics.  It would also allow for testing whether a claimed quantum computer is truly quantum.  Here we report a technique by which a classical system can certify the joint, entangled state of a bipartite quantum system, as well as command the application of specific operators on each subsystem.  This is accomplished by showing a strong converse to Tsirelson's optimality result for the Clauser-Horne-Shimony-Holt (CHSH) game: the only way to win many games is if the bipartite state is close to the tensor product of EPR states, and the measurements are the optimal CHSH measurements on successive qubits.  This leads directly to a scheme for device-independent quantum key distribution.  Control over the state and operators can also be leveraged to create more elaborate protocols for realizing general quantum circuits, and to establish that $\QMIP = \MIP^*$.  
\end{abstract}

\ifx\compilefullpaper\undefined  
\bibliographystyle{alpha-eprint}
\bibliography{q}

\end{document}
\fi

\clearpage
\tableofcontents
\clearpage

\ifx\compilefullpaper\undefined  
\documentclass[11pt]{article}

\begin{document}
\fi

\vspace*{-.8cm} 
\section{Introduction} \label{s:introduction}

Do the laws of quantum mechanics place any limits on how well a classical experimentalist can characterize the state and dynamics of a large quantum system?  As a thought experiment, consider that we are presented with a quantum system, together with instructions on how to control its evolution from a claimed initial state.  We make no assumptions about its inner structure, aside from its conforming to quantum mechanics.  Can we, as classical beings, possibly convince ourselves that the quantum system was indeed initialized as claimed, and that its state evolves as we instruct? 

More formally, model the quantum system as contained in a black box, and model our classical interactions with it as questions and answers across a digital interface, perhaps of buttons and light bulbs (\figref{f:blackbox}).  Using this limited interface, we wish to characterize the initial state of the system.  We also wish to certify that on command---by pressing a suitable sequence of buttons---the system applies a chosen local Hamiltonian, or equivalently a sequence of one- and two-qubit quantum gates, and outputs desired measurement results.  

\begin{figure}[b]
\centering
\raisebox{-.25cm}{\includegraphics[scale=.611]{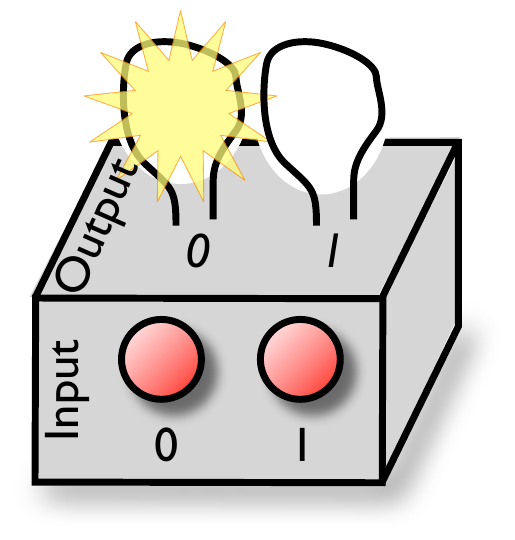}}
\caption{{\bf Classical interaction with a quantum system.} A~general system can be abstracted as a black box, with two buttons for accepting binary input and two light bulbs for output.  Using this interface, we wish to control fully the system's quantum dynamics.} \label{f:blackbox}
\end{figure}

Although partly a philosophical question, a positive resolution would have important consequences.  It is particularly relevant in quantum cryptography, where it is natural to model the quantum system as adversarial since the goal is to protect honest users from malicious adversaries.  Public-key distribution schemes have security based on the assumed difficulty of solving certain problems~\cite{DiffieHellman76, RivestShamirAdleman78rsa}, but quantum algorithms can violate these assumptions~\cite{Shor95factor}.  The \emph{raison d'\^etre} of quantum cryptography is to create a cryptographic system with security premised solely on basic laws of physics, and with quantum key distribution (QKD) and its security proofs~\cite{BennettBrassard84qkd, LoChau98qkdsecurity, ShorPreskill00qkdsecurity} it appeared to have achieved exactly this. However, attackers have repeatedly breached the security of QKD experiments, by exploiting imperfect implementations of the quantum devices~\cite{ZhaoFungQiChenLo07qkdattack, LydersenWiechersWittmannElserSkaarMakarov10qkdattack, GerhardtLiuLamasLinaresSkaarKurtsieferMakarov10qkdbroken}.  Rather than relying on ad hoc countermeasures, Mayers and Yao's $1998$ vision of {device-independent} (DI) QKD~\cite{MayersYao98chsh}, hinted at earlier by Ekert~\cite{Ekert91qkd}, relaxes all modeling assumptions on the devices, and even allows for them to have been constructed by an adversary.  It instead imagines giving the devices tests that cannot be passed unless they carry out the QKD protocol securely.  The challenge at the heart of this vision is for a classical experimentalist to force untrusted quantum devices to act according to certain specifications.  DIQKD has not been known to be possible; security proofs to date require the unrealistic assumption that the devices have no memory between trials, or that each party has many, strictly isolated devices~\cite{BarrettHardyKent04diqkd, MasanesRennerChristandlWinterBarrett06DIQKDnosignalingcomposablesecurity, AcinMassarPironio06DIQKDnosignaling, Masanes09diqkdnosignalingcomposablesecurity, HanggiRennerWolf10diqkd, AcinBrunnerGisinMassarPironioScarani07diqkdcollectiveattacks, PironioAcinBrunnerGisinMassarScarani09qkd, McKague09deviceindependent, HanggiRenner10deviceindependent, MasanesPironioAcin10deviceindependent}.\footnote{Refs.~\cite{HanggiRenner10deviceindependent, MasanesPironioAcin10deviceindependent} assume only that measurements for different games commute.  This is mathematically weaker than requiring measurements to lie in tensor product, but places the same constraints on an implementation.}  A scheme for characterizing and commanding a black-box quantum device would provide a novel approach to achieving DIQKD.  

Further, as the power of quantum mechanics is harnessed at larger scales, for example with the advent of quantum computers, it will be useful to evaluate whether a quantum device in fact carries out the claimed dynamics~\cite{AharonovBenOrEban08authenticated, BroadbentFitzsimonsKashefi08authenticated}.  Finally, we might wish to test the applicability of quantum mechanics for large systems, a situation in which Nature itself plays the role of the adversary~\cite{AharonovVazirani12quantummechanics}.

The existence of a general scheme for commanding an unmodeled quantum device appears singularly implausible.  For example, in an adversarial setting, experiments cannot be repeated exactly to gather statistics, since a system with memory could deliberately deceive the experimentalist.  More fundamentally, as macroscopic, classical entities, our access to a quantum system is extremely limited and indirect, and the measurements we apply collapse the quantum state.  We have never experienced quantum superposition---and likely nor have our cats.  Furthermore, whereas the dimension of the underlying Hilbert space scales exponentially in the number of particles or can be infinite, the information accessible via measurement only grows linearly. Indeed, as formulated it is impossible to command a single black-box system.  Quite simply, one cannot distinguish between a quantum system that evolves as desired and a device that merely simulates the desired evolution using a classical~computer.  

In this paper, we consider a closely related scenario.  Suppose we are instead given two devices, each modeled as a black box as above, and prevented from communicating with each other.  In this setting, with no further assumptions, we show how to classically command the devices.  That is, there is a strategy for pushing the buttons such that the answering light bulb flashes will satisfy a prescribed test only if the two devices started in a particular initial quantum state, to which they applied a desired sequence of quantum gates.  Moreover, though impractical, the scheme is theoretically efficient---in the sense that the total effort, measured by the number of button pushes, scales as a polynomial function of the size of the desired quantum~circuit.  Among other consequences, this result is still sufficiently powerful to imply a DIQKD scheme.  The necessary security assumptions are minimal: that the parties have isolated laboratories (as cryptography requires secrecy), they have local sources of random bits and share an authenticated classical communications channel (to prevent man-in-the-middle attacks), and quantum theory is correct.  

\begin{figure}
\centering
\raisebox{-.25cm}{\includegraphics[scale=.366]{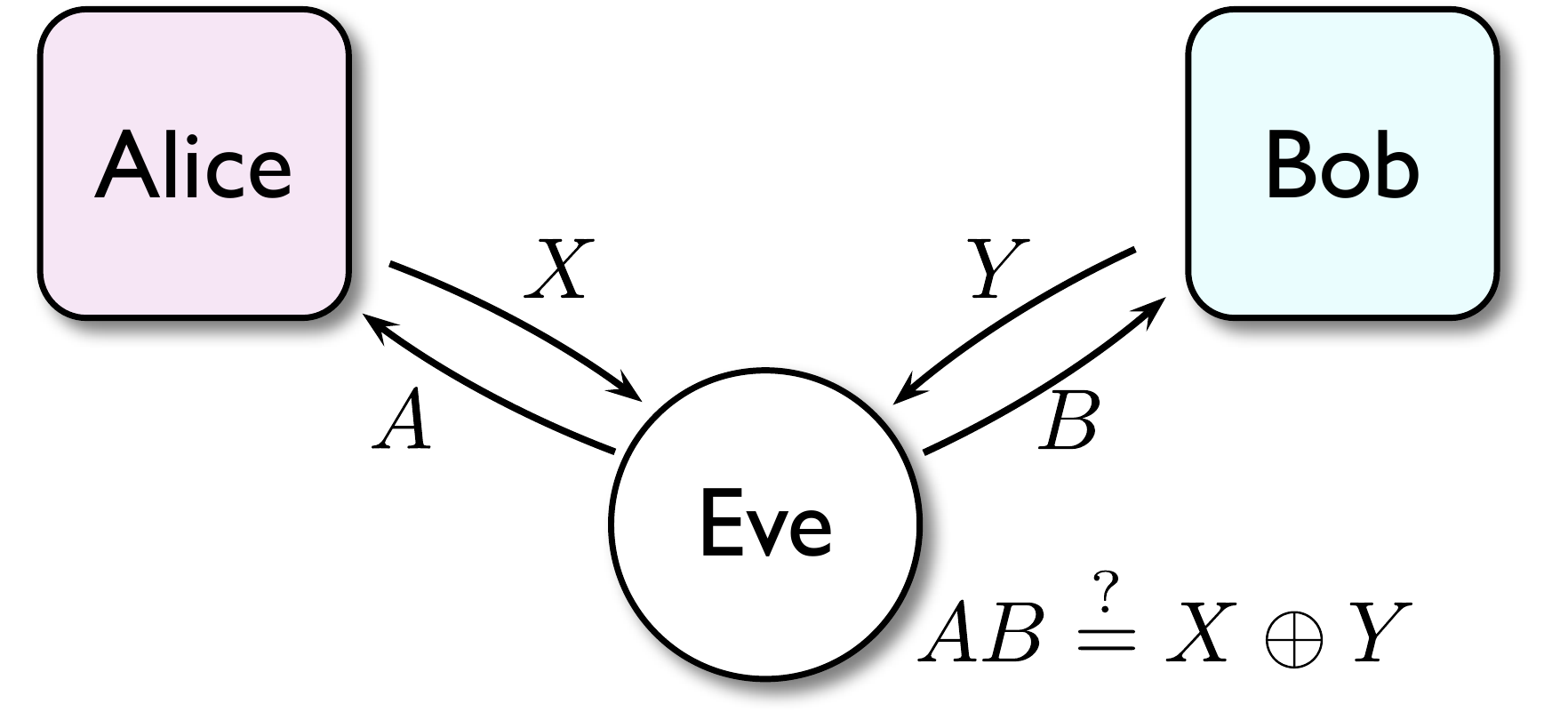}}
\caption{{\bf Test for quantumness.}  In a CHSH experiment, or ``game," the experimentalist Eve sends independent, uniformly random bits $A$ and $B$ to the devices Alice and Bob, respectively, who respond with bits $X$ and~$Y$.  The devices ``win" the game if $A B = X \oplus Y$.  By a Bell inequality, classical devices can win with probability at most $3/4$.  The probability of two classical devices winning $(3/4 + \epsilon) n$ out of $n$ games is therefore exponentially small.  Quantum devices can win the CHSH game with probability up to $\omega^* = \cos^2(\frac\pi8) \approx 85.4\%$, by Tsirelson's inequality~\cite{Tsirelson80inequality}, if they follow an ideal CHSH strategy: on a shared Einstein-Podolsky-Rosen (EPR) state $\ket \varphi = \tfrac{1}{\sqrt 2}(\ket{00} + \ket{11})$, Alice measures the Pauli operator $\sigma_z$ if $A = 0$ or $\sigma_x$ if $A = 1$, and Bob measures $\tfrac{1}{\sqrt 2}(\sigma_z + (-1)^B \sigma_x)$.} \label{f:chsh}
\end{figure}

The starting point for our protocol is the famous Bell experiment~\cite{Bell64epr}, and its subsequent distillation by Clauser, Horne, Shimony and Holt (CHSH)~\cite{ClauserHorneShimonyHolt69chshgame}.  Conceptually modeled as a game (\figref{f:chsh}), it provides a ``test for quantumness," a way for a classical experimentalist, whom we shall call Eve, to demonstrate the entanglement of two space-like separated devices, Alice and Bob.  Eve bases her decision, ``quantum" or ``not quantum," according to whether her interactions with the two devices satisfy non-local correlations, which are provably impossible to achieve in any local hidden variable theory.  Quantum devices can achieve such correlations, without any communication, by measuring two entangled qubits.  

Consider a protocol in which Eve plays a long sequence of CHSH games with Alice and Bob, and tests that they win close to the optimal fraction $\omega^*$ of the games.  This paper's main technical result establishes that if the devices pass Eve's test with high probability, then at the beginning of a randomly chosen long subsequence of games, Alice and Bob must share many EPR states in tensor product, that they measure one at a time using the single-game ideal CHSH operators of \figref{f:chsh}.  This is a step towards the general vision outlined above because it characterizes the initial state of many qubits, and allows Eve to command the devices to perform certain single-qubit operations.  
Of course, we cannot hope to characterize the devices' strategies exactly, but only for a suitable notion of approximation.  


In order to make a more precise statement, first consider a single CHSH game.  We show that if Alice and Bob win with probability $\omega^* - \epsilon$, then they must share a state that is $O(\sqrt \epsilon)$-close to an EPR state, possibly in tensor product with an additional state.  Moreover their joint measurement strategy is necessarily $O(\sqrt \epsilon)$-close to the ideal strategy.  (That is, applying Alice's measurement operator to the shared state gets within distance $O(\sqrt \epsilon)$ of her ideal measurement operator applied to the EPR state tensored with the ancilla; and similarly for Bob.)  Since each device can store its share of the EPR state in an arbitrary way, e.g., as a logical qubit spread over several physical qubits, these statements hold only up to local isometries.  This may be seen as a robust converse to Tsirelson's inequality, and as a rigidity property of the CHSH game: a nearly maximal Bell inequality violation rigidly locks into place the devices' shared state and measurement directions.  

A converse to Tsirelson's inequality for the CHSH game has been shown previously in the exact case~\cite{BraunsteinRevzen92tsirelsonconverse, PopescuRohrlich92tsirelsonconverse}.  Robustness is important for applications, however, because the success probability of a system can never be known exactly.  Robust, $\epsilon > 0$, converse statements have been shown based on a conjecture~\cite{BardynLiewMassarMcKagueScarani09deviceindependent} or under restrictive symmetry assumptions~\cite{AcinBrunnerGisinMassarPironioScarani07diqkdcollectiveattacks, PironioAcinBrunnerGisinMassarScarani09qkd}.\footnote{Similar $\epsilon = 0$ statements have been shown for other games~\cite{MayersYao03chsh, MayersYao98chsh, ColbeckKent10randomnessexpansion, Colbeck09thesis}, and Magniez et al.~\cite{MagniezMayersMoscaOllivier05selftest} have shown that the game in~\cite{MayersYao98chsh} is $O(\epsilon^{1/4})$-robust to error~$\epsilon > 0$.}  Recently, robustness has independently been shown for the CHSH game~\cite{McKagueYangScarani12chshrigidity, MillerShi12chshrigidity}.  

Scaling up to a sequence of $n$ CHSH games, suppose Alice and Bob use a strategy such that they win at least $(1 - \epsilon) \omega^* n$ of the games with high probability. By basic statistics, their strategy at the beginning of most games will win with probability at least $(1 - \epsilon^{\Omega(1)}) \omega^*$.  Rigidity for the one-shot game therefore applies.  However, their strategy for playing the $j$th game could depend on the previous games.  The states close to EPR states used in different games could overlap significantly, and their locations could depend on the history.  The multi-game rigidity theorem rules out such wayward behavior.  It says that for most random blocks of $m = n^{\Omega(1)}$ consecutive games, at the start of the block Alice and Bob must share a state that is close to a tensor product of $m$ EPR states, tensored with an additional state, and must play each $j$th game by making measurements that are close to the ideal CHSH strategy on the $j$th EPR state---different games being entirely independent.  

One way to view this theorem is that it scales up the CHSH test for quantumness and allows for identifying many qubits' worth of entanglement.  Much more than that, however, the multi-game rigidity theorem gives strong control over the devices' measurement operators for different games.  Combined with protocols for state and process tomography, and for computation by teleportation, this gives a method for realizing arbitrary dynamics in quantum systems without making assumptions about the internal structure or operations.  The dynamics are realized as the joint evolution of two isolated quantum systems, Alice and Bob, mediated by a classical experimentalist, Eve.  In order to realize the desired dynamics, Eve starts by testing the systems (devices) by playing with them many sequential CHSH games.  She gathers statistics and rejects if they lose too many games; by rigidity, this forces them to play nearly honestly.  At the beginning of a random block of~$m$ games, Eve stops playing with Alice but continues on with Bob.  Bob cannot tell that anything has changed, so continues playing the same way, measuring his halves of the EPR states.  Eve directs Alice to apply more complicated, multi-qubit operations, and she uses Bob's measurement results to tomographically certify Alice's compliance.  In a symmetrical manner, Eve can force Bob to follow her directions.  Finally, with a certain probability, Eve stops both Alice and Bob before the same block of~$m$ games, and she directs them both to apply multi-qubit operations on the next~$m$ EPR states.  The desired dynamics are implemented a step at a time, with the working qubits teleported back and forth between the two parties.  This zig-zagging evolution is natural because it allows complicated evolutions to be built out of simple, few-qubit operations; direct tomography on a many-qubit operation would be extremely inefficient.  Ultimately, should she wish, Eve can direct a full-scale quantum computation (\figref{f:computationprotocols}).  

\begin{figure*}
\centering
\subfigure[\label{f:circuitC} Circuit $\mathcal C$]{\raisebox{.5cm}{\includegraphics[scale=.7]{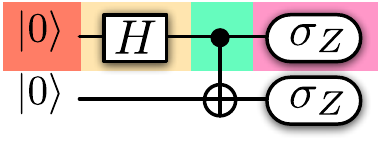}}} 
$\qquad\qquad\qquad$
\subfigure[\label{f:teleportH} Teleporting into $H$]{\raisebox{0cm}{\includegraphics[scale=.7]{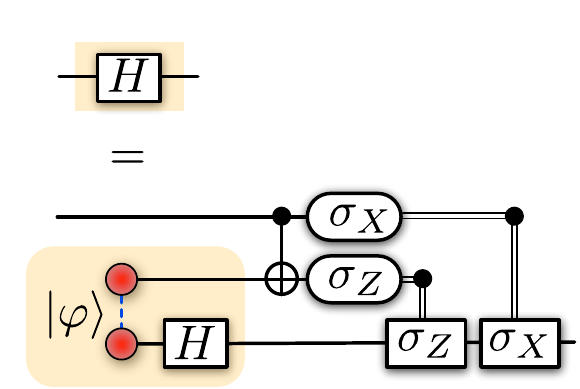}}}
\\[1cm]
\subfigure[\label{f:chshgames} CHSH games]{\includegraphics[scale=.5]{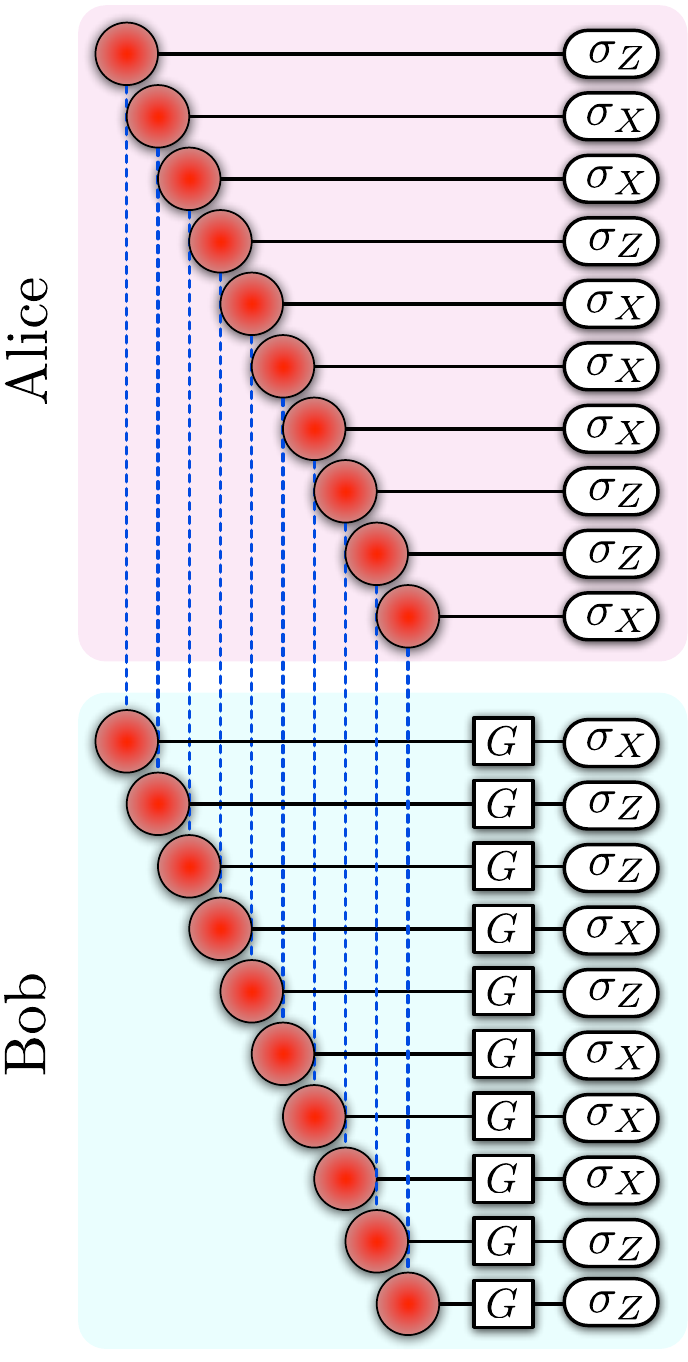}} 
$\quad\!$
\subfigure[\label{f:statetomography} State tomography]{\includegraphics[scale=.5]{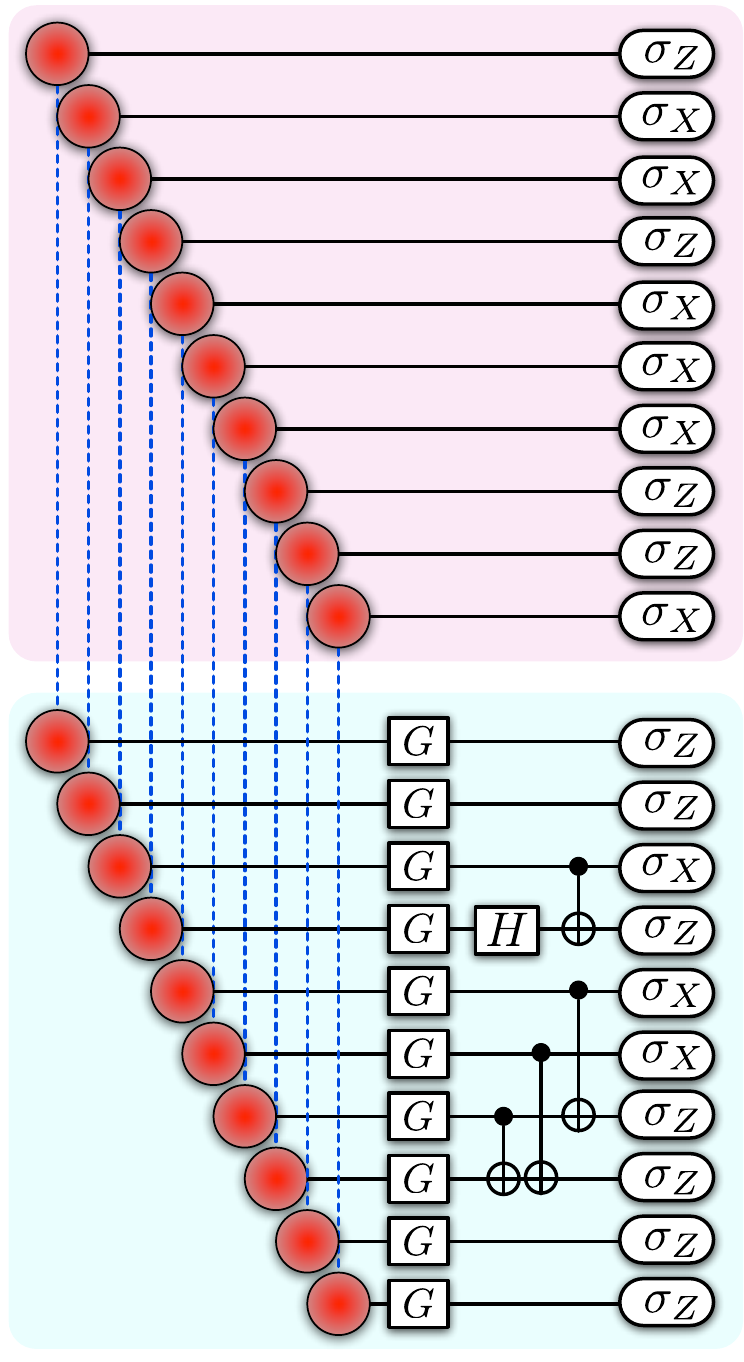}}
$\quad\!$
\subfigure[\label{f:processtomography} Process tomography]{\includegraphics[scale=.5]{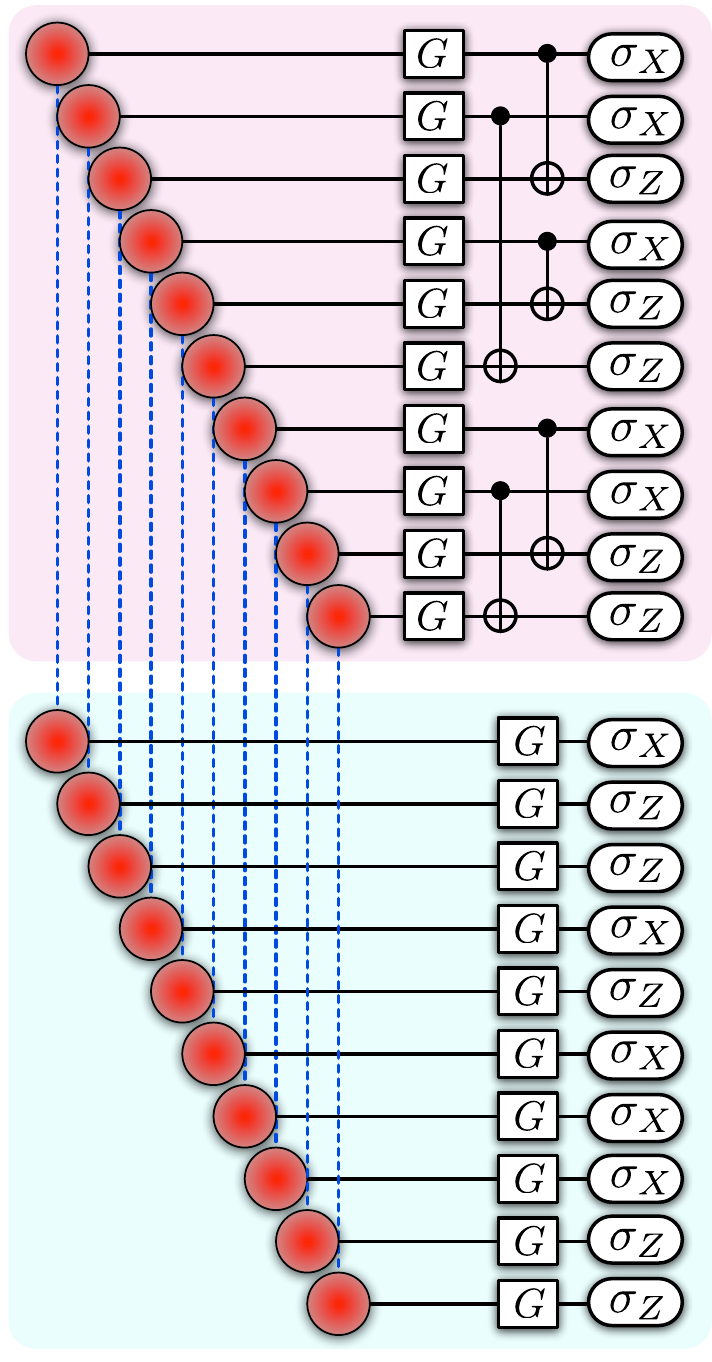}}
$\quad\!$
\subfigure[\label{f:computation} Computation]{\includegraphics[scale=.5]{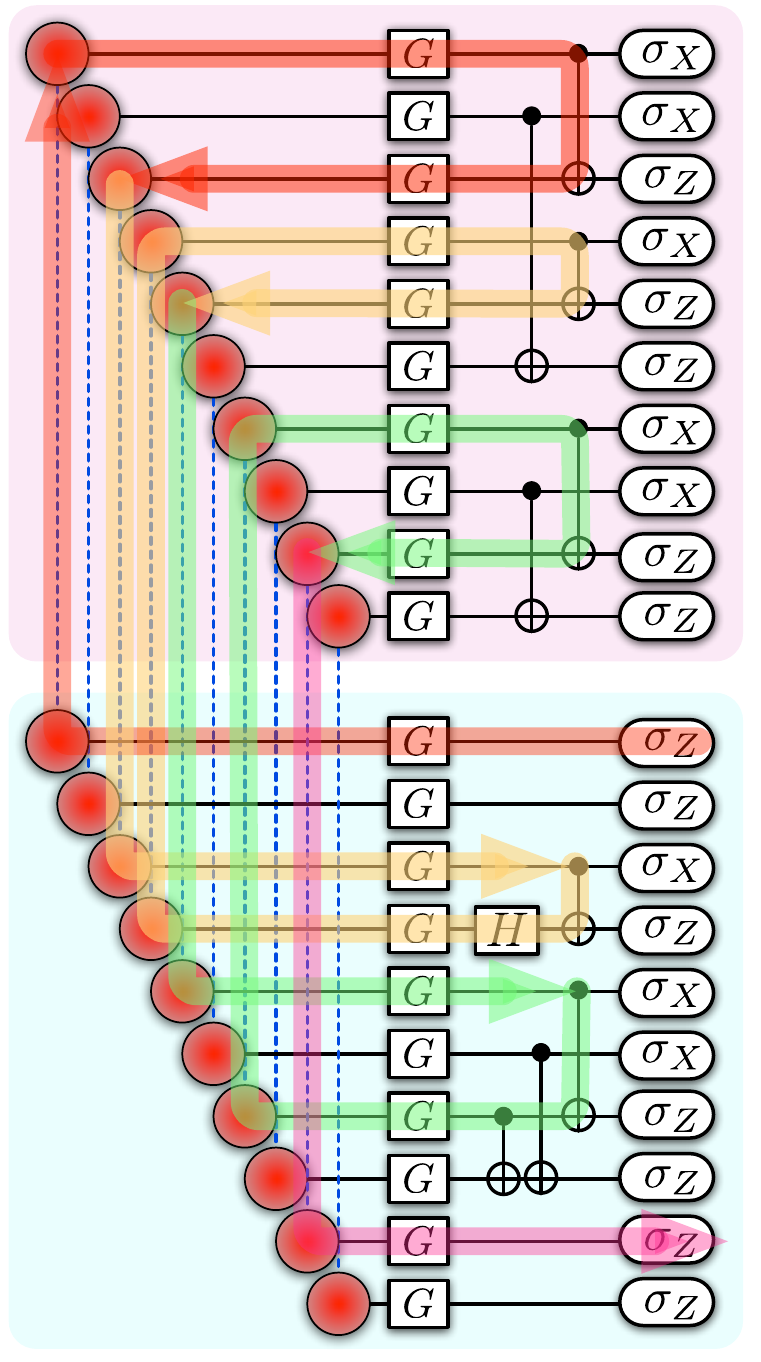}}
\caption{
{\bf Sub-protocols for verified quantum dynamics.}  
{\bf a}, Say that Eve wants to delegate to Alice and Bob a quantum circuit~$\mathcal C$, over the gate set $\{ H, G, \mathrm{CNOT} \}$, where $H$ is the Hadamard gate and $G$ a $\pi/4$ rotation about the $y$ axis.  {\bf b}, The idea is to use computation by teleportation~\cite{GottesmanChuang99teleportation}, which allows a gate, here $H$, to be implemented by a two-qubit Bell measurement on the input and half of a resource state, $(I \otimes H) \ket \varphi$.  Eve runs a random one of four sub-protocols with Alice and Bob.  {\bf c}, Playing many CHSH games ensures that the devices play honestly using shared EPR states.  {\bf d-e}, This lets Eve apply state or process tomography to characterize more complicated multi-qubit operations.  {\bf f}, By adaptively combining these operations, Eve directs the circuit~$\mathcal C$.  The zig-zagging logical path of the first qubit of~$\mathcal C$ is highlighted.  
} \label{f:computationprotocols}
\end{figure*}

The problem of controlling computationally powerful but untrusted resources lies at the foundation of computer science.  In the complexity class~$\NP$, for example, a polynomial-time routine---the ``verifier"---is allowed one round of interaction with an arbitrarily powerful, but malicious, ``prover."  We show that the same verifier can exploit the power of quantum-mechanical provers.  In particular: 

\begin{enumerate}
\item
A classical verifier can efficiently simulate a quantum computer by interacting with two untrusted, polynomial-time quantum provers.  This delegated computation scheme is also \emph{blind}, meaning that each prover learns no more than the length of the computation.  

\item
The verifier in any quantum multi-prover interactive proof (QMIP) system can be assumed to be classical.  Formally, the complexity classes $\QMIP$ and $\MIP^*$ are equal, where $\MIP^*$ is the class of languages decidable by a classical interactive protocol in which the provers share entanglement.  
\end{enumerate}

Previous work has considered a verifier who can store and control a constant number of qubits while interacting with a single prover~\cite{AharonovBenOrEban08authenticated, BroadbentFitzsimonsKashefi08authenticated, FitzsimonsKashefi12blind, BarzKashefiBroadbentFitzsimonsZeilingerWalther12blindexperiment}.  This makes controlling the system easier; for example, in the simplest scheme, the prover acts as an authenticated quantum memory and all computation is done by the verifier.  Our work is also inspired by a proposal~\cite{BroadbentFitzsimonsKashefi10qmip} that $\QMIP$ should equal $\MIP^*$.  The protocol introduced there can be attacked, however.  Our protocol has a very different form, based on the multi-game rigidity theorem.  

Thus a classical experimentalist can control quantum devices even under the weakest possible assumptions, in which the devices are not just imprecise or noisy, but are maliciously adversarial, and arbitrarily crafty.

\ifx\compilefullpaper\undefined  
\bibliographystyle{alpha-eprint}
\bibliography{q}

\end{document}
\fi

\ifx\compilefullpaper\undefined  
\documentclass[11pt]{article}

\begin{document}
\fi

\section{Proof sketches}

In this section, we sketch the main proofs, especially the characterization of strategies for sequential CHSH games.  The notation is presented intuitively, but of course precise definitions are given later.

\subsection{Rigidity of the CHSH game} \label{s:CHSHgamerigiditysketch}

The proof of the single-game rigidity theorem (\lemref{t:eprlemma}) is a good~place to start.  We show that nearly saturating Tsirelson's inequality nearly determines the devices' joint strategy.  To win the CHSH game with probability $\omega^* - \epsilon$, the devices' strategy must, up to local basis changes, be $O(\sqrt \epsilon)$-close to the ideal strategy of \figref{f:chsh}, involving measurements on two halves of an EPR state.  

A general strategy for Alice and Bob consists of some shared mixed state in $\H_A \otimes \H_B$, and two-outcome projective measurements for each of Eve's possible questions.  Truncate the devices' Hilbert spaces to finitely many dimensions, then decompose each space by Jordan's Lemma (\lemref{t:jordanslemma}) into the direct sum of two-dimensional spaces invariant under the projections.  The probability of winning is a convex combination of the success probabilities of the strategies that restrict the shared state to a two-dimensional space on each device's side, $\C^2 \otimes \C^2$.  Therefore it suffices to analyze the two-dimensional case, which we do by adjusting the angles between the projections to match the ideal strategy.  The resulting operators define the underlying qubits.

\subsection{Tensor-product structure for repeated CHSH games} \label{s:sequentialstructuredCHSHgameshavetensorproductstructuresketch}

A strategy $\S$ for playing $n$ sequential CHSH games specifies Alice and Bob's initial joint state as well as their measurement operators for every possible situation.  That is, for $X \in \{A, B\}$ and each $j = 1, \ldots, n$, $\S$ specifies the measurement operators used by device~$X$ in game $(j, \hX{j-1})$, where $\hX{j-1}$ is any transcript of the device's input and output bits for the first $j-1$ games.  For two strategies to be ``close" means that the distributions of game transcripts they induce should be close in total variation distance; and that for most transcripts (drawn from either distribution), the resulting quantum states should be close in a suitable norm.  We combine these conditions into one by defining for any strategy a block-diagonal density matrix that stores both the classical transcript and the resulting quantum~state: 
\begin{equation} \label{e:examplecombinedtranscriptstatedensitymatrix}
\rhoj{j} = \bigoplus_{h_{j-1}} \Pr[h_{j-1}] \, \rhoj{j}(\h{j-1})
 \enspace .
\end{equation}
Here $\h{j-1} = (\hA{j-1}, \hB{j-1})$ is the full transcript for the first $j-1$ games and $\rhoj{j}(\h{j-1})$ is the state at the beginning of game~$j$ conditioned on $\h{j-1}$.  Two strategies~$\S$ and~$\tilde \S$ are close if the~associated $\rhoj{j}$ and $\rhodecj{\tilde}{j}$ are close in trace distance, for every~$j$.  

Assume that for every $j$ and most $\h{j-1}$, the devices' conditional joint strategy at the beginning of game~$j$ is ``$\epsilon$-structured," meaning that it wins with probability at least $\omega^* - \epsilon$.  Our key theorem establishes that up to local basis changes, the devices' initial state must be close to $n$ EPR states, possibly in tensor product with an irrelevant extra state, and that their total strategy $\S$ must be close to an ideal strategy $\hat{\S}$ that plays game~$j$ using the $j$th EPR state.  Since the structure assumption can be established by martingale arguments on $\poly(n)$ sequential CHSH games, this implies the multi-game rigidity theorem.  See Theorems~\ref{t:sequentialCHSHgames} and~\ref{t:sequentialstructureandobservedcorrelationsapplied} for precise statements.

\subsubsection{Construction of the ideal strategy \texorpdfstring{$\hat S$}{S-hat}}

The main challenge is to ``locate" the ideal strategy $\hat \S$ within Alice and Bob's Hilbert space, i.e., to find an isometry on each of their spaces under which their states and measurement operators are close to ideal.  However, a priori, we do not know whether $\S$ calls for the devices to measure actual qubits in each step, or even if so whether the qubits form EPR states, qubits for different games overlap each other, or the locations of the qubits depend on the outcomes of previous games.  

The given strategy~$\S$ can be transformed into a nearby ideal strategy $\hat \S$ by a three-step sequence:

\smallskip

1. First, replace each device's measurement operators by the ideal operators promised by the single-game rigidity theorem.  In the resulting strategy $\tilde \S$, each device $X$ plays every game~$(j, \hX{j-1})$ using the ideal CHSH game operators on some qubit, up to a local change in basis.  However, the basis change can depend arbitrarily on $\hX{j-1}$, and the qubits for different~$j$ need not be in tensor product.  

\smallskip

2. In a ``multi-qubit ideal strategy" $\bar \S$, the qubits used in each game can still depend on the local transcripts but must at least lie in tensor product with the qubits from previous games. This imposes a tensor-product subsystem structure that previous DIQKD proofs have assumed.  The tensor-product structure is constructed beginning with a trivial transformation on~$\tilde \S$: to each device, add~$n$ ancilla qubits each in state~$\ket 0$.  Next, after a qubit has been measured, say as~$\ket{\alpha_j}$ in game~$j$, swap it with the $j$th ancilla qubit, then rotate this fresh qubit from $\ket 0$ to $\ket{\alpha_j}$ and continue playing games $j+1, \ldots, n$.  This defines a unitary change of basis that places the outcomes for games $1$ to~$j$ in the first $j$ ancilla qubits, and leaves the state in the original Hilbert space unchanged.  Since qubits are set aside after being measured, the qubits for later games are automatically in tensor product with those for earlier games; the resulting strategy $\bar \S$ is multi-qubit ideal. At the end of the $n$ games, swap back the ancilla qubits and undo their rotations, using the transcript.  

\smallskip

3. In the last step, we replace $\bar \S$ with an ideal strategy $\hat{\S}$, in which Alice and~Bob each play using a fixed set of $n$ qubits.  Fix a transcript $\hdec{\hat}{n}$, chosen at random from the distribution of transcripts for $\bar \S$.  For the first time, change the devices' initial state: replace $\rhoone$ with $\rhodecone{\hat}$, a state having $n$ EPR states in the locations determined by $\hdec{\hat}{n}$ in $\bar \S$.  In~$\hat \S$, the devices play using these EPR states, regardless of the actual transcript.  This $\hat \S$ is the desired ideal strategy.

\subsubsection{Ideal strategy \texorpdfstring{$\hat S$}{S-hat} is close to $\S$} \label{s:close}

It remains to show that the transformation's three steps incur a small error: $\hat \S$ is close to $\S$.  A major theme in the analysis is to leverage the known tensor-product structure between $\H_A$ and $\H_B$ to extract a tensor-product structure within $\H_A$ and~$\H_B$.  The steps are illustrated schematically in \figref{f:proofsketch}.  

\begin{figure*}
\centering
\def\qubitsballscale{.26}
\subfigure[General strategy]{\includegraphics[scale=\qubitsballscale]{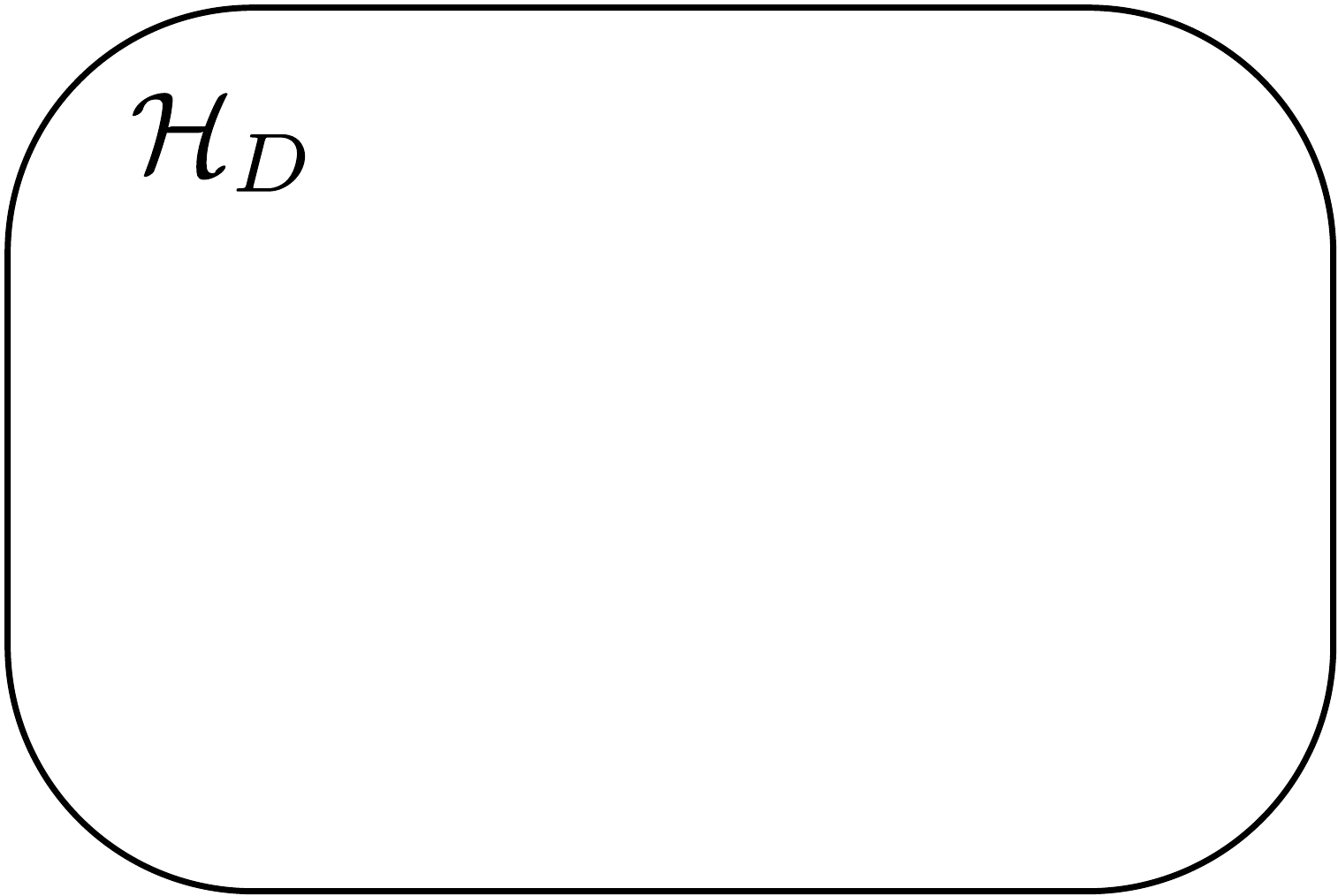}}
\subfigure[$\!\!$Single-qubit ideal strategy]{\includegraphics[scale=\qubitsballscale]{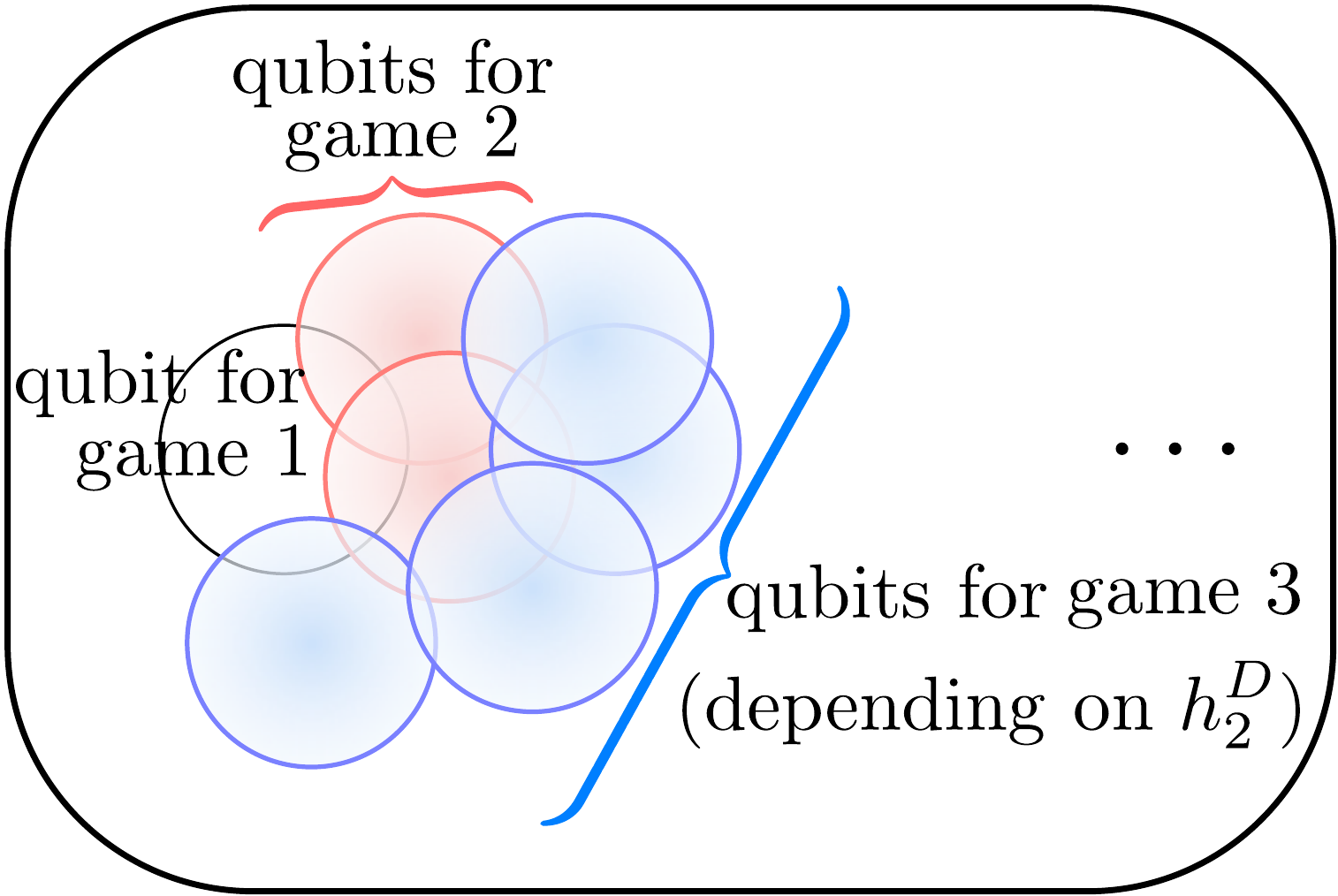}}
\subfigure[$\!\!$Multi-qubit ideal strategy]{\includegraphics[scale=\qubitsballscale]{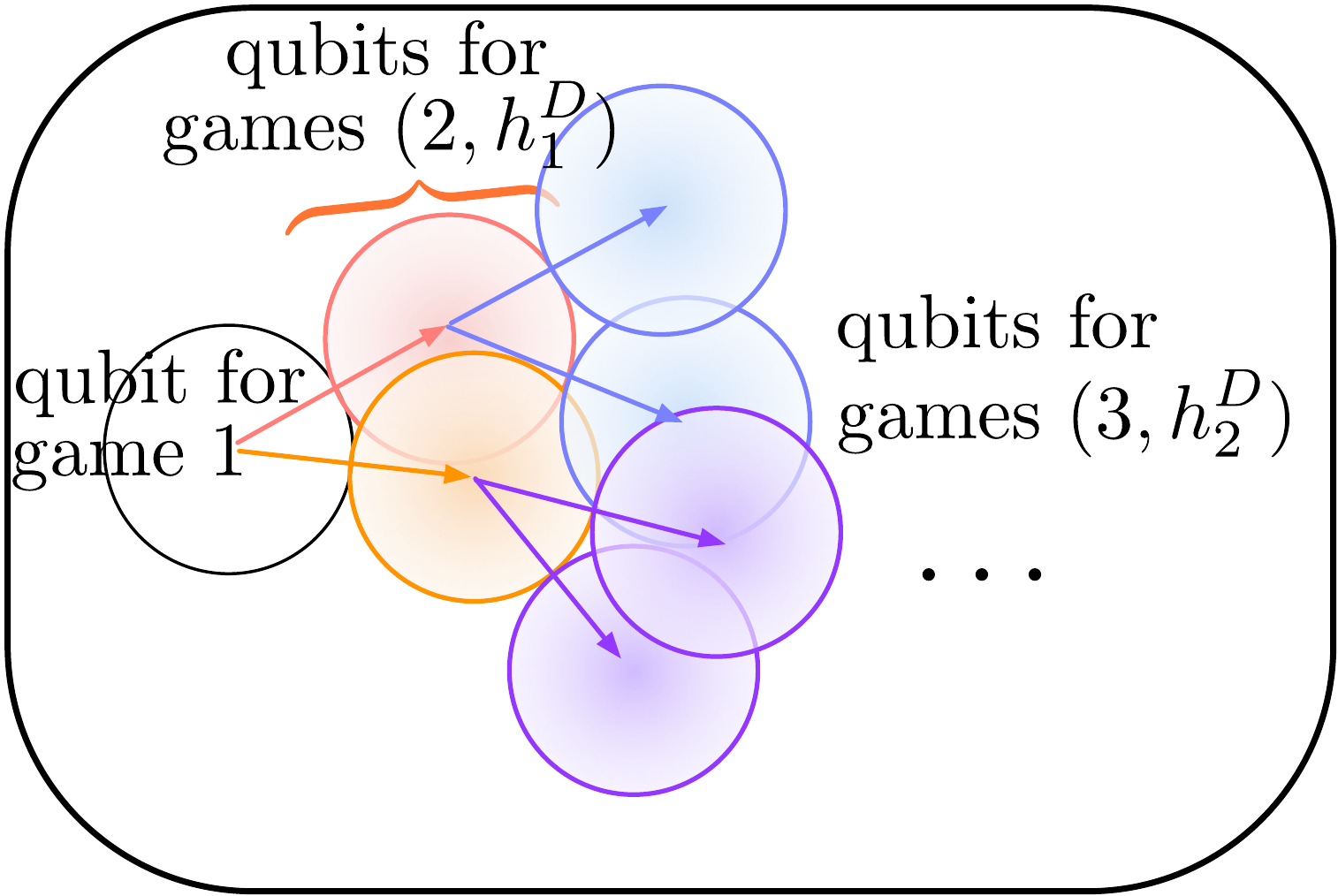}}
\subfigure[Ideal strategy]{\includegraphics[scale=\qubitsballscale]{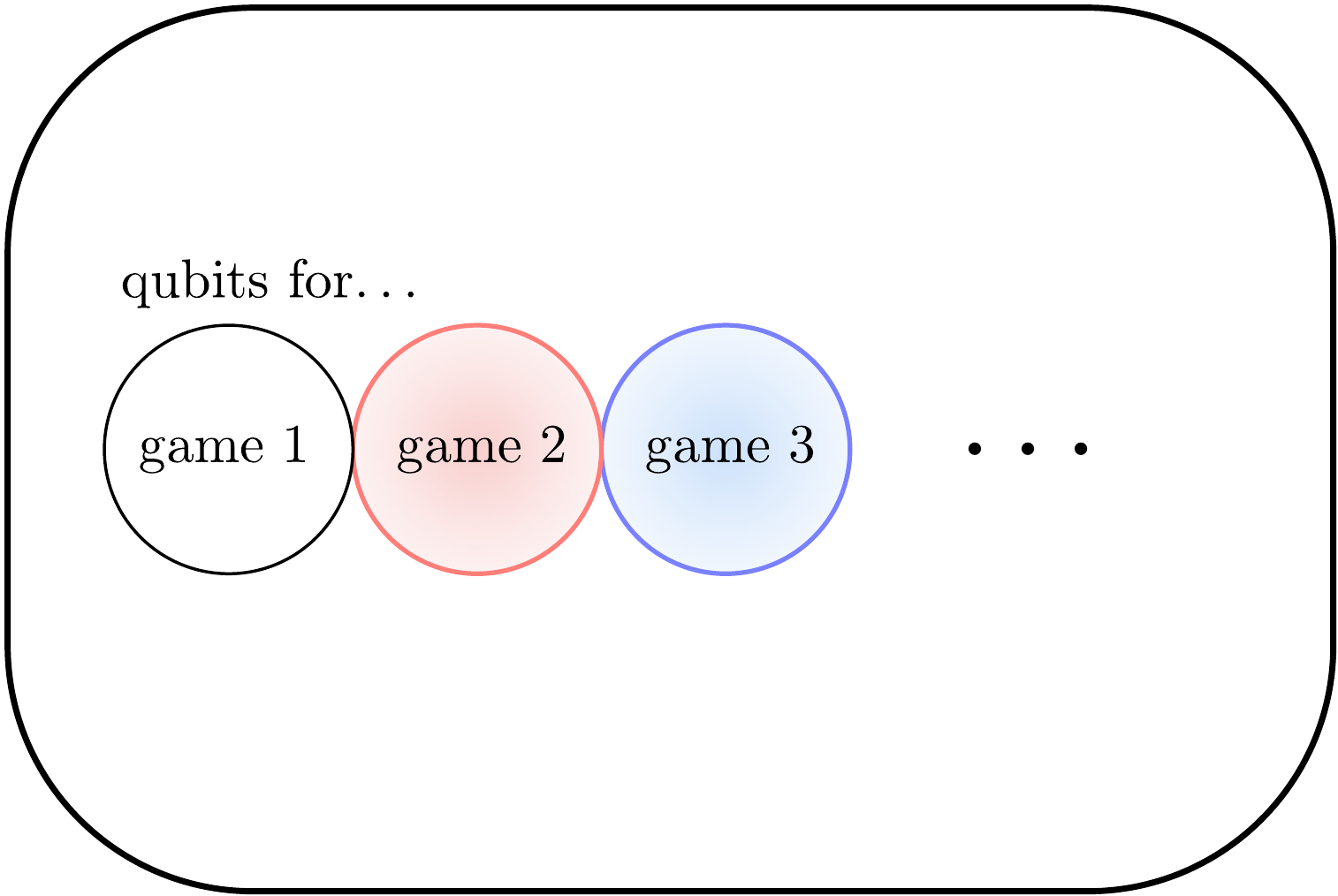}}
\caption{Proof outline for the multi-game rigidity theorem.  
{\bf a}, Initially, each device $D \in \{A, B\}$ can play arbitrarily, measuring in game~$j$ one of two reflections that depend on the local transcript~$\hX{j-1}$ for the previous games.  No structure is given for the Hilbert space $\H_D$.  
{\bf b},~We first show that $D$'s strategy is close to a ``single-qubit ideal strategy," in which for every game it measures some qubit using the ideal CHSH game strategy, but the qubit locations can be arbitrary.  Here, the qubits are illustrated as balls, and the overlaps indicate that they need not be in tensor product.  
{\bf c},~We then construct a nearby ``multi-qubit ideal strategy," in which the qubits used in each game must lie in tensor product with the qubits from previous games, but can overlap qubits used along other transcripts.  
{\bf d}, Finally, we argue that the qubit locations cannot depend significantly on the transcript, and therefore that the original strategy is well-approximated by an ideal strategy that measures a fixed set of $n$ qubits in sequence.  
(Note that these visualizations, representing qubits as balls, are inherently imprecise.  A qubit's location in a Hilbert space is given not by a ball, but by the two anti-commuting reflection operators $\sigma_x$ and $\sigma_z$.)  
} \label{f:proofsketch}
\end{figure*}

\smallskip

1. $\S \approx \tilde \S$: 
Although elementary, explaining this step is useful for establishing some notation.  Let $\rhoone$ be the devices' initial shared state, possibly entangled with the environment.  Let~$\EAj{j}$ and $\EBj{j}$ be the super-operators that implement Alice and Bob's respective strategies for game~$j$, $\EABj{j} = \EAj{j} \otimes \EBj{j}$ and $\EABj{j,k} = \EABj{k} \cdots \EABj{j}$ for $j \leq k$; thus the state $\rhoj{j}$ of Eq.~\eqnref{e:examplecombinedtranscriptstatedensitymatrix} equals $\EABj{1,j-1}(\rhoone)$.  For $D \in \{A, B\}$, let $\EXdecj{\tilde}{j}$ be the super-operator that replaces the actual measurement operators with the ideal operators promised by the CHSH rigidity theorem.  $\tilde \S$ is given by $\rhoone$, $\{ \EAdecj{\tilde}{j} \}$ and $\{ \EBdecj{\tilde}{j} \}$.  If $\Pr[\text{game~$j$ is $\epsilon$-structured}] \geq 1-\delta$, then $\trnorm{\EABj{j}(\rhoj{j}) - \EABdecj{\tilde}{j}(\rhoj{j})} \leq 2 \delta + O(\sqrt \epsilon)$.  (This expression uses Eq.~\eqnref{e:examplecombinedtranscriptstatedensitymatrix} to combine bounds on the probability of the bad event and the $O(\sqrt \epsilon)$ error from the good event.)  To show our goal, that $\EABj{1,n}(\rhoone) \approx \EABdecj{\tilde}{1,n}(\rhoone)$ in trace distance, use a hybrid argument that works backwards from game~$n$ to game~$1$ fixing each game's measurement operators one at a time.  The error introduced from fixing a game~$j$, by moving from $\EABj{j}(\rhoj{j})$ to $\EABdecj{\tilde}{j}(\rhoj{j})$, does not increase in later games because applying a super-operator cannot increase the trace distance.  Mathematically, this hybrid argument is simply a triangle inequality using the expansion 
\begin{equation*}
\EABj{1,n}(\rhoone) - \EABdecj{\tilde}{1,n}(\rhoone) = 
\sum_{j \in [n]} \EABdecj{\tilde}{j+1,n} \big(\EABj{j}(\rhoj{j}) - \EABdecj{\tilde}{j}(\rhoj{j})\big)
 \enspace .
\end{equation*}

\smallskip

2. $\tilde \S \approx \bar \S$: 
The key to showing that $\bar \S$ is close to $\tilde \S$ is the fact that operations on one half of an EPR state can equivalently be performed on the other half, since for any $2 \times 2$ matrix~$M$, $(M \otimes I)(\ket{00} + \ket{11}) = (I \otimes M^T)(\ket{00} + \ket{11})$.  This means that the outcome of an $\epsilon$-structured CHSH game would be nearly unchanged if Bob were hypothetically to perform Alice's measurement before his own.  By moving Alice's measurement operators for games $j+1$ to $n$ over to Bob's side, we see that they cannot significantly affect the qubit $\ket{\alpha_j}$ from game~$j$ on her side.  Therefore, undoing the original change of basis restores the ancilla qubits nearly to their initial state $\ket{0^n}$, and $\tilde \S \approx \bar \S$.  

\def\AmeasBBdecj #1#2{#1{\F}^{AB}_{#2}}

Formally, define a unitary super-operator ${\cal V}_j$ that rotates the $j$th ancilla qubit to $\ket{\alpha_j}$, depending on Alice's local transcript~$\hA{j}$.  Define a unitary super-operator ${\cal T}_j$ to apply ${\cal V}_j$ and swap the $j$th ancilla qubit with the qubit Alice uses in game~$j$ (depending on~$\hA{j-1}$).  Alice's multi-qubit ideal strategy is given by 
\begin{equation}
\EAdecj{\bar}{j} = {\cal T}_{1,j-1}^{-1} (\identity_{\C^{2^n}} \otimes \EAdecj{\tilde}{j}) {\cal T}_{1,j-1}
 \enspace .
\end{equation}
We aim to show that the strategy given by $\rhoone$, $\{ \EAdecj{\bar}{j} \}$ and $\{ \EBdecj{\tilde}{j} \}$ is close to~$\tilde \S$ up to the fixed isometry that prepends $\ketbra{0^n}{0^n}$ to the state.  Define a super-operator $\AmeasBBdecj{\tilde}{j}$, in which Alice's measurements are made on \emph{Bob's} Hilbert space~$\H_B$, on the qubit determined by Bob's local transcript $\hB{j-1}$.  Since most games are $\epsilon$-structured, by the CHSH rigidity theorem, $\AmeasBBdecj{\tilde}{j+1,k}(\rhodecj{\tilde}{j+1}) \approx \EABdecj{\tilde}{j+1,k}(\rhodecj{\tilde}{j+1}) = \rhodecj{\tilde}{k+1}$ for any $j \leq k$.  Since $\AmeasBBdecj{\tilde}{j+1,k}$ acts on $\H_B$, it does not affect Alice's qubit $\ket{\alpha_j}$ from game~$j$ at all, and so this qubit must stay near $\ket{\alpha_j}$ in $\rhodecj{\tilde}{k+1}$ as well, i.e., the trace of the reduced density matrix against the projection $\ketbra{\alpha_j}{\alpha_j}$ stays close to one.  As this holds for every~$j$, ${\cal T}_{1,n}^{-1}$ indeed returns the ancillas almost to their initial state~$\ket{0^n}$.  

In more detail, let~$X_j$ be the operator that projects onto Alice's $j$th ancilla qubit and the qubit she uses in the $j$th game being $\ket{0} \otimes \ket{\alpha_j}$.  By definition, $\Tr (X_j \, \rhodecj{\tilde}{j+1}) = 1$.  By the Gentle Measurement Lemma (\lemref{t:gentlemeasurement}), it suffices to show that $\Tr (X_j \, \rhodecj{\tilde}{k+1}) = \Tr X_j \EABdecj{\tilde}{j+1,k}(\rhodecj{\tilde}{j+1}) \approx 1$.  This is not obvious; since the operators for games~$j+1$ to $k$ do not act in tensor product, they can disturb the qubit measured in game~$j$.  However, since a super-operator on~$\H_B$ cannot affect the expectation of an operator supported on~$\H_A$, we find 
\begin{equation*}
\Tr (X_j \, \rhodecj{\tilde}{k+1}) = \Tr X_j \EABdecj{\tilde}{j+1,k}(\rhodecj{\tilde}{j+1}) \approx \Tr X_j \AmeasBBdecj{\tilde}{j+1,k}(\rhodecj{\tilde}{j+1}) = \Tr(X_j  \, \rhodecj{\tilde}{j+1}) = 1
 \enspace .
\end{equation*}

The $\{ \EBdecj{\tilde}{j} \}$ are symmetrically adjusted to $\{ \EBdecj{\bar}{j} \}$.  

\smallskip

3. $\bar \S \approx \hat \S$: 
Intuitively, if the location of Alice's $j$th qubit depended on $\hA{j-1}$, then without any communication Bob could not know which of his qubits to measure.  However, Alice and Bob's transcripts are significantly correlated, and we must show that they cannot use these correlations to coordinate the locations of their qubits.  

We argue that $\hat \S$ closely approximates $\bar \S$, provided that~$\hdec{\hat}{n}$ satisfies: for every~$j$, conditioned on the partial transcript~$\hdec{\hat}{j-1}$, 
(a) game~$j$ is $\epsilon$-structured, and 
(b)~there is a high probability that every subsequent game is $\epsilon$-structured.  By Markov inequalities, most transcripts satisfy these conditions.  

\def\rhodecjh #1#2#3{#1{\rho}_{#2}({#3})}	
\def\AmeasBBdecjh #1#2#3{#1{\F}^{AB \vert \smash{#3}}_{#2}}	
\def\EABdecjh #1#2#3{#1{\E}^{AB \vert \smash{#3}}_{#2}}

We connect $\bar \S$ to $\hat \S$ by an argument that one game at a time switches play to locate qubits according to $\hdec{\hat}{n}$.  The intermediate steps relate strategies in which the devices locate their qubits using a hybrid $(\hdec{\hat}{j}, \h{j+1,n})$ of $\hdec{\hat}{n}$ and the actual transcript $\h{n}$.  

Consider a partial transcript $\h{j}$ that differs from $\hdec{\hat}{j}$ only in the $j$th game, say on Alice's side.  By~(a) and the CHSH rigidity theorem, Alice's $j$th qubit is collapsed and nearly in tensor product with the rest of the state.  Therefore, there exists a unitary~$\AAunitaryAj{j}$ acting on this qubit such that 
\begin{equation} \label{e:proofidealocalgluingassumption}
\rhodecjh{\bar}{j+1}{\h{j}} \approx \AAunitaryAj{j} \rhodecjh{\bar}{j+1}{\hdec{\hat}{j}} \AAunitaryAj{j}{}^\dagger
 \enspace ,
\end{equation}
up to error $O(\sqrt \epsilon)$.  Since applying a super-operator cannot increase trace distance and on Bob's side $\hB{j} = \hBdec{\hat}{j}$, therefore 
\begin{equation*}
\AmeasBBdecjh{\bar}{j+1,n}{\hB{j}}\big(\rhodecjh{\bar}{j+1}{\h{j}}\big) \approx \AAunitaryAj{j} \AmeasBBdecjh{\bar}{j+1,n}{\hBdec{\hat}{j}}\big(\rhodecjh{\bar}{j+1}{\hdec{\hat}{j}}\big) \AAunitaryAj{j}{}^\dagger
 \enspace .
\end{equation*}
Here, $\AmeasBBdecjh{\bar}{j+1,n}{\hB{j}}$ is the same super-operator used in the multi-qubit ideal strategy simulation step---that plays Alice's games on Bob's qubits---except conditioned on the local transcript~$\hB{j}$.  By condition (b), these super-operators can be pulled back to Alice's side, to give 
\begin{equation*}
\EABdecjh{\bar}{j+1,n}{\h{j}}\big(\rhodecjh{\bar}{j+1}{\h{j}}\big) \approx \AAunitaryAj{j} \EABdecjh{\bar}{j+1,n}{\hdec{\hat}{j}}\big(\rhodecjh{\bar}{j+1}{\hdec{\hat}{j}}\big) \AAunitaryAj{j}{}^\dagger
 \enspace .
\end{equation*}
Note that this approximation does not follow immediately from Eq.~\eqnref{e:proofidealocalgluingassumption}, because Alice's super-operators conditioned on $\hA{j}$ can be very different from her super-operators conditioned on~$\hAdec{\hat}{j}$.  

By fixing the coordinates one at a time in this way, we find that for a typical transcript~$\h{n}$, $\rhodecjh{\bar}{n+1}{\h{n}} \approx V^{AB}_{1,n} \rhodecjh{\bar}{n+1}{\hdec{\hat}{n}} V^{AB}_{1,n}{}^\dagger$, and we conclude that $\EABdecj{\bar}{1,n}(\rhoone) \approx \EABdecj{\hat}{1,n}(\rhoone)$.  

Since $\EABdecj{\hat}{1,n}$ measures qubits in tensor product with each other, by using the CHSH rigidity theorem one last time, it is not difficult to show that $\EABdecj{\hat}{1,n}(\rhoone) \approx \EABdecj{\hat}{1,n}(\rhodecone{\hat})$, where $\rhodecone{\hat}$ has~$n$ EPR states in the qubit positions determined by~$\hdec{\hat}{n}$.  Thus the devices' actual strategy $\S = (\rhoone, \{ \EAj{j} \}, \{ \EBj{j} \})$ is close to the ideal strategy $\hat \S = (\rhodecone{\hat}, \{ \EAdecj{\hat}{j} \}, \{ \EBdecj{\hat}{j} \})$, as desired.  

\smallskip

The conclusion that the devices' joint strategy is close to ideal is not strong enough for our applications, in which sometimes Eve plays CHSH games with only one of the two devices.  We need to show that the devices' strategies are \emph{separately} close to ideal, i.e., 
\begin{equation}
\EAj{1,n}(\rhoone) \approx \EAdecj{\hat}{1,n}(\rhodecone{\hat})
\quad\qquad \text{and} \quad\qquad
\EBj{1,n}(\rhoone) \approx \EBdecj{\hat}{1,n}(\rhodecone{\hat})
 \enspace .
\end{equation}
These estimates cannot be obtained directly because our main assumption, that every game~$j$ is usually $\epsilon$-structured, is only of use if both devices have played games~$1$ through~$j-1$---it gives information about $\EXj{j}$ applied to $\EABj{1,j-1}(\rhoone)$, not about~$\EXj{j}$ applied to $\EXj{1,j-1}(\rhoone)$.  The key idea to obtain separate estimates is that applying both devices' super-operators is almost equivalent to applying Alice's super-operator, \emph{guessing} Bob's measurement outcome from the ideal conditional distribution, and based on the guess applying a controlled unitary correction to his qubit.  Since Alice's super-operator collapses both qubits of the EPR state, it is not actually necessary to measure Bob's qubit.  Defining $\Bguessj{j}$ to be this guess-and-correct super-operator, two hybrid arguments give $\EABj{1,n}(\rhoone) \approx \Bguessj{1,n} \EAj{1,n}(\rhoone)$ and $\EAdecj{\tilde}{1,n} \EBj{1,n}(\rhoone) \approx \Bguessj{1,n} \EAdecj{\tilde}{1,n}(\rhoone)$.  Thus, 
\begin{equation*}
\Bguessj{1,n} \EAj{1,n}(\rhoone) \approx \Bguessj{1,n} \EAdecj{\tilde}{1,n}(\rhoone)
 \enspace .
\end{equation*}
The same super-operator $\Bguessj{1,n}$ appears on both the left- and right-hand sides above.  In general, applying a super-operator can reduce the trace distance.  In this case, however, it does not; the correction part of $\Bguessj{1,n}$ is unitary, and the guessing part is a stochastic map acting on a {copy} of Alice's classical transcript register.  Therefore, indeed $\EAj{1,n}(\rhoone) \approx \EAdecj{\tilde}{1,n}(\rhoone)$.  The third step of the proof uses a similar, but more involved, argument.

\subsection{Verified quantum dynamics} \label{s:tomographysketch}

Our scheme for verified quantum dynamics is based on the idea of computation by teleportation~\cite{GottesmanChuang99teleportation}.  Say that Eve wants to simulate a quantum circuit~$\mathcal C$, over the gate set $\{ H, G, \mathrm{CNOT} \}$, where $H$ is the Hadamard gate and $G = \exp(- i \frac\pi8 \sigma_y)$ is a $\pi/4$ rotation about the $y$ axis of the Bloch sphere.  Eve asks Bob to prepare many copies of the resource state $\ket 0 \otimes (I \otimes H) \ket{\varphi} \otimes (I \otimes \phasegate) \ket{\varphi} \otimes \mathrm{CNOT}_{2,4} (\ket{\varphi} \otimes \ket{\varphi})$.  He can do so by applying one-, two- and four-qubit measurements to his halves of the shared EPR states and reporting the results to Eve.  If he plays honestly, Alice's shares of the EPR states collapse into the desired resource states, up to simple corrections.  Each resource state corresponds to a basic operation in~$\mathcal C$.  Eve wires these up by repeatedly directing Alice to make a Bell measurement connecting the output of one operation to the input of the next operation in~$\mathcal C$.  After each $G$ gate, an $H$ correction might be required.  

Of course, Alice and Bob might not follow directions.  To enforce honest play, Eve runs this protocol only a small fraction of the time, and otherwise chooses uniformly between three alternative protocols sketched in \figref{f:computationprotocols}.  Let $m = {\abs C}^{O(1)}$ and $n = m^{O(1)}$.  
\begin{enumerate}
\item
In the ``state tomography" protocol, Eve chooses $K$ uniformly from $\{1, \ldots, n/m\}$.  She referees $(K-1) m$ CHSH games with both devices.  Then in the $K$th block of~$m$, Eve asks Bob to prepare the resource states, in a random order, while continuing to play CHSH games with Alice.  Eve rejects if the tomography statistics are inconsistent.  We prove that if Alice plays honestly and Eve accepts with high probability, then on most randomly chosen small subsets of the resource state positions, Alice's reduced state is close to the correct tensor product of resource states.  
\item
In the ``process tomography" protocol, Eve again chooses $K$ uniformly from $\{1, \ldots, n/m\}$ and referees $(K-1) m$ CHSH games.  In the $K$th block of~$m$, Eve asks Alice to make Bell measurements on random pairs of qubits, while continuing to play CHSH games with Bob.  If Alice's reported result for \emph{any} pair of qubits is inconsistent with Bob's outcomes, Eve rejects.  Then if Bob plays honestly and Eve accepts with high probability, Alice must also have applied the Bell measurements honestly.  
\item
In the third protocol, Eve simply referees $n$ sequential CHSH games with both devices and rejects if they do not win at least $(1 - \epsilon) \omega^* n$ games.   
\end{enumerate}

From Bob's perspective the process tomography and computation protocols are indistinguishable, as are the state tomography and CHSH game protocols.  From Alice's perspective, the state tomography and computation protocols are indistinguishable, as are the process tomography and CHSH game protocols.  The devices must behave identically in indistinguishable protocols.  The multi-game rigidity theorem therefore provides the base for a chain of implications that implies that if Eve accepts with high probability, then the devices must implement $\mathcal C$ honestly.  

Four main technical problems obstruct these claims.  

First, in the state tomography protocol, if Bob is dishonest, then Alice gets an arbitrary $m$-qubit state, and there is no reason why it should split into a tensor product of constant-qubit states.  Standard state tomography and certification arguments require many copies of a state and so do not apply.  Nonetheless, we argue using martingales that if the counts of Alice's different measurement outcomes roughly match  their expectations with high probability, then for most reported measurement outcomes from Bob and for most subsystems~$j$, Alice's conditional state reduced to her $j$th subsystem is close to what it should be.  

Furthermore, saturating Tsirelson's inequality for the CHSH game only implies that Alice is honestly making Pauli $\sigma_x$ and~$\sigma_z$ measurements on her half of an EPR state.  Tomography also requires~$\sigma_y$ measurements.  To sidestep this issue, we generalize a theory introduced by McKague~\cite{McKague10thesis} and prove that there is a large class of states, including the necessary resource states, that are all robustly determined by only $\sigma_x$ and~$\sigma_z$ measurements.  

A bigger problem, though, is that we want to characterize the operations that the devices apply to their shared EPR states, and not just the states that these operations create on the other side.  The distinction is the same as that between process and state tomography.  Essentially, the problem is that the correct states could be generated by incorrect processes.  Moreover, as for sequential CHSH games, Bob's strategy in early tomography rounds might be sufficiently dishonest as to allow him in later rounds to apply completely dishonest operators.  For example, Bob could cheat in the first requested round by cyclically shifting all of his EPR state halves.  A statistical test will not suffice to detect one round of cheating.  However, if after this first round he plays using the shifted ideal operators, his operations will all be completely dishonest even though they have the correct effect on Alice's side.  

A key observation to avoid this problem is that it is enough to certify the states prepared by one device and the processes applied by the other.  Then since a broad class of states can be certified, for applications it suffices to certify a much smaller set of operations.  We restrict consideration to Pauli stabilizer measurements~\cite{Gottesman97thesis}.  For Pauli operators in the stabilizer of a state, the measurement outcome is deterministic.  Therefore if Alice reports the wrong stabilizer syndrome in even a single round, Eve can reject.  Our process certification analysis is similar to some of the arguments used above.  We argue that Alice's earlier measurements cannot usually overly disturb the qubits intended for use in later measurements, by pulling Alice's measurement super-operators over onto Bob's halves of the EPR states.  

Finally, the verifier's questions in the state and process tomography protocols are non-adaptive, whereas in computation by teleportation the questions must be chosen adaptively based on previous responses.  This is an attack vector in some related protocols.  However, we argue that the devices can learn nothing from the adaptive questions.  

More formally, let $\rho$ be the initial state, and let~$\B$ be the super-operator describing Eve's interactions with Bob in state tomography.  Roughly, state tomography implies that the states Bob prepares on Alice's side are correct up to a small error in trace distance, or 
\begin{equation}
\Tr_B \B(\rho) \approx \Tr_B \hat \B(\hat \rho)
 \enspace ,
\end{equation}
where $\hat \B$ is the ideal super-operator and $\hat \rho$ is an ideal initial state consisting of perfect EPR states.  Similarly, let $\A$ be the super-operator describing Eve's interactions with Alice in a process tomography protocol on Alice's operations; we have 
\begin{equation}
\A(\rho) \approx \hat \A(\rho)
 \enspace .
\end{equation}
Computation by teleportation can be implemented either by choosing Bob's state preparation questions non-adaptively and Alice's process questions adaptively, or vice versa.  We show that these are exactly equivalent regardless of the devices' strategies, i.e., 
\begin{equation}
\Aad \B = \Bad \A
 \enspace ,
\end{equation}
where $\Aad$ and $\Bad$ are the same as~$\A$ and~$\B$, respectively, except with Eve choosing her questions adaptively based on the previous messages.  Combining these steps, we therefore obtain 
\begin{equation*}\begin{split}
\Tr_B \Bad \A (\rho) 
&\approx \Tr_B \Bad \hat \A (\rho) \\
&= \Aadhat \Tr_B \B (\rho) \\
&\approx \Aadhat \Tr_B \Badhat (\hat \rho)
 \enspace ,
\end{split}\end{equation*}
and thus the actual computation by teleportation protocol leaves on Alice's side nearly the ideal output.  

The proof that $\QMIP = \MIP^*$ follows along similar lines.  Begin with a $k$-prover protocol.  We may assume that it has two rounds of quantum messages from the provers, before and after the verifier broadcasts a random bit~\cite{KempeKobayashiMatsumotoVidick07qmip}.  To convert to an MIP$^*$ protocol, with classical messages, add two additional provers, Alice and Bob.  Eve teleports the original $k$ provers' messages to Alice, and directs Alice and Bob together to apply the quantum verifier's acceptance predicate.

\ifx\compilefullpaper\undefined  
\bibliographystyle{alpha-eprint}
\bibliography{q}

\end{document}
\fi	

\ifx\compilefullpaper\undefined  
\documentclass[11pt]{article}

\begin{document}
\fi

\section{Background and notation}

For a natural number $n$, let $[n] = \{1, 2, \ldots, n\}$.  Let~$S_n$ be the symmetric group of degree~$n$.  Let~$\delta_{a,b}$ be the Kronecker delta function.  The {Pauli operators} are tensor products of the matrices $I = \smatrx{1&0\\0&1}$, $X = \smatrx{0&1\\1&0}$, $Y = \smatrx{0&-i\\i&0}$ and $Z = \smatrx{1&0\\0&-1}$.  The latter three matrices were earlier termed $\sigma_x, \sigma_y, \sigma_z$.  Let $H = \frac{1}{\sqrt 2} \smatrx{1&1\\1&-1}$, the Hadamard gate, and $G = \exp(-i \frac{\pi}{8} Y) = \Big(\begin{smallmatrix}\cos\frac\pi8 & -\sin\frac\pi8 \\ \sin\frac\pi8 & \cos\frac\pi8 \end{smallmatrix}\Big)$.  

The complex and real numbers are denoted by~$\C$ and~$\R$, respectively.  For a finite set~$S$, let~$\C^S$ be the complex Hilbert space $\C^{\abs S}$ with orthonormal basis $\{ \ket x : x \in S \}$.  We assume familiarity with ket notation, e.g., $\sum_{x \in S} \ketbra x x = \identity$, the identity on~$\C^S$.  For vector spaces~$V$ and $W$ over~$\C$, let $\L(V, W)$ denote the set of all linear transformations from $V$ into~$W$, and let~$\L(V) = \L(V, V)$.  For an operator~$A$, denote by $\norm A$ its spectral norm, and by $\trnorm{A}$ its trace norm, i.e., the sum of its singular values.  

We assume familiarity with the basics of quantum computation as found, e.g., in~\cite{NielsenChuang00}.  In particular, for a Hilbert space~$\H$, a (mixed) state is a positive semi-definite operator $\rho \in \L(\H)$ with trace one, and a pure state is a rank-one state.  The evolution of a quantum system is described by a super-operator~$\E$, a map from states on $\H$ to states on~$\H'$, which can in general be specified by a set $\{E_k\} \subset \L(\H, \H')$ of ``Kraus operators" satisfying $\sum_k E_k^\dagger E_k = \identity_\H$: $\E(\rho) = \sum_k E_k \rho E_k^\dagger$.  Applying a super-operator cannot increase the trace distance between two states: 

\begin{fact} \label{t:superoperatorreducestracedistance}
For a super-operator $\E$ and density matrices $\rho$ and $\sigma$, $\trnorm{\E(\rho)-\E(\sigma)} \leq \trnorm{\rho - \sigma}$.  
\end{fact}
\begin{proof}
Let $\delta = \rho - \sigma$ and let $\delta_\pm = \frac12(\abs \delta \pm \delta)$.  Then $\delta_\pm \succeq 0$, $\delta = \delta_+ - \delta_-$ and $\abs \delta = \delta_+ + \delta_-$, implying $\trnorm{\E(\rho) - \E(\sigma)} \leq \trnorm{\sum_k E_k \delta_+ E_k^\dagger} + \trnorm{\sum_k E_k \delta_- E_k^\dagger} = \Tr \sum_k E_k \abs \delta E_k^\dagger = \Tr \abs \delta = \trnorm{\delta}$.  
\end{proof}

\noindent 
An isometric super-operator not change the trace distance: $\trnorm{E A E^\dagger} = \trnorm{A}$ for an isometry~$E$.  

A measurement with finitely many outcomes can be defined as a super-operator~$\E$ in which the Kraus operators have the form $E_k = \ket k \otimes F_k \in \L(\H, \C^{[d]} \otimes \H')$, for $k \in [d]$.  Then $\E(\rho) = \sum_k \ketbra k k \otimes F_k \rho F_k^\dagger$ is a block-diagonal matrix, known as a classical-quantum state or cq-state, in which the first register labels the classical measurement outcome~$k$, and the block $F_k \rho F_k^\dagger$ is the resulting quantum state times its probability.  

The Holevo-Helstrom theorem~\cite{NielsenChuang00} states that for any states~$\rho$ and~$\sigma$, the maximum over all possible measurements~$\E$ of the total variation distance between the distributions of outcomes for~$\E(\rho)$ and~$\E(\sigma)$ is $\frac12 \trnorm{\rho - \sigma}$.  This can be most compactly phrased as 
\begin{equation} \label{e:holevohelstromheart}
\sup_{0 \preceq \Pi \preceq \identity} \Tr (\Pi A) = \frac12 \trnorm{A}
\end{equation}
for any Hermitian operator~$A$ with $\Tr A = 0$.  Since the trace distance between two states that are block-diagonal in the same basis is the sum of the trace distances between the corresponding blocks, one can also bound the expected trace distance between the resulting states $F_k \rho F_k^\dagger / \Tr (F_k^\dagger F_k \rho)$ and $F_k \sigma F_k^\dagger / \Tr (F_k^\dagger F_k \sigma)$: 

\begin{lemma} \label{t:blockdiagonaltracedistance}
Let $\rho^{(i)} = \sum_k \ketbra k k \otimes \rho^{(i)}_k$, for $i = 1, 2$.  Let $\epsilon = \bigtrnorm{\rho^{(1)} - \rho^{(2)}} = \sum_k \bigtrnorm{\rho^{(1)}_k - \rho^{(2)}_k}$.  Let~$K^{(i)}$ be a random variable distributed according to $\Pr[K^{(i)} = k] = \Tr \rho^{(i)}_k$.  Then the total variation distance between the distributions of $K^{(1)}$ and $K^{(2)}$ satisfies 
\begin{equation}
\frac12 \sum_k \bigabs{\Tr (\rho^{(1)}_k - \rho^{(2)}_k)} \leq \epsilon/2
 \enspace .
\end{equation}
Furthermore, letting $\bar \rho^{(i)}_k = \rho^{(i)}_k / \Tr \rho^{(i)}_k$, if $\rho^{(i)}_k \neq 0$, and $0$ otherwise, the expected trace distance between $\bar \rho^{(1)}_{K^{(1)}}$ and $\bar \rho^{(2)}_{K^{(1)}}$ satisfies 
\begin{equation}
\Ex\big[ \bigtrnorm{ \bar \rho^{(1)}_{K^{(1)}} - \bar \rho^{(2)}_{K^{(1)}} } \big] \leq 2 \epsilon
 \enspace .
\end{equation}
\end{lemma}

\begin{proof}
The bound on the total variation distance is a special case of the Holevo-Helstrom theorem, and follows directly from the inequality $\bigabs{\Tr(\rho^{(1)}_k - \rho^{(2)}_k)} \leq \Tr \bigabs{\rho^{(1)}_k - \rho^{(2)}_k} = \bigtrnorm{\rho^{(1)}_k - \rho^{(2)}_k}$.  

For the second part of the lemma, observe: 

\begin{claim}
For any $c \geq 0$ and any two density matrices $\sigma$ and $\tau$, $\trnorm{\sigma - \tau} \leq 2 \trnorm{\sigma - c \tau}$.  
\end{claim}

\begin{proof}
By symmetry, we may assume without loss of generality that $c \in [0,1]$.  Indeed, if $c > 1$, then $\trnorm{\sigma - c \tau} = c \trnorm{\tau - \frac{1}{c} \sigma} \geq \trnorm{\tau - \frac{1}{c} \sigma}$, and $1/c \in [0,1]$.  

For a Hermitian matrix $M$, let $M_\pm = \frac12 (\abs M \pm M) \succeq 0$.  Then 
\begin{align*}
\trnorm{\sigma - c \tau} 
&= \Tr (\sigma - c \tau)_+ + \Tr (\sigma - c \tau)_- \geq \Tr (\sigma - c \tau)_+ \geq \Tr (\sigma - \tau)_+ 
 \enspace .
\end{align*}
Here the second inequality follows since by Schur's Theorem~\cite{Bhatia07} and as $(1-c) \tau \succeq 0$, $\Tr (\sigma - c \tau)_+ = \max_{0 \preceq \Pi \preceq \identity} \Tr \Pi (\sigma - \tau + (1-c) \tau) \geq \max_{0 \preceq \Pi \preceq \identity} \Tr \Pi (\sigma - \tau) = \Tr (\sigma - \tau)_+$.  Finally, $\Tr (\sigma - \tau)_+ = \frac12 \trnorm{\sigma - \tau}$ since $\Tr \sigma = \Tr \tau$.  
\end{proof}
 
Therefore, 
\begin{align*}
\Ex\big[ \bigtrnorm{ \bar \rho^{(1)}_{K^{(1)}} - \bar \rho^{(2)}_{K^{(1)}} } \big]
&= \sum_k \Tr \rho^{(1)}_k \bigtrnorm{ \bar \rho^{(1)}_k - \bar \rho^{(2)}_k } \\
&\leq 2 \sum_{k : \, \rho^{\smash{(1)}}_k \neq 0} \Tr \rho^{(1)}_k \bigtrnorm{ \bar \rho^{(1)}_k - \tfrac{1}{\Tr \rho^{(1)}_k} \rho^{(2)}_k } \\
&\leq 2 \trnorm{ \rho^{(1)} - \rho^{(2)} }
 \enspace .  \qedhere
\end{align*}
\end{proof}

By a triangle inequality, a converse statement also holds: 
\begin{equation}
\trnorm{\rho^{(1)} - \rho^{(2)}} \leq \sum_k \abs{ \Tr (\rho^{(1)}_k - \rho^{(2)}_k) } + \sum_k \Tr \rho^{(1)}_k \trnorm{\bar \rho^{(1)}_k - \bar \rho^{(2)}_k}
 \enspace .
\end{equation}
Thus for measurement super-operators $\E$ and $\F$, $\E(\rho)$ is close to $\F(\sigma)$ in trace distance if and only if the distributions of measurement outcomes are close in total variation distance and the expected trace distance (under either measurement distribution) between the corresponding resulting states is small.  In general, both of the latter conditions are required for the implication that $\E(\rho) \approx \F(\sigma)$, but we will argue later that in certain special cases, e.g., $\rho = \sigma = \identity_\H / \dim \H$, the maximally mixed state, and~$\F$ a computational-basis measurement, it suffices that the expected trace distance between the resulting states be small.  See \lemref{t:blockscloseinexpectationimpliesmatricescloseandmore}.  

\medskip

An essential proposition in our analyses of sequential CHSH games and state and process tomography is the so-called Gentle Measurement Lemma.  It states that if a particular measurement outcome occurs with high probability on a given state~$\rho$, then that measurement does not much disturb~$\rho$: 

\begin{lemma}[{Gentle measurement~\cite{Winter99coding, OgawaNagaoka07channelcoding}}] \label{t:gentlemeasurement}
Let $\rho$ be a state, and $\Pi$ an operator with $0 \preceq \Pi \preceq \identity$.  Then 
\begin{equation}
\trnorm{\rho - \sqrt{\Pi} \rho \sqrt{\Pi}} \leq 2 \sqrt{1 - \Tr(\Pi \rho)}
 \enspace .
\end{equation}
\end{lemma}

A useful special case is when $\Pi$ can be written as $\pi \otimes \identity$ for a rank-one projection~$\pi$.  Then the Gentle Measurement Lemma implies that $\rho$ is close to a product state: 

\begin{corollary} \label{t:gentlemeasurementpurestate}
Let $\rho$ be a state on $\H_1 \otimes \H_2$, and let~$\pi$ be a pure state on~$\H_1$.  If for some $\delta \geq 0$, $\Tr (\pi \Tr_2 \rho) \geq 1 - \delta$, then 
\begin{equation}
\trnorm{\rho - \pi \otimes \Tr_1 \rho} \leq 2 \sqrt \delta + \delta
 \enspace .
\end{equation}
\end{corollary}

\begin{proof}
Substitute into \lemref{t:gentlemeasurement} $\Pi = \pi \otimes \identity$.  Since $\Tr (\Pi \rho) = \Tr(\pi \Tr_2 \rho) \geq 1 - \delta$, we obtain 
\begin{equation*}
\bigtrnorm{\rho - \pi \otimes \Tr_1 \!\big( (\pi \otimes \identity) \rho \big)} \leq 2 \sqrt \delta
 \enspace .
\end{equation*}
To finish, use $\bigtrnorm{\Tr_1 \rho - \Tr_1 \!\big( (\pi \otimes \identity) \rho \big)} = \bigtrnorm{\Tr_1 \! \big((\identity - \pi) \otimes \identity \big) \rho} = \Tr \! \big((\identity - \pi) \otimes \identity \big) \rho \leq \delta$.  
\end{proof}

This corollary can be generalized to say that if $\rho$ is a multi-partite state whose partial traces are close to pure states~$\pi_j$, then $\rho$ must be close to the tensor product $\bigotimes_j \pi_j$: 

\begin{lemma} \label{t:purepartsdeterminethewhole}
Let $\rho \in \L(\H_1 \otimes \cdots \otimes \H_m)$ be a quantum state.  For $j \in [m]$, let $\rho_j = \Tr_{1 \ldots \hat j \ldots m} \rho \in \L(\H_j)$ be its reduced density matrix on~$\H_j$.  Assume that for some $\delta \geq 0$ and for each $j \in [m]$ there exists a pure state $\pi_j \in \L(\H_j)$ such that $\Tr (\pi_j \rho_j) \geq 1 - \delta$.  Then 
\begin{equation}
\bigtrnorm{\rho - \pi_1 \otimes \cdots \otimes \pi_m} \leq m (2 \sqrt \delta + \delta)
 \enspace .
\end{equation}
\end{lemma}

\begin{proof}
By \corref{t:gentlemeasurementpurestate}, $\trnorm{\rho - \pi_j \otimes \Tr_j \rho} \leq 2 \sqrt \delta + \delta$ for all~$j$.  
Putting these bounds together, 
\begin{align*}
\bigtrnorm{\rho - \pi_1 \otimes \cdots \otimes \pi_m}
&\leq \sum_{j \in [m]} \bigtrnorm{ \pi_1 \otimes \cdots \pi_{j-1} \otimes (\Tr_{1\ldots j-1} \rho - \pi_j \otimes \Tr_{1\ldots j} \rho) } \\
&\leq \sum_{j \in [m]} \bigtrnorm{\rho - \pi_j \otimes \Tr_j \rho} \\
&\leq m (2 \sqrt \delta + \delta)
 \enspace .
\end{align*}
The second inequality holds because a partial trace cannot increase the trace distance.  
\end{proof}

Note that the lemma would also hold, with the same basic proof, if one of the states $\pi_j$, say $\pi_m$, were mixed and satisfied the assumption $\trnorm{\rho_m - \pi_m} \leq 2 \sqrt \delta + \delta$ instead of $\Tr (\pi_m \rho_m) \geq 1 - \delta$.  

\medskip

Let us conclude this section with two more straightforward technical claims about the trace~norm.  

\begin{lemma} \label{t:traceandtracenorm}
For a Hilbert space~$\H$ and linear operators~$A$ and $\Delta$ in $\L(\H)$, $\abs{\Tr (A \Delta)} \leq \norm{A} \trnorm{\Delta}$.  
\end{lemma}

\begin{proof}
Let $\Delta = \sum_j \lambda_j \ketbra{j}{j'}$ be the singular-value decomposition for~$\Delta$, for singular values $\lambda_j > 0$ and orthonormal sets $\{\ket j\}$ and $\{\ket{j'}\}$.  Then, since $\trnorm{\Delta} = \sum_j \lambda_j$, 
\begin{align*}
\abs{\Tr (A \Delta)}
&= \abs{ {\textstyle \sum}_j \lambda_j \bra{j'} A \ket j }
\leq {\textstyle \sum}_j \lambda_j \abs{\bra{j'} A \ket j}
\leq \norm{A} \trnorm{\Delta}
 \enspace . \qedhere
\end{align*}
\end{proof}

\begin{claim} \label{t:vectortraceversusl2distance}
For any two unit vectors $\ket a$ and $\ket b$, 
\begin{equation}
\sqrt 2 \min_{\phi \in [0, 2\pi)} \norm{\ket a - e^{i \phi} \ket b} \leq \trnorm{\ketbra a a - \ketbra b b} \leq 2 \norm{\ket a - \ket b}
 \enspace .
\end{equation}
For arbitrary vectors $\ket a$ and $\ket b$ with $\norm{\ket a - \ket b} \leq \delta$, $\trnorm{\ketbra a a - \ketbra b b} \leq \sqrt{4 \norm{\ket a}{}^2 \delta^2 + 4 \norm{\ket a} \delta^3 + \delta^4}$.  
\end{claim}

\begin{proof}
Calculate $\trnorm{\ketbra a a - \ketbra b b} = \sqrt{(\norm{\ket a}{}^2 + \norm{\ket b}{}^2)^2 - 4 \abs{\braket a b}{}^2}$.  For unit vectors, therefore, with $\theta = \arccos \abs{\braket a b}$, $\min_\phi \norm{\ket a - e^{i \phi} \ket b} = \sqrt{2 - 2 \cos \theta}$ and $\trnorm{\ketbra a a - \ketbra b b} = 2 \sin \theta$.  
The assertions follow.  
\end{proof}

Thus the trace distance between two pure states is closely related to their Euclidean vector distance up to a choice of phase.

\ifx\compilefullpaper\undefined  
\bibliographystyle{alpha-eprint}
\bibliography{q}

\end{document}
\fi

\ifx\compilefullpaper\undefined  
\documentclass[11pt]{article}

\begin{document}
\fi

\section{The CHSH game is rigid: A robust converse to Tsirelson's \mbox{inequality}} \label{s:CHSHgamerigidity}

In this section, we will study the CHSH game of \figref{f:chsh}.  We will argue that nearly optimal quantum strategies must, up to local changes of basis, be close to the ideal strategy that uses a shared EPR state, possibly in tensor product with an ancillary state.  

To avoid conflicting with the Pauli matrices $X$ and~$Y$, we will use lower-case letters $a, b, x, y$ for the random transcript in this section.  Recall that Alice and Bob win the game if the exor of their responses equals the product of Eve's questions, $x \oplus y = a b$.  In computer science terminology, the devices Alice and Bob are referred to as ``provers," and the experimentalist Eve is a ``verifier."  

In a general quantum strategy, Alice and Bob have Hilbert spaces $\H_A$ and $\H_B$, respectively, and a shared pure quantum state $\ket \psi \in \H_A \otimes \H_B \otimes \H_C$.  Here $\H_C$ is an inaccessible third Hilbert space used to purify the shared state.  Alice and Bob determine their outputs by applying POVMs, depending on~$a$ and~$b$, respectively, to their portions of $\ket \psi$.  By possibly appending ancilla states, we may without loss of generality assume that they apply two-outcome projective measurements.  (See also~\cite[Prop.~2]{CleveHoyerTonerWatrous04nonlocal}.)  For $\device \in \{A, B\}$ and $\alpha, \chi \in \{0,1\}$, let $P^\device(\alpha, \chi)$ be the projection applied by prover $\device$ for question $\alpha$ and answer~$\chi$.  Let $\RXa{\alpha} = P^\device(\alpha, 0) - P^\device(\alpha, 1)$.  Since $P^\device(\alpha, 1) = \identity_{\H_\device} - P^\device(\alpha, 0)$, $\RXa{\alpha}$ is a reflection.  Define the strategy's \emph{correlation value} to be 
\begin{equation} \label{e:chshcorrelationvaluedef}
4 \big(2 \Pr[a b = x \oplus y] - 1\big)
= \bra \psi \Big( \sum_{a, b \in \{0,1\}} (-1)^{a b} \RAa{a} \otimes \RBa{b} \Big) \otimes \identity_{\H_C} \ket \psi
 \enspace .
\end{equation}

An example of a strategy that uses a shared EPR state~$\ket \psi = \frac{1}{\sqrt 2}(\ket{00} + \ket{11})$ is given in \tabref{f:optimalCHSHstrategy}.  This strategy satisfies that conditioned on any fixed values for $a$, $b$ and $x$, the probability over $y$ that the provers win is $\cos^2 \frac\pi8$.  Tsirelson's inequality~\cite{Tsirelson80inequality} states that this strategy is optimal: for any quantum strategy, $\Pr[ab = x \oplus y] \leq \cos^2 \frac\pi8 = \frac12(1 + \frac{1}{\sqrt 2}) \approx 85.4\%$.  Therefore, the correlation value is at most $2 \sqrt 2$.  In contrast, for any classical strategy, based on a shared random string instead of a shared quantum state, the maximum probability of winning is~$3/4$.  

\begin{table}
\centering
\begin{tabular}{c c@{$\qquad\qquad$}c c}
\multicolumn{2}{c}{\bf Alice's strategy$\qquad\qquad$} & \multicolumn{2}{c}{\bf Bob's strategy} \\
\underline{$a = 0$} & \underline{$a = 1$} & \underline{$b = 0$} & \underline{$b = 1$} \\
$\ketbra 0 0 \mapsto x = 0$ & $\ketbra + + \rightarrow x = 0$ & $G^\dagger \ketbra + + G \rightarrow y = 0$ & $G^\dagger \ketbra 0 0 G \rightarrow y = 0$ \\
$\ketbra 1 1 \mapsto x = 1$ & $\ketbra - - \rightarrow x = 1$ & $G^\dagger \ketbra - - G \rightarrow y = 1$ & $G^\dagger \ketbra 1 1 G \rightarrow y = 1$ \\
\includegraphics[scale=1]{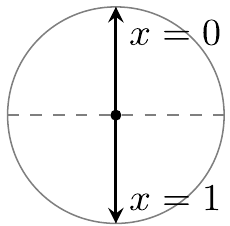} & \includegraphics[scale=1]{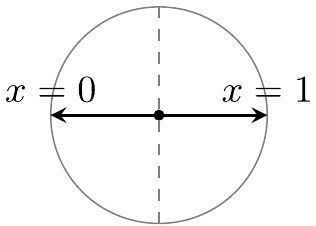} & \includegraphics[scale=1]{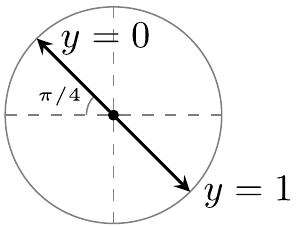}$\!\!\!\!\!\!\!\!\!$ & \includegraphics[scale=1]{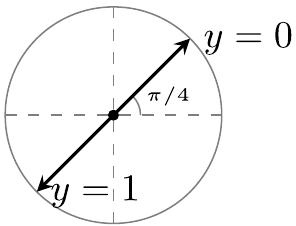}$\!\!\!\!\!\!\!\!\!$ 
\end{tabular}
\caption{An optimal quantum strategy for the CHSH game.  Alice and Bob each have one qubit of a shared EPR state $\frac{1}{\sqrt 2} (\ket{00} + \ket{11})$.  On each input $a$ or $b$, they make the two-outcome projective measurements listed above.  Here, $\ket{\pm} = \frac{1}{\sqrt 2}(\ket 0 \pm \ket 1)$ and $G = \exp(- i \frac\pi8 Y)$.  Thus $\RAa{0} = Z$, $\RAa{1} = X$, $\RBa{0} = G^\dagger X G$ and $\RBa{1} = G^\dagger Z G$.  The measurements are also illustrated on a cross-section through the $xz$-plane of the Bloch sphere.} \label{f:optimalCHSHstrategy}
\end{table}

Our CHSH rigidity lemma, a robust converse to Tsirelson's inequality, states that any strategy that achieves correlation value at least $2 \sqrt 2 - \epsilon$ must be $O(\sqrt \epsilon)$ close to the ideal strategy of \tabref{f:optimalCHSHstrategy}.  
In \appref{s:generalizedeprlemma}, we prove a similar statement for an extended CHSH game in which the ideal strategy also includes measurements in the $y$ direction of the Bloch sphere.  

\begin{definition} \label{t:chshgamestructuredef}
For $\epsilon \geq 0$, a quantum strategy for the CHSH game is \emph{$\epsilon$-structured} if the correlation value is at least $2 \sqrt 2 - \epsilon$.  
\end{definition}

\begin{lemma}[CHSH game rigidity] \label{t:eprlemma}
There exists a constant $c > 0$ such that the following statements hold.  
Consider a quantum strategy for the CHSH game, specified by Hilbert spaces $\H_A$, $\H_B$ and $\H_C$, a state $\ket \psi \in \H_A \otimes \H_B \otimes \H_C$, and reflections $\RXa{\alpha} \in \L(\H_\device)$ for $\device \in \{A, B\}$ and $\alpha \in \{0,1\}$.  Let $\epsilon > 0$ and assume that the strategy is $\epsilon$-structured.  

Then there are extensions of the Hilbert spaces $\H_A, \H_B$, and extensions of the reflections $\RXa{\alpha}$ by a direct sum with other reflections, so that the following properties hold: 
\begin{itemize}
\item
There is an isomorphism between Alice's extended space and $\C^2 \otimes \hat \H_A$, under which $\RAa{0} = Z \otimes \identity$ and $\bignorm{(\RAa{1} - X \otimes \identity)_A \otimes \identity_{BC} \ket \psi} < c \sqrt \epsilon$.  
\item 
Bob's space is isomorphic to $\C^2 \otimes \hat \H_B$, with $\RBa{0} = Z \otimes \identity$ and $\bignorm{(\RBa{1} - X \otimes \identity)_B \ket \psi} < c \sqrt \epsilon$.  
\item 
Finally, letting 
\begin{equation} \label{e:perfectBellpairalmost}
\ket{\psi^*} = \big(I \otimes (H G)\big) \frac{1}{\sqrt 2}\big(\ket{00} + \ket{11}\big)
 \enspace ,
\end{equation}
there exists a unit vector $\ket{\psi^\times} \in \hat \H_A \otimes \hat \H_B \otimes \H_C$ with $\norm{\ket \psi - \ket{\psi^*} \otimes \ket{\psi^\times}} < c \sqrt \epsilon$.  
\end{itemize}
Furthermore, if $\H_A$ and $\H_B$ are finite-dimensional, then the isomorphisms into $\C^2 \otimes \hat \H_A$ and into $\C^2 \otimes \hat \H_B$ depend only on $\RAa{0}, \RAa{1}$ and on $\RBa{0}, \RBa{1}$, respectively.  
\end{lemma}

Up to the constant factor, the $O(\sqrt \epsilon)$ dependence of the error terms is tight.  Indeed, if one starts with the ideal strategy of \tabref{f:optimalCHSHstrategy} and perturbs either the shared state or the measurements by~$\delta$, the correlation value will generically decrease by $\Theta(\delta^2)$; first-order corrections must cancel.  

In our main applications of \lemref{t:eprlemma}, the spaces $\H_A$ and $\H_B$ will be finite-dimensional.  The final statement in the lemma is important because we would like the isomorphisms into $\C^2 \otimes \hat \H_A$ and $\C^2 \otimes \hat \H_B$ to be computable locally even without knowing the underlying state~$\ket \psi$.  Indeed, after playing multiple CHSH games in sequence, neither prover knows~$\ket \psi$.  However, the dimension-truncation argument given below depends on~$\ket \psi$.  

For the proof of \lemref{t:eprlemma} we will use the following characterization of the eigen-decomposition of the product of reflections due to Jordan~\cite{Jordan75projections}.  Its use is common in quantum computation, including in algorithms~\cite{Szegedy04walkfocs, Reichardt10advtight, LeeMittalReichardtSpalekSzegedy11stateconversion}, in amplification of QMA in complexity theory~\cite{MarriottWatrous05qma, NagajWocjanZhang09qma}, and in the study of Bell inequalities in entanglement theory and device-independent QKD~\cite{Masanes06bellinequalities, PironioAcinBrunnerGisinMassarScarani09qkd, McKague09deviceindependent}.  

\begin{lemma}[Jordan's Lemma] \label{t:jordanslemma}
Let $\Pi$ and $\Delta$ be projections acting on a {finite-dimensional} Hilbert space~$\H$.  Then $\H$ can be decomposed into orthogonal one- and two-dimensional subspaces invariant under $\Pi$ and $\Delta$.  
\end{lemma}

Before beginning the proof of \lemref{t:eprlemma}, let us sketch the argument for the case that $\H_A = \H_B = \C^2$, $\H_C = \C$ and $\epsilon = 0$.  The rest of the proof essentially works by applying Jordan's Lemma to $\RAa{0}$ and~$\RAa{1}$, and again to $\RBa{0}$ and $\RBa{1}$, to locate Alice and Bob's qubits for the game and therefore reduce to this two-dimensional case.  However, achieving the optimal $O(\sqrt \epsilon)$ error dependence requires more work.  

If the $\RXa{\alpha}$ reflections act on $\C^2$ and are not equal to $\pm I$, then we can choose a basis such that $\RXa{0} = Z$, $\RAa{1} = \smatrx{\cos 2 \theta & \sin 2 \theta \\ \sin 2 \theta & -\cos 2 \theta}$ and $\RBa{1} = \smatrx{\cos 2 \theta' & \sin 2 \theta' \\ \sin 2 \theta' & -\cos 2 \theta'}$ for certain angles $\theta, \theta' \in [0, \frac\pi2]$.  Letting $M_0 = \frac12 (\RAa{0} + \RAa{1}) \otimes I - \frac{1}{\sqrt 2} I \otimes \RBa{0}$ and $M_1 = \frac12 (\RAa{0} - \RAa{1}) \otimes I - \frac{1}{\sqrt 2} I \otimes \RBa{1}$, the correlation value satisfies 
\begin{equation*}
2 \sqrt 2 - \epsilon 
\leq \bra \psi \Big( \sum_{a, b \in \{0,1\}} (-1)^{a b} \RAa{a} \otimes \RBa{b} \Big) \ket \psi
= 2 \sqrt 2 - \sqrt 2 \bra \psi (M_0^2 + M_1^2) \ket \psi
 \enspace .
\end{equation*}
For $\epsilon = 0$, this means that $\ket \psi$ must lie in the intersection of the kernels of $M_0$ and~$M_1$.  The four eigenvalues of $M_0$ are $\pm \cos \theta \pm \frac{1}{\sqrt 2}$.  For the kernel to be nonempty, it must be that $\theta = \frac\pi4$.  A symmetrical argument implies that $\theta' = \frac\pi4$.  For small $\epsilon > 0$, $\ket \psi$ must lie close to small-eigenvalue subspaces of both $M_0$ and $M_1$, implying that $\theta$ and~$\theta'$ are close to $\frac\pi4$.  Thus the measurement operators are rigidly determined.  

For $\theta = \theta' = \frac\pi4$, the kernel of $\sqrt 2 ((H G) \otimes I) M_0 ((G^\dagger H) \otimes I) = Z \otimes I - I \otimes Z$ is spanned by the vectors $\ket{00}$ and $\ket{11}$.  The kernel of $\sqrt 2 ((H G) \otimes I) M_1 ((G^\dagger H) \otimes I) = X \otimes I - I \otimes X$ is spanned by the vectors $\ket{+} \otimes \ket{+} = \frac12 (\ket{00} + \ket{01} + \ket{10} + \ket{11})$ and $\ket{-} \otimes \ket{-} = \frac12 (\ket{00} - \ket{01} - \ket{10} + \ket{11})$.  For the $\ket{01}$ and $\ket{10}$ terms to cancel out, a linear combination of these vectors must have equal coefficients.  The intersection between the two kernels is therefore spanned by $\ket{00} + \ket{11}$.  Thus the state $\ket \psi$ is rigidly determined.  

The above argument, together with Jordan's Lemma, conveys much of the intuition for the CHSH rigidity lemma.  However, we have not explained the derivation of the operators $M_0$ and~$M_1$, chosen to satisfy $\sum_{a, b \in \{0,1\}} (-1)^{a b} \RAa{a} \otimes \RBa{b} = 2 \sqrt 2 I \otimes I - \sqrt 2 (M_0^2 + M_1^2)$.  In general, for a game in which Eve draws her questions from the distribution $p(a,b)$ and accepts if $x \oplus y = V(a,b)$, let $\Theta = \sum_{a,b} p(a,b) (-1)^{V(a,b)} \ketbra a b$ and $\hat \Theta = \smatrx{0&\Theta\\\Theta^\dagger&0}$.  Let $\omega^*$ be the optimal success probability.  By the Tsirelson semi-definite program~\cite{CleveSlofstraUngerUpadhyay06parallelXOR}, the optimal bias is $2 \omega^* - 1 = \frac12 \max_{\Gamma \succeq 0, \Gamma \circ I = I} \langle \hat \Theta, \Gamma \rangle = \frac12 \min_{\Delta = \Delta \circ I \succeq \hat \Theta} \Tr \Delta$.  $\Gamma$ is the Gram matrix of the vectors $\RAa{a} \ket \psi$ and $\RBa{b} \ket \psi$.  Letting $\Delta^*$ achieve the second optimum, we have $\frac12 \langle \hat \Theta, \Gamma \rangle = (2\omega^*-1) - \frac12 \langle \Delta^* - \hat \Theta, \Gamma \rangle$.  For the CHSH game, $\Delta^* = \frac{1}{2 \sqrt 2} \identity$, and the matrices $M_0, M_1$ correspond to eigenvectors of $\Delta^* - \hat \Theta$.  

\begin{proof}[Proof of \lemref{t:eprlemma}]
We begin the proof by truncating the Hilbert spaces $\H_A$ and $\H_B$ to finite dimensions, in order to apply Jordan's Lemma.  Jordan's Lemma is false for infinite-dimensional Hilbert spaces.  

\begin{claim} \label{t:truncatedimensions}
For any $\delta > 0$, there are finite-dimensional subspaces $\bar \H_A \subseteq \H_A$, $\bar \H_B \subseteq \H_B$ such that: 
\begin{itemize}
\item 
For $\device \in \{A, B\}$, $\bar \H_\device$ is closed under $\RXa{0}$.  
\item 
For $\device \in \{A, B\}$, there exists a reflection $\RXdeca{\bar}{1} \in \L(\H_\device)$ with $\norm{(\RXdeca{\bar}{1} - \RXa{1}) \otimes \identity \ket \psi} < \delta$ and under which $\bar \H_\device$ is closed.  
\item
Letting $\ket{\bar \psi}$ be $\ket \psi$ projected to $\bar \H_A \otimes \bar \H_B \otimes \H_C$ and renormalized, $\norm{\ket{\bar \psi} - \ket \psi} < \delta$.  
\item
The joint strategy specified by Alice's reflections $\RAa{0}$, $\RAdeca{\bar}{1}$, Bob's reflections $\RBa{0}$, $\RBdeca{\bar}{1}$, and the joint state $\ket{\bar \psi}$ has correlation value at least $2 \sqrt 2 - \epsilon - \delta$.  
\end{itemize}
\end{claim}

\begin{proof}
First truncate the spaces $\H_A$ and $\H_B$ to finite dimensional spaces $\tilde \H_A$ and $\tilde \H_B$ that are closed under $\RAa{1}$ and $\RBa{1}$, respectively, and such that $\ket \psi$ is almost entirely supported on $\tilde \H_A \otimes \tilde \H_B \otimes \H_C$.  For $\device \in \{A, B\}$, let $\bar \H_\device$ be the closure of $\tilde \H_\device$ under $\RXa{0}$.  Let $\RXdeca{\bar}{1}$ be $\RXa{1}$ on $\tilde \H_\device$ extended by the identity on $\tilde \H_\device^\perp$.  In this way, $\bar \H_\device$ is closed under both $\RXa{0}$ and $\RXdeca{\bar}{1}$.  
\end{proof}

Using the assumption $\epsilon > 0$, apply \claimref{t:truncatedimensions} with $\delta = \epsilon$.  By Jordan's Lemma, $\bar \H_A$ can be decomposed into the direct product of a set of one- and two-dimensional subspaces invariant under both $\RAa{0}$ and $\RAdeca{\bar}{1}$.  For notational convenience, add dimensions and extend the reflections if necessary, so each subspace is two dimensional and includes both $+1$-eigenvalue and $-1$-eigenvalue eigenvectors for both reflections.  Index these subspaces by~$i$.  Similarly decompose $\bar \H_B$ according to $\RBa{0}$ and $\RBdeca{\bar}{1}$, indexing the invariant two-dimensional subspaces by~$i'$.  

Let $\theta_i \in [0, \frac{\pi}{2}]$ be the angle between the $+1$ eigenvectors of $\RAa{0}$ and $\RAdeca{\bar}{1}$ on the $i$th subspace, and let $C_i = \cos 2 \theta_i$ and $S_i = \sin 2 \theta_i$.  Define the angles $\theta_{i'}$ similarly, and let $C_{i'} = \cos 2 \theta_{i'}$, $S_{i'} = \sin 2 \theta_{i'}$.  
Choose orthonormal basis vectors $\ket 0 = \smatrx{1\\0}$, $\ket 1 = \smatrx{0\\1}$ for each subspace, so 
\begin{align*}
\RAa{0} \vert_{\bar \H_A} &= \sum_i \ketbra i i \otimes Z 
&
\RAdeca{\bar}{1} \vert_{\bar \H_A} &= 
\sum_i \ketbra i i \otimes \left(\begin{smallmatrix} C_i & S_i \\ S_i & -C_i \end{smallmatrix}\right)
\\
\RBa{0} \vert_{\bar \H_B} &= \sum_{i'} \ketbra{i'}{i'} \otimes Z 
&
\RBdeca{\bar}{1} \vert_{\bar \H_B} &= 
\sum_{i'} \ketbra {i'}{i'} \otimes \left(\begin{smallmatrix} C_{i'} & S_{i'} \\ S_{i'} & -C_{i'} \end{smallmatrix}\right)
 \enspace .
\end{align*}
Since each $\bar \H_\device^\perp \subset \H_\device$ is closed under $\RXa{0}$, we can choose a basis for $\bar \H_\device^\perp \subset \H_\device$ so that $\H_\device \cong \hat \H_\device \otimes \C^2$ and $\RXa{0} = \identity \otimes Z$ everywhere, by if necessary extending the Hilbert space $\H_\device$.  This gives two of the claims of \lemref{t:eprlemma}.  

With this decomposition, and letting $\{\ket c\}$ be an orthonormal basis for $\H_C$, the shared state~$\ket{\bar \psi}$ can be written 
\begin{equation*}\begin{split}
\ket{\bar \psi}_{ABC} &= \sum_{c, i, i'} \ket{c}_C \otimes \ket{i}_A \otimes \ket{i'}_B \otimes \ket{\psi_{c i i'}}_{AB} \\
\ket{\psi_{c i i'}}_{AB} &= \sum_{b, b' \in \{0,1\}} \alpha_{c i i' b b'} \ket{b}_A \otimes \ket{b'}_B
 \enspace .
\end{split}\end{equation*}
Let $\ket{\tilde \psi_{c i i'}} = \ket{\psi_{c i i'}} / \norm{\ket{\psi_{c i i'}}}$.  

\smallskip

For $j \in \{0,1\}$, let $M_j = \frac12 (\RAa{0} + (-1)^j \RAdeca{\bar}{1}) \otimes \identity_{BC} - \frac{1}{\sqrt 2} \RBa{j} \otimes \identity_{AC}$.  Then 
\begin{equation}\begin{split} \label{e:eprlemmaproofsdpdual}
2 \sqrt 2 - 2\epsilon 
&\leq \bra{\bar \psi} \big( \RAa{0} \otimes \RBa{0} + \RAa{0} \otimes \RBdeca{\bar}{1} + \RAdeca{\bar}{1} \otimes \RBa{0} - \RAdeca{\bar}{1} \otimes \RBdeca{\bar}{1} \big) \ket{\bar \psi} \\
&= 2 \sqrt 2 - \sqrt 2 \langle M_0^2 + M_1^2 \rangle_{\ket{\bar \psi}}
 \enspace .
\end{split}\end{equation}
In particular, letting $\beta_{cii'} = \langle M_0^2 + M_1^2 \rangle_{\ket{ii'} \otimes \ket{\tilde \psi_{c i i'}}}$, 
\begin{equation} \label{e:averageCHSHviolation}
0 \leq \big\langle M_0^2 + M_1^2 \big\rangle_{\ket{\bar \psi}} = \sum_{c, i, i'} \norm{\ket{\psi_{cii'}}}^2 \beta_{c i i'} \leq \sqrt 2 \epsilon
 \enspace .
\end{equation}

\begin{proposition} \label{t:goodsubspacestate}
For any $c, i, i'$, $\sin 2 \theta_i \geq 1 - O(\beta_{cii'})$, $\sin 2 \theta_{i'} \geq 1 - O(\beta_{cii'})$ and for some phase~$\phi_{cii'}$, 
\begin{equation}
\norm{\ket{\tilde \psi_{cii'}} - e^{i \phi_{cii'}} \ket{\psi^*}}{}^2 \leq O(\beta_{cii'})
 \enspace .
\end{equation}
\end{proposition}

\begin{proof}
To simplify notation, we will suppress the $c, i, i'$ dependence, and restrict the operators~$M_j$ to the invariant $i, i'$ subspace.  Then $\beta = \beta_{cii'} = \norm{M_0 \ket{\tilde \psi}}{}^2 + \norm{M_1 \ket{\tilde \psi}}{}^2$.  Assume that $\beta \leq 2 \cdot 10^{-6}$; by fixing the hidden constants in our desired inequalities to be sufficiently large, the claims are trivial for larger~$\beta$.  

Expanding $\ket{\tilde \psi}$ as $\ket{\tilde \psi} = \alpha_{00} \ket{00} + \alpha_{01} \ket{01} + \alpha_{10} \ket{10} + \alpha_{11} \ket{11}$, and letting $C = C_i$, $S = S_i$, 
\begin{equation*}\begin{split}
2 M_0 \ket{\tilde \psi}
&= \alpha_{00} (1 - \sqrt 2) \ket{00} + \alpha_{01} (1 + \sqrt 2) \ket{01} - \alpha_{10} (1 + \sqrt 2) \ket{10} - \alpha_{11} (1 - \sqrt 2) \ket{11} \\
&\quad + \alpha_{00} \binomial{C \ket 0}{+S \ket 1} \ket 0 + \alpha_{01} \binomial{C \ket 0}{+S \ket 1} \ket 1 + \alpha_{10} \binomial{S \ket 0}{-C \ket 1} \ket 0 + \alpha_{11} \binomial{S \ket 0}{-C \ket 1} \ket 1 \\
&= \ket{00} \big[\alpha_{00} (1 + C - \sqrt 2) + \alpha_{10} S\big] + \ket{01} \big[\alpha_{01} (1 + C + \sqrt 2) + \alpha_{11} S\big] \\
&\quad + \ket{10} \big[\alpha_{00} S - \alpha_{10} (1 + C + \sqrt 2)\big] + \ket{11} \big[\alpha_{01} S - \alpha_{11} (1 + C - \sqrt 2)\big]
 \enspace .
\end{split}\end{equation*}
For $b, b' \in \{0,1\}$, let $\delta_{bb'} = 2 \bra{bb'} M_0 \ket{\tilde \psi}$.  Then $\sum_{bb'} \abs{\delta_{bb'}}{}^2 = 4 \norm{M_0 \ket{\tilde \psi}}{}^2 \leq 4 \beta$.  
We find 
\begin{align} \label{e:eprlemmaproofcoefficients}
\alpha_{01} &= \frac{-\alpha_{11} S + \delta_{01}}{1 + C + \sqrt 2} &
\alpha_{10} &= \frac{\alpha_{00} S - \delta_{10}}{1 + C + \sqrt 2}
 \enspace ,
\end{align}
implying 
\begin{align} \label{e:eprlemmaproofcoefficientsbound}
\alpha_{00} \Big[ (1+C-\sqrt 2) + \frac{S^2}{1 + C + \sqrt 2} \Big] &= \delta_{00} + \frac{S \delta_{10}}{1 + C + \sqrt 2} \leq 3 \sqrt \beta
 \enspace .
\end{align}
The inequality uses $\abs C, \abs S \leq 1$ and a Cauchy-Schwarz inequality.  

\begin{claim}
Either $\abs{\alpha_{00}} \geq 1/2$ or $\abs{\alpha_{11}} \geq 1/2$.  
\end{claim}

\begin{proof}
Assume $\abs{\alpha_{00}} < 1/2$.  By Eq.~\eqnref{e:eprlemmaproofcoefficients} and since $\beta < 1/100$, 
\begin{equation*}
\abs{\alpha_{10}} \leq \frac{\abs{\alpha_{00}} + \abs{\delta_{10}}}{\sqrt 2} < \frac{\frac12 + 2 \sqrt \beta}{\sqrt 2} < \frac12 
 \enspace .
\end{equation*}
Symmetry under switching Alice and Bob implies $\abs{\alpha_{01}} < \frac12$.  Since $\sum_{bb'} \abs{\alpha_{bb'}}^2 = 1$, $\abs{\alpha_{11}} > 1/2$.  
\end{proof}

Thus from Eq.~\eqnref{e:eprlemmaproofcoefficientsbound} we determine 
\begin{equation*}
\Bigabs{1 + C - \sqrt 2 + \frac{S^2}{1 + C + \sqrt 2}} \leq 6 \sqrt \beta 
 \enspace .
\end{equation*}
Multiplying both sides by $1 + C + \sqrt 2 \leq 2 + \sqrt 2$ gives $\abs{(1 + C)^2 - 2 + S^2} = 2 \abs C \leq 21 \sqrt \beta$.  As $\beta \leq 2 \cdot 10^{-6}$, $S = \sqrt{1 - C^2} \geq 1 - 61 \beta$.  By symmetry, $\sin 2 \theta_{i'} \geq 1 - 61 \beta$, too.  

Now let us read off bounds for the coefficients $\alpha_{bb'}$ from Eq.~\eqnref{e:eprlemmaproofcoefficients}.  First,  
\begin{equation*}\begin{split}
\Bigabs{\alpha_{10} - \frac{\alpha_{00}}{1 + \sqrt 2}} 
= \Bigabs{\frac{\alpha_{00} S - \delta_{10}}{1 + C + \sqrt 2} - \frac{\alpha_{00}}{1 + \sqrt 2}} 
\leq \abs{\alpha_{00}} \cdot \Bigabs{\frac{S}{1 + C + \sqrt 2} - \frac{1}{1 + \sqrt 2}} + \sqrt \beta 
\leq 48 \sqrt \beta
 \enspace .
\end{split}\end{equation*}
By symmetry, $\abs{\alpha_{01} - \frac{\alpha_{00}}{1 + \sqrt 2}} \leq 48 \sqrt \beta$.  By the same steps, $\abs{\alpha_{10} + \frac{\alpha_{11}}{1 + \sqrt 2}} \leq 48 \sqrt \beta$.  Hence, 
\begin{equation*}\begin{split}
\abs{\alpha_{11} + \alpha_{00}}
\leq \bigabs{\alpha_{11} + (1 + \sqrt 2) \alpha_{10}} + \bigabs{\alpha_{00} - (1 + \sqrt 2) \alpha_{10}} 
\leq 2(1 + \sqrt 2) 48 \sqrt \beta \leq 232 \sqrt \beta
 \enspace .
\end{split}\end{equation*}

Putting these coefficient bounds together, we have 
\begin{equation*}\begin{split}
\Bigabs{1 - \frac{\abs{\alpha_{00}}^2}{\abs{\braket{00}{\psi^*}}^2}} 
&= 
\Bigabs{\sum_{b,b'} \abs{\alpha_{bb'}}^2 - \abs{\alpha_{00}}^2 \Big( 2 + \frac{2}{(1 + \sqrt 2)^2} \Big) } \\
&\leq \bigabs{\abs{\alpha_{11}}^2 - \abs{\alpha_{00}}^2} + \Bigabs{\abs{\alpha_{01}}^2 - \frac{\abs{\alpha_{00}}^2}{(1 + \sqrt 2)^2}} + \Bigabs{\abs{\alpha_{10}}^2 - \frac{\abs{\alpha_{00}}^2}{(1 + \sqrt 2)^2}} \\
&\leq 2 \cdot 232 \sqrt \beta + 2 \cdot 48 \sqrt \beta + 2 \cdot 48 \sqrt \beta = 656 \sqrt \beta
 \enspace ,
\end{split}\end{equation*}
where we have used $\abs{x^2 - y^2} = \abs{x-y} \cdot \abs{x+y}$.  In particular, it follows that if we let $\phi = \phi_{c i i'}$ be the argument of $\alpha_{00}$, so that $e^{-i \phi} \alpha_{00} = \abs{\alpha_{00}}$, 
\begin{equation*}\begin{split}
\abs{\alpha_{00} - e^{i \phi} \braket{00}{\psi^*}} 
\leq \frac{\abs{\braket{00}{\psi^*}}^2 \cdot 656 \sqrt \beta}{\abs{\alpha_{00}} + \braket{00}{\psi^*}} 
\leq 429 \sqrt \beta
 \enspace ,
\end{split}\end{equation*}
and   
\begin{equation*}\begin{split}
\norm{\ket{\tilde \psi} - e^{i \phi} \ket{\psi^*}}{}^2
&= \sum_{b,b'} \abs{\alpha_{bb'} - e^{i \phi} \braket{bb'}{\psi^*}}{}^2 \\
&\leq \Big(429^2 \! + 2 \Big(48 + \frac{429}{1 + \sqrt 2}\Big)^2 \!\! + \! (232 + 429)^2 \Big) \beta 
\leq 10^6 \beta
 \enspace . \qedhere
\end{split}\end{equation*}
\end{proof}

We now collect together our calculations to prove \lemref{t:eprlemma}.  Let 
\begin{equation*}	
\ket{\psi^\times} 
= \sum_{c,i,i'} e^{i \phi_{c i i'}} \norm{\ket{\psi_{cii'}}} \, \ket{cii'}_{CAB}
 \enspace ,
\end{equation*}
a unit vector.  We have, by \propref{t:goodsubspacestate} and Eq.~\eqnref{e:averageCHSHviolation}, 
\begin{align*}
\norm{\ket{\bar \psi} - \ket{\psi^*} \otimes \ket{\psi^\times}}{}^2
&= \sum_{c,i,i'} \bignorm{\ket{\psi_{cii'}} - e^{i \phi_{cii'}} \norm{\ket{\psi_{cii'}}} \cdot \ket{\psi^*}}^2 
= \sum_{c, i, i'} \norm{\ket{\psi_{cii'}}}^2 \cdot O(\beta_{cii'}) = O(\epsilon)
 \enspace .  
\end{align*}

This establishes the last claim in \lemref{t:eprlemma}.  It remains to argue that $\bignorm{(\RAa{1} - X \otimes \identity)_A \ket \psi}$ and $\bignorm{(\RBa{1} - X \otimes \identity)_B \ket \psi}$ are each of order $\sqrt \epsilon$.  
From \claimref{t:truncatedimensions} and a triangle inequality, we bound $\bignorm{(\RAa{1} - X \otimes \identity)_A \ket \psi} \leq 2 \delta + 2 \norm{\ket{\bar \psi} - \ket{\psi^*} \ket{\psi^\times}} + \bignorm{(\RAdeca{\bar}{1} - X \otimes \identity)_A \ket{\psi^*} \ket{\psi^\times}}$.  To bound the last term, recall that on each subspace~$i$, $\RAdeca{\bar}{1}$ acts as $\smatrx{C_i & S_i \\ S_i & -C_i}$, and so 
\begin{align*}
\bignorm{(\RAdeca{\bar}{1} - X \otimes \identity)_A \ket{\psi^*} \ket{\psi^\times}}^2
&= \sum_{c,i,i'} \norm{\ket{\psi_{cii'}}}^2 \Bignorm{\Big[ \smatrx{C_i & S_i \\ S_i & -C_i} - X \Big] \otimes \identity \ket{\psi^*}}^2 
= 2 \sum_{c,i,i'} \norm{\ket{\psi_{cii'}}}^2 (1 - S_i) 
 \enspace ,
\end{align*}
which is again of order~$\epsilon$ by \propref{t:goodsubspacestate} and Eq.~\eqnref{e:averageCHSHviolation}.  A symmetrical argument bounds $\bignorm{(\RBa{1} - X \otimes \identity)_B \ket \psi}$.  
\end{proof}

For later convenience in \secref{s:sequentialstructuredCHSHgameshavetensorproductstructure}, let us state several simple technical corollaries of \lemref{t:eprlemma}.  

\begin{corollary} \label{t:structuredgameprobabilitylowerbound}
There exists a constant $c > 0$ such that under the conditions of \lemref{t:eprlemma}, it further holds that for $\epsilon < c$, each of the sixteen possible outcomes of the game, $(a, x, b, y) \in \{0,1\}^4$, occurs with probability at least $1/60$.  
\end{corollary}

\begin{proof}
In an ideal CHSH game the probability of an outcome $(a, x, b, y)$ is $\frac{1}{16} (1 + \frac{1}{\sqrt 2})$ if $a b = x \oplus y$ and is $\frac{1}{16} (1 - \frac{1}{\sqrt 2}) > \frac{1}{60}$ otherwise.  As $\epsilon$ tends to $0$, \lemref{t:eprlemma} implies that the probabilities of the different outcomes converge to these values, so for sufficiently small $\epsilon$ all probabilities will be at least~$1/60$.  
\end{proof}

\begin{corollary} \label{t:pullmeasurementstotheotherside}
Under the conditions of \lemref{t:eprlemma}, it further holds that for $(\device, \device') \in \{(A,B), (B,A)\}$ and $\alpha \in \{0,1\}$, 
\begin{equation}
\bignorm{ \big[ \RXa{\alpha} \otimes \identity_{\device' C} - \tfrac{1}{\sqrt 2} ((Z+ (-1)^\alpha X) \otimes \identity)_{\device'} \otimes \identity_{\device C} \big] \ket \psi } = O(\sqrt \epsilon) 
 \enspace .
\end{equation}
\end{corollary}

\begin{proof}
This corollary says that Bob's measurements can be pulled over to Alice's side, or Alice's measurements pulled over to Bob's side.  For an ideal CHSH game, this fact is a consequence of the identity $\big(\begin{smallmatrix}a&b\\c&d\end{smallmatrix}\big) \otimes I (\ket{00} + \ket{11}) = I \otimes \big(\begin{smallmatrix}a&c\\b&d\end{smallmatrix}\big) (\ket{00} + \ket{11})$.  The claimed bounds come from combining this with triangle inequalities from \lemref{t:eprlemma}.  For example, for $\alpha = 1$, letting $R = \frac{1}{\sqrt 2} (Z - X)$, 
\begin{align*}
\bignorm{ [ \RXa{1} - (R \otimes \identity)_{\device'} ] \ket \psi }
&\leq \bignorm{ (\RXa{1} - X \otimes \identity)_\device \ket \psi } + \bignorm{ [ (X \otimes \identity)_\device - (R \otimes \identity)_{\device'} ] (\ket \psi - \ket{\psi^*} \ket{\psi^\times}) } \\ &\quad + \bignorm{ [ (X \otimes \identity)_\device - (R \otimes \identity)_{\device'} ] \ket{\psi^*} \ket{\psi^\times} }
 \enspace .
\end{align*}
The first two terms on the right are each of order $\sqrt \epsilon$, and the third term is zero.  
\end{proof}

\corref{t:pullmeasurementstotheotherside} in turn implies that the expectation value of an observable localized to $\H_A$ cannot change much, on average, over an $\epsilon$-structured CHSH game, since Alice's measurement can be pulled over to Bob's side.  Recall that $P^\device(\alpha, \chi) = \frac12 (\identity + (-1)^\chi \RXa{\alpha})$.  Let $\ket{\phi(a,x,b,y)} = P^A(a, x) \otimes P^B(b, y) \ket \psi$ and $\ket{\psi(a,x,b,y)} = \ket{\phi(a,x,b,y)} / \norm{\ket{\phi(a,x,b,y)}}$.   

\begin{corollary} \label{t:localobservablesunderstructuredCHSH}
Under the conditions of \lemref{t:eprlemma}, let $\ket \phi = \frac12 \sum_{a, x, b, y \in \{0,1\}} \ket{\phi(a,x,b,y)} \otimes \ket{a,x,b,y}$.  Then for any operator $M$ supported on $\H_A$, 
\begin{equation}
\abs{ \Tr (M \otimes \identity \ketbra \psi \psi) - \Tr (M \otimes \identity \ketbra \phi \phi) } = O(\sqrt \epsilon) \norm{M}
 \enspace .
\end{equation}
\end{corollary}

\begin{proof}
By \corref{t:pullmeasurementstotheotherside}, there exist reflections $R_a$ supported on $\H_B$, namely $R_0 = \frac{1}{\sqrt 2}(X + Z) \otimes \identity$ and $R_1 = \frac{1}{\sqrt 2}(-X+Z) \otimes \identity$, such that $\norm{ P^A(a, x) \ket \psi - \tfrac12 ( \identity + (-1)^x R_a)_B \ket \psi } = \frac12 \norm{ ((\RAa{a})_A - (R_a)_B) \ket \psi } = O(\sqrt \epsilon)$.  Let $\ket{\phi'(a,x,b,y)} = P^B(b, y) \frac12 (\identity + (-1)^x R_a)_B \ket \psi$ and $\ket{\phi'} = \frac12 \sum_{a,x,b,y} \ket{\phi'(a,x,b,y)} \ket{a,x,b,y}$.  Then $\norm{ \ket{\phi(a,x,b,y)} - \ket{\phi'(a,x,b,y)} } = O(\sqrt \epsilon)$ and so $\norm{\ket\phi - \ket{\phi'}} = O(\sqrt \epsilon)$.  Now $\Tr M \ketbra \psi \psi = \Tr M \ketbra{\phi'}{\phi'}$, since $M$ is supported on $\H_A$ and the projections $P^B(b, y)$ and $\frac12(\identity + (-1)^x R_a)$ are supported on $\H_B$.  Therefore, by a triangle inequality, 
\begin{equation*}
\bigabs{ \Tr M \ketbra \psi \psi - \Tr M \ketbra \phi \phi }
\leq \abs{ \Tr M (\ketbra \phi \phi - \ketbra{\phi'}{\phi'}) }
\leq 2 \norm{M} \norm{\ket \phi - \ket{\phi'}}
= O(\sqrt \epsilon) \norm{M}
 \enspace .  \qedhere
\end{equation*}
\end{proof}

Similarly, in any structured CHSH game, the states resulting from different game outcomes are approximately related by single-qubit unitaries: 

\begin{corollary} \label{t:CHSHgameoutcomesrelatedbysinglequbitunitary}
For $a, a', \Delta \in \{0,1\}$, there exist single-qubit unitaries $U(a, a', \Delta)$ such that for sufficiently small $\epsilon > 0$ and any $\epsilon$-structured CHSH game where Alice plays using the ideal measurements from \tabref{f:optimalCHSHstrategy} on her first qubit, it holds that for all $a, x, a', x', b, y \in \{0,1\}$, 
\begin{equation} \label{e:CHSHgameoutcomesrelatedbysinglequbitunitary}
\Bignorm{ 
\ket{\psi(a,x,b,y)} 
- (U(a, a', x \oplus x') \otimes \identity)_A \ket{\psi(a',x',b,y)}
} = O(\sqrt \epsilon)
 \enspace .
\end{equation}
\end{corollary}

\begin{proof}
Consider an ideal CHSH game, with $\epsilon = 0$.  Then the unitaries $U(a, a', x \oplus x')$ can be read off \tabref{f:optimalCHSHstrategy}.  For example, if $a = a'$, the Pauli~$X$ operator switches the two outcomes of a Pauli~$Z$ measurement and vice versa.  Thus $U(a, a, 0) = \identity$, $U(0, 0, 1) = X$ and $U(1, 1, 1) = Z$.  The specific forms of the unitaries are not important.  

Now for an $\epsilon$-structured CHSH game with $\epsilon$ smaller than a certain positive constant, both denominators in Eq.~\eqnref{e:CHSHgameoutcomesrelatedbysinglequbitunitary} are nonzero, so the left-hand side is at least well-defined.  The claimed bound follows by applying the CHSH rigidity lemma to relate $\ket \psi$ to $\ket{\psi^*} \otimes \ket{\psi'}$ for $\ket{\psi^*}$ an EPR state and some~$\ket{\psi'}$, and several triangle inequalities.  
\end{proof}

Finally, the effect of the game can be duplicated by having only Alice make her measurements and then applying a unitary correction to Bob's qubit: 

\begin{corollary} \label{t:CHSHBobgameoutcomesrelatedbysinglequbitunitary}
For $a, b, \Delta \in \{0,1\}$, there exist single-qubit unitaries $V(a, b, \Delta)$ such that for sufficiently small $\epsilon > 0$ and any $\epsilon$-structured CHSH game where Bob plays using the ideal measurements from \tabref{f:optimalCHSHstrategy} on his first qubit, it holds that for all $a, x, a', x', b, y \in \{0,1\}$, 
\begin{equation} \label{e:CHSHBobgameoutcomesrelatedbysinglequbitunitary}
\Bignorm{ 
\ket{\psi(a,x,b,y)} 
- (V(a, b, x \oplus y) \otimes \identity)_B \frac{P^A(a, x) \ket \psi}{\norm{P^A(a, x) \ket \psi}} 
} = O(\sqrt \epsilon)
 \enspace .
\end{equation}
\end{corollary}

\begin{proof}
In an ideal CHSH game, there certainly exist unitaries $V(a,b,x \oplus y)$ relating, up to normalization, $P^A(a, x) \otimes P^B(b, y) \ket \psi$ to $P^A(a, x) \ket \psi$, since Alice's measurement collapses the shared EPR state leaving Bob's qubit in a tensor-product state.  The argument for an $\epsilon$-structured CHSH game is now the same as in the proof of \corref{t:CHSHgameoutcomesrelatedbysinglequbitunitary}.  
\end{proof}

\ifx\compilefullpaper\undefined  
\bibliographystyle{alpha-eprint}
\bibliography{q}

\end{document}
\fi

\ifx\compilefullpaper\undefined  
\documentclass[11pt]{article}

\begin{document}
\tableofcontents
\fi

\section[Sequential structured CHSH games have a tensor-product structure]{Sequential structured CHSH games have a tensor-product \\ structure} \label{s:sequentialstructuredCHSHgameshavetensorproductstructure}

In this section, we will argue that if two provers play $n$ sequential CHSH games in such a way that for every~$j$, game~$j$ is $\epsilon$-structured most of the time, then the provers must share a state close to~$n$ EPR states and most of the time their strategy for game~$j$ must be nearly equivalent to the ideal CHSH game strategy acting on the $j$th EPR state.  

The proof uses \lemref{t:eprlemma} repeatedly to simplify the provers' strategies and extract EPR states.  It is not enough to correct the games one at a time, in sequence.  Although the first game is $\epsilon$-structured, and therefore $O(\sqrt \epsilon)$-close to the ideal strategy on an EPR state, correcting this first game will introduce an $O(\sqrt \epsilon)$ error into all subsequent games.  This leaves the second game only $O(\sqrt \epsilon)$-structured, so correcting it introduces an $O(\epsilon^{1/4})$ error into subsequent games.  Thus with this na{\"i}ve argument, the error snowballs, both from the $\epsilon \rightarrow \sqrt \epsilon$ dependence of \lemref{t:eprlemma} and from the exponentially accumulating renormalization factors.  Other natural arguments face similar problems.  To obtain only a polynomial blowup in the error parameter~$\epsilon$, the argument is surprisingly involved, following the proof sketch in \secref{s:sequentialstructuredCHSHgameshavetensorproductstructuresketch}.

\subsection{Notation}

To make our claims precise, we begin with some notation for CHSH games played in sequence, one following the next, with no communication between games.  

\begin{definition}[Notation for sequentially repeated CHSH games] \label{t:CHSHserialnotation}
A strategy~$\S$ for two provers, Alice and Bob, to play $n$ sequential CHSH games consists of the provers' Hilbert spaces, their initial state and the reflections they use to play each game.  Fix the following notation: 

\begin{description}
\item[Transcripts:] 
Denote questions asked to Alice by $a_1, \ldots, a_n$, questions asked to Bob by $b_1, \ldots, b_n$, and possible answers by $x_1, \ldots, x_n$ and $y_1, \ldots, y_n$, respectively.  Write $\hA{j} = (a_1, \ldots, a_j, x_1, \ldots, x_j)$, $\hB{j} = (b_1, \ldots, b_j, y_1, \ldots, y_j)$ and $\h{j} = (\hA{j}, \hB{j})$, a full transcript for games $1$ through~$j$.  Similarly write $\h{j,k}$ and $\hX{j,k}$ for the full or partial transcripts for games~$j$ through~$k$, inclusive.  

\item[Hilbert spaces:]
Let $\H_A$ and $\H_B$ be the two provers' Hilbert spaces, and $\H_C$ any external space.  

\item[Reflection and projection operators:]
In game $j$, for questions $a_j$ and $b_j$, let $\RAjah{j}{a_j}{\hA{j-1}}$ and $\RBjah{j}{b_j}{\hB{j-1}}$ be the reflections specifying Alice and Bob's respective strategies.\footnote{Although the provers' reflections for game~$j$ may, without loss of generality, be taken to be independent of previous measurement outcomes, allowing such dependence will be convenient for specifying alternative strategies.}  Let $\PAjh{j}{\hA{j}} = \tfrac12 (\identity + (-1)^{x_j} \RAjah{j}{a_j}{\hA{j-1}})$ and $\PBjh{j}{\hB{j}} = \tfrac12 (\identity + (-1)^{y_j} \RBjah{j}{b_j}{\hB{j-1}})$.  For $\device \in \{A, B\}$ and $j \leq k$, let $\PXjh{j,k}{\hX{k}} = \PXjh{k}{\hX{k}} \cdots \PXjh{j+1}{\hX{j+1}} \PXjh{j}{\hX{j}}$.  Let $\PABjh{j, k}{\h{k}} = \PAjh{j, k}{\hA{k}} \otimes \PBjh{j, k}{\hB{k}}$.  

\item[Super-operators:]
For $j < k$ and partial transcript~$\h{j}$, define super-operators~$\EAjh{k}{\hA{j}}$ and~$\EBjh{k}{\hB{j}}$ by 
\begin{equation}\begin{split}
\EAjh{k}{\hA{j}} (\ketbra{\hA{j+1,k-1}}{\hA{j+1,k-1}} \otimes \rho) &= \frac12 \sum_{a_k, x_k} \ketbra{\hA{j+1,k}}{\hA{j+1,k}} \otimes \PAjh{k}{\hA{k}} \rho \PAjh{k}{\hA{k}} \\
\EBjh{k}{\hB{j}}(\ketbra{\hB{j+1,k-1}}{\hB{j+1,k-1}} \otimes \rho) &= \frac12 \sum_{b_k, y_k} \ketbra{\hB{j+1,k}}{\hB{j+1,k}} \otimes \PBjh{k}{\hB{k}} \rho \PBjh{k}{\hB{k}}
 \enspace .  
\end{split}\end{equation}
These super-operators capture the effects of Alice and Bob playing game~$k$, where games~$j+1$ to~$k-1$ of the transcript are stored in a separate register.  For $\ell \geq k$ and $\device \in \{A, B\}$, let $\EXjh{k, \ell}{\hX{j}} = \EXjh{\ell}{\hX{j}} \cdots \EXjh{k+1}{\hX{j}} \EXjh{k}{\hX{j}}$.  Let $\EABjh{k, \ell}{\h{j}} = \EAjh{k, \ell}{\hA{j}} \otimes \EBjh{k, \ell}{\hB{j}}$.  

\item[States:]
Let $\ket{\psione} \in \H_A \otimes \H_B \otimes \H_C$ be the provers' initial shared state, and let $\ket{\psijh{j}{\h{j-1}}}$ be the shared state at beginning of game $j$ conditioned on the transcript $\h{j-1}$; it is given by 
\begin{equation} \label{e:psij}
\ket{\psijh{j}{\h{j-1}}} 
= \frac{ \PABjh{1,j-1}{\h{j-1}} \ket{\psione} }{ \norm{ \PABjh{1,j-1}{\h{j-1}} \ket{\psione} } }
 \enspace .
\end{equation}
We adopt the convention that if the numerator above is $0$, then $\ket{\psijh{j}{\h{j-1}}} = 0$.  For notational brevity, we will commonly suppress the dependence on the transcript and write simply~$\psij{j}$.  

Let $\rhoone = \ketbra{\psione}{\psione}$, $\rhoj{j} = \EABj{1,j-1}(\rhoone) = \frac{1}{4^{j-1}} \sum_{\h{j-1}} \ketbra{\h{j-1}}{\h{j-1}} \otimes \PABjh{1,j-1}{\h{j-1}} \rhoone \PABjh{1,j-1}{\h{j-1}}^\dagger$, and $\rhoh{\h{j-1}} = \ketbra{\psijh{j}{\h{j-1}}}{\psijh{j}{\h{j-1}}}$.  

\item[Random variables:]
We use $A_j, B_j, X_j, Y_j$ to denote the random variables for the questions and answers in game~$j$, and $H_j$ for the transcript up through game~$j$.  $A_j$ and $B_j$ are distributed independently and uniformly at random.  Conditioned on the transcript~$\h{j-1}$ for the first~$j-1$ games and the questions $a_j$ and~$b_j$, $X_j$ and $Y_j$ are distributed according to $\Pr[X_j = x_j, Y_j = y_j \vert H_{j-1} = \h{j-1}, A_j = a_j, B_j = b_j] = \norm{ \PAjh{j}{\hA{j}} \otimes \PBjh{j}{\hB{j}} \ket{\psijh{j}{\h{j-1}}} }{}^2$.  Then $\rhoj{j} = \sum_{\h{j-1}} \Pr[H_{j-1} = \h{j-1}] \ketbra{\h{j-1}}{\h{j-1}} \otimes \ketbra{\psijh{j}{\h{j-1}}}{\psijh{j}{\h{j-1}}}$.  

\item[Other strategies:]
When considering multiple strategies, say $\S$ and $\tilde \S$, we will decorate the above notation to indicate the corresponding strategy.  For example, $\ket{\tilde \psi(\h{j-1})}$ denotes the shared state at the beginning of game~$j$ conditioned on the transcript~$\h{j-1}$, with play according to~$\tilde \S$.  
\end{description}
\end{definition}

Recall from \defref{t:chshgamestructuredef} that an $\epsilon$-structured CHSH game is one with a correlation value at least $2 \sqrt 2 - \epsilon$.  In our theorem, we will assume that most games the provers play are $\epsilon$-structured, in the following sense: 

\begin{definition}[Structured strategy] \label{t:structuredstrategydef}
A strategy~$\S$ for $n$ sequential CHSH games is \emph{$(\delta, \epsilon)$-structured} if for every~$j$, $\Pr[$game~$(j, H_{j-1})$ is $\epsilon$-structured$] \geq 1-\delta$.  $\S$ is \emph{$\epsilon$-structured} if it is $(\epsilon, \epsilon)$-structured.  
\end{definition}

Our goal is to show that the provers play close to an ideal strategy, defined as in \tabref{f:optimalCHSHstrategy} by: 

\begin{definition}[Notation for an ideal CHSH game] \label{t:idealgame}
Let $\ket{\psi^*} = \frac{1}{\sqrt 2}(\ket{00} + \ket{11}) \in \C^2 \otimes \C^2$.  
For $x \in \{0,1\}$, let $\ket{(0,x)_A} = \ket x$ and $\ket{(1,x)_A} = \frac{1}{\sqrt 2}(\ket 0 + (-1)^x \ket 1)$.  For $b \in \{0,1\}$, let $\ket{(b,0)_B} = \cos\frac\pi8 \ket 0 +(-1)^b \sin\frac\pi8 \ket 1$ and $\ket{(b,1)_B} = \sin\frac\pi8 \ket 0 -(-1)^b \cos\frac\pi8 \ket 1$.  For $\device \in \{A, B\}$ and $\alpha \in \{0,1\}$, let $\RXa{\alpha} = \ketbra{(\alpha,0)_\device}{(\alpha,0)_\device} - \ketbra{(\alpha,1)_\device}{(\alpha,1)_\device}$.  

Let $\hat A, \hat B, \hat X, \hat Y$ be random variables distributed according to the outcomes of the ideal CHSH game, $\Pr[(\hat A, \hat B, \hat X, \hat Y) = (a, b, x, y)] = \frac14 \bignorm{ \frac12(\identity + (-1)^x \RAa{a}) \otimes \frac12(\identity + (-1)^y \RBa{b}) \ket{\psi^*} }{}^2$, which equals $\frac14 \cos^2 \frac\pi8$ if $a b = x \oplus y$ and $\frac14 \sin^2 \frac\pi8$ otherwise.  
\end{definition}

\begin{definition}[Ideal strategy] \label{t:idealstrategy}
A strategy~$\S$ for $n$ sequential CHSH games is an \emph{ideal strategy} if there exist isometries $\UidealA: \H_A \hookrightarrow (\C^2)^{\otimes n} \otimes \H_A'$ and $\UidealB : \H_B \hookrightarrow (\C^2)^{\otimes n} \otimes \H_B'$ and a state $\ket{\psione'} \in \H_A' \otimes \H_B' \otimes \H_C$ such that for every $j$ and $\h{j-1}$, 
\begin{align}
\UidealA \otimes \UidealB \ket{\psione} &= \ket{\psi^*}^{\otimes n} \otimes \ket{\psione'} 
& \RXjah{j}{\alpha}{\hX{j-1}} &= \UidealX{}^\dagger (\RXa{\alpha})_j \UidealX 
 \enspace ,
\end{align}
where $(\RXa{\alpha})_j$ denotes the ideal operator~$\RXa{\alpha}$ from \defref{t:idealgame} acting on the $j$th qubit.  
\end{definition}

We will want to compare strategies in order to argue that the provers' actual strategy is ``nearby" a better-behaved strategy.  For this purpose, we introduce the following notion of strategy simulation: 

\begin{definition}[Strategy simulation] \label{t:simulationdef}
Let $\S$ and $\tilde \S$ be two strategies for playing $n$ sequential CHSH games.  For $\epsilon \geq 0$, we say that strategy $\tilde \S$ \emph{$\epsilon$-simulates} strategy~$\S$ if they both use the same Hilbert spaces and for all~$j$, 
\begin{equation}
\max_{\device \in \{A, B\}} \trnorm{\EXj{1,j}(\rhoone) -  \EXdecj{\tilde}{1,j}(\rhodecone{\tilde})} \leq \epsilon
 \enspace .
\end{equation}
Say that $\tilde \S$ \emph{weakly $\epsilon$-simulates}~$\S$ if only the weaker inequality $\trnorm{ \EABj{1,j}(\rhoone) - \EABdecj{\tilde}{1,j}(\rhodecone{\tilde}) } \leq 2 \epsilon$ holds.  
\end{definition}

It is also convenient to allow a basis change by local unitaries or local isometries: 

\begin{definition} \label{t:isometricextensiondef}
A strategy $\tilde \S$ is an \emph{isometric extension} of $\S$ if there exist isometries $\XX: \H_\device \hookrightarrow \tilde \H_\device$, for $\device \in \{A, B\}$, such that $\ket{\tilde \psione} = \XA \otimes \XB \ket{\psione}$ and $\XX \RXjah{j}{\alpha}{\hX{j-1}} = \RXdecjah{\tilde}{j}{\alpha}{\hX{j-1}} \XX$ always.  (Thus $\XX \EXj{j} = \EXdecj{\tilde}{j} \XX$ and $\XA \otimes \XB \ket{\psijh{j}{\h{j-1}}} = \ket{\psidecjh{\tilde}{j}{\h{j-1}}}$.)  
\end{definition}

\subsection{Main rigidity theorem and proof outline}

Our main theorem states that a structured strategy can be closely simulated by an ideal strategy: 

\begin{theorem}[Main rigidity theorem for sequential CHSH games] \label{t:sequentialCHSHgames}
There exists a constant~$\kappaEPR$ such that for any $\epsilon$-structured strategy~$\S$ for $n$ sequential CHSH games, 
letting $\zeta = \kappaEPR n^{\kappaEPR} \epsilon^{1/\kappaEPR}$, there exists an ideal strategy $\hat \S$ that $\zeta$-simulates an isometric extension of~$\S$.  
\end{theorem}

The proof of \thmref{t:sequentialCHSHgames} is sufficiently involved that an outline should be useful.  See \figref{f:proofsketch}.  The first step of the proof is to replace the structured strategy $\S$ with one in which the provers play every game using the ideal CHSH game operators on some qubit, up to a local change in basis.  

\begin{definition}[Single-qubit ideal strategy] \label{t:singlequbitidealstrategy}
A strategy~$\S$ is a \emph{single-qubit ideal strategy} if there exist unitaries $\UsingleXjh{j}{\hX{j-1}} : \H_\device \overset{\cong}{\rightarrow} \C^2 \otimes \H_\device'$ such that always 
\begin{equation} \label{e:singlequbitidealstrategy}
\RXjah{j}{\alpha}{\hX{j-1}} = \UsingleXjh{j}{\hX{j-1}}^\dagger (\RXa{\alpha} \otimes \identity) \UsingleXjh{j}{\hX{j-1}} 
 \enspace .
\end{equation}
That is, each prover's reflections for game $(j, \hX{j-1})$ are equivalent up to local unitaries to the ideal CHSH game reflections of \defref{t:idealgame}, but the qubits used need not be in tensor product.  
\end{definition}

\begin{theorem} \label{t:singlequbitidealstrategysimulation}
There exists a constant~$\kappa$ such that if $\S$ is an $\epsilon$-structured strategy for~$n$ sequential CHSH games, then there is a single-qubit ideal strategy~$\tilde \S$ that $\kappa n^\kappa \epsilon^{1/\kappa}$-simulates an isometric extension of~$\S$.  
\end{theorem}

Next, we find a nearby strategy in which the qubits for successive games are in tensor product.  

\begin{definition}[Multi-qubit ideal strategy] \label{t:multiqubitidealstrategy}
A strategy~$\S$ is a \emph{multi-qubit ideal strategy} if there is a unitary isomorphism $\XmultiX : \H_\device \overset{\cong}{\rightarrow} (\C^2)^{\otimes n} \otimes \H_\device'$ under which for unitaries $\UmultiXjh{j}{\hX{j-1}} \in \L((\C^2)^{\otimes (n-j+1)} \otimes \H_\device')$ such that 
\begin{equation} \label{e:multiqubitidealstrategy}
\RXjah{j}{\alpha}{\hX{j-1}} = 
\XmultiX{}^\dagger
\UmultiXj{1}{}^\dagger \ldots \big(\identity_{(\C^2)^{\otimes (j-1)}} \otimes \UmultiXjh{j}{\hX{j-1}}^\dagger\big) (\RXa{\alpha})_j \big(\identity_{(\C^2)^{\otimes (j-1)}} \otimes \UmultiXjh{j}{\hX{j-1}}\big) \ldots \UmultiXj{1}
\XmultiX
 \enspace .
\end{equation}
That is, $\S$ is a single-qubit ideal strategy in which the qubits used in each game must lie in tensor product with the qubits from previous games.  
\end{definition}

\begin{theorem} \label{t:multiqubitidealstrategysimulation}
There exists a constant~$\kappa$ such that if $\S$ is an $\epsilon$-structured single-qubit ideal strategy for~$n$ sequential CHSH games, then there is a multi-qubit ideal strategy~$\tilde \S$ that $\kappa n^\kappa \epsilon^{1/\kappa}$-simulates an isometric extension of~$\S$.  
\end{theorem}

The last major step in the proof of \thmref{t:sequentialCHSHgames} is to argue that the qubit locations cannot depend significantly on the local transcripts, and therefore simulate a multi-qubit ideal strategy with an ideal strategy.  

\begin{theorem} \label{t:globalgluing}
There exists a constant~$\kappa$ such that if~$\S$ is a $(\delta, \epsilon)$-structured multi-qubit ideal strategy for~$n$ sequential CHSH games, then there exists a transcript~$\hdec{\hat}{n}$ such that~$\S$ is $\kappa n^\kappa (\delta + \epsilon)^{1/\kappa}$-simulated by the ideal strategy~$\hat \S$ that uses the qubits defined by $\hdec{\hat}{n}$ in~$\S$.  That is, using the notation of Definitions~\ref{t:idealstrategy} and~\ref{t:multiqubitidealstrategy}, $\hat \S$ is defined by 
\begin{equation}
\UidealXdec{\hat} = \big(\identity_{(\C^2)^{\otimes (n-1)}} \otimes \UmultiXjh{n}{\hX{n-1}}\big) \ldots \UmultiXj{1} \XmultiX
 \enspace .
\end{equation}
\end{theorem}

The proofs of Theorems~\ref{t:singlequbitidealstrategysimulation}, \ref{t:multiqubitidealstrategysimulation} and~\ref{t:globalgluing} are given, respectively, in Sections~\ref{s:singlequbitidealstrategysimulation}, \ref{s:historydependenttensorproductstrategy} and~\ref{s:gluing} below.  To chain these theorems together, we will use: 

\begin{lemma} \label{t:simulationpreservesstructure}
Let $\S$ be a $(\delta, \epsilon)$-structured strategy for $n$ sequential CHSH games.  If~$\tilde \S$ is a strategy that weakly $\eta$-simulates~$\S$, then $\tilde \S$ is $(\delta + 2 \sqrt{\eta}, \epsilon + 16 \sqrt{\eta})$-structured.  
\end{lemma}

\begin{proof}
By definition of weak $\eta$-simulation and \lemref{t:blockdiagonaltracedistance}, $d_{TV}(H_{j-1}, \tilde H_{j-1}) \leq \trnorm{\rhoj{j} - \rhodecj{\tilde}{j}} / 2 \leq \eta$ for all~$j$.  In particular, therefore $\Pr[\text{game $(j, \tilde H_{j-1})$ is $\epsilon$-structured in $\S$}] \geq 1 - \delta - \eta$.  

For random variables $(A, B)$ and $(A', B')$ in the same space, $\sum_a \Pr[A = a] d_{TV}(B \vert A=a, B' \vert A'=a) \leq 2 d_{TV}\big((A,B), (A',B')\big)$.  Applying this and a Markov inequality to $(\tilde H_{j-1}, \tilde H_j)$ and $(H_{j-1}, H_j)$, we get that with at most a $\sqrt{\eta}$ probability over $\tilde H_{j-1}$ can the total variation distance between the outcomes of playing strategy~$\S$ and of playing strategy~$\tilde \S$ in game $(j, \tilde H_{j-1})$ be greater than~$2 \sqrt{\eta}$.  By \defref{t:chshgamestructuredef} for structure and a union bound, therefore $\Pr[$game $(j, \tilde H_{j-1})$ is $\epsilon$-structured in $\S$ and $(\epsilon + 16 \sqrt{\eta})$-structured in $\tilde \S] \geq 1 - \delta - \eta - \sqrt{\eta}$.  
\end{proof}

\thmref{t:sequentialCHSHgames} therefore follows from Theorems~\ref{t:singlequbitidealstrategysimulation}, \ref{t:multiqubitidealstrategysimulation} and~\ref{t:globalgluing}.  

\smallskip

We begin the proofs of the latter theorems by reducing to the case where the Hilbert spaces $\H_A$ and $\H_B$ are finite dimensional.  This is necessary to ensure that the strategies given by the CHSH rigidity lemma applied to games $2$ through~$n$ depend only on the local transcript~$\hX{}$ and not on the full transcript~$\h{}$.  Although the CHSH rigidity lemma itself holds even for infinite-dimensional Hilbert spaces, the dimension-truncation argument it uses depends on the underlying state and therefore potentially on the full transcript.  

\begin{lemma} \label{t:reductiontofinitedimensions}
Assuming that Theorem~\ref{t:sequentialCHSHgames}, \ref{t:singlequbitidealstrategysimulation}, \ref{t:multiqubitidealstrategysimulation} and~\ref{t:globalgluing} hold whenever the provers' Hilbert spaces $\H_A$ and $\H_B$ are finite dimensional, the theorems also hold in general.  
\end{lemma}

\begin{proof}
We claim that for any strategy $\S$ on possibly infinite-dimensional Hilbert spaces $\H_A$ and~$\H_B$, and for any parameter $\delta > 0$, there exists a strategy $\tilde \S$ that $\delta$-simulates $\S$, such that there exist finite-dimensional subspaces $\H_A' \subseteq \H_A$ and $\H_B' \subseteq \H_B$ that are closed under all of the operators $\tilde R_{j,a}$ and $\tilde R_{j,b}$ and such that $\ket{\tilde \psi}$ is entirely supported on $\H_A' \otimes \H_B' \otimes \H_C$.  

Provided this claim holds, \thmref{t:sequentialCHSHgames} can be applied to $\tilde \S$ restricted to $\H_A'$ and $\H_B'$, yielding an ideal strategy $\hat \S$ that $\epsilon$-simulates an isometric extension of $\tilde \S$.  \thmref{t:sequentialCHSHgames} follows by extending the isometries to all of $\H_A$ and~$\H_B$.  The other theorems follow similarly.  

To establish the claim, assume that in fact one or both of $\H_A$ and $\H_B$ are infinite dimensional.  Let $\H_A^{(0)} \subseteq \H_A$ and~$\H_B^{(0)} \subseteq \H_B$ be finite-dimensional subspaces such that $\rho^{(0)}$, the renormalized projection of $\ketbra \psi \psi$ to $\H_A^{(0)} \otimes \H_B^{(0)} \otimes \H_C$, is $\delta$-close to $\ketbra \psi \psi$.  Then Alice and Bob's actual strategy is $\delta$-simulated by the same set of measurements applied to~$\rho^{(0)} \in \L(\H_A \otimes \H_B \otimes \H_C)$.  

The proof is not yet complete, since $\H_A^{(0)}$ and $\H_B^{(0)}$ will generally not be closed under the provers' operators $R_{j,a}$ and $R_{j,b}'$.  We fix this one operator at a time.  Order Alice's reflection operators as $S_1, S_2, \ldots, S_m$, such that reflections for earlier games come before reflections for later games.  For~$k$ from $1$ to~$m$, let $\H_A^{(k)}$ be the closure of $\H_A^{(k-1)}$ under $S_k$ and extend each $S_j$ for $j < k$ by the identity on $(\H_A^{(k-1)})^\perp$.  By this construction, for any vector $\ket v \in \H_A^{(0)}$ and any $a_1, \ldots, a_m \in \{0,1\}$, the state $S_m^{a_m} \ldots S_2^{a_2} S_1^{a_1} \ket v$ lies in~$\H_A^{(m)}$.  Therefore there is no loss in truncating $\H_A$ to $\H_A^{(m)}$.  Finally apply the analogous procedure for Bob to obtain space $\H_B^{(m)}$ that is finite dimensional and closed under each of Bob's operators.  
\end{proof}

Henceforth we will always assume that $\H_A$ and $\H_B$ are finite dimensional.

\subsection{How to play for Bob without measuring \texorpdfstring{$\H_B$}{H\_B}} \label{s:howtoplayforBob}

Before continuing the proof of \thmref{t:sequentialCHSHgames}, it will be useful to argue that a structured strategy can be closely simulated by alternative protocols in which only one of the two provers makes measurements.  

Since for any matrix $M \in \L(\C^2)$, $(M \otimes \identity) (\ket{00} + \ket{11}) = (\identity \otimes M^T) (\ket{00} + \ket{11})$, operations on one half of an EPR state can equivalently be performed on the other half.  This allows us to show that a structured protocol for playing sequential CHSH games can be simulated by either of two hypothetical protocols in which Alice receives Bob's questions and answers for him.  Studying these simulations has a key conceptual advantage over studying the actual protocol: if when given~$b_j$ only, Alice can play for Bob in game~$j$, then Bob's strategy for game~$j$ intuitively cannot depend on the outcomes $\hB{j-1}$ of the prior games.  There are also technical advantages.  For example, one of our main concerns is that the qubits Bob uses in two successive CHSH games might overlap.  However, if we switch Bob's measurements for the second game over to Alice's side, then, since $\H_A$ is in tensor product with $\H_B$, they necessarily act on qubits in tensor product with Bob's qubits for the first game.

\subsubsection{First hypothetical protocol: Alice guesses Bob's measurement outcomes}

In the first hypothetical protocol, Alice plays her games as usual, but also receives Bob's questions~$b_j$.  Alice guesses Bob's answers and gives them to Bob, who merely applies certain unitary corrections.  We will argue that this protocol generates states nearly indistinguishable from the results of a structured strategy for $n$ sequential CHSH games.  

This alternative protocol is only hypothetical, since it requires information and communication not allowed in sequential CHSH games.  However, it is technically simpler to analyze since only one of the two provers makes any measurements.  

Before defining the alternative protocol, it will be useful to define the qubit used in game~$(j, \hX{j-1})$ for each prover~$\device$: 

\begin{definition}[Game qubits] \label{t:gamequbitsdef}
Let~$\S$ be a strategy for $n$ sequential CHSH games.  For $\device \in \{A, B\}$ and for each partial transcript~$\hX{j-1}$, let~$\beta$ index a complete, irreducible set of orthogonal one- or two-dimensional subspaces of $\H_\device$ that are invariant under $\RXjah{j}{\alpha}{\hX{j-1}}$ for $\alpha \in \{0,1\}$.  Let $\UXjh{j}{\hX{j-1}} : \H_\device \hookrightarrow \C^2 \otimes \H_\device'$ be an isometry such that $\UXjh{j}{\hX{j-1}}^\dagger (\identity \otimes \ketbra \beta \beta) \UXjh{j}{\hX{j-1}}$ is a projection onto subspace~$\beta$, chosen so that for dihedral angles $\theta_\beta(\hX{j-1}) \in [0, \pi/2]$, 
\begin{equation}\begin{split} \label{e:gamequbits}
\RXjah{j}{0}{\hX{j-1}} &= \UXjh{j}{\hX{j-1}}^\dagger (\RXa{0} \otimes \identity) \UXjh{j}{\hX{j-1}} \\
\RAjah{j}{1}{\hA{j-1}} &= \UAjh{j}{\hA{j-1}}^\dagger \Big(\sum_\beta \smatrx{\cos 2 \theta_\beta & \sin 2 \theta_\beta \\ \sin 2 \theta_\beta & -\cos 2 \theta_\beta} \otimes \ketbra{\beta}{\beta}\Big) \UAjh{j}{\hA{j-1}} \\
\RBjah{j}{1}{\hB{j-1}} &= \UBjh{j}{\hB{j-1}}^\dagger \Big(\sum_\beta G^\dagger H \smatrx{\cos 2 \theta_\beta & \sin 2 \theta_\beta \\ \sin 2 \theta_\beta & -\cos 2 \theta_\beta} H G \otimes \ketbra{\beta}{\beta}\Big) \UBjh{j}{\hB{j-1}}
 \enspace .
\end{split}\end{equation}
We refer to the first register in the codomain of $\UXjh{j}{\hX{j-1}}$ as the ``qubit used in game~$(j, \hX{j-1})$."  
\end{definition}

Such isometries exist, provided $\H_A$ and $\H_B$ are finite dimensional, by Jordan's Lemma (\lemref{t:jordanslemma}) but they are generally not unique.  Eq.~\eqnref{e:gamequbits} simply specifies a convenient basis for the two-dimensional subspace $\Range(\identity \otimes \ketbra \beta \beta)$.  However, up to this freedom in choosing the subspaces, and up to the choice of basis within each subspace, the isometries $\UAjh{j}{\hA{j-1}}$ and~$\UBjh{j}{\hB{j-1}}$ are the same isometries as promised by the CHSH rigidity lemma, \lemref{t:eprlemma}, for game~$(j, \h{j-1})$.  

The maps $\UXjh{j}{\hX{j-1}}$ are generally isometries and not unitaries because there may be some one-dimensional invariant subspaces~$\beta$ and it is notationally inconvenient to have a separate term for this case in Eq.~\eqnref{e:gamequbits}.  It is not difficult to argue, though: 

\begin{proposition} \label{t:wlogunitaryqubits}
Provided $\H_A$ and $\H_B$ are finite dimensional, there exists an isometric extension of the provers' strategy~$\S$ into finite-dimensional spaces such that the operators $\UXjh{j}{\hX{j-1}}$ are all simultaneously unitary.  
\end{proposition}

\begin{proof}
Choose an arbitrary order for all the partial transcripts~$\hX{j-1}$, $j \in [n]$.  One transcript at a time, add dimensions to $\H_\device$ so that $\UXjh{j}{\hX{j-1}}$ is unitary.  The concern is that this could break previously considered $\UXj{j}$ operators.  After the first extension $\dim(\H_\device)$ is even, however, so there are always an even number of one-dimensional invariant subspaces~$\beta$, so always an even number of dimensions are added to~$\H_\device$.  These dimensions can be paired up arbitrarily in the previous~$\UXj{j}$ operators, so that they still each unitarily expose a qubit.  
\end{proof}

If~$\S$ is a single-qubit ideal strategy, then the $\UXjh{j}{\hX{j-1}}$ satisfy Eq.~\eqnref{e:singlequbitidealstrategy} in \defref{t:singlequbitidealstrategy}, and in this case the prover's reflections for game~$(j, \hX{j-1})$ are only supported on that one qubit.  In general, though, the angles $\theta_\beta$ will depend on~$\beta$ and there does not exist a basis change under which the prover's reflections for a game are supported on just one qubit.  

Now we are ready to define the super-operators for the alternative protocol mentioned above.  

\begin{definition} \label{t:ABunitarydef}
Let $\S$ be a strategy such that the operators $\UBjh{j}{\hB{j-1}}$ are unitary.  For $a, b, \Delta \in \{0,1\}$ and a partial transcript~$\hB{j-1}$, define unitaries $\ABunitary(a, b, \Delta) \in \L(\C^2)$ and $\ABunitaryBj{j}(\hB{j-1}, a, b, \Delta)$ by  
\begin{equation}\begin{split}
\ABunitary(a, b, \Delta) &= \sum_{x \in \{0,1\}} \ketbra{(b, x \oplus \Delta)_B}{(a, x)_A} \\
\ABunitaryBj{j}(\hB{j-1}, a, b, \Delta) &= \UBjh{j}{\hB{j-1}}^\dagger \big(\ABunitary(a, b, \Delta) \otimes \identity\big) \UBjh{j}{\hB{j-1}}
 \enspace ,
\end{split}\end{equation}
i.e., $\ABunitaryBj{j}(\hB{j-1}, a, b, \Delta)$ is $\ABunitary(a, b, \Delta)$ acting on Bob's qubit for game~$(j, \hB{j-1})$.  

Recall the distribution of the ideal CHSH game outcomes~$(\hat A, \hat B, \hat X, \hat Y)$ from \defref{t:idealgame}.  Define a super-operator~$\Bguessj{j}$ by 
\begin{multline}
\Bguessj{j}(\ketbra{\h{j-1}, a_j, x_j}{\h{j-1}, a_j, x_j} \otimes \rho) \\
= \sum_{b_j, y_j} \bigl(\begin{aligned}
\Pr[(\hat B, \hat Y) = (b_j, y_j) \vert (\hat A, \hat X) = (a_j, x_j)] \,
\ketbra{\h{j}}{\h{j}} \otimes (\identity \otimes \ABunitaryBj{j}) \rho (\identity \otimes \ABunitaryBj{j}{}^\dagger)
\end{aligned}\bigr)
 \enspace ,
\end{multline}
where $\ABunitaryBj{j} = \ABunitaryBj{j}(\hB{j-1}, a_j, b_j, x_j \oplus y_j)$.  Let $\ABguessj{j} = \Bguessj{j} \circ \EAj{j}$, and for $\ell \geq k$, let $\Bguessj{k,\ell} = \Bguessj{\ell} \cdots \Bguessj{k}$ and $\ABguessj{k,\ell} = \ABguessj{\ell} \cdots \ABguessj{k}$.  Thus, 
\begin{equation}
\ABguessj{1,j}(\rho)
= \frac{1}{2^j} \sum_{\h{j}} \biggl(\begin{aligned}
{\begingroup\textstyle \prod_{k=1}^{j} \endgroup} \Pr[(\hat B, \hat Y) = (b_j, y_j) \vert (\hat A, \hat X) = (a_j, x_j)] \\ \ketbra{\h{j}}{\h{j}} 
\otimes \big(\PAjh{1,j}{\hA{j}} \otimes \ABunitaryBj{1,j}(\h{j})\big) \,\rho\, \big(\PAjh{1,j}{\hA{j}} \otimes \ABunitaryBj{1,j}(\h{j})^\dagger\big)
\end{aligned}\biggr)
 \enspace ,
\end{equation}
where $\ABunitaryBj{1,j}(\h{j}) = \ABunitaryBj{j}(\hB{j-1}, a_j, b_j, x_j \oplus y_j) \cdots \ABunitaryBj{2}(\hB{1}, a_2, b_2, x_2 \oplus y_2) \ABunitaryBj{1}(a_1, b_1, x_1 \oplus y_1)$.  This non-local super-operator has the effect of measuring Alice's qubits, guessing Bob's answers according to the appropriate ideal conditional distribution, and then applying a unitary to correct Bob's qubits.  
\end{definition}

In the above definition, it is worth remarking that the super-operator $\ABguessj{j}$ guesses Bob's measurement result $y_j$ according to its distribution in the ideal CHSH game, and not according to its distribution in~$\S$.  This type of approximation is inevitable because a super-operator that only measures Alice's qubits cannot precisely capture the transcript distribution's dependence on Bob's measurement outcomes for previous games.  

The super-operators $\Bguessj{j}$ are useful because they do not affect the trace distance between matrices that are block-diagonal in the computational basis for transcripts: 

\begin{claim} \label{t:Bguesstracedistance}
For any~$\sigma = \sum_h \ketbra h h \otimes \sigma_h$, $\trnorm{ \Bguessj{j}(\sigma) } = \trnorm{ \sigma }$.  
\end{claim}

\begin{proof}
$\Bguessj{j}$ can be split into two super-operators: the first adds $\ketbra{b_j, y_j}{b_j, y_j}$ to the transcript register, weighted by a certain probability; and the second applies a controlled isometry to the state register.  Neither operation changes the trace norm.  
\end{proof}

Observe that if $\S$ is an ideal strategy, projecting Alice's half of an EPR state onto $\ket{(a,x)_A}$ also collapses Bob's half to the same state~$\ket{(a,x)_A}$.  $\ABunitary(a, b, x \oplus y)$ corrects this to $\ket{(b,y)_B}$.  Hence $\rhoj{j+1} = \EABj{1,j}(\rhoone) = \ABguessj{1,j}(\rhoone)$.  We next show that if most games are structured, then $\rhoj{j+1}$ is close to $\ABguessj{1,j}(\rhoone)$ in trace distance: 

\begin{lemma} \label{t:allCHSHBobgameoutcomesrelatedbysinglequbitunitary}
Let $\S$ be a $(\delta, \epsilon)$-structured strategy.  Then for all~$j$, 
\begin{equation} \label{e:allCHSHBobgameoutcomesrelatedbysinglequbitunitary}
\bigtrnorm{
\EABj{1,j}(\rhoone) - \ABguessj{1,j}(\rhoone)
} \leq j \big( 2 \delta + O(\sqrt \epsilon) \big)
 \enspace .
\end{equation}
In particular, letting $\ket{\psijh{j}{\hA{j}}} =  \PAjh{1,j}{\hA{j}} \ket{\psione} / \norm{ \PAjh{1,j}{\hA{j}} \ket{\psione} }$ and $\density(\ket a) = \ketbra a a$, 
\begin{equation} \label{e:allCHSHBobgameoutcomesrelatedbysinglequbitunitaryExpectation}
\Ex\!\big[ \bigtrnorm{ \density\big( \ket{\psijh{j+1}{H_j}} \big) - \density\big( \ABunitaryBj{1,j}(H_j) \ket{\psiAh{H^A_j}} \big) } \big] \leq 2 n (2\delta + O(\sqrt \epsilon))
 \enspace .
\end{equation}
\end{lemma}

\begin{proof}
Using a hybrid argument, expand the difference $\rhoj{j+1} - \ABguessj{1,j}(\rhoone)$ as 
\begin{align*}
\rhoj{j+1} - \ABguessj{1,j}(\rhoone) 
&= \big( \rhoj{j+1} - \ABguessj{j}(\rhoj{j}) \big) + \ABguessj{j} \big( \rhoj{j} - \ABguessj{j-1}(\rhoj{j-1}) \big) + \cdots + \ABguessj{2,j} \big( \rhoj{2} - \ABguessj{1}(\rhoone) \big)
 \enspace .
\end{align*}
By a triangle inequality, and since applying a super-operator cannot increase the trace distance, $\bigtrnorm{ \rhoj{j+1} - \ABguessj{1,j}(\rhoone) } \leq j \max_{k \in [j]} \bigtrnorm{ \rhoj{k+1} - \ABguessj{k}(\rhoj{k}) }$.  

Next, expand $\bigtrnorm{ \rhoj{k+1} - \ABguessj{k}(\rhoj{k}) }$ as 
\begin{align*}
\bigtrnorm{ \rhoj{k+1} - \ABguessj{k}(\rhoj{k}) }
&= \sum_{\h{k-1}} \Pr[H_{k-1} = \h{k-1}] \bigtrnorm{ \EAj{k} \EBj{k}(\rhoh{\h{k-1}}) - \ABguessj{k}(\rhoh{\h{k-1}}) }
 \enspace .
\end{align*}
If game~$(k, \h{k-1})$ is $\epsilon$-structured, then the total variation distance between the distribution of outcomes $(a_k, b_k, x_k, y_k)$ generated by $\EAj{k} \EBj{k}$ and the distribution generated by $\ABguessj{k}$ is at most~$O(\epsilon)$.  Moreover, by \corref{t:CHSHBobgameoutcomesrelatedbysinglequbitunitary}, the resulting states are within $O(\sqrt \epsilon)$ in trace distance of each other.  Therefore $\bigtrnorm{ \EAj{k} \EBj{k}(\rhoh{\h{k-1}}) - \ABguessj{k}(\rhoh{\h{k-1}}) } = O(\sqrt \epsilon)$.  On the other hand, the total contribution from terms for games~$(k, \h{k-1})$ that are not $\epsilon$-structured is at most $2 \delta$.  This implies $\bigtrnorm{ \rhoj{k+1} - \ABguessj{k}(\rhoj{k}) } \leq 2 \delta + O(\sqrt \epsilon)$, and yields Eq.~\eqnref{e:allCHSHBobgameoutcomesrelatedbysinglequbitunitary}.  

Applying \lemref{t:blockdiagonaltracedistance} to Eq.~\eqnref{e:allCHSHBobgameoutcomesrelatedbysinglequbitunitary} gives Eq.~\eqnref{e:allCHSHBobgameoutcomesrelatedbysinglequbitunitaryExpectation}.  
\end{proof}

We will use \lemref{t:allCHSHBobgameoutcomesrelatedbysinglequbitunitary} four times below, in the proofs of single- and multi-qubit ideal strategy simulation, and local and global gluing (Theorems~\ref{t:singlequbitidealstrategysimulation}, \ref{t:multiqubitidealstrategysimulation},  \ref{t:localgluing} and~\ref{t:globalgluing}).  It allows~for turning weak simulation statements, i.e., bounds on $\trnorm{\EABj{1,j}(\rhoone) - \EABdecj{\tilde}{1,j}(\rhodecone{\tilde})}$, into simulation statements, i.e., bounds on $\trnorm{\EXj{1,j}(\rhoone) - \EXdecj{\tilde}{1,j}(\rhodecone{\tilde})}$: 

\begin{corollary} \label{t:weaktostrongsimulation}
There exists a contant~$\kappa$ such that if $\S = (\ket{\psione}, \{\EAj{j}\}, \{\EBj{j}\})$ is an $\epsilon$-structured strategy that is weakly $\delta$-simulated by $\tilde \S = (\ket{\psione}, \{\EAdecj{\tilde}{j}\}, \{\EBj{j}\})$, a strategy differing only in Alice's reflection operators, then $\tilde \S$ also $\kappa n^\kappa (\delta + \epsilon)^{1/\kappa}$-simulates~$\S$.  
\end{corollary}

\begin{proof}
The idea is that \lemref{t:allCHSHBobgameoutcomesrelatedbysinglequbitunitary} allows for replacing Bob's measurement super-operators with an isometry.  Since the isometry is the same for $\S$ as for~$\tilde \S$, it can be removed without affecting the trace distance (\claimref{t:Bguesstracedistance}), and so $\tilde \S$ simulates~$\S$.  Formally, we have 
\begin{align*}
\trnorm{\EAj{1,j}(\rhoone) - \EAdecj{\tilde}{1,j}(\rhoone)}
&= \trnorm{ \Bguessj{1,j} \EAj{1,j}(\rhoone) - \Bguessj{1,j} \EAdecj{\tilde}{1,j} } && \text{by \claimref{t:Bguesstracedistance}} \\
&= \trnorm{ \ABguessj{1,j}(\rhoone) - \ABguessdecj{\tilde}{1,j}(\rhoone) } && \text{since $\Bguessj{1,j} = \Bguessdecj{\tilde}{1,j}$} \\ 
&\leq \trnorm{ \ABguessj{1,j}(\rhoone) - \EABj{1,j}(\rhoone) } + \trnorm{ \ABguessdecj{\tilde}{1,j}(\rhoone) - \EABdecj{\tilde}{1,j}(\rhoone) } \\ &\quad+ \trnorm{ \EABj{1,j}(\rhoone) - \EABdecj{\tilde}{1,j}(\rhoone) }
 \enspace .
\end{align*}
By \lemref{t:simulationpreservesstructure}, $\tilde \S$ is $(\epsilon + 16 \sqrt \delta)$-structured, so \lemref{t:allCHSHBobgameoutcomesrelatedbysinglequbitunitary} gives bounds for $\trnorm{ \ABguessj{1,j}(\rhoone) - \EABj{1,j}(\rhoone) }$ and $\trnorm{ \ABguessdecj{\tilde}{1,j}(\rhoone) - \EABdecj{\tilde}{1,j}(\rhoone) }$.  Thus $\trnorm{\EAj{1,j}(\rhoone) - \EAdecj{\tilde}{1,j}(\rhoone)} \leq \kappa n^\kappa (\delta + \epsilon)^{1/\kappa}$ for a certain fixed constant~$\kappa$.  Of course, $\trnorm{\EBj{1,j}(\rhoone) - \EBdecj{\tilde}{1,j}(\rhoone)} = 0$.  
\end{proof}

\subsubsection{Second hypothetical protocol: Alice measures for Bob}

The super-operators $\ABguessj{j}$ correspond to a hypothetical protocol in which Alice applies her measurements and then guesses Bob's measurement outcomes, after which Bob applies a unitary correction.  A similar idea is that Alice could herself first apply Bob's measurement operators to her own qubits, collapsing both provers' qubits, and then she could either apply her own measurement operators or, equally well, simply guess a unitary correction.  We next show that this second hypothetical protocol also accurately simulates the actual protocol.  The applications of this claim (in Theorems~\ref{t:multiqubitidealstrategysimulation} and~\ref{t:localgluing}) are to show that Bob's measurement super-operators do not depend much on his local transcript---since Alice can apply them herself without even knowing his transcript.  This is not a purpose that \lemref{t:allCHSHBobgameoutcomesrelatedbysinglequbitunitary} can serve, since there the prover who measures is allowed arbitrary dependence on her local transcript.  Nor does the claim replace \lemref{t:allCHSHBobgameoutcomesrelatedbysinglequbitunitary}.  For our applications, it will be convenient to state the claim with the two provers switched from the above description, i.e., with Bob measuring for Alice.  

\begin{definition} \label{t:canpullovereverythingdef}
Let $\S$ be a strategy such that the operators $\UBjh{k}{\hB{k-1}}$ are unitary.  Let $\PAmeasBjh{k}{\hB{k-1}, a_k, x_k}$ be the projection of \emph{Bob's} qubit $(k, \hB{k-1})$ according to \emph{Alice's} ideal reflection $\RAa{a_k}$, and let $\AmeasBjh{k}{\hB{j}}$ be the corresponding measurement super-operator.  That is, letting $\density(\ket a) = \ketbra a a$, 
\begin{equation}\begin{split}
\PAmeasBjh{k}{\hB{k-1}, a_k, x_k} 
&= \UBjh{k}{\hB{k-1}}^\dagger \big( \tfrac12 (\identity + (-1)^{x_k} \RAa{a_k}) \otimes \identity  \big) \UBjh{k}{\hB{k-1}} \\
\AmeasBjh{k}{\hB{j}}(\density(\ket{\hB{j+1,k-1}}) \otimes \rho) 
&= \frac12 \sum_{a_k, x_k} \big(\begin{aligned}
\density(\ket{\hB{j+1,k-1}, a_k, x_k}) \otimes \PAmeasBj{k} \rho \, \PAmeasBj{k}
\end{aligned}\big) 
 \enspace .
\end{split}\end{equation}
Let $\AmeasBBjh{k}{\hB{j}} = \EBjh{k}{\hB{j}} \AmeasBjh{k}{\hB{j}}$ and $\AmeasBBjh{k, \ell}{\hB{j}} = \AmeasBBjh{\ell}{\hB{j}} \cdots \AmeasBBjh{k}{\hB{j}}$.  These super-operators capture the effects of playing Alice's ideal reflections on Bob's qubits before making Bob's own measurements.  
\end{definition}

Observe that if $\S$ is an ideal strategy, then since a measurement on one half of an EPR state can be made equivalently on the other half, $\rhoj{j+1} = \EABj{1,j}(\rhoone) = \AmeasBBj{1,j}(\rhoone)$.  If most games are $\epsilon$-structured, then $\rhoj{j+1}$ is close to $\AmeasBBj{1,j}(\rhoone)$ in trace distance: 

\begin{lemma} \label{t:canpullovereverything}
Let $\S$ be a strategy and $\h{\ell}$ a partial transcript such that for every $j > \ell$, $\Pr[$game $(j, H_{j-1})$ is $\epsilon$-structured $\vert H_\ell = \h{\ell}] \geq 1-\delta$.  Then for all~$k > j > \ell$, letting $\rhojh{j}{\h{\ell}} = \EABjh{\ell+1,j-1}{\h{\ell}}(\rhoh{\h{\ell}})$, $\bigtrnorm{ \EAjh{j}{\hA{\ell}}(\rhojh{j}{\h{\ell}}) - \AmeasBjh{j}{\hB{\ell}}(\rhojh{j}{\h{\ell}}) } \leq O(\sqrt \epsilon) + 4 \delta$, and 
\begin{equation}
\bigtrnorm{ \EABjh{\ell+1,k-1}{\h{\ell}}(\rhoh{\h{\ell}}) - \AmeasBBjh{j,k-1}{\hB{\ell}} \EABjh{\ell+1,j-1}{\h{\ell}} (\rhoh{\h{\ell}}) } 
\leq (k-j) (O(\sqrt \epsilon) + 4 \delta)
 \enspace .  
\end{equation}
\end{lemma}

\begin{proof}
Let $\density(\ket a) = \ketbra a a$.  Then begin by placing an upper bound on 
\begin{align*}
\bigtrnorm{\rhojh{j+1}{\h{\ell}} - \AmeasBBjh{j}{\hB{\ell}}(\rhojh{j}{\h{\ell}})} 
&= \bigtrnorm{ \EABjh{j}{\h{\ell}}(\rhojh{j}{\h{\ell}}) - \AmeasBBjh{j}{\hB{\ell}}(\rhojh{j}{\h{\ell}}) } \\
&\leq \bigtrnorm{ \EAjh{j}{\hA{\ell}}(\rhojh{j}{\h{\ell}}) - \AmeasBjh{j}{\hB{\ell}}(\rhojh{j}{\h{\ell}}) }
 \enspace .
\end{align*}
Since $\rhojh{j}{\h{\ell}} = \sum_{\h{j-1}} \Pr[H_{j-1} = \h{j-1} \vert H_\ell = \h{\ell}] \density(\ket{\h{\ell+1,j-1}}) \otimes \rhoh{\h{j-1}}$, we can expand the right-hand side of this bound as 
\begin{align*}
\big\lVert &\EAjh{j}{\hA{\ell}}(\rhojh{j}{\h{\ell}}) - \AmeasBjh{j}{\hB{\ell}}(\rhojh{j}{\h{\ell}}) \big\rVert_{\mathrm{tr}} \\
&= \sum_{\h{j-1}} \Pr[H_{j-1} = \h{j-1} \vert H_\ell = \h{\ell}] \Bigtrnorm{ \EAjh{j}{\hA{j-1}}(\rhoh{\h{j-1}}) - \AmeasBjh{j}{\hB{j-1}}(\rhoh{\h{j-1}}) } \\
&= \frac{1}{2} \sum_{\h{j-1}, a_j, x_j} \Pr[H_{j-1} = \h{j-1} \vert H_\ell = \h{\ell}] \Bigtrnorm{ \density(\PAjh{j}{\hA{j}} \ket{\psijh{j}{\h{j-1}}}) - \density(\PAmeasBjh{j}{\hB{j-1}, a_j, x_j} \ket{\psijh{j}{\h{j-1}}}) } 
.
\end{align*}
Now if game $(j, \h{j-1})$ is $\epsilon$-structured, then by \corref{t:pullmeasurementstotheotherside}, $\norm{ \PAmeasBjh{j}{\hB{j-1}, a_j, x_j} \ket{\psij{j}} - \PAjh{j}{\hA{j}} \ket{\psij{j}} } = O(\sqrt \epsilon)$, so $\trnorm{ \density(\PAmeasBjh{j}{\hB{j-1}, a_j, x_j} \ket{\psij{j}}) - \density(\PAjh{j}{\hA{j}} \ket{\psij{j}}) } = O(\sqrt \epsilon)$ (\claimref{t:vectortraceversusl2distance}).  As $a_j$ and $x_j$ are each summed over $\{0,1\}$, it follows that $\bigtrnorm{\rhojh{j+1}{\h{\ell}} - \AmeasBBjh{j}{\hB{\ell}}(\rhojh{j}{\h{\ell}})} \leq O(\sqrt \epsilon) + 4 \delta$.  

Our claim now follows by a sequence of triangle inequalities in each step of which one of Alice's measurements is pulled over to Bob's side.  Write 
\begin{equation*}\begin{split}
\rhojh{k}{\h{\ell}} - \AmeasBBjh{j,k-1}{\hB{\ell}}(\rhojh{j}{\h{\ell}})
&= \Big( \rhojh{k}{\h{\ell}} - \AmeasBBjh{k-1}{\hB{\ell}}(\rhojh{k-1}{\h{\ell}}) \Big) + \AmeasBBjh{k-1}{\hB{\ell}} \Big( \rhojh{k-1}{\h{\ell}} - \AmeasBBjh{k-2}{\hB{\ell}}(\rhojh{k-2}{\h{\ell}}) \Big) \\
&\quad + \cdots + \AmeasBBjh{j+1,k-1}{\hB{\ell}} \Big( \rhojh{j+1}{\h{\ell}} - \AmeasBBjh{j}{\hB{\ell}}(\rhojh{j}{\h{\ell}}) \Big)
\end{split}\end{equation*}
By our above calculation, the trace norm of each term is at most $O(\sqrt \epsilon) + 4 \delta$.  
\end{proof}

\begin{corollary} \label{t:pullfromBobtoAlice}
Let $\S$ be a strategy and $\h{\ell}$ a partial transcript such that for every $j > \ell$, $\Pr[$game $(j, H_{j-1})$ is $\epsilon$-structured $\vert H_\ell = \h{\ell}] \geq 1-\delta$.  Then for all~$j > \ell$, 
\begin{equation} \label{e:pullfromBobtoAlice}
\Bignorm{ \EAjh{j}{\hA{\ell}} \AmeasBBjh{\ell+1,j-1}{\hB{\ell}}(\rhoh{\h{\ell}}) - \AmeasBjh{j}{\hB{\ell}} \AmeasBBjh{\ell+1,j-1}{\hB{\ell}}(\rhoh{\h{\ell}}) } \leq \big( 2(j-\ell) - 1 \big) (O(\sqrt \epsilon) + 4 \delta)
 \enspace .  
\end{equation}
\end{corollary}

\begin{proof}
Let $\delta' = O(\sqrt \epsilon) + 4 \delta$.  By \lemref{t:canpullovereverything}, $\bigtrnorm{ \EAjh{j}{\hA{\ell}}(\rhojh{j}{\h{\ell}}) - \AmeasBjh{j}{\hB{\ell}}(\rhojh{j}{\h{\ell}}) } \leq \delta'$.  Also $\bigtrnorm{ \EAjh{j}{\hA{\ell}} \AmeasBBjh{\ell+1,j-1}{\hB{\ell}}(\rhoh{\h{\ell}}) - \EAjh{j}{\hA{\ell}}(\rhojh{j}{\h{\ell}}) }$ and $\bigtrnorm{ \AmeasBjh{j}{\hB{\ell}} \AmeasBBjh{\ell+1,j-1}{\hB{\ell}}(\rhoh{\h{\ell}}) - \AmeasBjh{j}{\hB{\ell}}(\rhojh{j}{\h{\ell}}) }$ are both at most $\bigtrnorm{ \AmeasBBjh{\ell+1,j-1}{\hB{\ell}}(\rhoh{\h{\ell}}) - \rhojh{j}{\h{\ell}} }$, which by \lemref{t:canpullovereverything} is at most $(j - \ell - 1) \delta'$.  Combining these bounds gives our claim.  
\end{proof}

Measuring a qubit a second time does not change the trace distance.  Therefore, as in \claimref{t:Bguesstracedistance}, for a single-qubit ideal strategy~$\S$, we can replace $\AmeasBBj{j}$ with $\AmeasBj{j}$ without affecting the trace distance: 

\begin{claim} \label{t:measuringaqubittwice}
Let $\S$ be a single-qubit ideal strategy.  Then for any density matrices $\sigma$ and~$\tau$, 
\begin{equation}
\bigtrnorm{ \AmeasBBjh{j}{\hB{j-1}}(\sigma - \tau) } = \bigtrnorm{ \AmeasBjh{j}{\hB{j-1}}(\sigma - \tau) }
 \enspace .
\end{equation}
\end{claim}

\begin{proof}
Since $\AmeasBBjh{j}{\hB{j-1}} = \EBjh{j}{\hB{j-1}} \circ \AmeasBjh{j}{\hB{j-1}}$, and application of a super-operator cannot increase trace distance, $\bigtrnorm{ \AmeasBBjh{j}{\hB{j-1}}(\sigma - \tau) } \leq \bigtrnorm{ \AmeasBjh{j}{\hB{j-1}}(\sigma - \tau) }$.  The reason that this is an equality is that $\EBjh{j}{\hB{j-1}}$ measures the same qubit that $\AmeasBjh{j}{\hB{j-1}}$ already measured.  Since $\AmeasBjh{j}{\hB{j-1}}$ stores its measurement result in a transcript register, no information is lost by measuring the qubit a second time or even discarding it.  Slightly more formally, observe that $\AmeasBjh{j}{\hB{j-1}}(\sigma - \tau) = \frac{1}{2} \sum_{a_j, x_j} \ketbra{a_j, x_j}{a_j, x_j} \otimes \PAmeasBjh{j}{\hB{j-1}, a_j, x_j} (\sigma - \tau) \PAmeasBjh{j}{\hB{j-1}, a_j, x_j}$, so 
\begin{equation*}
\bigtrnorm{ \AmeasBjh{j}{\hB{j-1}}(\sigma - \tau) } = \frac{1}{2} \sum_{a_j, x_j} \bigtrnorm{\PAmeasBjh{j}{\hB{j-1}, a_j, x_j} (\sigma - \tau) \PAmeasBjh{j}{\hB{j-1}, a_j, x_j}}
 \enspace .
\end{equation*}
The expression $\PAmeasBjh{j}{\hB{j-1}, a_j, x_j} (\sigma - \tau) \PAmeasBjh{j}{\hB{j-1}, a_j, x_j}$ factors as the tensor product between a single-qubit state $\ketbra{(a_j, x_j)_A}{(a_j, x_j)_A}$ and another matrix.  The single-qubit state does not affect the trace distance, even after it is measured again by $\EBjh{j}{\hB{j-1}}$.  
\end{proof}

\subsection{Proof of \texorpdfstring{\thmref{t:singlequbitidealstrategysimulation}}{Theorem~\ref{t:singlequbitidealstrategysimulation}}: Simulation by single-qubit ideal strategies} \label{s:singlequbitidealstrategysimulation}

\begin{proof}[Proof of \thmref{t:singlequbitidealstrategysimulation}]
By \propref{t:wlogunitaryqubits}, we may assume without loss of generality that the isometries~$\UXjh{j}{\hX{j-1}}$ from \defref{t:gamequbitsdef} are actually unitary.  Let~$\tilde \S$ be the strategy with the same initial state~$\ket{\psione}$ as~$\S$, but that uses the reflections $\RAdecjah{\tilde}{j}{\alpha}{\hA{j-1}} = \UAjh{j}{\hA{j-1}}^\dagger (\RAa{\alpha} \otimes \identity) \UAjh{j}{\hX{j-1}}$ for Alice.  Then $\tilde \S$ is a single-qubit ideal strategy on Alice's side.  

Our proof that $\tilde \S$ closely simulates $\S$ is based on \lemref{t:allCHSHBobgameoutcomesrelatedbysinglequbitunitary} and the following claim: 

\begin{claim}
For every~$j$, $\trnorm{\EABj{j}(\rhoj{j}) - \EABdecj{\tilde}{j}(\rhoj{j})} = O(\sqrt \epsilon)$.  
\end{claim}

\begin{proof}
Expand $\trnorm{\EABj{j}(\rhoj{j}) - \EABdecj{\tilde}{j}(\rhoj{j})} = \sum_{\h{j-1}} \Pr[H_{j-1} = \h{j-1}] \trnorm{\EABj{j}(\rhoh{\h{j-1}}) - \EABdecj{\tilde}{j}(\rhoh{\h{j-1}})}$.  Split the sum according to whether game~$j$ is played with $\epsilon$-structure on transcript~$\h{j-1}$.  The total contribution from unstructured games is at most $2 \Pr[\text{game $(j,H_{j-1})$ is not $\epsilon$-structured}] \leq 2 \epsilon$.  On the other hand, by the CHSH rigidity lemma, \lemref{t:eprlemma}, for any $\epsilon$-structured game, we have, using \claimref{t:vectortraceversusl2distance} and letting $\density(\ket a) = \ketbra a a$, 
\begin{align*}
\trnorm{\density(\PABj{j} \ket{\psij{j}}) - \density(\PABdecj{\tilde}{j} \ket{\psij{j}})} 
&\leq 2 \norm{\PABj{j} \ket{\psij{j}} - \PABdecj{\tilde}{j} \ket{\psij{j}}} \\
&= \norm{ ( \RAj{j} \otimes \RBj{j} - \RAdecj{\tilde}{j} \otimes \RBdecj{\tilde}{j} ) \ket{\psij{j}} } \\
&= O(\sqrt \epsilon) 
 \enspace .
\end{align*}
Thus $\trnorm{\EABj{j}(\rhoj{j}) - \EABdecj{\tilde}{j}(\rhoj{j})} \leq 2 \epsilon + O(\sqrt \epsilon)$.  
\end{proof}

From the expansion of $\EABj{1,j}(\rhoone) - \EABdecj{\tilde}{1,j}(\rhoone)$ as 
\begin{equation*}
\big( \EABj{j}(\rhoj{j}) - \EABdecj{\tilde}{j}(\rhoj{j}) \big)
+ \EABdecj{\tilde}{j} \big( \EABj{j-1}(\rhoj{j-1}) - \EABdecj{\tilde}{j-1}(\rhoj{j-1}) \big)
+ \cdots 
+ \EABdecj{\tilde}{2,j} \big( \EABj{j}(\rhoone) - \EABdecj{\tilde}{j}(\rhoone) \big)
 \enspace ,
\end{equation*}
it follows that $\trnorm{\EABj{1,j}(\rhoone) - \EABdecj{\tilde}{1,j}(\rhoone)} \leq j O(\sqrt \epsilon)$.  Thus $\tilde \S$ weakly $(n O(\sqrt \epsilon))$-simulates~$\S$.  By \corref{t:weaktostrongsimulation}, there is a constant~$\varkappa$ such that $\tilde \S$ $\varkappa n^\varkappa \epsilon^{1/\varkappa}$-simulates~$\S$.  Since $\tilde \S$ is structured (\lemref{t:simulationpreservesstructure}), we can repeat the argument, but this time changing Bob's reflections, to get simulation by a single-qubit ideal strategy for both provers.  
\end{proof}

This completes the first part of the proof of \thmref{t:sequentialCHSHgames}.  In the remainder of the proof, we will restrict consideration to single-qubit ideal strategies.  This is okay since the strategy~$\tilde \S$ is structured by \lemref{t:simulationpreservesstructure}.  Furthermore, simulation is transitive; if we find a strategy~$\hat \S$ that $\eta$-simulates~$\tilde \S$, then $\hat \S$ $(\kappa n^\kappa \epsilon^{1/\kappa} + \eta)$-simulates~$\S$.

\subsection{Proof of \texorpdfstring{\thmref{t:multiqubitidealstrategysimulation}}{Theorem~\ref{t:multiqubitidealstrategysimulation}}: Simulation by multi-qubit ideal strategies} \label{s:historydependenttensorproductstrategy}

\begin{proof}[Proof of \thmref{t:multiqubitidealstrategysimulation}]
As in the proof of \thmref{t:singlequbitidealstrategysimulation}, it suffices to show that an isometric extension of $\S$ can be weakly simulated by a strategy~$\tilde \S$ in which Alice plays according to a multi-qubit ideal strategy and Bob plays the same as in~$\S$.  Indeed, \corref{t:weaktostrongsimulation} then turns weak simulation into a simulation statement.  By \lemref{t:simulationpreservesstructure}, $\tilde \S$ is structured, so repeating the argument implies that Bob can also play according to a multi-qubit ideal strategy.  

Let us begin by defining Alice's strategy in~$\tilde \S$.  Alice uses the Hilbert space $(\C^2)^{\otimes n} \otimes \H_A$, with the extra $n$ qubits providing convenient workspace.  The isometry $\XA: \H_A \hookrightarrow (\C^2)^{\otimes n} \otimes \H_A$ from \defref{t:isometricextensiondef} simply prepends $\ket{0}^{\otimes n}$.  Thus the initial state for~$\tilde \S$ is $\ket{0}^{\otimes n} \otimes \ket{\psione}$.  Since~$\S$ is a single-qubit ideal strategy (\defref{t:singlequbitidealstrategy}), there exist unitaries~$\UsingleAjh{j}{\hA{j-1}} : \H_A \overset{\cong}{\rightarrow} \C^2 \otimes \H_A'$ such that  $\RAjah{j}{\alpha}{\hA{j-1}} = \UsingleAjh{j}{\hA{j-1}}^\dagger (\RAa{\alpha} \otimes \identity) \UsingleAjh{j}{\hA{j-1}}$.  In particular, we can fix a basis so $\H_A = \C^2 \otimes \H_A'$.  Number this qubit~$0$, and the other qubits from~$1$ to~$n$.  Then, in \defref{t:multiqubitidealstrategy}, let $\XmultiA$ be the identity, and define the operators~$\UmultiAjh{j}{\hA{j-1}}$ by 
\begin{equation}\begin{split} \label{e:Alicesmultiqubitidealstrategy}
\UmultiAj{1} &= S_1 \UsingleAj{1} \\
\UmultiAjh{j}{\hA{j-1}} &= S_j \UsingleAjh{j}{\hA{j-1}} \UsingleAjh{j-1}{\hA{j-1}}^\dagger V(a_{j-1}, x_{j-1})_0
 \enspace .
\end{split}\end{equation}
Here, the operators $\UsingleAjh{k}{\hA{k-1}}$ are understood to act on the $\H_A$ register.  $S_k$ denotes the swap operator between qubit~$0$ and qubit~$k$.  For $a, x \in \{0,1\}$, $V(a, x)$ is a fixed one-qubit unitary that maps $\ket 0$ to $\ket{(a, x)_A}$; the subscript $0$ in the expression above indicates that it acts on qubit~$0$.  Since~$\UmultiAj{j}$ does not involve qubits~$1$ through~$j-1$ (nor qubits~$j+1$ through~$n$), Eq.~\eqnref{e:Alicesmultiqubitidealstrategy} defines a valid multi-qubit ideal strategy for Alice, using Eq.~\eqnref{e:multiqubitidealstrategy}.  

Eq.~\eqnref{e:Alicesmultiqubitidealstrategy} deserves some explanation.  
First of all, $\EAj{1}$ and $\EAdecj{\tilde}{1}$ act in exactly the same way: for any $\sigma \in \L(\H_A)$, $\ketbra{0^n}{0^n} \otimes \EAj{1}(\sigma) = \EAdecj{\tilde}{1}(\ketbra{0^n}{0^n} \otimes \sigma)$.  They both expose a qubit, with~$\UsingleAj{1}$, measure that qubit, and then put it back, with~$\UsingleAj{1}{}^\dagger$.  To understand $\UmultiAj{j}$, notice that there is a trivial way of forcing a tensor-product structure for Alice's measurements: after a qubit has been measured, say as $\ket{(a_j, x_j)_A}$, put that qubit to the side, rotate a fresh ancilla qubit $\ket 0$ into $\ket{(a_j, x_j)_A}$, and continue playing using the ancilla in place of the measured qubit.  This is how $\UmultiAj{j}$ works; $\UsingleAjh{j-1}{\hA{j-1}}^\dagger V(a_{j-1}, x_{j-1})$ rotates the ancilla qubit to $\ket{(a_{j-1}, x_{j-1})_A}$ and puts it into the place of the measured qubit for game~$j-1$, and $S_j \UsingleAjh{j}{\hA{j-1}}$ exposes the qubit for the next game.  Thus Eq.~\eqnref{e:multiqubitidealstrategy} seems to be the obvious way of defining a multi-qubit ideal strategy for Alice.  It is not obvious, however, that $\tilde \S$ simulates~$\S$.   The reason is that $\tilde \S$ does not just set measured qubits to the side---which would make simulation according to \defref{t:simulationdef} hopeless.  It also tries to restore the qubits, by applying $\UmultiAj{1}{}^\dagger \cdots \UmultiAj{j}{}^\dagger$.  Our claim that $\tilde \S$ simulates~$\S$ will boil down to showing that the qubit $\ket{(a_j, x_j)_A}$ measured in game~$j$ will stay close to that through all later games (\lemref{t:latergamesdontchangeearlierresults}), and therefore when $\UmultiAj{j}{}^\dagger$ is applied it returns qubit~$j$ to its initial state~$\ket 0$.  

\smallskip 

To prove \thmref{t:multiqubitidealstrategysimulation}, we need to bound $\bigtrnorm{ \ketbra{0^n}{0^n} \otimes \EABj{1,k}(\rhoone) - \EABdecj{\tilde}{1,k}(\ketbra{0^n}{0^n} \otimes \rhoone) }$.  By a hybrid argument, this is at most $k \max_{j \in [k]} \bigtrnorm{ \ketbra{0^n}{0^n} \otimes \EABj{j}(\rhoj{j}) - \EABdecj{\tilde}{j}(\ketbra{0^n}{0^n} \otimes \rhoj{j}) }$.  

\def\TAj #1{T^A_{#1}}
\def\TAjh #1{T^A_{#1}(\hA{#1})}

At this point, we need to define some new notation.  To save space, let us henceforth assume that the $n$ prepended qubits have been incorporated into Alice's operators~$\EAj{j}$.  Therefore we will write simply $\rhoone$ instead of $\ketbra{0^n}{0^n} \otimes \rhoone$ and $\EABj{1,j}(\rhoone)$ instead of $\ketbra{0^n}{0^n} \otimes \EABj{1,j}(\rhoone)$.  We aim to bound $\trnorm{ \EABj{j}(\rhoj{j}) - \EABdecj{\tilde}{j}(\rhoj{j})}$.  
Let 
\begin{equation}
\TAjh{j} = \UsingleAjh{j}{\hA{j-1}}^\dagger S_j V(a_j, x_j)_j \UsingleAjh{j}{\hA{j-1}}
 \enspace .
\end{equation}
Then Eq.~\eqnref{e:multiqubitidealstrategy}, $\RAjah{j}{a}{\hA{j-1}} = \UmultiAj{1}{}^\dagger \cdots \UmultiAj{j}{}^\dagger (\RAa{a})_j \UmultiAj{j} \cdots \UmultiAj{1}$ can be equivalently rewritten as 
\begin{equation*}
\RAjah{j}{a}{\hA{j-1}}
= \TAj{1}{}^\dagger \cdots \TAj{j}{}^\dagger (\RAa{a})_j \TAj{j} \cdots \TAj{1}
 \enspace ,
\end{equation*}
since $\TAj{j} \cdots \TAj{1} = \UsingleAj{j}{}^\dagger V(a_j, x_j)_0 \UmultiAj{j} \cdots \UmultiAj{1}$ and the extra $\UsingleAj{j}{}^\dagger V(a_j, x_j)_0$ factor cancels out.  These $\TAjh{j}$ operators are more convenient to work with than the $\UmultiAjh{j}{\hA{j-1}}$ operators.  (It is their dependence on $a_j$ and $x_j$ that disallows using them directly in the definition of~$\tilde \S$.)  Define super-operators ${\cal U}^A_j$, ${\cal V}_j$, ${\cal S}_j$ and ${\cal T}_j$ by, for $\sigma \in \L((\C^2)^{\otimes n} \otimes \H_A)$, 
\begin{equation}\begin{split}
{\cal U}_j(\ketbra{\hA{}}{\hA{}} \otimes \sigma) &= \ketbra{\hA{}}{\hA{}} \otimes \UsingleAjh{j}{\hA{j-1}} \sigma \UsingleAjh{j}{\hA{j-1}}^\dagger \\
{\cal V}_j(\ketbra{\hA{}}{\hA{}} \otimes \sigma) &= \ketbra{\hA{}}{\hA{}} \otimes V(a_j, x_j)_j \sigma V(a_j, x_j)_j^\dagger \\
{\cal S}_j(\ketbra{\hA{}}{\hA{}} \otimes \sigma) &= \ketbra{\hA{}}{\hA{}} \otimes S_j \sigma S_j \\
{\cal T}_j &= {\cal U}_j^{-1} {\cal S}_j {\cal V}_j {\cal U}_j
 \enspace .
\end{split}\end{equation}
Let ${\cal T}_{j,k} = {\cal T}_k \cdots {\cal T}_{j+1} {\cal T}_j$ and ${\cal V}_{j,k} = {\cal V}_k \cdots {\cal V}_{j+1} {\cal V}_j$.  Observe then that 
\begin{equation}
\EAdecj{\tilde}{j} = {\cal T}_{1,j-1}^{-1} \EAj{j} {\cal T}_{1,j-1}
 \enspace .
\end{equation}
Therefore, $\trnorm{ \EABdecj{\tilde}{j}(\rhoj{j}) - \EABj{j}(\rhoj{j}) } = \trnorm{ \EABj{j} {\cal T}_{1,j-1}(\rhoj{j}) - {\cal T}_{1,j-1} \EABj{j}(\rhoj{j}) }$, as ${\cal T}_j$, being unitary, does not affect the trace norm.  

Next we claim that ${\cal T}_{1,j-1}(\rhoj{j}) \approx {\cal V}_{1,j-1}(\rhoj{j})$ and ${\cal T}_{1,j-1}(\rhoj{j+1}) \approx {\cal V}_{1,j-1}(\rhoj{j+1})$, where by $\approx$ we mean that the difference is at most $\varkappa n^\varkappa \epsilon^{1/\varkappa}$ (in trace norm) for some constant~$\varkappa$.  In other words, super-operator ${\cal T}$, when applied to a $\rhoj{k}$, effectively only rotates the extra $\ket 0$ qubits at the beginning of $\rhoj{k}$.  If these approximations hold, then our theorem is proved: 
\begin{equation*}
\bigtrnorm{ \EABj{j} {\cal T}_{1,j-1}(\rhoj{j}) - {\cal T}_{1,j-1} \EABj{j}(\rhoj{j}) }
\approx \bigtrnorm{ \EABj{j} {\cal V}_{1,j-1}(\rhoj{j}) - {\cal V}_{1,j-1} \EABj{j}(\rhoj{j}) }
= 0
 \enspace ,
\end{equation*}
since $\EABj{j}$ commutes with ${\cal V}_{1,j-1}$.  

\smallskip

Both approximations are shown by a hybrid argument.  For $k < j$, expand 
\begin{equation*}\begin{split}
{\cal T}_{1,k}(\rhoj{j}) - {\cal V}_{1,k}(\rhoj{j})
&= {\cal T}_{2,k} \big( {\cal T}_1(\rhoj{j}) - {\cal V}_1(\rhoj{j}) \big)
+ {\cal T}_{3,k} \big( {\cal T}_2 {\cal V}_1 (\rhoj{j}) - {\cal V}_{1,2}(\rhoj{j}) \big) \\
&\quad + {\cal T}_{4,k} \big( {\cal T}_3 {\cal V}_{1,2} (\rhoj{j}) - {\cal V}_{1,3}(\rhoj{j}) \big)
+ \cdots + \big( {\cal T}_k {\cal V}_{1,k-1}(\rhoj{j}) - {\cal V}_{1,k}(\rhoj{j}) \big)
 \enspace .
\end{split}\end{equation*}
For $k \neq \ell$, ${\cal T}_k$ and ${\cal V}_\ell$ commute.  Thus, 
\begin{equation*}
\bigtrnorm{ {\cal T}_{1,k}(\rhoj{j}) - {\cal V}_{1,k}(\rhoj{j}) }
\leq k \max_\ell \bigtrnorm{ {\cal T}_\ell(\rhoj{j}) - {\cal V}_\ell(\rhoj{j}) }
 \enspace .
\end{equation*}
Our main lemma places a bound on $\bigtrnorm{ {\cal T}_\ell(\rhoj{j}) - {\cal V}_\ell(\rhoj{j}) }$ for $\ell < j$.  This means that later games do not much change qubits that have been measured earlier.  

\begin{lemma} \label{t:latergamesdontchangeearlierresults}
There exists a constant~$\varkappa$ such that for $\ell < j$, $\bigtrnorm{ {\cal T}_\ell(\rhoj{j}) - {\cal V}_\ell(\rhoj{j}) } < \varkappa n^\varkappa \epsilon^{1/\varkappa}$.  
\end{lemma}

\begin{proof}
Since $\S$ is a single-qubit ideal strategy, Alice's measurement in game~$(j, \hA{j-1})$ projects her qubit for that game into exactly $\ket{(a_j, x_j)_A}$.  In particular, therefore 
\begin{equation*}
\Tr\!\big[ \big( (\ketbra 0 0)_0 \otimes (\ketbra 0 0)_\ell \otimes \identity \big) \cdot {\cal V}_\ell^{-1} {\cal S}_\ell \, {\cal U}_\ell(\rhoj{\ell+1}) \big] = 1
 \enspace .
\end{equation*}
Recall from \defref{t:canpullovereverythingdef} the super-operators~$\AmeasBBj{j}$.  As these super-operators act on~$\H_B$, they commute with ${\cal V}_\ell$, ${\cal S}_\ell$ and ${\cal U}_\ell$, and do not change the above trace, so 
\begin{equation*}
\Tr\!\big[ \big( (\ketbra 0 0)_0 \otimes (\ketbra 0 0)_\ell \otimes \identity \big) \cdot {\cal V}_\ell^{-1} {\cal S}_\ell \, {\cal U}_\ell \AmeasBBj{\ell+1,j-1}(\rhoj{\ell+1}) \big] = 1
 \enspace .
\end{equation*}
By \lemref{t:canpullovereverything}, $\trnorm{ \rhoj{j} - \AmeasBBj{\ell+1,j-1}(\rhoj{\ell+1}) } = O(n \sqrt \epsilon)$, and so by \lemref{t:traceandtracenorm}, 
\begin{equation*}
\Tr\!\big[ \big( (\ketbra 0 0)_0 \otimes (\ketbra 0 0)_\ell \otimes \identity \big) \cdot {\cal V}_\ell^{-1} {\cal S}_\ell \, {\cal U}_\ell(\rhoj{j}) \big] \geq 1 - O(n \sqrt \epsilon)
 \enspace .
\end{equation*}
Applying \corref{t:gentlemeasurementpurestate}, there exists a state~$\sigma = \sum_{\hA{j-1}} \ketbra{\hA{j-1}}{\hA{j-1}} \otimes \sigma_h$ such that for $\tau = (\ketbra 0 0)_0 \otimes (\ketbra 0 0)_\ell \otimes \sigma$, 
\begin{equation*}
\bigtrnorm{
{\cal V}_\ell^{-1} {\cal S}_\ell \, {\cal U}_\ell(\rhoj{j}) - \tau
} \leq O(\sqrt{n} \sqrt{\epsilon})
 \enspace .
\end{equation*}
(The state~$\sigma$ is block diagonal because it is a partial trace of the block-diagonal matrix ${\cal V}_\ell^{-1} {\cal S}_\ell \, {\cal U}_\ell(\rhoj{j})$.)  
It remains only to substitute the definition ${\cal T}_\ell = {\cal U}_\ell^{-1} {\cal S}_\ell {\cal V}_\ell {\cal U}_\ell$ and apply two last triangle inequalities: 
\begin{equation*}
\bigtrnorm{ {\cal T}_\ell(\rhoj{j}) - {\cal V}_\ell(\rhoj{j}) }
= \bigtrnorm{ {\cal S}_\ell {\cal V}_\ell {\cal U}_\ell(\rhoj{j}) - {\cal V}_\ell {\cal U}_\ell(\rhoj{j}) }
\leq O(\sqrt{n} \epsilon^{1/4}) + \bigtrnorm{ {\cal S}_\ell {\cal V}_\ell {\cal S}_\ell {\cal V}_\ell(\tau) - {\cal V}_\ell {\cal S}_\ell {\cal V}_\ell(\tau) }
 \enspace .
\end{equation*}
Since ${\cal S}_\ell \tau = \tau$, the final term is zero.  
\end{proof}

This completes the proof of \thmref{t:multiqubitidealstrategysimulation}.  
\end{proof}

\subsection{Proof of \texorpdfstring{\thmref{t:globalgluing}}{Theorem~\ref{t:globalgluing}}: Gluing together multi-qubit ideal strategies} \label{s:gluing}

\thmref{t:multiqubitidealstrategysimulation} shows that Alice and Bob are close to playing according to a strategy in which every game uses a qubit in tensor product with the previous games' qubits.  However, the qubit used can depend on previous games' outcomes.  Next, in the third and last part of the proof of \thmref{t:sequentialCHSHgames}, we will argue that Alice and Bob must play using a single set of $n$ qubits, fixed in advance independent of the transcript.  The reason is essentially that the players cannot communicate with each other and their local transcripts are insufficiently correlated to coordinate a dynamic strategy.  

For a toy example of the issue, consider two provers who play the first $n-1$ games honestly and who at the beginning of the last game share two EPR states, $\ket{\psi^*}^{\otimes 2}$.  Say that for certain functions~$f$ and~$g$, Alice uses EPR state $f(\hA{n-1}) \in \{1,2\}$ in game~$n$, and Bob uses pair $g(\hB{n-1}) \in \{1,2\}$.  For game~$n$ to be structured, they need $f(\hA{n-1}) = g(\hB{n-1})$ so that they measure the same EPR state.  Now Alice and Bob's local transcripts are each uniformly random, separately, but they have a constant correlation in every game coordinate.  It is straightforward to argue based on coordinate influence that if $\Pr[f(H^A_{n-1}) \neq g(H^B_{n-1})]$ is small, then~$f$ and~$g$ must both be nearly constant.  In particular, although the majority function is the stablest balanced function~\cite{MosselODonellOleszkiewicz05majorityisstablest}, it is not stable enough.  Thus one of the two EPR states is used almost always.  

This example is of an essentially classical cheating strategy.  The actual provers we face may be significantly more sophisticated.  In particular, by cheating in small amounts in the first games, they potentially can drastically change the underlying quantum state.  For example, Alice might have knowingly managed to swap her halves of the two last EPR states along some transcripts~$\hA{n-1}$.  Then she can use completely different strategies for the last game, depending on whether or not there has been a swap, without having to coordinate any classical information with Bob.  There may also be much more sophisticated ways of cheating than this example.  We worry especially that small amounts of cheating in earlier games might enable an avalanche of more and more blatant cheating in later games.  

Our ``gluing" argument has two parts, that we term local and global gluing.  In the local gluing argument, we show that most of the time, for two typical partial transcripts $\hA{n}$ and $\hA{n}{}'$ that differ in only one game coordinate~$j$, the states created by Alice measuring along these transcripts are close to each other (up to unitary corrections on Alice and Bob's $j$th qubits).  See \thmref{t:localgluing} for a precise statement.  Essentially, this means that Alice's strategy for games~$j+1, \ldots, n$ along $\hA{n}$ does not depend much on game~$j$.  In the global gluing argument, we connect together all of the transcripts, by connecting far away transcripts with a sequence of local gluing steps.

\subsubsection{Local gluing}

Similar to \defref{t:ABunitarydef}, we define unitary operators $\AAunitary(a, a', \Delta)$ that rotate between Alice's different measurement bases (see \corref{t:CHSHgameoutcomesrelatedbysinglequbitunitary}): 

\begin{definition} \label{t:AAunitarydef}
Let $\S$ be a strategy such that the operators $\UXjh{j}{\hX{j-1}}$ are unitary, for $\device \in \{A, B\}$.  For $a, a', \Delta \in \{0,1\}$, define unitaries $\AAunitary(a, a', \Delta) \in \L(\C^2)$ and $\AAunitaryXj{j}(\hX{j-1}, a, a', \Delta)$ by  
\begin{equation}\begin{split}
\AAunitary(a, a', \Delta) &= \sum_{x \in \{0,1\}} \ketbra{(a', x \oplus \Delta)_A}{(a, x)_A} \\
\AAunitaryXj{j}(\hX{j-1}, a, a', \Delta) &= \UXjh{j}{\hX{j-1}}^\dagger \big(\AAunitary(a, a', \Delta) \otimes \identity\big) \UXjh{j}{\hX{j-1}}
 \enspace ,
\end{split}\end{equation}
i.e., $\AAunitaryXj{j}(\hX{j-1}, a, a', \Delta)$ is $\AAunitary(a, a', \Delta)$ acting on the qubit in~$\H_\device$ for game~$(j, \hX{j-1})$.  
\end{definition}

\begin{theorem} \label{t:localgluing}
For $a_j', x_j' \in \{0,1\}$ and a partial transcript $\hA{k}$, let $\hA{k}{}'$ denote the same transcript except with the question and outcome for game~$j$ replaced by $a_j'$ and~$x_j'$.  

There exists a constant~$\kappa$ such that, for $p(n, \delta, \epsilon) = \kappa n^\kappa (\delta + \epsilon)^{1/\kappa}$, if $\S$ is a $(\delta, \epsilon)$-structured multi-qubit ideal strategy for $n$ sequential CHSH games, then there is at least a $1 - p(n, \delta, \epsilon)$ probability that $H_j$ lies in the set of $\h{j}$ that satisfy, for all~$a_j'$ and~$x_j'$, 
\begin{equation} \label{e:localgluing}
\Bigtrnorm{
\EAjh{j+1,k}{\hA{j}}(\rhoAh{\hA{j}})
- \AAunitaryAj{j} \AAunitaryBj{j} \EAjh{j+1,k}{\hA{j}{}'}(\rhoAh{\hA{j}{}'}) \AAunitaryAj{j}{}^\dagger \AAunitaryBj{j}{}^\dagger
} 
\leq p(n, \delta, \epsilon)
 \enspace ,
\end{equation}
where $\AAunitaryXj{j} = \AAunitaryXj{j}(\hX{j-1}, a_j', a_j, x_j' \oplus x_j)$.  
\end{theorem}

\begin{proof}
There are three parts to the proof.  \corref{t:CHSHgameoutcomesrelatedbysinglequbitunitary} begins the gluing: if game $(j, \h{j-1})$ is $\epsilon$-structured, then $\bignorm{\ket{\psijh{j+1}{\h{j}}} - \AAunitaryAj{j} \ket{\psijh{j+1}{\h{j}{}'}}} = O(\sqrt \epsilon)$.  In trace distance, 
\begin{equation*}
\rhoh{\h{j}} \approx \AAunitaryAj{j} \rhoh{\h{j}{}'} \AAunitaryAj{j}{}^\dagger
 \enspace .
\end{equation*}
Since applying a super-operator cannot increase trace distance, therefore $\AmeasBBjh{j+1,k}{\hB{j}}(\rhoh{\h{j}})$ is close to $\AAunitaryAj{j} \AmeasBBjh{j+1,k}{\hB{j}}(\rhoh{\h{j}{}'}) \AAunitaryAj{j}{}^\dagger$.  
In the second part of the proof, we use \corref{t:pullfromBobtoAlice} of \lemref{t:canpullovereverything} to pull Alice's measurement super-operators back to her side, simultaneously eliminating Bob's measurements for games~$j+1$ and later; thus 
\begin{equation*}
\EAjh{j+1,k}{\hA{j}}(\rhoh{\h{j}}) \approx \AAunitaryAj{j} \EAjh{j+1,k}{\hA{j}{}'}(\rhoh{\h{j}{}'}) \AAunitaryAj{j}{}^\dagger
 \enspace .
\end{equation*}
This equation says that Alice's actions along the transcript~$\hA{j}$ have nearly the same effect as along the transcript~$\hA{j}{}'$.  It holds essentially because both super-operators can be pulled to Bob's side in the same way, if Bob also measures.  
In the third part of the proof, we apply \lemref{t:allCHSHBobgameoutcomesrelatedbysinglequbitunitary}.  The lemma shows that $\rhoh{\h{j}}$ is close to $\rhoAh{\hA{j}}$, up to certain unitary corrections on $\H_B$.  This allows us to eliminate Bob's measurement super-operators for games up to~$j$, thus establishing the claim that $\EAjh{j+1,k}{\hA{j}}(\rhoAh{\hA{j}})$ is close to $\AAunitaryAj{j} \EAjh{j+1,k}{\hA{j}{}'}(\rhoAh{\hA{j}{}'}) \AAunitaryAj{j}{}^\dagger$, up to certain unitary corrections on $\H_B$.  These corrections are the same for the first $j-1$ games, and since~$\S$ is a multi-qubit ideal strategy they can be canceled out, leaving only a correction $\AAunitaryBj{j}$ for game~$j$.  

It will be convenient to establish the notation that for a vector~$\ket a$, $\density(\ket a) = \ketbra a a$.  

The next proposition combines the first two steps: 

\begin{proposition} \label{t:removedmostBobmeasurements}
Under the conditions of \thmref{t:localgluing}, there is at least a $1 - \sqrt{n \delta}$ probability that $H_j$ lies in the set of~$\h{j}$ satisfying, for all $a_j'$ and~$x_j'$, 
\begin{equation} \label{e:removedmostBobmeasurements}
\Bigtrnorm{
\EAjh{j+1,k}{\hA{j}}(\rhoh{\h{j}}) 
- \AAunitaryAj{j} \EAjh{j+1,k}{\hA{j}{}'}(\rhoh{\h{j}{}'}) \AAunitaryAj{j}{}^\dagger
} 
\leq O(\sqrt \epsilon) + 4n (O(\sqrt \epsilon) + 4 \cdot 60 \sqrt{n \delta}) (k - j)
.
\end{equation}
\end{proposition}

Very roughly, this inequality means that Alice's actions in games~$j+1$ to~$k$ are almost the same starting with $\hA{j}$ as starting with the perturbed transcript~$\hA{j}{}'$.  

\begin{proof}
By a union bound, $\Pr[\text{all games are $\epsilon$-structured along $H_n$}] \geq 1 - n \delta$.  By a Markov inequality, then, there is at least a $1 - \sqrt{n \delta}$ probability that $H_{j-1}$ lies in the set $S' = \{ \h{j-1} : \Pr[$all games are $\epsilon$-structured along $H_n \vert H_{j-1} = \h{j-1}] \geq 1 - \sqrt{n \delta} \}$.  Let $S = \{ \h{j} : \forall \, a_j', x_j', \, \Pr[$all games are $\epsilon$-structured along $H_n \vert H_j = (\h{j-1}, a_j', x_j, b_j', y_j) ] \geq 1 - 60 \sqrt{n \delta} \}$.  When game~$(j, \h{j-1})$ is structured, all outcomes occur with probability at least $1/60$ (\corref{t:structuredgameprobabilitylowerbound}).  Therefore any $\h{j}$ whose prefix $\h{j-1}$ lies in $S'$ itself lies in~$S$, so $\Pr[H_j \in S] \geq 1 - \sqrt{n \delta}$.  

Now for $\h{j} \in S$, since game~$(j, \h{j-1})$ is structured, \corref{t:CHSHgameoutcomesrelatedbysinglequbitunitary} gives $\bignorm{\ket{\psijh{j+1}{\h{j}}} - \AAunitaryAj{j} \ket{\psijh{j+1}{\h{j}{}'}}} = O(\sqrt \epsilon)$, where $\AAunitaryAj{j} = \AAunitaryAj{j}(\hA{j-1}, a_j', a_j, x_j' \oplus x_j)$.  Notice that since Bob's view along the two transcripts is the same, i.e., $\hB{j} = \hB{j}{}'$, and measurements on Bob's side commute with $\AAunitaryAj{j}$, we therefore have 
\begin{equation} \label{e:statesrelatedatperturbation}
\trnorm{\AmeasBBjh{j+1,k}{\hB{j}}(\rhoh{\h{j}}) - \AAunitaryAj{j} \AmeasBBjh{j+1,k}{\hB{j}}(\rhoh{\h{j}{}'}) \AAunitaryAj{j}{}^\dagger} \leq \bigtrnorm{\density(\ket{\psijh{j+1}{\h{j}}}) - \density(\AAunitaryAj{j} \ket{\psijh{j+1}{\h{j}{}'}})} = O(\sqrt \epsilon)
,
\end{equation}
using \claimref{t:vectortraceversusl2distance}.  

We complete the proof with an inductive argument that pulls Alice's measurement super-operators back over to her side and eliminates Bob's measurements in games~$j+1$ and later.  

\begin{claim} \label{t:eliminateBobbyinduction}
For $\ell \in \{j, \ldots, k\}$, 
\begin{equation} \label{e:eliminateBobbyinduction}
\Bigtrnorm{
\EAjh{\ell+1,k}{\hA{j}} \AmeasBBjh{j+1,\ell}{\hB{j}}(\rhoh{\h{j}}) 
- \AAunitaryAj{j} \big[ \EAjh{\ell+1,k}{\hA{j}{}'} \AmeasBBjh{j+1,\ell}{\hB{j}}(\rhoh{\h{j}{}'}) \big]\AAunitaryAj{j}{}^\dagger
}
\leq O(\sqrt \epsilon) + 4n (O(\sqrt \epsilon) + 4 \cdot 60 \sqrt{n \delta}) (k - \ell)
.
\end{equation}
\end{claim}

\begin{proof}
The proof is by induction in $(k-\ell)$, starting with Eq.~\eqnref{e:statesrelatedatperturbation} for $\ell = k$.  

Assume we are given Eq.~\eqnref{e:eliminateBobbyinduction} for some $\ell > j$.  By \claimref{t:measuringaqubittwice}, since the qubit has been measured already we can eliminate Bob's final measurement super-operator without affecting the trace distance: 
\begin{gather*}
\Bigtrnorm{
\EAjh{\ell+1,k}{\hA{j}} \AmeasBBjh{j+1,\ell}{\hB{j}}(\rhoh{\h{j}}) 
- \AAunitaryAj{j} \big[ \EAjh{\ell+1,k}{\hA{j}{}'} \AmeasBBjh{j+1,\ell}{\hB{j}}(\rhoh{\h{j}{}'}) \big] \AAunitaryAj{j}{}^\dagger
} \\
= \Bigtrnorm{
\EAjh{\ell+1,k}{\hA{j}} \AmeasBjh{\ell}{\hB{j}} \AmeasBBjh{j+1,\ell-1}{\hB{j}}(\rhoh{\h{j}}) 
- \AAunitaryAj{j} \big[ \EAjh{\ell+1,k}{\hA{j}{}'} \AmeasBjh{\ell}{\hB{j}} \AmeasBBjh{j+1,\ell-1}{\hB{j}}(\rhoh{\h{j}{}'}) \big] \AAunitaryAj{j}{}^\dagger
}
 \enspace .
\end{gather*}
By \corref{t:pullfromBobtoAlice}, we can pull Alice's last measurement on Bob's side back to Alice's side: letting $\delta' = 2n (O(\sqrt \epsilon) + 4 \cdot 60 \sqrt{n \delta})$, $\Bignorm{ \AmeasBjh{\ell}{\hB{j}} \AmeasBBjh{j+1,\ell-1}{\hB{j}}(\rhoh{\h{j}}) - \EAjh{\ell}{\hA{j}} \AmeasBBjh{j+1,\ell-1}{\hB{j}}(\rhoh{\h{j}}) } \leq \delta'$.  The same bound holds for the transcript~$\h{j}{}' = (\hA{j}{}', \hB{j})$.  Since applying $\EAjh{\ell+1,k}{\hA{j}}$ or $\EAjh{\ell+1,k}{\hA{j}{}'}$ only decreases these trace distances, therefore, 
\begin{equation*}
\Bigtrnorm{
\EAjh{\ell,k}{\hA{j}} \AmeasBBjh{j+1,\ell-1}{\hB{j}}(\rhoh{\h{j}}) 
- \AAunitaryAj{j} \big[ \EAjh{\ell,k}{\hA{j}{}'} \AmeasBBjh{j+1,\ell-1}{\hB{j}}(\rhoh{\h{j}{}'}) \big] \AAunitaryAj{j}{}^\dagger
} \leq O(\sqrt \epsilon) + 2 \delta' (k - \ell - 1)
 \enspace ,
\end{equation*}
as claimed.  
\end{proof}

In particular, letting $\ell = j$ in \claimref{t:eliminateBobbyinduction}, we obtain Eq.~\eqnref{e:removedmostBobmeasurements}.  
\end{proof}

Letting $\delta' = 2n (2 \delta + O(\sqrt \epsilon))$ and $f(\h{j}) = \trnorm{ \rhoh{\h{j}} - \ABunitaryBj{1,j}(\h{j}) \rhoAh{\hA{j}} \ABunitaryBj{1,j}(\h{j})^\dagger }$, by \lemref{t:allCHSHBobgameoutcomesrelatedbysinglequbitunitary}, 
\begin{equation*}
\Ex\!\big[ f(H_j) \big\vert \text{game~$(j, H_{j-1})$ is $\epsilon$-structured} \big] \leq \frac{\delta'}{1-\delta} \leq \delta' (1 + 2 \delta)
 \enspace .
\end{equation*}
Thus, given that game~$(j, H_{j-1})$ is $\epsilon$-structured, there is at least a $1 - \sqrt{\delta' (1+2\delta)}$ probability that $H_{j-1}$ lies in the set of $\h{j-1}$ with $\Ex[f(H_j) \vert H_{j-1} = \h{j-1}] \leq \sqrt{\delta' (1+2\delta)}$.  By \corref{t:structuredgameprobabilitylowerbound}, this implies that for all $a_j, x_j, b_j, y_j$, $f(\h{j}) \leq 60 \sqrt{\delta' (1+2\delta)}$.  

Combined with \propref{t:removedmostBobmeasurements}, there is at least a $(1-\delta)(1 - \sqrt{\delta'(1+2\delta)}) - \sqrt{n \delta}$ probability that $H_j$ lies in the set of $\h{j}$ satisfying, for all $a_j'$ and~$x_j'$, 
\begin{equation*}\begin{aligned}
\Bigtrnorm{
\ABunitaryBj{1,j}(\h{j}) \EAjh{j+1,k}{\hA{j}}(\rhoAh{\hA{j}}) \ABunitaryBj{1,j}(\h{j})^\dagger
- \AAunitaryAj{j} \ABunitaryBj{1,j}(\h{j}{}') \EAjh{j+1,k}{\hA{j}{}'}(\rhoAh{\hA{j}{}'}) \AAunitaryAj{j}{}^\dagger \ABunitaryBj{1,j}(\h{j}{}')^\dagger
} \\
\leq O(\sqrt \epsilon) + 4n^2 (O(\sqrt \epsilon) + 4 \cdot 60 \sqrt{n \delta}) + 2 \cdot 60 \sqrt{\delta'(1+2\delta)} 
 \enspace .
\end{aligned}\end{equation*}
So far, we have only used that~$\S$ is a single-qubit ideal strategy.  Since $\S$ is in fact a multi-qubit ideal strategy, there is a basis in which $\ABunitaryBj{1,j}(\h{j})$ and $\ABunitaryBj{1,j}(\h{j}{}')$ are both tensor-products of $j$ one-qubit unitaries, with the same unitaries on the first $j-1$ coordinates.  Removing these unitaries does not affect the trace distance in the above inequality and thus it is equivalent to 
\begin{equation*}
\Bigtrnorm{
\EAjh{j+1,k}{\hA{j}}(\rhoAh{\hA{j}})
- \AAunitaryAj{j} {\cal V}^B_j \EAjh{j+1,k}{\hA{j}{}'}(\rhoAh{\hA{j}{}'}) \AAunitaryAj{j}{}^\dagger {\cal V}^B_j{}^\dagger
} 
\leq p(n, \delta, \epsilon)
\end{equation*}
for some polynomial $p(n, \delta, \epsilon)$ that tends to zero with $\delta$ and~$\epsilon$, and ${\cal V}^B_j = \ABunitaryBj{j}(\hB{j-1}, a_j, b_j, x_j \oplus y_j)^\dagger \ABunitaryBj{j}(\hB{j-1}, a_j', b_j, x_j' \oplus y_j)$.  Finally, observe that for all $\chi \in \{0,1\}$, ${\cal V}^B_j$ maps $\ket{(a_j', \chi)_A}$ to $\ket{(a_j, \chi \oplus x_j' \oplus x_j)_A}$, and so ${\cal V}^B_j = \AAunitaryBj{j}(\hB{j-1}, a_j', a_j, x_j' \oplus x_j)$, as claimed.  This completes the proof of \thmref{t:localgluing}.  
\end{proof}

Of course, a symmetrical statement to \thmref{t:localgluing} holds also for Bob's super-operators.

\subsubsection{Global gluing}

Our global gluing argument will fix a gluing target, a transcript~$\hdec{\hat}{n}$.  For partial transcripts~$\hA{n}$, we consider the path $\lambda^{(0)} = \hA{k}, \lambda^{(1)}, \lambda^{(2)}, \ldots, \lambda^{(k)} = \hdec{\hat}{k}$ where $\lambda^{(j-1)}$ and $\lambda^{(j)}$ differ only possibly in the outcomes for game~$j$.  We will compare $\rhoAh{\hA{k}}$ to $\rhoAh{\hAdec{\hat}{k}}$ by applying local gluing comparisons along each step of the path.  (It is important that the coordinates be changed in increasing order, so that the unitary corrections for each comparison depend only on $\hdec{\hat}{n}$.)  We will choose $\hdec{\hat}{n}$ so that all local gluing steps along the path succeed, for most transcripts~$\hA{n}$.  Therefore, we will end up showing that the provers' strategy~$\S$ is simulated by an ideal strategy in which they use the qubits defined by the fixed transcript~$\hdec{\hat}{n}$, regardless of the observed transcript.  

\begin{definition} \label{t:BBunitarydef}
Let $\S$ be a strategy such that the operators $\UXjh{j}{\hX{j-1}}$ are unitary, for $\device \in \{A, B\}$.  Similar to \defref{t:AAunitarydef}, define unitary operators that rotate between Bob's different measurement bases:  
\begin{equation}\begin{split}
\BBunitary(b, b', \Delta) &= \sum_{y \in \{0,1\}} \ketbra{(b', y \oplus \Delta)_B}{(b, y)_B} \\
\BBunitaryXj{j}(\hX{j-1}, b, b', \Delta) &= \UXjh{j}{\hX{j-1}}^\dagger \big(\BBunitary(b, b', \Delta) \otimes \identity\big) \UXjh{j}{\hX{j-1}}
 \enspace .
\end{split}\end{equation}
Furthermore, for notational brevity, define super-operators $\AAunitarysupABj{j}(\h{j-1}, a, a', \Delta)$ and $\BBunitarysupABj{j}(\h{j-1}, b, b', \Delta)$ by 
\begin{equation}\begin{aligned}
\AAunitarysupABj{j}(\h{j-1}, a, a', \Delta)(\sigma) &= \AAunitaryABj{j} \sigma \AAunitaryABj{j}{}^\dagger \\
\BBunitarysupABj{j}(\h{j-1}, b, b', \Delta)(\sigma) &= \BBunitaryABj{j} \sigma \BBunitaryABj{j}{}^\dagger
 \enspace ,
\end{aligned}\end{equation}
where $\AAunitaryABj{j} = \AAunitaryAj{j}(\hA{j-1}, a, a', \Delta) \AAunitaryBj{j}(\hB{j-1}, a, a', \Delta)$ and $\BBunitaryABj{j} = \BBunitaryAj{j}(\hA{j-1}, a, a', \Delta) \BBunitaryBj{j}(\hB{j-1}, a, a', \Delta)$.  
\end{definition}

In the global gluing argument, we will need to handle various conditional probability distributions, such as the distribution of outcomes for game~$k$, $H_{k,k}$ conditioned on $H_j = \h{j}$ for different values of~$j$.  Unfortunately, for some transcripts~$\h{n}$, these distributions can depend heavily on~$j$, preventing us from coupling them together.  This is a minor technical difficulty, not a serious obstacle.  To get around it, we will move to the distribution $\hat H_n$ of transcripts for $n$ ideal CHSH games, in which each game is independent.  This can be done at little cost if most games are structured: 

\begin{lemma} \label{t:structuretovariationdistance}
If $\Pr[\text{every game along~$H_n$ is $\epsilon$-structured} \vert H_j = \h{j}] \geq 1-\delta$, then the total variation distance between the distribution of $H_n$ conditioned on $H_j = \h{j}$ and the distribution of $\hat H_n$, from an ideal CHSH strategy, conditioned on $\hat H_j = \h{j}$, satisfies 
\begin{equation}
d_{TV}(H_n \vert H_j = \h{j}, \hat H_n \vert \hat H_j = \h{j}) \leq \delta + 2 \cdot 60 n \epsilon
 \enspace .
\end{equation}
\end{lemma}

\begin{proof}
Except for notational complications, the proof is the same whether or not we condition on a partial transcript $\h{j}$.  Therefore for simplicity assume $j = 0$.  

Let $\mu(\h{n}) = \Pr[H_n = \h{n}]$ and $\nu(\h{n}) = \Pr[\hat H_n = \h{n}]$.  Then $d_{TV}(\mu, \nu) = \sum_{\h{n} : \mu(\h{n}) > \nu(\h{n})} (\mu(\h{n}) - \nu(\h{n}))$, which is at most $\delta$ plus the same sum restricted further to transcripts~$\h{n}$ along which all games are $\epsilon$-structured.  If all games along~$\h{n}$ are $\epsilon$-structured, then by definition $\mu(\h{n}) \leq \prod_j (p_j + \epsilon)$, whereas $\nu(\h{n}) = \prod_j p_j$, where $p_j = \Pr[\hat H_{j,j} = \h{j,j}] \geq 1/60$ (\corref{t:structuredgameprobabilitylowerbound}).  Therefore, $d_{TV}(\mu, \nu) \leq \delta + \sum_{\h{n}} \mu(\h{n}) \big(1 - \prod_j p_j / (p_j + \epsilon)\big) \leq \delta + 1 - (1-2 \cdot 60 \epsilon)^n \leq \delta + 2 \cdot 60 n \epsilon$.  
\end{proof}

\begin{lemma} \label{t:goodgluingtargetsexist}
If $\S$ is a $(\delta, \epsilon)$-structured multi-qubit ideal strategy for $n$ sequential CHSH games, such that the operators~$\UXjh{j}{\hX{j-1}}$ are unitary, then for $p(n, \delta, \epsilon)$ the polynomial from \thmref{t:localgluing}, there is at least a $1 - 2 n^2 p(n, \delta, \epsilon) - 2 n \sqrt{n \delta}$ probability that $H_n$ lies in the set 
\begin{equation} \label{e:goodgluingtargetsexist}
S = \Bigg\{ \h{n} : \forall k , \, \min\!\Bigg\{ \begin{aligned}
\Pr\!\big[ \bigtrnorm{ \rhoAh{\hA{k}} - \AAunitarysupABj{1,k}( \rhoAh{\hat H^A_k} ) } \leq n \sqrt{2 p(n, \delta, \epsilon)} \big], \\
\Pr\!\big[ \bigtrnorm{ \rhoAh{\hB{k}} - \BBunitarysupABj{1,k}( \rhoBh{\hat H^B_k} ) } \leq n \sqrt{2 p(n, \delta, \epsilon)} \big] 
\end{aligned}\Bigg\} \geq 1 - n \delta' \Bigg\}
 \enspace .
\end{equation}
Here $\AAunitarysupABj{1,k} = \AAunitarysupABj{k}(\h{k-1}, \hat A_k, a_k, \hat X_k \oplus x_k) \cdots \AAunitarysupABj{1}(\hat A_1, a_1, \hat X_1 \oplus x_1)$ and $\delta' = \sqrt{2 p(n, \delta, \epsilon)} + (\sqrt{n \delta} + 2 \cdot 60 n \epsilon)$.  
\end{lemma}

\begin{proof}
By \thmref{t:localgluing} and a union bound over $j$ and $k$, there is at least a $1 - n^2 p(n, \delta, \epsilon)$ probability that~$H_n$ lies in the set 
\begin{equation*}
S_1 = \Big\{ \h{n} : \forall j, k, a_j', x_j', \, \bigtrnorm{
\EAjh{j+1,k}{\hA{j}}(\rhoAh{\hA{j}})
- \AAunitarysupABj{j} \EAjh{j+1,k}{\hA{j}{}'}(\rhoAh{\hA{j}{}'})
} 
\leq p(n, \delta, \epsilon) \Big\}
 \enspace ,
\end{equation*}
with $\AAunitarysupABj{j} = \AAunitarysupABj{j}(\h{j-1}, a_j', a_j, x_j' \oplus x_j)$.  
By \lemref{t:blockdiagonaltracedistance} and a Markov inequality, $S_1$ is a subset of 
\begin{equation*}
S_2 = \Big\{ \h{n} : \forall j, k, a_j', x_j', \, 
\Pr\!\big[ \trnorm{ \rhoAh{H^A_k} - \AAunitarysupABj{j}( \rhoAh{H^A_k{}'} ) } 
\leq \sqrt{2 p(n, \delta, \epsilon)} \,\big\vert\, H_j = \h{j} \big] \geq 1 - \sqrt{2 p(n, \delta, \epsilon)}
\Big\}
.
\end{equation*}

Furthermore, there is at least a $1 - n \sqrt{n \delta}$ probability that $H_n$ lies in the set 
\begin{equation*}\begin{split}
S_3 = \Big\{ \h{n} : \forall j, \, \Pr[\text{$\forall \, i > j$, game~$(i, H_{i-1})$ is $\epsilon$-structured} \,\vert\, H_j = \h{j}] &\geq 1 - \sqrt{n \delta} \Big\}
 \enspace .
\end{split}\end{equation*}
By \lemref{t:structuretovariationdistance}, $S_3$ is a subset of 
\begin{equation*}\begin{split}
S_4 = \Big\{ \h{n} : \forall j, \, 
d_{TV}(H^A_n \vert H_j = \h{j}, \hat H^A_n \vert \hat H_j = \h{j})
\leq \sqrt{n \delta} + 2 \cdot 60 n \epsilon \Big\}
 \enspace .
\end{split}\end{equation*}

Taking the intersection of $S_2$ and $S_4$, we obtain that there is at least a $1 - n^2 p(n, \delta, \epsilon) - n \sqrt{n \delta}$ probability that $H_n$ lies in the set 
\begin{equation*}
S_5 = \Big\{
\h{n} : \forall j, k, a_j', x_j', \, 
\Pr\!\big[ 
\bigtrnorm{ \rhoAh{\hat H^A_k} - \AAunitarysupABj{j}\big( \rhoAh{\hat H^A_k{}'} \big) } 
\leq \sqrt{2 p(n, \delta, \epsilon)} \,\big\vert\, \hat H_j = \h{j}\big]
\geq 1 - \delta'
\Big\}
 \enspace .
\end{equation*}

Let $\h{n} \in S_5$.  Since the different game coordinates are independent of each other in the ideal distribution~$\hat H_n$, we have 
\begin{equation*}
\Pr\!\Big[ 
\bigtrnorm{ \rhoAh{\hA{j}, \hat H^A_{j+1,k}} - \AAunitarysupABj{j}\big( \rhoAh{\hA{j}{}', \hat H^A_{j+1,k}} \big) } 
\leq \sqrt{2 p(n, \delta, \epsilon)} \Big] \geq 1 - \delta'
 \enspace ,
\end{equation*}
without conditioning on $\hat H_j = \h{j}$.  Since this holds for all~$a_j', x_j'$, in particular we find 
\begin{equation*}
\Pr\!\Big[ 
\bigtrnorm{ \rhoAh{\hA{j}, \hat H^A_{j+1,k}} - \AAunitarysupABj{j}\big( \rhoAh{\hA{j-1}, \hat H^A_{j,k}} \big) } 
\leq \sqrt{2 p(n, \delta, \epsilon)} \Big]
\geq 1 - \delta'
 \enspace ,
\end{equation*}
where now $\AAunitarysupABj{j} = \AAunitarysupABj{j}(\h{j-1}, \hat A_j, a_j, \hat X_j \oplus x_j)$.  
By a union bound, 
\begin{equation*}
\Pr\!\Big[ \forall j, \;
\bigtrnorm{ \rhoAh{\hA{j}, \hat H^A_{j+1,k}} - \AAunitarysupABj{j}\big( \rhoAh{\hA{j-1}, \hat H^A_{j,k}} \big) } 
\leq \sqrt{2 p(n, \delta, \epsilon)} \Big]
\geq 1 - n \delta'
 \enspace .
\end{equation*}

For $i \leq j$, let $\AAunitarysupABj{i,j} = \AAunitarysupABj{j}(\h{j-1}, \hat A_j, a_j, \hat X_j \oplus x_j) \cdots \AAunitarysupABj{i}(\h{i-1}, \hat A_i, a_i, \hat X_i \oplus x_i)$.  A triangle inequality based on the expansion $\rhoAh{\hA{k}} - \AAunitarysupABj{1,k}(\rhoAh{\hat H^A_k}) = \sum_{j \in [k]} \AAunitarysupABj{j+1,k}\big( \rhoAh{\hA{j}, \hat H^A_{j+1,k}} - \AAunitarysupABj{j} \rhoAh{\hA{j-1}, \hat H^A_{j,k}} \big)$ implies that 
\begin{equation*}
\Pr\!\Big[ 
\bigtrnorm{ \rhoAh{\hA{k}} - \AAunitarysupABj{1,k}( \rhoAh{\hat H^A_k} ) } 
\leq n \sqrt{2 p(n, \delta, \epsilon)} \Big]
\geq 1 - n \delta'
 \enspace .
\end{equation*}
This is one of the two bounds needed in the definition of~$S$, Eq.~\eqnref{e:goodgluingtargetsexist}.  Symmetrical arguments from Bob's perspective, and one final union bound, complete the proof of \lemref{t:goodgluingtargetsexist}.  
\end{proof}

Before proving \thmref{t:globalgluing}, we need one last lemma, that characterizes the states at the beginning of each game along a structured transcript in a multi-qubit ideal strategy: 

\begin{lemma} \label{t:structuredtranscripthasnEPRpairs}
Fix a transcript $\h{n}$ along which every game is $\epsilon$-structured according to the multi-qubit ideal strategy~$\S$.  For $\device \in \{A, B\}$, let $\UmultiX(\hX{n}) = \big(\identity_{(\C^2)^{\otimes (n-1)}} \otimes \UmultiXjh{n}{\hX{n-1}}\big) \ldots \UmultiXj{1} \XmultiX$.  Then there exists a state $\ket{\psione'}$ such that for all~$k$,
\begin{equation}
\Bignorm{\UmultiA \UmultiB \ket{\psijh{k}{\h{k}}} - {\textstyle \bigotimes}_{j \in [k]} (\ket{(a_j, x_j)_A} \ket{(b_j, y_j)_B}) \otimes \ket{\psi^*}{}^{\otimes (n-k)} \otimes \ket{\psione'}} \leq n \, O(\sqrt \epsilon)
 \enspace .
\end{equation}
\end{lemma}

\begin{proof}
Alice and Bob play each game~$k$ along~$\h{n}$ according to the ideal CHSH game strategy on their $k$th qubits.  The state at the beginning of game~$(k+1, \h{k})$ is $\UmultiA \UmultiB \ket{\psijh{k+1}{\h{k}}} = \bigotimes_{j \in [k]} (\ket{(a_j, x_j)_A} \ket{(b_j, y_j)_B}) \otimes \ket{\psij{k+1}'}$ for some state~$\ket{\psij{k+1}'}$.  By the CHSH rigidity lemma, \lemref{t:eprlemma}, there exists a state $\ket{\psij{k+1}''}$ such that $\norm{ \ket{\psij{k+1}'} - \ket{\psi^*} \otimes \ket{\psij{k+1}''} } = O(\sqrt \epsilon)$.  Since for sufficiently small~$\epsilon$ every outcome of the game occurs with probability at least~$1/60$ (\corref{t:structuredgameprobabilitylowerbound}), it follows too that $\norm{\ket{\psij{k+1}''} - \ket{\psij{k+2}'}} = O(\sqrt \epsilon)$.  Thus $\norm{ \ket{\psij{k+1}'} - \ket{\psi^*} \otimes \ket{\psij{k+2}'} } = O(\sqrt \epsilon)$.  Chain together these inequalities for $\norm{\ket{\psij{k+1}'} - \ket{\psi^*}^{\otimes (n-k)} \ket{\psij{n+1}'}} \leq n \, O(\sqrt \epsilon)$.  
\end{proof}

\begin{proof}[Proof of \thmref{t:globalgluing}]
If $\hdec{\hat}{n}$ belongs to the set~$S$ from Eq.~\eqnref{e:goodgluingtargetsexist}, then 
\begin{equation*}
\sum_{\hA{k}} \Pr[\hat H^A_k = \hA{k}] \bigtrnorm{ \rhoAh{\hA{k}} - \AAunitarysupABj{1,k} \rhoAh{\hAdec{\hat}{k}} } 
\leq n \sqrt{2 p(n, \delta, \epsilon)} + 2 n \delta'
 \enspace ,
\end{equation*}
where $\AAunitarysupABj{1,k} = \AAunitarysupABj{j}(\hdec{\hat}{j-1}, \hat a_j, a_j, \hat x_j \oplus x_j) \cdots \AAunitarysupABj{i}(\hat a_1, a_1, \hat x_1 \oplus x_1)$.  
Also, by \lemref{t:structuretovariationdistance}, 
\begin{equation*}
\bigtrnorm{ \EAj{1,k}(\rhoone) - \sum_{\hA{k}} \Pr[\hat H^A_k = \hA{k}] \ketbra{\hA{k}}{\hA{k}} \otimes \rhoAh{\hA{k}} } 
= 2 d_{TV}(H^A_k, \hat H^A_k)
\leq 2 (n \delta + 2 \cdot 60 n \epsilon)
 \enspace .
\end{equation*}
Therefore, 
\begin{equation} \label{e:globalgluing1}
\bigtrnorm{
\EAj{1,k}(\rhoone) - \sum_{\hA{k}} \Pr[\hat H^A_k = \hA{k}] \ketbra{\hA{k}}{\hA{k}} \otimes \AAunitarysupABj{1,k} \rhoAh{\hAdec{\hat}{k}}
} \leq n \sqrt{2 p(n, \delta, \epsilon)} + 2 n \delta' + 2 (n \delta + 2 \cdot 60 n \epsilon)
 \enspace .
\end{equation}
Of course, a symmetrical bound holds from Bob's perspective.  

Therefore, to bound $\trnorm{\EAj{1,k}(\rhoone) - \EAdecj{\hat}{1,k}(\rhodecone{\hat})}$, and by symmetry $\trnorm{\EBj{1,k}(\rhoone) - \EBdecj{\hat}{1,k}(\rhodecone{\hat})}$, in order to prove \thmref{t:globalgluing}, we need only to bound the trace distance from $\EAdecj{\hat}{1,k}(\rhodecone{\hat})$ to $\sum_{\hA{k}} \Pr[\hat H^A_k = \hA{k}] \ketbra{\hA{k}}{\hA{k}} \otimes \AAunitarysupABj{1,k} \rhoAh{\hAdec{\hat}{k}}$.  For this, we will apply Lemmas~\ref{t:allCHSHBobgameoutcomesrelatedbysinglequbitunitary} and~\ref{t:structuredtranscripthasnEPRpairs}.  

By Eq.~\eqnref{e:allCHSHBobgameoutcomesrelatedbysinglequbitunitaryExpectation} in \lemref{t:allCHSHBobgameoutcomesrelatedbysinglequbitunitary} and a Markov inequality, there exists a constant~$\varkappa$ such that, for $\delta'' = \sqrt{2 n (2\delta + \varkappa \epsilon^{1/2})}$, 
\begin{equation*}
\Pr\!\big[ \bigtrnorm{ \density\big( \ket{\psijh{k+1}{H_k}} \big) - \density\big( \ABunitaryBj{1,k}(H_k) \ket{\psiAh{H^A_k}} \big) } \leq \delta'' \big] \geq 1 - \delta''
 \enspace .
\end{equation*}
Therefore, if we let 
\begin{equation*}
T = \Bigg\{
\h{n} : \begin{aligned}
&\forall k, \, \text{game $(k, \h{k-1})$ is $\epsilon$-structured}, \,
\bigtrnorm{ \density\big( \ket{\psijh{k+1}{\h{k}}} \big) - \density\big( \ABunitaryBj{1,k}(\h{k}) \ket{\psiAh{\hA{k}}} \big) } \leq \delta'', \\
&\quad \text{and the symmetrical bound from Bob's perspective holds}
\end{aligned}
\Bigg\}
 \enspace ,
\end{equation*}
then by a union bound, $\Pr[H_n \in T] \geq 1 - n \delta - 2 n \delta''$.  

Assume that $\hdec{\hat}{n} \in T$.  For $\device \in \{A, B\}$, let $\UmultiX = \UmultiX(\hXdec{\hat}{n})$ be the operators defined by \lemref{t:structuredtranscripthasnEPRpairs}.  Let the initial state for~$\hat \S$ be $\ket{\psidecone{\hat}} = \UmultiA{}^\dagger \UmultiB{}^\dagger \ket{\psi^*}^{\otimes n} \otimes \ket{\psione'}$.  Then by \defref{t:ABunitarydef} for $\ABunitaryBj{1,k}(\hdec{\hat}{k})$, a triangle inequality, and \lemref{t:structuredtranscripthasnEPRpairs}, 
\begin{equation*}\begin{split}
&\bigtrnorm{
\density\big(\ket{\psiAh{\hAdec{\hat}{k}}}\big)
- \density\big(\ket{\psiAdech{\hat}{\hAdec{\hat}{k}}}\big)
} \\
&\qquad= \Bigtrnorm{
\density\big(\ket{\psiAh{\hAdec{\hat}{k}}}\big)
- \density\big(\UmultiA{}^\dagger \UmultiB{}^\dagger \bigotimes_{j \in [k]} (\ket{(\hat a_j, \hat x_j)_A} \ket{(\hat a_j, \hat x_j)_A}) \otimes \ket{\psi^*}{}^{\otimes (n-k)} \otimes \ket{\psione'}\big)
} \\
&\qquad= 
\Bigtrnorm{
\density\big(\ket{\psiAh{\hAdec{\hat}{k}}}\big)
- \density\big(\ABunitaryBj{1,k}(\hdec{\hat}{k})^\dagger \UmultiA{}^\dagger \UmultiB{}^\dagger \bigotimes_{j \in [k]} (\ket{(\hat a_j, \hat x_j)_A} \ket{(\hat b_j, \hat y_j)_A}) \ket{\psi^*}{}^{\otimes (n-k)} \ket{\psione'}\big)
} \\
&\qquad\leq \bigtrnorm{ \density(\ket{\psiAh{\hAdec{\hat}{k}}}) - \density(\ABunitaryBj{1,k}(\hdec{\hat}{k})^\dagger \ket{\psijh{k+1}{\hdec{\hat}{k}}}) } \\
&\qquad\quad + \Bigtrnorm{ \density(\ket{\psijh{k+1}{\hdec{\hat}{k}}}) - \density\big( \UmultiA{}^\dagger \UmultiB{}^\dagger \bigotimes_{j \in [k]} (\ket{(\hat a_j, \hat x_j)_A} \ket{(\hat b_j, \hat y_j)_A}) \ket{\psi^*}{}^{\otimes (n-k)} \ket{\psione'} \big) } \\
&\qquad\leq \delta' + n \, O(\sqrt \epsilon) 
 \enspace .
\end{split}\end{equation*}
Since for any transcript $\hA{k}$, $\AAunitarysupABj{1,k}\density\big(\ket{\psiAdech{\hat}{\hAdec{\hat}{k}}}\big) = \density\big(\UmultiA{}^\dagger \UmultiB{}^\dagger \bigotimes_{j \in [k]} (\ket{(a_j, x_j)_A} \ket{(a_j, a_j)_A}) \ket{\psi^*}{}^{\otimes (n-k)} \ket{\psione'}\big) = \density\big(\ket{\psiAdech{\hat}{\hA{k}}}\big)$, it follows that $\bigtrnorm{\AAunitarysupABj{1,k}\density\big(\ket{\psiAh{\hAdec{\hat}{k}}}\big) - \density\big(\ket{\psiAdech{\hat}{\hA{k}}}\big)} \leq \delta' + n \, O(\sqrt \epsilon)$.  Thus, 
\begin{align}
\Bigtrnorm{
\EAdecj{\hat}{1,k}(\rhodecone{\hat}) 
- \sum_{\hA{k}} \Pr[\hat H^A_k = \hA{k}] \ketbra{\hA{k}}{\hA{k}} \otimes \AAunitarysupABj{1,k} \rhoAh{\hAdec{\hat}{k}}
}
&= \sum_{\hA{k}} \Pr[\hat H^A_k = \hA{k}] \bigtrnorm{ \rhoAdech{\hat}{\hA{k}} - \AAunitarysupABj{1,k} \rhoAh{\hAdec{\hat}{k}} } \nonumber \\
&\leq \delta' + n \, O(\sqrt \epsilon) \label{e:globalgluing2}
 \enspace .
\end{align}

Putting together Eqs.~\eqnref{e:globalgluing1} and~\eqnref{e:globalgluing2}, we obtain that $\EXj{1,k}(\rhoone) \approx \EXdecj{\hat}{1,k}(\rhodecone{\hat})$ for $\device \in \{A, B\}$, provided that~$\hdec{\hat}{n} \in S \cap T$.  
This completes the proof of \thmref{t:globalgluing}.  
\end{proof}

\subsection{Converse to Tsirelson's inequality based on observed correlations}

By combining \thmref{t:sequentialCHSHgames} and some simple statistics, we can extend \lemref{t:eprlemma} to obtain a converse to Tsirelson's inequality that depends on the \emph{observed} correlations in a repeated game---\thmref{t:sequentialstructureandobservedcorrelations} below.  As a consequence, we will also derive efficient ``self-testing" for sequential CHSH games, in \thmref{t:sequentialCHSHselftesting} below.  

Let Alice and Bob be the two entangled provers playing $n$ CHSH games, in sequence, with independent questions refereed by the verifier Eve.  Let $W = \abs{\{j \in [n] : A_j B_j = X_j \oplus Y_j \}}$ be the number of games that Alice and Bob win.  If Alice and Bob use an ideal strategy, i.e., a $0$-structured strategy, for all games, then by Hoeffding's inequality they are likely to win nearly $\cos^2(\pi/8) n$ games: 

\begin{lemma} \label{t:honestproverscentrallimit}
If Alice and Bob use an ideal strategy for $n$ sequential CHSH games, then 
\begin{equation}
\Pr[W \geq (\cos^2(\pi/8) - \delta) n] \geq 1 - e^{-2 \delta^2 n}
 \enspace .
\end{equation}
\end{lemma}

Conversely, let $S$ be the number of games in which the provers' joint strategies are $\epsilon$-structured.  We first claim that with high probability, either nearly all games are $\epsilon$-structured or Alice and Bob win significantly fewer than $\cos^2(\pi/8) n$ games.  This lemma is a warm-up to \thmref{t:sequentialstructureandobservedcorrelations} below.  

\begin{lemma} \label{t:structureandobservedcorrelations}
Let $\epsilon, \eta > 0$ and $\delta \leq \eta \epsilon/8$.  Then  
\begin{equation}
\Pr\!\big[ \text{$W \geq (\cos^2(\pi/8) - \delta) n$ and $S < (1 - \eta) n$} \big] 
\leq e^{-2 n (\eta \epsilon/8 - \delta)^2}
 \enspace .
\end{equation}
\end{lemma}

\begin{proof}
Let $S_1, S_2, \ldots, S_n$ and $W_1, W_2, \ldots, W_n$ be $0/1$-valued random variables, $S_j$ being an indicator for whether the $j$th game is played in an $\epsilon$-structured fashion, and $W_j$ an indicator for $A_j \wedge B_j = X_j \oplus Y_j$, i.e., for the provers winning the $j$th game.  Then $S = \sum_j S_j$ and $W = \sum_j W_j$.  Let $p = \cos^2(\pi/8)$ and $\epsilon' = \epsilon/8$.  We know 
\begin{align*}
\Pr[W_j = 1 \vert S_j = 1] &\leq p & 
\Pr[W_j = 1 \vert S_j = 0] &\leq p - \epsilon' 
 \enspace .
\end{align*}

Let $\Gamma_1, \ldots, \Gamma_n$ be independent Bernoulli($p$) random variables, and $\Lambda_1, \ldots, \Lambda_n$ be independent Bernoulli($p - \epsilon'$) random variables.  Couple $W_j$ for the first structured game to~$\Gamma_1$ such that $W_j \leq \Gamma_1$, for the second structured game to~$\Gamma_2$, and so on.  Similarly, couple $W_j$ for the first unstructured game to~$\Lambda_n$ such that $W_j \leq \Lambda_n$, for the second unstructured game to~$\Lambda_{n-1}$, and so on.  This yields the bound 
\begin{align*}
\Pr\!\big[ W \geq (p - \delta) n,\, S < (1 - \eta) n \big] 
&\leq \Pr\!\Big[ \sum_{j \leq S} \Gamma_j + \sum_{j > S} \Lambda_j \geq ( p - \delta ) n,\, S < (1 - \eta)n \Big] \\
&\leq \Pr\!\Big[ \sum_{j \leq (1-\eta)n} \Gamma_j + \sum_{j > (1-\eta)n} \Lambda_j \geq (p - \delta) n \Big]
 \enspace .
\end{align*}
Let $X = \sum_{j \leq (1-\eta)n} \Gamma_j + \sum_{j > (1-\eta)n} \Lambda_j$ and $\mu = \Ex[X] = (p - \eta \epsilon') n$.  Hoeffding's inequality implies that if $\delta \leq \eta \epsilon'$, then $\Pr[X \geq (p - \delta) n] \leq \exp(-2 n (\eta \epsilon' - \delta)^2)$.  
\end{proof}

This lemma can be seen as a weak converse to Tsirelson's inequality based on the observed correlations for a sequence of CHSH games.  It says that if the provers do not use a structured strategy most of the time, then they are unlikely to win too many games.  Our goal, though, is to prove a stronger statement, based on \thmref{t:sequentialCHSHgames}: If the provers do not use a nearly ideal strategy for most subsequences of games, then they are unlikely to win too many games.  The logic behind this claim will be essentially the same as that behind \lemref{t:structureandobservedcorrelations}.  

\begin{definition} \label{t:epsilonideal}
For $\epsilon > 0$, call a strategy $\S$ for $n$ sequential CHSH games \emph{$\epsilon$-ideal} if an isometric extension of $\S$ is $\epsilon$-simulated by an ideal strategy.  $\S$ is $\epsilon$-ideal with respect to the isometries $\XX : \H_\device \hookrightarrow (\C^2)^{\otimes n} \otimes \H_\device'$, for $\device \in \{A, B\}$, if the isometric extension of~$\S$ by $\XA$ and~$\XB$ is $\epsilon$-simulated by an ideal strategy.  
\end{definition}

In particular, if $\S$ is $\epsilon$-ideal with respect to $\XA$ and $\XB$, then for $\ket{\psione} \in \H_A \otimes \H_B \otimes \H_C$ the initial state, there exists a state $\ket{\psione'} \in \H_A' \otimes \H_B' \otimes \H_C$ such that, letting $\density(\ket a) = \ketbra a a$ and $\rhoone = \density(\ket \psione)$, $\rhodecone{\hat} = \density(\ket{\psi^*}{}^{\otimes n} \otimes \ket{\psione'})$ and $\XAB(\rho) = (\XA \otimes \XB) \rho (\XA \otimes \XB)^\dagger$, 
\begin{align}
\bigtrnorm{ \XAB(\rhoone) - \rhodecone{\hat} } &\leq \epsilon 
& \text{and }\qquad\quad
\bigtrnorm{ \XAB \EXj{1,n}(\rhoone) - \EXdecj{\hat}{1,n}(\rhodecone{\hat}) } &\leq 2 \epsilon
\end{align}
for $\device \in \{A, B\}$.  Here, $\EXj{1,n}$ is the measurement super-operator for prover~$\device$, and $\EXdecj{\hat}{1,n}$ is the ideal measurement super-operator that uses the $j$th qubit in game~$j$ of the set (\defref{t:CHSHserialnotation}).  

\begin{theorem} \label{t:sequentialstructureandobservedcorrelations}
Let Alice and Bob play in sequence $N$ sets each of $n$ sequential CHSH games.  Let $W \leq N n$ be the total number of games that Alice and Bob win.  Fix $\epsilon > 0$, and let $G \leq N$ be the number of sets of games for which the provers' joint strategy for that set, conditioned on the previous games' outcomes, is $\kappaEPR n^{\kappaEPR} \epsilon^{1/\kappaEPR}$-ideal, where~$\kappaEPR$ is the constant from \thmref{t:sequentialCHSHgames}.  Let $\eta > 0$.  Then for any $\delta$ such that $t = \frac{1}{8} \epsilon^2 \eta N - \delta N n \geq 0$, 
\begin{equation}
\Pr\!\big[ \text{$W \geq (\cos^2(\pi/8) - \delta) N n$ and $G < (1 - \eta) N$} \big] 
\leq \exp( -t^2 / (2 N n) )
 \enspace .
\end{equation}
\end{theorem}

\begin{proof}
For $j \in [N n]$, let $W_j = 1$ if $X_j \wedge Y_j = A_j \oplus B_j$, i.e., if the provers win game~$j$, and let $W_j = 0$ otherwise.  Let $S_j$ be the indicator variable for game~$j$ being $\epsilon$-structured.  For $k \in [n]$, let $G_k = 1$ if after $k-1$ sets of games, the provers' strategy for the next set is $\epsilon$-ideal, and let $G_k = 0$ otherwise.  Then $W = \sum_{j \in [N n]} W_j$ and $G = \sum_{k \in [n]} G_k$.  

For $k \in [N]$, let $H_k$ be the indicator variable for the $k$th set of games being $\epsilon$-structured (\defref{t:structuredstrategydef}).  By the contrapositive to \thmref{t:sequentialCHSHgames}, if $G_k = 0$ then the strategy for the $k$th set of games cannot be $\epsilon$-structured, so $H_k = 0$ also.  That is, $H_k \leq G_k$, so letting $H = \sum_k H_k$, $\Pr\!\big[ W \geq (\cos^2(\pi/8) - \delta) N n, \, G < (1 - \eta) N \big] \leq \Pr\!\big[ W \geq (\cos^2(\pi/8) - \delta) N n, \, H < (1 - \eta) N \big]$.  

Let $p = \cos^2(\pi/8)$ and $\epsilon' = \epsilon/8$.  Let $\Gamma_1, \ldots, \Gamma_{N n}, \Lambda_1, \ldots, \Lambda_{N n}$ be independent random variables, with $\Gamma_j \sim \text{Bernoulli}(p)$ and $\Lambda_j \sim \text{Bernoulli}(p - \epsilon')$.  Let $T_1, \ldots, T_{N n}$ be Bernoulli($1 - \epsilon$) random variables.  As in the proof of \lemref{t:structureandobservedcorrelations}, the idea now is to design an appropriate coupling from the~$W_j$ to these simpler random variables.  

For $k \in [N]$, let $\varsigma(k) = \{(k-1)n+1, \ldots, k n\}$ be the set of games in the $k$th set.  If $H_k = 1$, then let $J_k = k n$.  If $H_k = 0$, then let $J_k$ be the largest index $j \in \varsigma(k)$ for which the probability that game~$j$ is $\epsilon$-structured is less than $1-\epsilon$.  By \defref{t:structuredstrategydef}, such an index exists, so $J_k$ is well-defined.  

Define the coupling as follows.  First, for $j \in \varsigma(k)$ and $k \in [N]$, couple $W_j$ to a random variable $\Xi_j$ such that $W_j \leq \Xi_j$ and 
\begin{equation*}
\Xi_j = \begin{cases}
T_j \Gamma_j + (1-T_j) \Lambda_j & \text{if $H_k = 0$, $j = J_k$ and $k - \sum_{k' \leq k} H_{k'} \leq \lceil \eta N \rceil$} \\
\Gamma_j & \text{otherwise} 
 \enspace .
\end{cases}
\end{equation*} 
In the case $H_k = 0$ and $j = J_k$, this coupling can be achieved by first coupling $S_j \leq T_j$.  Note that although the $\Gamma$ and $\Lambda$ variables are fully independent, the $T_j$ variables are not necessarily independent of each other or of the $\Gamma$ and $\Lambda$ variables.  Call a random variable $\Xi_j$ unstructured if it has the form $T_j \Gamma_j + (1-T_j) \Lambda_j$.  The condition $k - \sum_{k' \leq k} H_{k'} \leq \lceil \eta N \rceil$ ensures that we couple at most $\lceil \eta N \rceil$ $W_j$ variables to unstructured $\Xi_j$.  

The $\Xi_j$ variables have a Markov structure that we will use to define a martingale.  Before doing so, we need to make use of the condition $H < (1 - \eta) N$.  To this purpose, we next define a set of random variables~$\{ \Xi_j' \}$, such that \emph{exactly} $\lceil \eta N \rceil$ of them are unstructured.  For all~$j$ such that $\Xi_j$ is unstructured, let $\Xi_j' = \Xi_j$ be unstructured as well.  If $N - H \geq \lceil \eta N \rceil$, then no more unstructured variables are needed; let $\Xi_j' = \Xi_j = \Gamma_j$ for all remaining~$j$.  Otherwise, we are still missing $\lceil \eta N \rceil - (N - H)$ unstructured variables.  Work backward starting with $j = N n$, setting~$\Xi_j'$ to be unstructured so long as the total number of unstructured $\Xi_j'$ is less than $\lceil \eta N \rceil$.  For all remaining~$j$, let $\Xi_j' = \Xi_j = \Gamma_j$.  Under this construction, notice that $\Xi_j' = \Xi_j$ for all of the initial games, certainly for all $j \in \varsigma(k)$ with $k \leq N - \lceil \eta N / n \rceil$.  Because we set the additional unstructured variable starting from the end, the variables $\Xi_j'$ still form a Markov sequence.  (This would not have been the case had we started with $j = 1$ because $H$ is not then determined.)  

In general, it need not hold that $\sum_j W_j \leq \sum_j \Xi_j'$.  If $H < (1 - \eta) N$, though, then indeed $\sum_j W_j \leq \sum_j \Xi_j'$, since in this case $\Xi_j' = \Xi_j$ for all~$j$.  Therefore, letting $\Xi' = \sum_j \Xi_j'$, 
\begin{equation*}\begin{split}
\Pr\!\big[ W \geq (p - \delta) N n, \, H < (1 - \eta) N \big]
&\leq \Pr\!\big[ \Xi' \geq (p - \delta) N n \big]
 \enspace .
\end{split}\end{equation*}
We will bound this probability using Azuma's inequality for martingales.  The sequence of variables $\Pi_j = \sum_{i \leq j} (\Xi_i' - \Ex[\Xi_i' \vert \Xi_1', \ldots, \Xi_{i-1}'])$ form a martingale, with $\abs{\Pi_j - \Pi_{j-1}} \leq 1$.  By Azuma's inequality, therefore, for any~$t > 0$, 
\begin{equation*}\begin{split}
e^{-t^2 / (2 N n)} 
\geq 
\Pr[\Pi_{N n} \geq t] 
= \Pr\!\big[\Xi' \geq t + {\textstyle \sum}_j \Ex[\Xi_j' \vert \Xi_1', \ldots, \Xi_{j-1}']\big]
 \enspace .
\end{split}\end{equation*}
By construction, there are always $\lceil \eta N \rceil$ unstructured variables $\Xi_j'$, meaning that with probability one, $\sum_j \Ex[\Xi_j' \vert \Xi_1', \ldots, \Xi_{j-1}']\big] = (N n - \lceil \eta N \rceil) p + \lceil \eta N \rceil \big( (1-\epsilon) p + \epsilon (p - \epsilon'') \big) = p N n - \epsilon \epsilon' \lceil \eta N \rceil$.  Therefore set $t = \epsilon \epsilon' \lceil \eta N \rceil - \delta N n \geq \frac18 \epsilon^2 \eta N - \delta N n$ to conclude the proof.  
\end{proof}

Typical values for the parameters in \thmref{t:sequentialstructureandobservedcorrelations} are $\delta \sim 1/\sqrt{N n}$ and $\epsilon^2 \eta \sim \sqrt{n / N}$.  There is of course some freedom in choosing the parameters' exact values.  To simplify later applications, though, we will restate \thmref{t:sequentialstructureandobservedcorrelations} with particular parameter choices, and in a more easily applied form.  We make no attempt to optimize the parameters.  

\begin{theorem} \label{t:sequentialstructureandobservedcorrelationsapplied}
Let $\kappaEPR > 1$ be the constant from \thmref{t:sequentialCHSHgames}.  For $\alpha \geq 16 \kappaEPR^2$ and~$n \geq 100$, let Alice and Bob play in sequence $N \geq n^{\alpha - 1}$ sets each of $n$ sequential CHSH games.  Let $W = \abs{\{ j \in [N n] : A_j B_j = X_j \oplus Y_j \}}$ be the total number of games that Alice and Bob win.  
Say that Eve accepts at the end of the protocol if 
\begin{equation}
W \geq 
\cos^2(\pi/8) N n - \tfrac{1}{2 \sqrt 2} \sqrt{N n \log(N n)}
 \enspace .
\end{equation}

This protocol satisfies the following completeness and soundness conditions: 
\begin{description}
\item[Completeness:]
If Alice and Bob play using an ideal strategy for all $N n$ games, then 
\begin{equation}
\Pr[\text{Eve accepts}] 
\geq 1 - \frac{1}{n^{\alpha/4}}
 \enspace .
\end{equation}
\item[Soundness:]
Assume that $\Pr[\text{Eve accepts}] \geq 1 - \epsilon$.  
Let $\zeta = n^{-\alpha / (32 \kappaEPR)}$.  
Then for $K \in [N]$ chosen uniformly at random, the probability that after $(K-1)n$ games the provers' strategy for the $K$th set of $n$ games is~$\zeta$-ideal satisfies 
\begin{equation}\begin{split}
\Pr[\text{$K$th set of games has $\zeta$-ideal strategy}] 
&\geq 1 - \epsilon - n^{-\alpha/8}
 \enspace .
\end{split}\end{equation}
\end{description}
\end{theorem}

\begin{proof}
The completeness condition follows by \lemref{t:honestproverscentrallimit}.  Therefore, we will only argue soundness.  

Apply \thmref{t:sequentialstructureandobservedcorrelations} with parameters $\delta = k \sqrt{\log (N n) / (N n)}$, $\eta = 24 k \sqrt{\log (N n)} n / (N n)^{1/4}$ and $\epsilon' = 1/(N n)^{1/8}$, with $k = 1/(2 \sqrt 2)$.  Then $t := \frac18 \epsilon'^2 \eta N - \delta N n = 2 k \sqrt{N n \log (N n)} \geq 0$.  Let $\xi = \kappaEPR n^{\kappaEPR} \epsilon'^{1/\kappaEPR}$.  We obtain that, for $G$ being the number of $\xi$-ideal sets of games, 
\begin{equation*}\begin{split}
\Pr[\text{$K$th set is $\xi$-ideal}] 
&\geq (1 - \eta) \Pr[G \geq (1-\eta) N] \\
&\geq (1 - \eta) \big( \Pr[\text{Eve accepts}] - \Pr[\text{Eve accepts}, G < (1-\eta) N] \big) \\
&\geq 1 - \epsilon - \eta - \exp(-t^2 / (2 N n)) \\
&\geq 1 - \epsilon - 10 \sqrt{\log (N n)} n / (N n)^{1/4}
 \enspace .
\end{split}\end{equation*}
Finally, $10 \sqrt{\log (N n)} n / (N n)^{1/4} \leq 10 \sqrt{\alpha \log n} / n^{\alpha/4-1}$, which is at most $n^{-\alpha/8}$ for~$\alpha \geq 16$ and~$n \geq 85$.  Since $\alpha \geq 16 \kappaEPR^2$, $\xi = \kappaEPR n^{\kappaEPR} / (N n)^{1/(8 \kappaEPR)} \leq \kappaEPR n^{-\alpha / (16 \kappaEPR)}$, which is at most $n^{-\alpha / (32 \kappaEPR)}$ for $\alpha \geq 16 \kappaEPR^2$ and~$n \geq 3$.  
\end{proof}

The sequential CHSH game theorems assume that there are only two provers, Alice and Bob.  This setting holds for the applications to device-independent quantum key distribution and blind, verified computation.  However, to show that $\QMIP = \MIP^*$, in \thmref{t:qmiptomipstarconversion} below, we will need Alice to share entanglement with multiple provers, say $B_1, \ldots, B_\ell$.  \thmref{t:sequentialstructureandobservedcorrelationsapplied} still applies, if we group $B_1, \ldots, B_\ell$ together into a conglomerate prover, but we need to ensure that it respects the tensor-product decomposition of $\H_{B_1} \otimes \cdots \otimes \H_{B_\ell}$.  This is straightforward to see for sequential CHSH games, since only one of the steps in the proof of \thmref{t:sequentialCHSHgames} involves operations that can cross between $\H_{B_j}$ spaces: the truncation to finitely many dimensions (\lemref{t:reductiontofinitedimensions}).  By separately truncating the spaces $\H_{B_1}, \ldots, \H_{B_\ell}$, i.e., applying \lemref{t:reductiontofinitedimensions} in $\ell$ steps, we obtain that the ideal strategy $\hat \S$ that closely simulates the provers' strategy~$\S$ obeys the same locality constraints as~$\S$: 

\begin{proposition} \label{t:sequentialCHSHgamesmultipleBobs}
If ``Bob" is actually a collection of separate provers $B_1, \ldots, B_\ell$, where $B_j$ plays on $\H_{B_j}$ $n_j$ out of every set of~$n$ CHSH games, then in the conclusions of Theorems~\ref{t:sequentialCHSHgames} and~\ref{t:sequentialstructureandobservedcorrelationsapplied}, we may assume that the isometry $\XB : \H_{B_1} \otimes \cdots \otimes \H_{B_\ell} \hookrightarrow (\C^2)^{\otimes n} \otimes \H_B'$, with respect to which the provers' strategy is $\zeta$-ideal, factors as the tensor product of isometries $\XBj{j} : \H_{B_j} \hookrightarrow (\C^2)^{\otimes n_j} \otimes \H_{B_j}'$.  
\end{proposition}

\smallskip

In the ``self-testing" framework~\cite{MayersYao03chsh, DamMagniezMoscaSantha99selftesting, MagniezMayersMoscaOllivier05selftest}, one is allowed to reinitialize and run the same experiment multiple times in order to test its functionality.  By substituting the right parameter values into \thmref{t:sequentialstructureandobservedcorrelations}, we obtain as a corollary efficient self-testing for sequential CHSH games: 

\begin{theorem}[Self-testing sequential CHSH games] \label{t:sequentialCHSHselftesting}
Let $\S$ be an arbitrary strategy for~$n$ sequential CHSH games.  
Let $\epsilon > 0$ be at most a sufficiently small constant.  
Let $k > 0$ and let $n^* > 0$ solve $n^* / \log n^* = 256 k^2 (4 \kappaEPR^2 + 3) \kappaEPR^{4 \kappaEPR} / \epsilon^{4 \kappaEPR}$, where~$\kappaEPR$ is the constant from \thmref{t:sequentialCHSHgames}.  Let $N = (\max\{n, n^*\})^{4 \kappaEPR^2 + 2}$.  

Consider running the strategy~$\S$ $N$ times, reinitializing the joint state of the provers and the environment between sets.  Let $W \leq N n$ be the total number of games that Alice and Bob win.  Let~$\delta = k \sqrt{\log (N n) / (N n)}$ and $p = \cos^2(\pi/8)$.  
\begin{itemize}
\item 
If $\S$ is an ideal strategy, then 
\begin{equation} \label{e:sequentialCHSHselftestingcomplete}
\Pr[W \geq (p - \delta) N n] \geq 1 - n^{-2 k^2 (4 \kappaEPR^2 + 3)}
 \enspace .
\end{equation}
\item
If $\S$ is not $\epsilon$-ideal, then 
\begin{equation} \label{e:sequentialCHSHselftestingsound}
\Pr[W \geq (p - \delta) N n] \leq n^{-k^2 (4 \kappaEPR^2 + 3)/2}
 \enspace .
\end{equation}
\end{itemize}
\end{theorem}

\begin{proof}
For $\S$ ideal, the claim follows by \lemref{t:honestproverscentrallimit}.  

Consider next the case that $\S$ is not $\epsilon$-ideal.  Let $\epsilon' = (\epsilon / (\kappaEPR n^{\kappaEPR}))^{\kappaEPR}$, where~$\kappaEPR$ is the constant from \thmref{t:sequentialCHSHgames}; thus $\epsilon = \kappaEPR n^{\kappaEPR} \epsilon'^{1/\kappaEPR}$.  By \thmref{t:sequentialstructureandobservedcorrelations} with parameter $\eta$ tending to one, $\Pr[W \geq (p - \delta) N n] \leq \exp( -t^2 / (2 N n) )$, so long as $t = \frac{1}{8} \epsilon'^2 N - \delta N n \geq 0$.  Assume that $n \geq n^*$.  Substituting our parameter choices for $N$ and~$\delta$ gives 
\begin{equation*}\begin{split}
t 
&= n^{2 \kappaEPR^2 + \frac{3}{2}} \Big[ \frac{1}{8} \frac{\epsilon^{2 \kappaEPR}}{\kappaEPR^{2 \kappaEPR}} \sqrt n - k \sqrt{(4 \kappaEPR^2 + 3) \log n} \Big] \\
&\geq \frac{1}{16} \frac{\epsilon^{2 \kappaEPR}}{\kappaEPR^{2 \kappaEPR}} n^{2 \kappaEPR^2 + 2} \\
&= \sqrt{N n} \Big[ \frac{1}{16} \frac{\epsilon^{2 \kappaEPR}}{\kappaEPR^{2 \kappaEPR}} \sqrt n \Big] \\
&\geq \sqrt{N n} k \sqrt{(4 \kappaEPR^2 + 3) \log n}
 \enspace ,
\end{split}\end{equation*}
where the inequalities follow by using the definition of~$n^*$ to bound the bracketed expressions.  In particular, $t > 0$, so the bound from \thmref{t:sequentialstructureandobservedcorrelations} indeed holds.  It follows, too, that $\exp( -t^2 / (2 N n) ) \leq n^{-k^2 (4 \kappaEPR^2 + 3)/2} = (N n)^{-k^2 / 2}$.  

If $n < n^*$, then the same inequalities all hold using $n^*$ in place of~$n$ everywhere.  The final inequalities are $\exp( -t^2 / (2 N n) ) \leq (n^*)^{-k^2 (4 \kappaEPR^2 + 3)/2} < n^{-k^2 (4 \kappaEPR^2 + 3)/2}$.  
\end{proof}

By repeating the $N$ experiments in \thmref{t:sequentialCHSHselftesting} and taking the majority of the test results $W \overset{?}{\geq} (p - \delta) N n$, the completeness and soundness parameters in Eqs.~\eqnref{e:sequentialCHSHselftestingcomplete} and~\eqnref{e:sequentialCHSHselftestingsound} can efficiently be made exponentially close to one and exponentially close to zero, respectively.  

Along with other self-testing problems, Magniez et al.\ have previously studied self-testing for sequential CHSH games~\cite[Corollary~3]{MagniezMayersMoscaOllivier05selftest}.  Their result for sequential CHSH games is weaker than \thmref{t:sequentialCHSHselftesting} in two aspects.  First, they assume a fixed tensor-product structure for the provers' measurement operators for different games, whereas we derive this structure.  More precisely, they assume that the Hilbert space for prover~$\device$ is divided as $\H_\device = \H_\device^1 \otimes \cdots \otimes \H_\device^n$, and that the measurements for game~$j$ act only on~$\H_\device^j$.  Second, their analysis requires an overhead exponential in~$n$, whereas the overhead in \thmref{t:sequentialCHSHselftesting} is polynomial.

\ifx\compilefullpaper\undefined  
\bibliographystyle{alpha-eprint}
\bibliography{q}

\end{document}
\fi

\ifx\compilefullpaper\undefined  
\documentclass[11pt]{article}

\begin{document}
\tableofcontents
\fi

\section{Tomography} \label{s:tomography}

\thmref{t:sequentialCHSHgames} lets us test two entangled provers to gain confidence that they really do have a state close to $n$ shared EPR states that they nearly honestly measure one at a time in sequential CHSH games.  Even though the CHSH game is very simple, it is practically useful in quantum key distribution for extracting shared randomness that is guaranteed to be uncorrelated with any outside environment.  \thmref{t:sequentialCHSHgames} may also have other applications in cryptography; CHSH games are also used, for example, in randomness expansion~\cite{Pironioetal09randombell, AcinMassarPironio11randomnessexpansion, PironioMassar11randomnessexpansion, FehrGellesSchaffner11randomnessexpansion, VaziraniVidick11randomnessexpansion}.  

In this section, however, we will leverage the sequential CHSH game test to build tests for more complicated multi-qubit operations.  Given single-qubit measurements, the natural approach to test more complicated operations is to apply tomography.  We will therefore consider protocols in which one prover is asked to apply the single-qubit measurements of sequential CHSH games and the other prover is asked either to play sequential CHSH games, or to apply certain multi-qubit operations.  Success in the CHSH games assures us that the first prover is playing nearly honestly, which means that her measurement results give meaningful statistics for tomographically characterizing the multi-qubit operations of the second prover.  

In the state tomography problem, one is given $n$ copies of an unknown state~$\rho \in \L(\H)$, and can measure the states to roughly determine~$\rho$.  State certification is a promise, decision version of tomography.  In state certification, one is given the additional promise that for a fixed state~$\sigma$, either $\rho = \sigma$ or $\rho$ is far from~$\sigma$, and the goal is to determine which situation holds.  This model is insufficiently adversarial for our applications.  We allow the weaker promise, that for an \emph{arbitrary} $n$-system state $\rho \in \L(\H^{\otimes n})$, either $\rho = \sigma^{\otimes n}$ or there is a significant probability that its reduced density matrix on a random subsystem is far from~$\sigma$.  

Despite the power of \thmref{t:sequentialCHSHgames}, the tomography arguments are still surprisingly involved.  There are three essential problems: 
\begin{enumerate}
\item
Characterizing tomographically an $n$-qubit state generally requires collecting statistics on~$4^n$ separate observables, using exponentially many copies of the state~\cite{NielsenChuang00}.  Tomography is more efficient on restricted classes of quantum states.  Compressed sensing techniques allow low-rank states to be recovered with fewer experiments~\cite{GrossLiuFlammiaBeckerEisert09compressedsensing, Liu11tomography, Gross09tomography}.  For example, an $n$-qubit pure state can be characterized with only $\tilde O(2^n)$ different experiments.  An $n$-qubit matrix-product state with rank~$r$ can be characterized with only $O(n r^2)$ different experiments~\cite{CramerPlenioFlammiaSommaGrossBartlettLandonCardinalPoulinLiu10tomography}.  (These procedures also detect if the actual state is far from being pure or far from a rank-$r$ matrix-product state.)  The task of \emph{certifying} a state instead of tomographically characterizing it is still more efficient.  The fidelity of a state~$\sigma$ with a known pure state~$\ket \psi$ can be estimated with only a \emph{constant} number of different experiments, although still requiring a polynomial number of copies of~$\sigma$~\cite{FlammiaLiu11tomography, SilvaLandonCardinalPoulin11tomography}.  

Our setting is not compatible with this tomography and certification framework.  For example, we would like a test Eve can apply to gain confidence that, when she asks him to, Bob indeed applies a Bell basis measurement to two of his shared EPR states.\footnote{The Bell basis consists of the four orthonormal states $\frac{1}{\sqrt 2}(\ket{00} \pm \ket{11})$ and $\frac{1}{\sqrt 2}(\ket{01} \pm \ket{10})$.}  This operation involves only a constant number of qubits, but Eve might want Bob to apply it many times and there is no guarantee that the operations he actually applies are identical or even decided on non-adaptively.  The process being characterized therefore involves many qubits.  An exponential or even polynomial overhead is unacceptable---in fact, Eve can only ask Alice to make her CHSH game measurements on \emph{one} $n$-qubit state.  

This setting is therefore more adversarial than standard tomography, in which it is generally assumed that the same state can be prepared repeatedly.  The problem is similar to one we faced in the analysis of sequential CHSH games: we need to allow the adversary memory.  
\item 
A second problem is that we want to characterize the operations the provers apply to their shared EPR states, and not just the states that these operations create on the other side.  The distinction is the same as that between process and state tomography.  

This difference will turn out to be surprisingly important in our analysis.  We are be able to analyze state tomography for a broad class of states, and process tomography only for a very limited class of operations.  The protocol used for process tomography will also be more involved than that for state tomography, using some additional sequential CHSH games.  Essentially, the problem is that the correct states could be generated by incorrect processes.  For example, statistical tests are not sufficient to catch Bob cheating in just one of the many operations he is asked to apply.  In particular, he might cheat in the first requested operation, and instead of a Bell pair measurement might cyclically shift all of his EPR state halves.  If he subsequently plays honestly except taking this shift into account, then he can never be caught even though his overall strategy is highly dishonest---for example, when Eve asks him to apply a Bell measurement to his third and fourth qubits, he instead applies it to the fourth and fifth qubits.  

To avoid this problem, we will need to apply stronger tests that let us be sure that Bob cannot cheat in even one of the operations.  A Bell pair measurement is a stabilizer operation~\cite{NielsenChuang00}.  This allows Eve to reject if even a single experiment has an incorrect measurement outcome, instead of having to collect statistics on many experiments.  

This example suggests that perhaps we should use a weaker definition of process tomography, because Eve only cares that Bob applies a Bell measurement to two qubits that are maximally entangled with Alice's third and fourth qubits, and she does not care where in Bob's Hilbert space these qubits are kept.  Intuitively, it does not seem very reasonable for process tomography to be restricted to stabilizer operations, but this is the best analysis we have so far been able to apply.  
\item
A third problem is that saturating Tsirelson's inequality for the CHSH game only implies that Alice is honestly making $X$ and $Z$ measurements on her half of a shared EPR state.  For tomography, however, we also need measurements in the Pauli~$Y$ basis.  There is a technical solution that allows us to add~$Y$ operators to the game.  Instead, though, we will use a theory developed by McKague~\cite{McKague10thesis}, that shows the existence of a large class of states that are fully determined by only $X$ and $Z$ measurements.  
\end{enumerate}

We explain McKague's theory of states determined by $X$ and~$Z$ measurements in \secref{s:xzdeterminedstates} immediately below.  In \secref{s:statetomographyanalysis}, we study state tomography for states that are determined by~$X$ and~$Z$ measurements.  In \secref{s:processtomographyanalysis}, we study process tomography, specializing our discussion to commuting sets of~$X$ and~$Z$ Pauli stabilizer measurements.

\subsection{States fully determined by tomography in the $X$ and $Z$ bases} \label{s:xzdeterminedstates}

\def\operatorset{S}	

In the standard CHSH games that we have chosen to analyze, each prover has only two measurement settings, that in the honest strategy may be identified with $X$ and~$Z$ operators.  For carrying out tomography, however, it is generally necessary to be able to measure in the~$Y$ basis as well.  

One option we have, therefore, is to extend the CHSH game to add $Y$ operators.  The $y$ direction in the Bloch sphere can be fixed, up to a sign, by adding measurement directions intermediate between the $x$ and~$y$ axes in the Bloch sphere, and intermediate between the $z$ and~$y$ axes.  See \appref{s:generalizedeprlemma}.  It is not possible to fix the sign of the $Y$ operator, since a prover who consistently measures using $-Y$ will give indistinguishable statistics from one who uses $+Y$.  To force the provers to use the same choice of sign consistently, the verifier can ask one of the provers to measure random pairs of qubits in the Bell basis.  Intuitively, this will force the other prover to use the same sign choice for every qubit, since $\frac{1}{4}(I \otimes I + X \otimes X + Z \otimes Z - Y \otimes Y)$ is a valid state but $\frac{1}{4}(I \otimes I + X \otimes X + Z \otimes Z + Y \otimes Y)$ is not.  This approach is somewhat complicated, though, because it adds another step to the protocol.  

A simpler approach, that we follow here, is to argue that for certain states, reliable tomography can be accomplished without needing to measure in the $Y$ basis.  This observation is due to McKague~\cite{McKague10thesis} and was suggested earlier by Magniez et al.~\cite{MagniezMayersMoscaOllivier05selftest}.  McKague shows that for $\ket{\psi^*} = \frac{1}{\sqrt 2}(\ket{00} + \ket{11})$, an EPR state, the states $(I \otimes T) \ket{\psi^*}$, for any single-qubit real unitary~$T$, and $\text{CNOT}_{24} \ket{\psi^*}_{12} \otimes \ket{\psi^*}_{34}$, as well as finite tensor products of these states, are exactly determined by their traces against tensor products of $I$, $X$ and~$Z$ operators.  That is, they are determined by the expectations of observables that can be estimated using measurements in the $X$ and~$Z$ bases.  We call such states ``$X\!Z$-determined."  

In this section, we give simplified proofs that a much larger class of states is $X\!Z$-determined.  However, characterizing the full set of $X\!Z$-determined states remains an open problem.   

\begin{definition} \label{t:xzdetermineddef}
For a Hilbert space~$\H$, a set of operators $\operatorset \subseteq \L(\H)$ and $d > 0$, a state $\sigma \in \L(\H)$ is \emph{determined by~$\operatorset$ with exponent~$d$} if there exists $c > 0$ such that for all~$\epsilon \geq 0$ and any state $\rho \in \L(\H)$, 
\begin{equation} \label{e:xzdetermineddef}
\max_{P \in \operatorset} \abs{ \Tr P (\rho - \sigma) } \leq \epsilon
\qquad\Longrightarrow\qquad 
\trnorm{\rho - \sigma} \leq c \, \epsilon^d
 \enspace .
\end{equation}
The state~$\sigma$ is \emph{determined by~$\operatorset$} if there exists $d > 0$ such that $\sigma$ is determined by~$\operatorset$ with exponent~$d$.  

For $\H = (\C^2)^{\otimes n}$, a state $\sigma$ is \emph{$X\!Z$-determined} (with exponent~$d$) if it is determined (with exponent~$d$) by the Pauli operators~$\{I, X, Z\}^{\otimes n}$.  

Both definitions extend to pure states $\ket \psi \in \H$ by setting $\sigma = \ketbra \psi \psi$.  
\end{definition}

By basic algebraic geometry, robustness follows from the $\epsilon = 0$ case: 

\begin{lemma}
For a finite-dimensional Hilbert space~$\H$, a state $\sigma \in \L(\H)$ is determined by a finite set $\operatorset \subset \L(\H)$ if and only if for any state $\rho \in \L(\H)$, the implication of Eq.~\eqnref{e:xzdetermineddef} holds at $\epsilon = 0$.  
\end{lemma}

\begin{proof}
By Sylvester's criterion, the set of states in~$\H$ is a compact, semi-algebraic set.  The functions $f(\rho) = \max_{P \in \operatorset} \abs{ \Tr P (\rho - \sigma) }$ and $g(\rho) = \norm{\rho - \sigma}_F$, where $\norm{\cdot}_F$ is the Frobenius norm, are continuous, semi-algebraic functions.  If $f(\rho) = 0$ implies $g(\rho) = 0$ for all states~$\rho$, therefore by {\Lstroke}ojasiewicz's inequality~\cite[Prop.~2.3.11]{BenedettiRisler90semialgebraic} there exist $c > 0$ and an integer $d \geq 1$ such that $g(\rho) \leq c f(\rho)^{1/d}$ for all states~$\rho$.  
\end{proof}

\begin{lemma} \label{t:tomographicallydeterminedexamples}
The following are examples of tomographically determined states: 
\begin{enumerate}
\item
Any $n$-qubit state~$\sigma$ is determined with exponent~$1$ by the Pauli operators~$\{I, X, Y, Z\}^{\otimes n}$.  
\item 
The set of one-qubit $X\!Z$-determined states is exactly $\{ \frac12 (I + \cos(\theta) X + \sin(\theta) Z) : \theta \in [0, 2 \pi) \}$, i.e., the set of pure states in the $xz$-plane of the Bloch sphere.  
\item 
There exist two-qubit mixed states that are $X\!Z$-determined.  In particular, the state $\frac12 \ketbra00 \otimes \ketbra++ + \frac12 \ketbra++ \otimes \ketbra00$, where $\ket{+} = \frac{1}{\sqrt 2}(\ket 0 + \ket 1)$, is $X\!Z$-determined with exponent~$1/4$.  
\end{enumerate}
\end{lemma}

\begin{proof}[Proof sketch]
Any operator $\rho \in \L((\C^2)^{\otimes n})$ can be expanded in the Pauli basis as $\rho = \frac{1}{2^n} \sum_{P \in \mathcal P} \rho_P P$, where $\mathcal P = \{I, X, Y, Z\}^{\otimes n}$ and $\rho_P = \Tr(\rho P)$.  Since for $P \in \mathcal P$, $\trnorm{P} = 2^n$, by a triangle inequality, 
\begin{equation} \label{e:trnormtoPaulicoordinates}
\trnorm{\rho} \leq \frac{1}{2^n} \sum_{P \in \mathcal P} \trnorm{\rho_P P} = \sum_{P \in \mathcal P} \abs{\rho_P}
 \enspace .
\end{equation}
Furthermore, if~$\rho$ is a state, then $\frac{1}{2^n} \sum_{P \in \mathcal P} \rho_P^2 = \Tr(\rho^2) \leq 1$, with equality if $\rho$ is a pure state.  

1. If for a state~$\rho$ and for all~$P \in \mathcal P$, $\abs{\rho_P - \sigma_P} \leq \epsilon$, then by Eq.~\eqnref{e:trnormtoPaulicoordinates}, $\trnorm{\rho - \sigma} \leq 4^n \epsilon$.  

2. Similar calculations show that all of the one-qubit states $\{ \frac12 (I + \cos(\theta) X + \sin(\theta) Z) \}$ are all $X\!Z$-determined.  These are the only one-qubit, $X\!Z$-determined states since any state of the form $\frac12 (I + x X + y Y + z Z)$, with $x^2 + y^2 + z^2 \leq 1$, has the same $X$ and $Z$ Pauli coefficients as $\frac12 (I + x X + z Z)$.  

3. Let $\sigma = \frac12 (\ketbra{0+}{0{+}} + \ketbra{{+}0}{+0})$.  If $\rho$ is a state with the same $\{I, X, Z\}^{\otimes 2}$ coordinates as~$\sigma$, then $\frac12 (\rho + \text{SWAP} \rho \, \text{SWAP}^\dagger)$ is a state of the form $\sigma + \sum_{P \in \{I, X, Y, Z\}} \alpha_P (Y \otimes P + P \otimes Y)$ for some real coefficients $\alpha_P$.  Writing this matrix out in the computational basis, the requirement that each $2 \times 2$ block along the diagonal be positive semi-definite forces $\alpha_X = \alpha_Y = 0$ and all $\alpha_I = \alpha_Z$.  Then considering the first $3 \times 3$ block forces $\alpha_I = 0$.  This gives the $\epsilon = 0$ case, and a similar argument holds when $\abs{\rho_P - \sigma_P} \leq \epsilon$ for $\epsilon > 0$.  The stability exponent $d = 1/4$ may not be optimal.  
\end{proof}

Starting with the fact that $\ket 0$ is determined by $\{ Z \}$ with exponent~$1/2$, we will apply several closure properties to bootstrap into a large class of tomographically determined states.  

\begin{lemma}[General closure properties] \label{t:determinedtomographyclosure}
If $\sigma \in \L(\H)$ is a state determined by~$\operatorset = \{ P_1, \ldots, P_s \}$ with exponent~$d$, then: 
\begin{enumerate}
\item For any unitary~$U \in \L(\H)$, $U \sigma U^\dagger$ is determined by $\{ U P U^\dagger : P \in \operatorset \}$, with the same exponent~$d$.  
\item For any invertible $s \times s$ matrix~$V$, $\sigma$ is determined by $\{ \sum_{j \in [s]} V_{ij} P_j : i \in [s] \}$, with the same exponent~$d$.  
\item For $\ket{\psi'} \in \H'$ a pure state determined by~$\operatorset'$ with exponent~$d'$, $\sigma \otimes \ketbra{\psi'}{\psi'}$ is determined by $\{ P \otimes I : P \in \operatorset \} \cup \{ I \otimes P' : P' \in \operatorset' \}$, with exponent at least $\min \{ d, d'/2 \}$.  
\end{enumerate}
In particular, the set of $X\!Z$-determined pure states is closed under tensor products.  
\end{lemma}

\begin{proof}
1. Let $\rho$ be a state such that for all $P \in \operatorset$, $\abs{\Tr ((U P U^\dagger) (U \sigma U^\dagger - \rho))} \leq \epsilon$.  Since the trace is cyclic, this implies that $\max_P \abs{ \Tr P (\sigma - U^\dagger \rho U) } \leq \epsilon$.  Since $\sigma$ is determined by~$\operatorset$, therefore $\trnorm{\rho - U \sigma U^\dagger} = \trnorm{U^\dagger \rho U - \sigma} \leq c \, \epsilon^d$.  

2. Let $\rho$ be a state and define a vector $\vec x$ by $x_i = \Tr P_i (\rho - \sigma)$.  If for all~$i \in \operatorset$, $\bigabs{ \Tr \big( \sum_j V_{ij} P_j (\rho - \sigma) \big) } = \bigabs{ (V x)_i } \leq \epsilon$, then $\max_i \abs{x_i} \leq \norm{V^{-1}}_{1, 1} \epsilon$, where $\norm{V^{-1}}_{1, 1} = \max_{\vec y : \max \abs{y_i} \leq 1} \max_i \abs{(V^{-1} \vec y)_i}$.  Therefore, $\trnorm{\rho - \sigma} \leq c \norm{V^{-1}}_{1,1}^d \epsilon^d$.  

3. Let $\pi = \ketbra{\psi'}{\psi'}$.  Let $\rho$ be a state on $\H \otimes \H'$, let $\rho_\H = \Tr_{\H'} \rho$ and $\rho_{\H'} = \Tr_\H \rho$.  Assuming that for all $P \in \operatorset$ and $P' \in \operatorset'$, $\abs{\Tr (P \otimes \identity) (\sigma \otimes \pi - \rho)} = \abs{\Tr P (\sigma - \rho_\H)} \leq \epsilon$ and $\abs{\Tr (\identity \otimes P') (\sigma \otimes \pi - \rho)} = \abs{\Tr P' (\pi - \rho_{\H'})} \leq \epsilon$, it follows that $\trnorm{\rho_\H - \sigma} \leq c \, \epsilon^d$ and $\trnorm{\rho_{\H'} - \pi} \leq c' \epsilon^{d'}$.  Therefore, $\Tr (\pi \rho_{\H'}) \geq 1 - \delta$, where $\delta = c' \epsilon^{d'}$.  By \corref{t:gentlemeasurementpurestate} of the Gentle Measurement Lemma, $\trnorm{\rho - \rho_\H \otimes \pi} \leq 2 \sqrt \delta + \delta$, so $\trnorm{\rho - \sigma \otimes \sigma'} \leq 2 \sqrt \delta + \delta + c \epsilon^d$.  
\end{proof}

Closure under tensor products has been shown previously by McKague~\cite[Lemma~3.4]{McKague10thesis}.  Note that in this third statement, it is important that one of the two states in the tensor product be pure.  If $\sigma$ and~$\sigma'$ are two mixed states determined by $\operatorset$ and~$\operatorset'$, respectively, then $\sigma \otimes \sigma'$ is generally not determined by $\{ P \otimes I : P \in \operatorset \} \cup \{ I \otimes P' : P' \in \operatorset' \}$.  For an example, consider $\sigma = \sigma'$ given by the third example in \lemref{t:tomographicallydeterminedexamples}, and take $\rho = \frac12 \ketbra{0{+}}{0{+}} \otimes \ketbra{0{+}}{0{+}} + \frac12 \ketbra{{+}0}{{+}0} \otimes \ketbra{{+}0}{{+}0}$.  

\begin{corollary} \label{t:PaulisfixXZdeterminedstates}
If $\sigma$ is a state determined by $\operatorset \subseteq \{I, X, Y, Z\}^{\otimes n}$, and $Q \in \{I, X, Y, Z\}^{\otimes n}$ is any Pauli operator, then $Q \sigma Q^\dagger$ is also determined by~$\operatorset$, with the same exponent.  
\end{corollary}

\begin{proof}
By the first closure property of \lemref{t:determinedtomographyclosure}, $Q \sigma Q^\dagger$ is determined by $\{ Q P Q^\dagger : P \in \operatorset \}$.  For Pauli operators~$P$ and~$Q$, $Q P Q^\dagger$ is either $P$ or~$-P$, depending on whether $P$ and~$Q$ commute or anti-commute, respectively.  By the second closure property of \lemref{t:determinedtomographyclosure}, with $V$ a diagonal matrix with $\pm 1$ entries along the diagonal, $Q \sigma Q^\dagger$ is determined by~$\operatorset$.  
\end{proof}

\smallskip

Recall that a \emph{stabilizer state} is an $n$-qubit pure state~$\ket \psi$ for which there exists a set of $2^n$ distinct and pairwise commuting operators $\operatorset \subset \{ \pm P : P \in \{I, X, Y, Z\}^{\otimes n} \}$, the \emph{stabilizer group}, such that $P \ket \psi = \ket \psi$ for all~$P \in \operatorset$~\cite{NielsenChuang00}.  Any set of $n$ operators that generate the stabilizer group~$\operatorset$ are called stabilizer generators for $\ket \psi$.  

\begin{theorem} \label{t:stabilizerstatedeterminedbygenerators}
A stabilizer state is determined by any of its sets of stabilizer generators.  
\end{theorem}

\begin{proof}
For any stabilizer state $\ket \psi \in (\C^2)^{\otimes n}$ and set~$\operatorset$ of stabilizer generators, there exists a Clifford group unitary $U$ such that $U \ket \psi = \ket{0^n}$ and $\{ U P U^\dagger : P \in \operatorset \} = \{ Z_1, \ldots, Z_n \}$.  Indeed, to find such a~$U$, first choose a Clifford operator~$V$ such that $V \ket \psi = \ket{0^n}$.  $V$ conjugates $\operatorset$ to some set of independent operators in $\{I, Z\}^{\otimes n}$.  Using $\text{CNOT}$ gates, this set can then be conjugated to $\{Z_1, \ldots, Z_n\}$.  By the tensor-product closure property of \lemref{t:determinedtomographyclosure}, $\ket{0^n}$ is determined by $\{Z_1, \ldots, Z_n\}$.  By the unitary conjugation closure property of \lemref{t:determinedtomographyclosure}, therefore $\ket \psi$ is determined by~$\operatorset$.  
\end{proof}

\begin{theorem} \label{t:XZdeterminedstabilizerstates}
If $\ket \psi \in (\C^2)^{\otimes n}$ is a stabilizer state that has a set of stabilizer generators in $\{I, X, Z\}^{\otimes n}$, and if $U$ is the tensor product of any~$n$ single-qubit real unitaries, then $U \ket \psi$ is $X\!Z$-determined.  
\end{theorem}

Indeed, $\ket \psi$ is $X\!Z$-determined by \thmref{t:stabilizerstatedeterminedbygenerators}, and the set of $X\!Z$-determined states is closed under conjugation by tensor products of one-qubit real unitaries: 

\begin{lemma}
If $\sigma \in \L((\C^2)^{\otimes n})$ is an $X\!Z$-determined state and $U$ is the tensor product of~$n$ single-qubit real unitaries, then $U \sigma U^\dagger$ is $X\!Z$-determined.  
\end{lemma}

\begin{proof}
The key idea is that for any state~$\rho$, the $\{I, X, Z\}^{\otimes n}$ coefficients of~$U^\dagger \rho U$ are determined by, i.e., are a function of, the $\{I, X, Z\}^{\otimes n}$ coefficients of~$\rho$.  Therefore, if~$\rho$ has $\{I, X, Z\}^{\otimes n}$ coefficients close to those of~$U \sigma U^\dagger$, then $U^\dagger \rho U$ has $\{I, X, Z\}^{\otimes n}$ coefficients close to those of~$\sigma$.  Since~$\sigma$ is $X\!Z$-determined, therefore $U^\dagger \rho U \approx \sigma$ in trace distance, and so $\rho \approx U \sigma U^\dagger$.  Now let us give the formal proof, using the closure properties of \lemref{t:determinedtomographyclosure}.  

Without loss of generality, it suffices to consider the case that $U$ acts as the identity on all but the first qubit.  Any one-qubit unitary with real coefficients can be expanded as a product of operators of the form $e^{i \theta Y}$ and~$Z$.  Since Pauli operators fix the set of $X\!Z$-determined states by \corref{t:PaulisfixXZdeterminedstates}, it suffices to consider the case $U = e^{i \theta Y} \otimes I^{\otimes (n-1)}$.  

By the first closure property of \lemref{t:determinedtomographyclosure}, $U \sigma U^\dagger$ is determined by $\{ U P U^\dagger : P \in \{I, X, Z\}^{\otimes n} \}$.  
For $P \in \{I, X, Z\}^{\otimes (n-1)}$, 
\begin{align*}
U (I \otimes P) U^\dagger &= I \otimes P \\
U (X \otimes P) U^\dagger &= (\cos(2\theta) X + \sin(2\theta) Z) \otimes P \\
U (Z \otimes P) U^\dagger &= (-\sin(2\theta) X + \cos(2\theta) Z) \otimes P
 \enspace .
\end{align*}
Thus $U$ conjugates operators in~$\{I, X, Z\}^{\otimes n}$ to linear combinations of operators in $\{I, X, Z\}^{\otimes n}$.  
Since the matrix $\smatrx{\cos(2\theta)&\sin(2\theta)\\-\sin(2\theta)&\cos(2\theta)}$ is invertible, the second closure property of \lemref{t:determinedtomographyclosure} implies that $U \sigma U^\dagger$ is $X\!Z$-determined.  
\end{proof}

It is still unknown whether or not every state~$\ket \psi$ with real coefficients in the computational basis is $X\!Z$-determined, a question first posed in~\cite{MagniezMayersMoscaOllivier05selftest}.  The problem is that whereas single-qubit real unitaries conjugate $X$ and $Z$ to combinations of $X$ and~$Z$, multi-qubit real unitaries need not do so.  We can, however, show one last relevant closure property: 

\begin{lemma}
The set of states determined by $\{I, X\}^{\otimes n} \cup \{I, Z\}^{\otimes n}$ is closed under applying $\text{CNOT}$ gates.  
\end{lemma}

\begin{proof}
A CNOT gate conjugates operators in $\{I, X\}^{\otimes n}$ to $\{I, X\}^{\otimes n}$, and conjugates operators in $\{I, Z\}^{\otimes n}$ to $\{I, Z\}^{\otimes n}$.  Therefore this is a special case of the first closure property in \lemref{t:determinedtomographyclosure}.  
\end{proof}

This proof does not work for arbitrary $X\!Z$-determined states since $\text{CNOT}_{1,2} (X \otimes Z) \text{CNOT}_{1,2}^\dagger = - Y \otimes Y$.  

\smallskip

For later reference, let us state explicitly several special cases of \thmref{t:XZdeterminedstabilizerstates}: 

\begin{theorem} \label{t:xzdeterminedstates}
Letting $\ket{\psi^*} = \frac{1}{\sqrt 2}(\ket{00} + \ket{11})$, the following are complete, orthonormal sets of $X\!Z$-determined states: 
\begin{itemize}
\item
$\{ \ket 0, \ket 1 \}$, 
\item
$\{ U \otimes P \ket{\psi^*} : P \in \{I, X, Y, Z\} \}$, for any one-qubit real unitary~$U$, and 
\item
$\{ (P_{1,2} \otimes \mathrm{CNOT}_{3,4}) (\ket{\psi^*}_{1,3} \otimes \ket{\psi^*}_{2,4}) : P \in \{I,X,Y,Z\}^{\otimes 2} \}$.  
\end{itemize}
Finite tensor products of these states are $X\!Z$-determined, as are the same states multiplied by arbitrary single-qubit real unitaries.  
\end{theorem}

In our applications, we will use the states $\{\ket 0, \ket 1\}$ for initialization and readout, and will use the other two sets of states for teleporting into the gates of a quantum circuit.  The $\text{CNOT}$ gate and single-qubit real unitaries form a universal gate set for quantum computation.

\subsection{State tomography protocol} \label{s:statetomographyanalysis}

In this section, we present a protocol by which Eve can certify that Bob has nearly honestly prepared a set of $X\!Z$-determined states.  We first assume that Alice honestly measures her halves of the shared EPR states in either the~$X$ or~$Z$ eigenbases when requested.  We then combine the protocol with a set of sequential CHSH games to ensure that Alice plays honestly.  

\begin{definition} \label{t:permutedqubitstatetomographyprotocoldef}
A \emph{state tomography protocol} is parameterized by natural numbers~$q$, $n$ and~$m$, with $q n \leq m$, a $q$-qubit POVM $\xzdeterminedset$ with at most $2^q$ outcomes, and a list~$\sigma$ of $q n$ distinct indices from~$[m]$.  The protocol involves a verifier, Eve, and two provers, Alice and Bob.  Alice and Bob share a state in $\H_A \otimes \H_B$.  The protocol proceeds as follows: 
\begin{itemize}
\item
Eve's interaction with Alice has~$m$ rounds.  In round~$j$, Eve sends Alice an independent, uniformly random bit, $A_j$.  Alice applies a two-outcome projective measurement on~$\H_A$ to determine her reply $X_j \in \{0,1\}$.  
\item
Eve has one round of interaction with Bob.  First, Eve sends Bob the list~$\sigma$.  Bob returns to Eve a string $O_1, \ldots, O_n$, with the~$O_j \in [2^q]$ determined by successive $2^q$-outcome projective measurements on~$\H_B$.  
\end{itemize}
No other communication is allowed.  

Alice's strategy is \emph{ideal}, with respect to an isometry $U^A : \H_A \hookrightarrow (\C^2)^{\otimes m} \otimes \H_A'$, if in round~$j$ of her interaction with Eve, Alice returns the result of measuring the $j$th qubit in either the $\{\ket 0, \ket 1\}$ basis, if $A_j = 0$, or the $\{\ket +, \ket -\}$ basis, if $A_j = 1$.  

Alice and Bob's joint strategy is \emph{ideal}, with respect to the isometries $U^\device: \H_\device \hookrightarrow (\C^2)^{\otimes m} \otimes \H_\device'$, $\device \in \{A, B\}$, if Alice's strategy is ideal with respect to $U^A$ and if 
\begin{enumerate}
\item
The initial state consists of $m$ EPR states in tensor product with a state in $\H_A' \otimes \H_B'$, and 
\item
Bob returns the results of measuring with $\xzdeterminedset$ each successive block of~$q$ qubits specified in~$\sigma$.  
\end{enumerate}
\end{definition}

To specify Eve's acceptance criterion, we will need the following notation: 

\begin{definition}[Notation for a state tomography protocol]
For $j \in [n]$ and $i \in [q]$, let $\sigma(j, i) = \sigma_{(j-1)q + i} \in [m]$.  
For $o \in [2^q]^n$, let $\rho_o$ be the normalized state of the system conditioned on Bob outputting $(O_1, \ldots, O_n) = o$ but before any of Alice's measurements.  Let $(\rho_o)_{\sigma, j}$ be the same state reduced to Alice's qubits $\sigma(j, 1), \ldots, \sigma(j, q)$.  

Let $P_j^i = \delta_{A_{\sigma(j, i)}, 0} Z + \delta_{A_{\sigma(j, i)}, 1} X \in \{X,Z\}$ be the Pauli basis Alice is asked to measure in game~$\sigma(j, i)$.  
Further, for $o \in [2^q]$ and $P \in \{I, X, Z\}$, let 
\begin{equation}
I_j^{o, i, P} = \delta_{O_j, o} \big( \delta_{P, I} + \delta_{P, P_j^i} (-1)^{X_{\sigma(j, i)}} \big)
 \enspace .
\end{equation}
That is, $I_j^{O_j, i, I} = 1$, $I_j^{O_j, i, P_j^i} = (-1)^{X_{\sigma(j, i)}}$, and otherwise $I_j^{o, i, P} = 0$.  For $P \in \{I, X, Z\}^q$, let $I_j^{o, P} = \prod_{i \in [q]} I_j^{o, i, P_i}$, let $\abs P$ be the number of coordinates in which~$P$ is not the identity, and let $\tau^{o, P}$ be given by  
\begin{equation} \label{e:observedPaulitraces}
\tau^{o, P} = \frac{2^{q + \abs P}}{n} \sum_{j \in [n]} I^{o, P}_j
 \enspace .
\end{equation}
\end{definition}

The motivation for $\tau^{o, P}$ is to give an estimator for $\Tr(E_o^T P)$ when Alice and Bob use an ideal strategy for the POVM $\xzdeterminedset = \{E_o\}$: 

\begin{lemma} \label{t:idealstatetomographyexpectations}
If Alice and Bob's joint strategy is ideal and the POVM $\xzdeterminedset = \{E_o\}$, then 
\begin{itemize}
\item 
The~$O_j$ variables are independent of each other, and satisfy $\Pr[O_j = o] = \frac{1}{2^q} \Tr E_o$.  
\item 
For all~$j$, Alice's state $(\rho_{o_1\ldots o_n})_{\sigma, j}$ equals $E_{o_j}^T / \Tr E_{o_j}$.  
\item 
The~$I^{o, P}_j$ variables are independent for different~$j$, and satisfy $\Ex[I^{o, P}_j] = \frac{1}{2^{q + \abs P}} \Tr(E_o^T P) / \Tr E_o$.  
\end{itemize} 
\end{lemma}

\begin{proof}
For Pauli operators $P$ and~$Q$, say that $P \in Q$ if in every coordinate either $Q$ is the identity or~$P$ and~$Q$ agree.  Thus for $Q \in \{I, X, Z\}^{\otimes q}$, $\abs{ \{ P \in \{X, Z\}^{\otimes q} : P \in Q \} } = 2^{q - \abs Q}$.  Let $\chi_{P \in Q}$ equal~$1$ if $P \in Q$, and $0$ otherwise.  Notice that $I_j^{o, P} = \delta_{O_j, o} \chi_{\otimes_i P_j^i \in P} I_j^{o, P}$.  The state of Alice's $q$ qubits after Bob measures outcome~$o$ is $E_o^T / \Tr E_o$.  Then a calculation gives $\Ex[I^{o, P}_j] = \Pr[O_j = o] \Pr[\otimes_i P_j^i \in P] \Tr(E_o^T P) / \Tr E_o = \frac{1}{2^{q + \abs P}} \Tr(E_o^T P) / \Tr E_o$.  
\end{proof}

We study state tomography for a POVM $\xzdeterminedset = \{ \pi^o \}$ consisting of projections onto $X\!Z$-determined pure states.  In particular, this implies that each $\pi^o$ equals its transpose.  Our state tomography theorem shows that if Eve accepts with high probability, then for most of Bob's measurement outcomes $O_1, \ldots, O_n$ and most~$j \in [n]$, $(\rho_{O_1\ldots O_n})_{\sigma, j}$ is close to $\pi^{O_j}$.  We begin by analyzing a state tomography protocol in which Alice's strategy is ideal: 

\begin{theorem} \label{t:statetomography}
Fix $\xzdeterminedset = \{ \pi^1, \ldots, \pi^{2^q} \}$ a complete, orthonormal set of $q$-qubit $X\!Z$-determined pure states.  For $n$ sufficiently large, let $m = m(n) \geq q n$ and let $\sigma \in [m]^{q n}$ be a list of distinct indices.  Consider a state tomography protocol with parameters $q$, $n$, $m$, $\xzdeterminedset$ and~$\sigma$, in which Alice plays according to an ideal strategy.  Say that Eve accepts at the end of the protocol if the following two checks are satisfied: 
\begin{subequations} \label{e:tomographyacceptancecriteria}
\begin{align}
\max_{o \in [2^q]} \bigabs{\# \{ j : O_j = o \} - n/2^q} &\leq 4^q \sqrt{n \log n} \label{e:tomographyacceptanceenoughstatisticsperoutcome}
\\
\max_{o \in [2^q], P \in \{I, X, Z\}^{\otimes q}} \abs{\tau^{o, P} - \Tr (\pi^o P)} &\leq 4^q \sqrt{(\log n)/n} \label{e:tomographyacceptancecriterionestimator}
 \enspace .
\end{align}\end{subequations}

This protocol satisfies the following completeness and soundness conditions: 
\begin{description}
\item[Completeness:]
If the provers' joint strategy is ideal, then 
\begin{equation}
\Pr[\text{Eve accepts}] \geq 1 - O(n^{-1/2})
 \enspace .
\end{equation}
\item[Soundness:]
If $\Pr[\text{Eve accepts}] \geq 1 - n^{-1/4}$, then 
\begin{equation}
\Pr\!\Big[
\bigabs{ \big\{ j \in [n] : \Tr ((\rho_{O_1 \ldots O_n})_{\sigma, j} \pi^{O_j}) \geq 1 - O(n^{-1/16}) \big\} } \geq (1-O(n^{-1/16})) n
\Big] \geq 1 - n^{-1/8}
.
\end{equation}  
\end{description}
\end{theorem}

\begin{proof}
Let $k = 4^q$.  
Let us first show the completeness criterion.  Since each~$O_j$ is drawn independently and uniformly at random from~$[2^q]$, $\Pr[\max_o \abs{\#\{j : O_j = o\} - n/2^q} \leq k \sqrt{n \log n}] \geq 1 - 2^q \cdot 2 n^{-2 k^2}$, by Hoeffding's inequality and a union bound.  Since $\pi^o$ is an $X\!Z$-determined state, it necessarily has only real entries and therefore equals its transpose.  Thus again Hoeffding's inequality and a union bound imply that for any $t \geq 0$, 
\begin{equation*}
\Pr\!\big[ \max_{o, P} \abs{\tau^{o, P} - \Tr(\pi^o P)} \geq t \big] 
\leq 2^q 3^q \cdot 2 \exp( -t^2 n / 2^{4q + 1} )
 \enspace .
\end{equation*}
Substitute $t = k \sqrt{(\log n)/n}$ to get $\Pr[\text{Eve accepts}] \geq 1 - O(n^{-k^2/2^{4q+1}}) = 1 - O(n^{-1/2})$.  

\smallskip

Next we will argue soundness.  Let $\epsilon = n^{-1/4}$ and assume that $\Pr[\text{Eve accepts}] \geq 1 - \epsilon$.  Then there is at least a $1 - \sqrt \epsilon$ probability that Bob outputs a string $o_{1,n} = (o_1, \ldots, o_n)$ such that $\Pr[\text{Eve accepts} \,\vert\, O_{1,n} = o_{1,n}] \geq 1 - \sqrt \epsilon$.  Fix such a transcript.  

In Alice's actual interactions with Eve, she measures her qubits in order, $1, 2, 3, \ldots, m$.  We will analyze instead a hypothetical protocol in which Alice measures her qubits in order $\sigma(1, 1), \sigma(1, 2), \ldots,$ $\sigma(n, q-1), \sigma(n, q), \ldots$.  Since Alice's strategy is ideal, and in particular her measurements commute, the distributions of her measurement outcomes are the same in the hypothetical protocol as in the actual protocol, so Eve accepts with the same probability.  For $j \in [n]$, define the random variable $\sigma_j$ to be the reduced density matrix of Alice's qubits $\sigma(j, 1), \ldots, \sigma(j, q)$ immediately after completing the first $(j-1)q$ rounds of the hypothetical protocol.  Let~$P_j  = \otimes_i P_j^i \in \{X, Z\}^{\otimes q}$ be Alice's measurement bases for rounds~$(j, 1), \ldots, (j, q)$.  For $J \in [n]$, $o \in [2^q]$ and $Q \in \{I, X, Z\}^{\otimes q}$, define $\tau^{o, Q}_J$ and~$\rho^{o, Q}_J$ by 
\begin{align*}
\tau^{o, Q}_J
&= \frac{2^{q + \abs Q}}{n} \sum_{j \in [J]} I^{o, Q}_j &
\rho^{o, Q}_J 
&= \frac{2^q}{n} \sum_{j \in [J]} \delta_{O_j, o} \Tr(\sigma_j Q)
 \enspace .
\end{align*}
Observe that $\tau^{o, Q}_J - \rho^{o, Q}_J$, for $J \in [n]$, is a martingale.  (Achieving this property is the reason behind our definition for~$\sigma_j$.  Had we instead defined~$\sigma_j$ to be the state of Alice's qubits $\sigma(j, 1), \ldots, \sigma(j, q)$ at the beginning of her interactions with Eve, then $\tau^{o, Q}_J - \rho^{o, Q}_J$ would not define a martingale sequence.)  Successive terms of the sequence differ by at most $\frac{2^q}{n}(1 + 2^{\abs Q})$ in magnitude.  By Azuma's inequality, therefore, for any $t \geq 0$, 
\begin{equation*}
\Pr\!\big[\abs{\tau^{o, Q}_n - \rho^{o, Q}_n} \geq t \,\big\vert\, O_{1,n} = o_{1,n}\big] 
\leq 2 \exp\Big(- \frac{t^2 n}{8 \cdot 4^{q + \abs Q}}\Big)
 \enspace .
\end{equation*}

Let $N^o = \sum_{j \in [n]} \delta_{O_j, o}$ be the number of times Bob announces outcome~$o$, and let 
\begin{equation*}
\tau^o = \frac{1}{N^o} \sum_{j \in [n]} \delta_{O_j, o} \sigma_j
\end{equation*}
be the average of the states~$\sigma_j$ over games in which Bob's outcome is~$o$.  Note that $\tau^o$ is a density matrix, i.e., $\tau^o \succeq 0$ and $\Tr \tau^o = 1$.  
Note also that 
\begin{equation*}
\Tr(\tau^o Q) -\rho^{o, Q}_n = \Big(\frac{1}{N^o} - \frac{2^q}{n}\Big) \sum_{j \in [n]} \delta_{O_j, o} \Tr(\sigma_j Q)
 \enspace .
\end{equation*}
If Eve accepts, so $\abs{N^o - n/2^q} \leq k \sqrt{n \log n}$, it follows that $\abs{\Tr(\tau^o Q) -\rho^{o, Q}_n} \leq 2^q k \sqrt{(\log n)/n}$.  

Combining the above calculations, we find that for any $t > 1$ and $\delta = 2 t 2^q k \sqrt{(\log n)/n}$, 
\begin{equation*}\begin{split}
\Pr\!\big[\text{Eve accepts}&\text{ and $\max_{o, Q} \abs{ \Tr(\tau^o Q) - \tau^{o, Q}_n} \geq \delta$} \,\big\vert\, O_{1,n} = o_{1,n}\big] \\
&\leq 
\Pr[\max_{o, Q} \abs{\tau^{o, Q}_n - \rho^{o, Q}_n} \geq \delta/2 \,\vert\, O_{1,n} = o_{1,n}]
\\&\quad + \Pr[\text{Eve accepts and $\max_{o, Q} \abs{\Tr(\tau^o Q) - \rho^{o, Q}_n} \geq \delta/2$} \,\vert\, O_{1,n} = o_{1,n}] \\
&= 
\Pr[\max_{o, Q} \abs{\tau^{o, Q}_n - \rho^{o, Q}_n} \geq \delta/2 \,\vert\, O_{1,n} = o_{1,n}] \\
&\leq 
6^q \cdot 2 n^{-\frac{1}{2^{2q+3}} t^2 k^2}
 \enspace ,
\end{split}\end{equation*}
implying that for $\epsilon' = \sqrt{\epsilon} + 6^q \cdot 2 n^{-\frac{1}{2^{2q+3}} t^2 k^2}$, 
\begin{equation*}
\Pr\!\big[\text{Eve accepts and $\max_{o, Q} \abs{ \Tr(\tau^o Q) - \tau^{o, Q}_n} < \delta$} \,\big\vert\, O_{1,n} = o_{1,n}\big] \geq 1 - \epsilon'
 \enspace .
\end{equation*}
Substituting $t = 2$ and $k = 4^q$, note that $\epsilon' = \sqrt \epsilon + O(n^{-2^{2q-1}}) = O(\sqrt \epsilon)$ and $\delta = \tilde O(1/\sqrt n)$.  

Assume that Eve accepts and $\max_{o, Q} \abs{\Tr(\tau^o Q) - \tau^{o, Q}_n} < \delta$.  Letting $\delta' = \delta + k \sqrt{(\log n)/n} = \tilde O(1/\sqrt n)$, then for all~$o$ and all~$Q \in \{I, X, Z\}^{\otimes n}$, $\abs{\Tr Q (\tau^o - \pi^o)} < \delta'$, by Eq.~\eqnref{e:tomographyacceptancecriterionestimator}.  Since~$\tau^o$ is a density matrix and $\pi^o$ is $X\!Z$-determined, we conclude that there is a constant~$c$ such that $\trnorm{\tau^o - \pi^o} < c \sqrt{\delta'}$.  In particular, there is a constant~$c'$ such that $\max_{Q \in \{I, X, Y, Z\}^{\otimes q}} \abs{\Tr Q (\tau^o - \pi^o)} < c' \sqrt{\delta'}$.  Thus, 
\begin{equation*}
\Pr\!\big[\max_{o, Q} \abs{\Tr Q(\tau^o - \pi^o)} < c' \sqrt{\delta'} \,\big\vert\, O_{1,n} = o_{1,n}\big] 
\geq 
1 - \epsilon'
 \enspace .
\end{equation*}

Assume that $\max_{o, Q} \abs{\Tr Q (\tau^o - \pi^o)} < c' \sqrt{\delta'}$.  Since each $\tau^o$ is an average of states $\sigma_j$, we will argue next, using a version of Markov's inequality for points lying in the unit ball, that $\sigma_j$ is close to~$\pi^{O_j}$ for most~$j$.  

\begin{claim} \label{t:markovunitball}
Let $x^1, \ldots, x^n \in \R^d$ each satisfy $\norm{x^j} \leq 1$.  Let $x = \frac{1}{n} \sum_j x^j$.  If $\norm{x} \geq 1 - \delta$, then for any $p > 0$, at least $(1-p) n$ of the $x^j$ must satisfy $\norm{x^j - x} \leq \sqrt{2 \delta / p}$.  
\end{claim}

\begin{proof}
Let $v = x / \norm x$.  Then since $\norm v = 1$, all~$x^j$ satisfy $x_j \cdot v \leq 1$, whereas $x \cdot v = \norm x \geq 1 - \delta$.  By Markov's inequality, at least $(1-p) n$ of the $x^j$ must have $x^j \cdot v \geq 1 - \delta / p$.  For each such~$x^j$, simple geometry on the unit ball implies that $\norm{x^j - x} \leq \sqrt{2 \delta/p}$.  
\end{proof}

For a state $\rho \in \L((\C^2)^{\otimes q})$ and a Pauli operator~$Q \in \{I, X, Y, Z\}^{\otimes q}$, let $\rho_Q = \frac{1}{\sqrt{2^q}} \Tr (Q \rho)$.  Let $\vec \rho = ( \rho_Q : Q \in \{I, X, Y, Z\}^{\otimes q})$ be the vector of weighted Pauli coefficients.  Then $\norm{\vec \rho}^2 = \frac{1}{2^q} \sum_Q (\Tr Q \rho)^2 = \Tr (\rho^2) \leq 1$.  Since $\pi^o$ is a pure state, $\norm{\vec \pi^o} = 1$, implying that $\norm{\vec \tau^o} \geq 1 - 2^{q/2} c' \sqrt{\delta'}$.  Applying \claimref{t:markovunitball} for $p = n^{-1/8}$, at least $(1 - p) N^o$ of the~$j$ with $O_j = o$ must satisfy $\norm{\vec \sigma_j - \vec \tau^o}{}^2 \leq \frac{2}{p} 2^{q/2} c' \sqrt{\delta'}$.  By a triangle inequality, also $\norm{\vec \sigma_j - \vec \pi^o} \leq \delta''$, where $\delta'' = \sqrt{\frac{2}{p} 2^{q/2} c' \sqrt{\delta'}} + 2^q c' \sqrt{\delta'} = \tilde O(n^{-1/16})$.  Thus, for~$J$ drawn uniformly at random from~$[n]$, 
\begin{equation*}
\Pr\!\big[ \norm{\vec \sigma_J - \vec \pi^{o_J}} \leq \delta'' \,\big\vert\, O_{1,n} = o_{1,n} \big]
\geq (1 - \epsilon')(1-p) \geq 1 - (\epsilon' + p)
 \enspace .
\end{equation*}

By a Markov inequality, at least $(1 - \sqrt{\epsilon' + p}) n = (1 - O(n^{-1/16})) n$ of the coordinates $j \in [n]$ satisfy $\Pr[\norm{\vec \sigma_j - \vec \pi^{o_j}} \leq \delta'' \,\vert\, O_{1,n} = o_{1,n}] \geq 1 - \sqrt{\epsilon' + p}$.  To complete the theorem, we will show: 

\begin{claim}
Letting $1 - \eta = \Pr\!\big[\norm{\vec \sigma_j - \vec \pi^{o_j}} \leq \delta'' \,\big\vert\, O_{1,n} = o_{1,n}\big]$, $\Tr \!\big((\rho_o)_{\sigma, j} \pi^{o_j}\big) \geq 1 - \delta'' - 2 \eta$.  
\end{claim}

\begin{proof}
Observe that $\sigma_j$ is a fixed function of $A_{1,j-1}$, the random transcript of Alice's interactions with Eve for rounds $\sigma(1, 1), \sigma(1, 2), \ldots, \sigma(j-1, q-1), \sigma(j-1, q)$.  Furthermore, since Alice's measurements in these earlier rounds are on qubits in tensor product with qubits $\sigma(j, 1), \ldots, \sigma(j, q)$, it holds that 
\begin{equation} \label{e:partialtraceiscyclicforoperatorssupportedononehalf}
(\rho_o)_{\sigma, j} = \sum_{a_{1,j-1}} \Pr[A_{1,j-1} = a_{1,j-1}] \sigma_j(a_{1,j-1})
 \enspace .
\end{equation}
Indeed, in general, given a bipartite state $\rho \in \L(\H_1 \otimes \H_2)$ and a set of Kraus operators $E_i$ acting on~$\H_2$ and satisfying $\sum_i E_i^\dagger E_i = \identity$, it holds that $\Tr_2 \rho = \sum_i \Tr_2\!\big( (\identity \otimes E_i) \rho (\identity \otimes E_i^\dagger) \big)$.  Eq.~\eqnref{e:partialtraceiscyclicforoperatorssupportedononehalf} follows by letting $\H_1$ be the space of Alice's qubits $\sigma(j, 1), \ldots, \sigma(j, q)$ and $\H_2$ be everything else, letting~$\rho$ be the initial state~$\rho_o$, so $\Tr_2 \rho_o = (\rho_o)_{\sigma, j}$, and letting the $E_i$ be Eve and Alice's measurement operators for the earlier rounds.  

By linearity, Eq.~\eqnref{e:partialtraceiscyclicforoperatorssupportedononehalf} implies that also $\overrightarrow{(\rho_o)}_{\sigma, j} = \sum_{a_{1,j-1}} \Pr[A_{1,j-1} = a_{1,j-1}] \vec \sigma_j(a_{1,j-1})$.  Thus, 
\begin{equation*}\begin{split}
\norm{ \overrightarrow{(\rho_o)}_{\sigma, j} - \vec \pi^{o_j} }
&\leq \sum_{a_{1,j-1}} \Pr[A_{1,j-1} = a_{1,j-1}] \bignorm{\vec \sigma_j(a_{1,j-1}) - \vec \pi^{o_j}} 
\leq (1 - \eta) \delta'' + \eta \cdot 2
 \enspace .
\end{split}\end{equation*}
In particular, $\Tr ((\rho_o)_{\sigma, j} \pi^{o_j}) = \overrightarrow{(\rho_o)}_{\sigma, j} \cdot \vec \pi^{o_j} \geq 1 - (\delta'' + 2 \eta)$.  
\end{proof}

Thus for at least a $1 - O(n^{-1/16})$ fraction of the coordinates~$j$, $\Tr((\rho_o)_{\sigma, j} \pi^{o_j}) \geq 1 - O(n^{-1/16})$.  
\end{proof}

\thmref{t:statetomography} assumes that Alice's strategy is ideal.  Next we will relax this assumption and allow both provers to follow dishonest strategies.  To do so, we combine the state tomography protocol with a set of sequential CHSH games, analyzed in \thmref{t:sequentialstructureandobservedcorrelationsapplied}.  

\begin{theorem} \label{t:statetomographydishonestAlice}
Fix $\xzdeterminedset = \{ \pi^1, \ldots, \pi^{2^q} \}$ a complete, orthonormal set of $q$-qubit $X\!Z$-determined pure states.  For a sufficiently large constant~$\alpha$ and for sufficiently large~$n$, let $m = m(n) \geq q n$ and $N \geq m^{\alpha - 1}$.  Let $\sigma \in [m]^{q n}$ be a list of distinct indices.  Consider a combination of the following two protocols between the verifier, Eve, and the provers, Alice and Bob: 
\begin{enumerate}
\item
CHSH games: In the first protocol, Eve referees $N m$ sequential CHSH games.  She accepts if 
\begin{equation}
\bigabs{\{ j \in [N m] : A_j B_j = X_j \oplus Y_j \}} \geq \cos^2(\pi/8) N m - \tfrac{1}{2 \sqrt 2} \sqrt{N m \log(N m)}
 \enspace .
\end{equation}
\item
State tomography: In the second protocol, Eve chooses $K \in [N]$ uniformly at random.  She referees $(K-1) m$ CHSH games.  For the $K$th set, she referees a state tomography protocol with parameters~$q$, $n$, $m$, $\xzdeterminedset$ and $\sigma$.  She accepts if the criteria of Eq.~\eqnref{e:tomographyacceptancecriteria} are satisfied.  
\end{enumerate}

The combined protocol satisfies the following completeness and soundness conditions: 
\begin{description}
\item[Completeness:]
If Alice and Bob use $N m$ shared EPR states to play the CHSH games according to an ideal strategy, and if Bob uses an ideal strategy with respect to the projections~$\xzdeterminedset$ on the~$K$th set of $m$ EPR states in the state tomography protocol, then in both protocols, 
\begin{equation}
\Pr[\text{Eve accepts}] \geq 1 - O(n^{-1/2}) 
 \enspace .
\end{equation}
\item[Soundness:]
Assume that for both protocols, $\Pr[\text{Eve accepts}] \geq 1 - n^{-1/3}$.  Let $\rho$ be Alice's state in the second protocol after $(K-1) m$ games and conditioned on Bob's messages~$O_1, \ldots, O_n$.  Then there exists an isometry $\XA : \H_A \hookrightarrow (\C^2)^{\otimes m} \otimes \H_A'$ such that letting $\rho_{\sigma, j}$ be $\XA \rho \XA{}^\dagger$ reduced to Alice's qubits $\{ \sigma(j, i) : i \in [q] \}$, 
\begin{equation}
\Pr\!\Big[
\bigabs{ \big\{ 
j \in [n] : \Tr (\rho_{\sigma, j} \pi^{O_j}) \geq 1 - O(n^{-1/16})
\big\} } \geq \big(1 - O(n^{-1/16})\big)n
\Big] \geq 1 - 4 n^{-1/12}
 \enspace .
\end{equation}
Here, the probability is over $K$, the first $(K-1) m$ games and $O_1, \ldots, O_n$.  

The isometries~$\XA$ depend only on the first $(K-1) m$ games, not on $O_1, \ldots, O_n$, and are the isometries promised by \thmref{t:sequentialstructureandobservedcorrelationsapplied} for determining an $m^{-\alpha/(32 \kappaEPR)}$-ideal strategy for the $K$th set of~$m$ CHSH games.  
\end{description}
\end{theorem}

\begin{proof}
The completeness condition for sequential CHSH games follows by \thmref{t:sequentialstructureandobservedcorrelationsapplied} and \lemref{t:honestproverscentrallimit}.  The completeness condition for state tomography follows by \thmref{t:statetomography}.  

Next we will argue soundness.  Let $\epsilon = n^{-1/3}$ and $\zeta = m^{-\alpha / (32 \kappaEPR)}$, where~$\kappaEPR$ is the constant from \thmref{t:sequentialCHSHgames}.  Since the provers win the sequential CHSH games with probability at least $1 - \epsilon$, by \thmref{t:sequentialstructureandobservedcorrelationsapplied} there is at least a $1 - \epsilon - m^{-\alpha/8}$ probability that the provers' strategy for the $K$th set of $m$ games is $\zeta$-ideal.  

Whether or not the provers' strategy for a set of games is $\zeta$-ideal is a property determined at the beginning of that set of games.  It does not depend on any subsequent events.  Since Bob's strategy for the first $(K-1) m$ rounds is the same regardless of whether Eve is running CHSH games or state tomography, it therefore also holds that there is at least a $1 - \epsilon - m^{-\alpha/8}$ probability that the initial state and Alice's strategy is $\zeta$-ideal for the $K$th set of games in the state tomography protocol.  By a union bound, there is at least a $1 - \epsilon - m^{-\alpha/8} - 2 n^{-1/12} \geq 1 - 3 n^{-1/12}$ probability that, additionally, the probability that Eve accepts the state tomography protocol, conditioned on~$K$ and the $(K-1) m$ previous games, is at least $1 - \frac12 n^{-1/4}$.  

Assume that the provers' strategy for the $K$th set of CHSH games is $\zeta$-ideal and $\Pr[$Eve accepts $\vert$ $K$, previous $(K-1) m$ games$] \geq 1-\frac12 n^{-1/4}$.  Using the notation from \thmref{t:sequentialstructureandobservedcorrelationsapplied}, this implies that there exist isometries $\XX : \H_\device \hookrightarrow (\C^2)^{\otimes m} \otimes \H_\device'$ such that, letting $\XAB(\rho) = (\XA \otimes \XB) \rho (\XA \otimes \XB)^\dagger$, $\trnorm{\XAB(\rhoone) - \rhodecone{\hat}} \leq \zeta$ and $\trnorm{\XAB \EXj{1,m}(\rhoone) - \EXdecj{\hat}{1,m}(\rhodecone{\hat})} \leq 2 \zeta$.  For notational simplicity, we can embed $\H_\device$ into $(\C^2)^{\otimes m} \otimes \H_\device'$, extend the prover's measurements, and choose a basis so $\XX = \identity$.  Thus we have that 
\begin{equation*}\begin{split}
\trnorm{\rhoone - \rhodecone{\hat}} &\leq \zeta \\
\trnorm{\EAj{1,m}(\rhoone) - \EAdecj{\hat}{1,m}(\rhodecone{\hat})} &\leq 2 \zeta
 \enspace .
\end{split}\end{equation*}

\def\EB {\E^B}
Let $\EB$ be the measurement super-operator Bob uses to determine his responses $O_1, \ldots, O_n$.  We have that $\EB(\rhoone) = \sum_{o \in [2^q]^n} \ketbra o o \otimes \rho_o$ for matrices $\rho_o$ satisfying $\Tr \rho_o = \Pr[O_1 \ldots O_n = o]$.  Similarly, $\EB(\rhodecone{\hat}) = \sum_o \ketbra o o \otimes \hat \rho_o$ for certain matrices~$\hat \rho_o$.  Let $\hat O_1, \ldots, \hat O_n \in [2^q]$ be random variables distributed according to $\Pr[\hat O_1 \ldots \hat O_n = o] = \Tr \hat \rho_o$.  Then $\trnorm{\EB(\rhoone) - \EB(\rhodecone{\hat})} \leq \trnorm{\rhoone - \rhodecone{\hat}} \leq \zeta$.  By \lemref{t:blockdiagonaltracedistance}, the total variation distance between the distributions of $O = O_1 \ldots O_n$ and $\hat O = \hat O_1 \ldots \hat O_n$ is at most $\zeta/2$, and furthermore, letting $\rho_o' = \rho_o / \Tr \rho_o$ and $\hat \rho_o' = \hat \rho_o / \Tr \hat \rho_o$, $\Ex[\trnorm{\rho_O' - \hat \rho_O'}] \leq 2 \zeta$.  

For a state~$\rho$, let~$\rho_{\sigma, j}$ be its partial trace onto Alice's qubits $\sigma(j, 1), \ldots, \sigma(j, q)$.  For $\eta \geq 0$ and $o \in [2^q]^n$, define $\rho$ to be \emph{$\eta$-good for~$o$} if for at least a $1 - O(n^{-1/16})$ fraction of the coordinates~$j \in [n]$, $\Tr (\rho_{\sigma, j} \pi^{\smash{\hat O_j}}) \geq 1 - \eta$.  

Since Eve accepts $\EAj{1,m} \EB(\rhoone)$ with probability at least $1 - \frac12 n^{-1/4}$, and $\trnorm{\EAj{1,m} \EB(\rhoone) - \EAdecj{\hat}{1,m} \EB(\rhodecone{\hat})} \leq 2 \zeta$, by Eq.~\eqnref{e:holevohelstromheart} the same predicate accepts $\EAdecj{\hat}{1,m} \EB(\rhodecone{\hat})$ with probability at least $1 - \frac12 n^{-1/4} - \zeta \geq 1-n^{-1/4}$ (for sufficiently large~$n$).  Since $\EAdecj{\hat}{1,m}$ and~$\rhodecone{\hat}$ are ideal, \thmref{t:statetomography} applies.  
We obtain that there is at least a $1 - n^{-1/8}$ probability over $\hat O$ that $\hat \rho_{\smash{\hat O}}'$ is $O(n^{-1/16})$-good for~$\hat O$.  Since the distributions of~$O$ and $\hat O$ are $\zeta/2$-close, there is at least a $1 - n^{-1/8} - \zeta/2$ probability over~$O$ that $\hat \rho_{\smash{O}}'$ is $O(n^{-1/16})$-good for~$O$.  

Since $\trnorm{\EB(\rhoone) - \EB(\rhodecone{\hat})} \leq \zeta$, \lemref{t:blockdiagonaltracedistance} implies that with probability at least $1 - \sqrt{2 \zeta}$ over~$O$, $\trnorm{\rho_O' - \hat \rho_O'} \leq \sqrt{2 \zeta}$.  By a union bound, there is at least a $1 - n^{-1/8} - \zeta/2 - \sqrt{2 \zeta} \geq 1 - 2 n^{-1/8}$ probability that $\rho_O$ is $\eta$-good for~$O$, where $\eta = O(n^{-1/16}) + \frac12 \sqrt{2 \zeta} = O(n^{-1/16})$.  

The inequality $(1 - 3 n^{-1/12})(1 - 2n^{-1/8}) \geq 1 - 4 n^{-1/12}$ completes our proof.  
\end{proof}

In our application of \thmref{t:statetomographydishonestAlice}, we will sample a uniformly random set $S \subset [n]$ of fixed size~$s$.  With high probability, for all $j \in S$, $\Tr (\rho_{\sigma, j} \pi^{O_j}) \geq 1 - O(n^{-1/16})$.  \lemref{t:purepartsdeterminethewhole} implies that the reduction of~$\rho$ to Alice's qubits $\{ \sigma(j, i) : j \in S,\, i \in [q] \}$ is within $O(s n^{-1/32})$ from $\bigotimes_{j \in S} \pi^{O_j}$ in trace distance.  For this to be meaningful, we will pick~$s \ll n^{1/32}$ coordinates.  

A problem with \thmref{t:statetomographydishonestAlice} is that the soundness condition is hard to apply directly.  The theorem gives us control over Alice's state conditioned on Bob's messages, but it does not say anything about the distribution of Bob's messages.  The verification criterion of Eq.~\eqnref{e:tomographyacceptanceenoughstatisticsperoutcome} constrains Bob to report measuring $\pi^j$ on roughly a $1/2^q$ fraction of his messages, for $j \in [2^q]$.  However, he might, for example, output $O_1 = \cdots = O_{n/2^q} = 1$, $O_{n/2^q+1} = \cdots = O_{2 n / 2^q} = 2$, and so on, following a deterministic strategy.  Having to condition always on Bob's messages would severely complicate our later analysis.  Therefore, we next extend \thmref{t:statetomographydishonestAlice} to show that on a random subset of the coordinates~$j \in [n]$, with high probability both $\rho_{\sigma,j}$ is close to $\pi^{O_j}$ and $O_j$ is distributed nearly uniformly.  Thus the effect of Bob's super-operator on Alice's qubits for these coordinates is close to the effect of the ideal super-operator.  

It is possible to control the distribution of Bob's measurements because he shares with Alice a state that is close to a tensor product of EPR states, which to either party looks maximally mixed.  The more he controls his measurement outcome the less effect the measurement has on Alice's portion of the state.  The following lemma states this claim in a slightly more abstract setting: 

\begin{lemma} \label{t:blockscloseinexpectationimpliesmatricescloseandmore}
Let $\ket \psi = \frac{1}{\sqrt d} \sum_{i \in [d]} \ket{i, i}_{AB} \otimes \ket{\psi'}_{A'B'} \in \C^{[d]}_A \otimes \C^{[d]}_B \otimes \H_{A'} \otimes \H_{B'}$ and $\rho = \ketbra \psi \psi$, for Hilbert spaces $\H_{A'}$ and~$\H_{B'}$.  
Let~$\EBdecj{\hat}{}$ be the measurement super-operator for the computational-basis measurement on $\H_B$, i.e., its Kraus operators are $\hat E_i = \ket i \otimes (\ketbra{i}{i}_B \otimes \identity_{A A' B'})$ for~$i \in [d]$.  
Let $\{ \Pi_{i \ell} \}$, where $i \in [d]$ and $\ell$ varies over some finite set, be a complete set of orthogonal projections on $\C^{[d]}_B \otimes \H_{B'}$.  
Let $\EBj{}$ be the super-operator with Kraus operators $E_{i \ell} = \ket i \otimes (\Pi_{i \ell})_{BB'} \otimes \identity_{A A'}$; it corresponds to measuring~$i$ and~$\ell$, and then tracing out~$\ell$.  
Let $p_i = \sum_\ell \norm{\Pi_{i \ell} \ket \psi}{}^2$ be the probability of measuring~$i$, and when $p_i > 0$ let $\rho_i = \frac{1}{p_i} \Tr_{A'BB'} \sum_\ell \Pi_{i \ell} \rho$ be the resulting state reduced to $\H_A$.  

Assume that $\sum_{i : \trnorm{\rho_i - \ketbra i i} \leq \epsilon} p_i \geq 1 - \epsilon$.  Then, 
\begin{equation} \label{e:blockscloseinexpectationimpliesmatricescloseandmore}
\bigtrnorm{ \Tr_{BB'} \! \big( \EBj{}(\rho) - \EBdecj{\hat}{}(\rho) \big) } \leq 31 \epsilon^{1/3}
 \enspace .
\end{equation}
\end{lemma}

A state that is block diagonal defines a probability distribution over the blocks given by their traces, and defines conditional states given by the renormalized blocks.  For two states that are simultaneously block diagonal, the trace distance between them is small if and only if their distributions over blocks are close in total variation distance, and if for most blocks, drawn according to either distribution, the conditional states are close.  (See \lemref{t:blockdiagonaltracedistance}.)  In this lemma, however, we are only given that the conditional states are usually close, and we need to show that this implies the distributions are also close.  

\begin{proof}[Proof of \lemref{t:blockscloseinexpectationimpliesmatricescloseandmore}]
Let $\Pi_i = \sum_\ell \Pi_{i \ell}$.  The main claim puts an upper bound on the probability of any outcome~$i$ for which $\rho_i$ is close to~$\ketbra i i$: 

\begin{claim} \label{t:outcomescantbetoolikely}
For any~$i$ with $\trnorm{\rho_i - \ketbra i i} \leq 1$, 
\begin{equation}
p_i - \frac{1}{d} \leq \frac{1}{d} \bigtrnorm{\rho_i - \ketbra i i}
 \enspace .
\end{equation}
\end{claim}

\begin{proof}
Let $c_{ijk} = \bra{k, \psi'} \Pi_i \ket{j, \psi'}$.  Then $p_i \rho_i = \frac{1}{d} \sum_{j,k} c_{ijk} \ketbra j k$ and $p_i = \frac{1}{d} \sum_j c_{ijj}$.  Thus, using the general inequality $\trnorm{\sigma} \geq \sum_j \abs{ \bra j \sigma \ket j }$, 
\begin{equation*}
\bigtrnorm{\ketbra i i - \rho_i}
= \Bigtrnorm{ \ketbra i i - \frac{\sum_{j,k} c_{ijk} \ketbra j k}{\sum_j c_{ijj}} }
\geq \Big( 1 - \frac{c_{iii}}{\sum_j c_{ijj}} \Big) + \frac{\sum_{j \neq i} c_{ijj}}{\sum_j c_{ijj}}
= \frac{2 \sum_{j \neq i} c_{ijj}}{\sum_j c_{ijj}}
 \enspace ,
\end{equation*}
Let $\delta = \trnorm{\ketbra i i - \rho_i} \leq 1$ and $S = \sum_{j \neq i} c_{ijj}$.  Since $c_{iii} \leq 1$, we have $\frac{\delta}{2} \geq S / (1 + S)$, or $S \leq \delta / (2-\delta) \leq \delta$.  Thus, 
\begin{equation*}
p_i - \frac{1}{d} = \frac{1}{d} \big( c_{iii} + S - 1 \big) \leq \frac{\delta}{d}
 \enspace . \qedhere
\end{equation*}
\end{proof}

As a consequence of this claim, $\sum_i \bigabs{p_i - \tfrac{1}{d}} \leq 4 \epsilon$.  Indeed, call an $i \in [d]$ ``good" if $\trnorm{\rho_i - \ketbra i i} \leq \epsilon$, and ``bad" otherwise.  By assumption, $\sum_{\text{bad $i$}} p_i \leq \epsilon$.  Thus, by \claimref{t:outcomescantbetoolikely}, 
\begin{align}
\bigtrnorm{ \Tr_{AA'BB'} \! \big( \EBj{}(\rho) - \EBdecj{\hat}{}(\rho) \big) }
&= \sum_i \bigabs{p_i - \tfrac{1}{d}} \nonumber \\
&= 2 \sum_{i : p_i > 1/d} \big( p_i - \tfrac{1}{d} \big) \nonumber \\
&\leq 2 \sum_{\text{bad $i$}} p_i + 2 \sum_{\text{good $i$}} \big( p_i - \tfrac{1}{d} \big) \nonumber \\
&\leq 4 \epsilon \label{e:blockscloseinexpectationimpliesmatricescloseTVdistance}
 \enspace .
\end{align}

Therefore, we can immediately bound 
\begin{align}
\bigtrnorm{ \Tr_{A'BB'} \! \big( \EBj{}(\rho) - \EBdecj{\hat}{}(\rho) \big) } 
&= \Bigtrnorm{ \sum_i \ketbra i i \otimes p_i \rho_i - \frac{1}{d} \sum_i \ketbra i i \otimes \ketbra i i } \nonumber \\
&= \sum_i \bigtrnorm{ p_i \rho_i - \tfrac{1}{d} \ketbra i i } \nonumber \\
&\leq \sum_i p_i \bigtrnorm{ \rho_i - \ketbra i i } + \sum_i \bigabs{p_i - \tfrac{1}{d} } \nonumber \\
&\leq 7 \epsilon \label{e:blockscloseinexpectationimpliesmatricesclose}
 \enspace ,
\end{align}
where we have applied a triangle inequality and used $\bigtrnorm{p_i \ketbra i i - \tfrac{1}{d} \ketbra i i} = \bigabs{p_i - \tfrac{1}{d}}$.  

It takes more work to bound the trace distance without tracing out $\H_{A'}$.  For $i$ with $p_i > 0$, let $\tau_i = \frac{1}{p_i} \Tr_{BB'} \Pi_i \rho$.  Then $\rho_i = \Tr_{A'} \tau_i$.  Let $\rho_i' = \Tr_A \tau_i$.  Intuitively, we are given by assumption that for most~$i$, $\rho_i \approx \ketbra i i$, which means that $\tau_i$ must be close to a tensor product $\ketbra i i \otimes \rho_i'$.  
The additional conclusion of Eq.~\eqnref{e:blockscloseinexpectationimpliesmatricescloseandmore}, compared to Eq.~\eqnref{e:blockscloseinexpectationimpliesmatricesclose}, is that $\rho_i'$ is usually close to $\Tr_{B'} \ketbra{\psi'}{\psi'}$; whereas $\Tr_{A'BB'} \EBdecj{\hat}{}(\rho) = \frac{1}{d} \sum_i \ketbra i i \otimes \ketbra i i$, $\Tr_{BB'} \EBdecj{\hat}{}(\rho) = \frac{1}{d} \sum_i \ketbra i i \otimes \ketbra i i \otimes \Tr_{B'} \ketbra{\psi'}{\psi'}$.  That is, not only does Bob's super-operator properly collapse Alice's half of the maximally entangled state $\frac{1}{\sqrt d} \sum_{i \in [d]} \ket{i, i}_{AB}$, but also Bob's operation cannot significantly affect Alice's portion of the extra state $\ket{\psi'}$.  Essentially, this is because Eq.~\eqnref{e:blockscloseinexpectationimpliesmatricescloseTVdistance} implies that for most~$i$, $p_i$ is close to being uniform $1/d$---in fact, $d p_i \approx 1$ up to a small additive error.  However, a $1/d$ probability for outcome~$i$ already comes from the overlap of $\ketbra i i$ with the maximally mixed state $\frac{1}{d} \identity$.  For Bob's measurement to change substantially the state on the $A'$ register, outcome~$i$ would have to have a substantially lower probability.  

Let $\rho' = \ketbra{\psi'}{\psi'}$ and $\rho'_{A'} = \Tr_{B'} \rho'$.  Then we have 
\begin{align}
\bigtrnorm{ \Tr_{BB'} \! \big( \EBj{}(\rho) - \EBdecj{\hat}{}(\rho) \big) }
&= \Bigtrnorm{ \sum_i \ketbra i i \otimes p_i \tau_i - \frac{1}{d} \sum_i \ketbra i i \otimes \ketbra i i \otimes \rho'_{A'} } \nonumber \\
&= \sum_i \bigtrnorm{ p_i \tau_i - \tfrac{1}{d} \ketbra i i \otimes \rho'_{A'} } \nonumber \\
&\leq \sum_i p_i \trnorm{\tau_i - \ketbra i i \otimes \rho_i'} + \sum_i \trnorm{p_i \rho_i' - \tfrac{1}{d} \rho'_{A'}}
 \enspace . \label{e:blockscloseinexpectationimpliesmatricescloseandmorestepone}
\end{align}
By \corref{t:gentlemeasurementpurestate} of the Gentle Measurement Lemma, $\trnorm{\tau_i - \ketbra i i \otimes \rho_i'} \leq 3 \sqrt{1 - \bra i \rho_i \ket i}$.  

By definition, when $p_i > 0$, 
\begin{equation*}
\tau_i = \frac{1}{d p_i} \sum_{j, k \in [d]} \ketbra{j}{k}_A \otimes \Tr_{BB'} \! \big[ (\Pi_i)_{BB'} \ketbra{j}{k}_B \otimes \rho'_{A'B'} \big]
 \enspace .
\end{equation*}
Substituting $\rho_i = \Tr_{A'} \tau_i$ gives $\bra i \rho_i \ket i = \frac{1}{d p_i} \Tr \! \big[ (\Pi_i)_{BB'} \ketbra{i}{i}_B \otimes \rho'_{A'B'} \big]$, and so by the Gentle Measurement Lemma, 
\begin{equation*}
\bigtrnorm{
\ketbra{i}{i}_B \otimes \rho'_{A'B'}
- (\Pi_i)_{BB'} \ketbra{i}{i}_B \otimes \rho'_{A'B'} (\Pi_i)_{BB'}
} \leq 2 \sqrt{ 1 - d p_i \bra i \rho_i \ket i}
 \enspace .
\end{equation*}

Use $\rho_i' = \Tr_A \tau_i = \frac{1}{d p_i} \Tr_{BB'} \big[ (\Pi_i)_{BB'} \identity_B \otimes \rho'_{A'B'} \big]$ and expand $\identity = \ketbra i i - (\identity - \ketbra i i)$ to get 
\begin{align*}
\bigtrnorm{p_i \rho_i' - \tfrac{1}{d} \rho'_{A'}}
&= \frac{1}{d} \Bigtrnorm{ \Tr_{BB'} (\Pi_i)_{BB'} \identity_B \otimes \rho'_{A'B'} - \Tr_{BB'} \! \ketbra{i}{i}_B \otimes \rho'_{A'B'} } \\
&\leq \frac{1}{d} \Bigtrnorm{ (\Pi_i)_{BB'} \ketbra{i}{i}_B \otimes \rho'_{A'B'} (\Pi_i)_{BB'} - \ketbra{i}{i}_B \otimes \rho'_{A'B'} } \\ &\quad + \frac{1}{d} \bigtrnorm{ (\Pi_i)_{BB'} (\identity - \ketbra i i)_B \otimes \rho'_{A'B'} (\Pi_i)_{BB'} } \\
&\leq \frac{2}{d} \sqrt{ 1 - d p_i \bra i \rho_i \ket i} + p_i ( 1 - \bra i \rho_i \ket i )
 \enspace .
\end{align*}
In the last step, we have used that the trace norm of a positive semi-definite operator equals its trace.  

Letting $c_i = \bra i \rho_i \ket i$ and substituting into Eq.~\eqnref{e:blockscloseinexpectationimpliesmatricescloseandmorestepone}, we find 
\begin{align*}
\bigtrnorm{ \Tr_{BB'} \! \big( \EBj{}(\rho) - \EBdecj{\hat}{}(\rho) \big) }
&\leq \sum_i p_i \big( 3 \sqrt{1 - c_i} + 2 \sqrt{1 - d p_i c_i} + (1 - c_i) \big) + 2 \sum_i \bigabs{p_i - \tfrac{1}{d_i}}
 \enspace .
\end{align*}

We can next use the general inequality $1 - \bra i \rho_i \ket i \leq \frac12 \trnorm{\rho_i - \ketbra i i}$, but to make real progress we need to use the assumption $\sum_{i : \trnorm{\rho_i - \ketbra i i} > \epsilon} p_i \leq \epsilon$.  From Eq.~\eqnref{e:blockscloseinexpectationimpliesmatricescloseTVdistance}, this assumption implies that $\sum_i \abs{p_i - \frac{1}{d}} \leq 4 \epsilon$.  Let $\eta = (4 \epsilon)^{2/3}$.  Call an $i \in [d]$ ``great" if $\trnorm{\rho_i - \ketbra i i} \leq \epsilon$ and $d p_i \in \big[\frac{1}{1 + \eta}, \frac{1}{1 - \eta} \big]$.  By a Markov inequality and a union bound, $\sum_{\text{great $i$}} p_i \geq 1 - \epsilon - (4 \epsilon)^{1/3}$.  Thus, 
\begin{align*}
\bigtrnorm{ \Tr_{BB'} \! \big( \EBj{}(\rho) - \EBdecj{\hat}{}(\rho) \big) }
&\leq \Big( 3 \sqrt{\epsilon/2} + 2 \sqrt{1 - \frac{1 - \epsilon/2}{1 + \eta}} + \frac{\epsilon}{2} \Big) + (\epsilon + (4 \epsilon)^{1/3}) \cdot 6 + 2 \cdot 4 \epsilon \\
&\leq 31 \epsilon^{1/3}
 \enspace . \qedhere
\end{align*}
\end{proof}

For state tomography, the register~$A$ in \lemref{t:blockscloseinexpectationimpliesmatricescloseandmore} consists of Alice's qubits in the blocks indexed by the set~$S \subset [n]$ introduced above the lemma.  For the application of state tomography to blind, verified computation, it is enough to trace away all of the quantum registers aside from~$A$.  Alice can compute using the states prepared by Bob in this register.  Therefore, the bound in Eq.~\eqnref{e:blockscloseinexpectationimpliesmatricesclose} is sufficient.  However, for the application to simulating quantum multi-prover interactive protocols by classical protocols with entangled provers, we need Alice to work on additional input qubits, in the register~$A'$, that hold the quantum messages of the original QMIP system.  

For applying \lemref{t:blockscloseinexpectationimpliesmatricescloseandmore}, it is convenient to make two minor technical modifications: first, to allow the initial state to differ from the ideal state, and second, to allow Bob to make more measurements.  

\begin{corollary} \label{t:blockscloseinexpectationimpliesmatricescloseandmoretechnical}
Let $\rho$, $\EBdecj{\hat}{}$, $\{ \Pi_{i \ell} \}$ and $\EBj{}$ be as in \lemref{t:blockscloseinexpectationimpliesmatricescloseandmore}.  Let $\bar \rho = \ketbra{\bar \psi}{\bar \psi}$ be a state with $\trnorm{\bar \rho - \rho} \leq \zeta$.  Let $\EBdecj{\bar}{}$ be the super-operator with Kraus operators $\bar E_{i \ell} = \ket{i}_I \otimes \ket{\ell}_L \otimes (\Pi_{i \ell})_{BB'} \otimes \identity_{AA'}$.  Let $\bar p_{i \ell} = \Tr (\Pi_{i \ell} \bar \rho)$ and $\bar \rho_{i \ell} = \frac{1}{\bar p_{i \ell}} \Tr_{A'BB'} (\Pi_{i \ell} \bar \rho)$.  Then $\Tr_{LBB'} \EBdecj{\bar}{}(\bar \rho) = \Tr_{BB'} \EBj{}(\bar \rho)$ and, assuming $\sum_{i, \ell : \trnorm{\bar \rho_{i \ell} - \ketbra i i} \leq \epsilon} \bar p_{i \ell} \geq 1 - \epsilon$, 
\begin{equation}\begin{split} \label{e:blockscloseinexpectationimpliesmatricescloseandmoretechnical}
\bigtrnorm{ \Tr_{LBB'} \EBdecj{\bar}{}(\bar \rho) - \Tr_{BB'} \EBdecj{\hat}{}(\rho) } 
&\leq 42 (\epsilon + \zeta)^{1/6}
 \enspace .
\end{split}\end{equation}
\end{corollary}

\begin{proof}
Let $\bar p_i = \sum_\ell \bar p_{i \ell}$ and $\bar \rho_i = \frac{1}{\bar p_i} \sum_\ell \bar p_{i \ell} \bar \rho_{i \ell}$, so $\Tr_{LA'BB'} \EBdecj{\bar}{}(\bar \rho) = \Tr_{A'BB'} \EBj{}(\bar \rho) = \sum_i \ketbra i i \otimes \bar p_i \bar \rho_i$.  

Let us make the changes one at a time.  The first extension, to the case of $\bar \rho \approx \rho$ is a simple corollary of \lemref{t:blockdiagonaltracedistance}.  Let $\varepsilon \geq 0$ and assume for the moment that $\sum_{i : \trnorm{\bar \rho_i - \ketbra i i} \leq \varepsilon} \bar p_i \geq 1 - \varepsilon$.  Since $\trnorm{\bar \rho - \rho} \leq \zeta$, $\sum_i \abs{\bar p_i - p_i} \leq \zeta$ and $\sum_i p_i \trnorm{\bar \rho_i - \rho_i} \leq 2 \zeta$.  Therefore, for any $\delta > 0$, 
\begin{equation*}\begin{split}
\sum_{\substack{i : \\ \trnorm{\rho_i - \ketbra i i} \leq \varepsilon + \delta}} p_i
&\geq 
\sum_{\substack{i : \\ \trnorm{\rho_i - \ketbra i i} \leq \varepsilon + \delta \\ \trnorm{\bar \rho_i - \rho_i} \leq \delta}} p_i \\
&\geq
\sum_{\substack{i : \\ \trnorm{\bar \rho_i - \ketbra i i} \leq \varepsilon \\ \trnorm{\bar \rho_i - \rho_i} \leq \delta}} p_i \\
&\geq 
\sum_{\substack{i : \\ \trnorm{\bar \rho_i - \ketbra i i} \leq \varepsilon}} \bar p_i - \sum_{\substack{i : \\ \trnorm{\bar \rho_i - \rho_i} > \delta}} \bar p_i - \sum_i \abs{\bar p_i - p_i} \\
&\geq 1 - \varepsilon - \frac{2 \zeta}{\delta} - \zeta
 \enspace .
\end{split}\end{equation*}
Fixing $\delta = 2 \sqrt \zeta$, it follows from \lemref{t:blockscloseinexpectationimpliesmatricescloseandmore} and a triangle inequality that 
\begin{equation}\begin{split} \label{e:blockscloseinexpectationimpliesmatricescloseandmoretechnicalstepone}
\bigtrnorm{ \Tr_{BB'} \! \big( \Tr_L \EBdecj{\bar}{}(\bar \rho) - \EBdecj{\hat}{}(\rho) \big) } 
&= \bigtrnorm{ \Tr_{BB'} \! \big( \EBj{}(\bar \rho) - \EBdecj{\hat}{}(\rho) \big) } \\
&\leq 31 (\varepsilon + 2 \sqrt \zeta)^{1/3} + \zeta
 \enspace .
\end{split}\end{equation}

It remains to use the assumption that $\sum_{i, \ell : \trnorm{\bar \rho_{i \ell} - \ketbra i i} \leq \epsilon} \bar p_{i \ell} \geq 1 - \epsilon$ to determine an~$\varepsilon$ such that $\sum_{i : \trnorm{\bar \rho_i - \ketbra i i} \leq \varepsilon} \bar p_i \geq 1 - \varepsilon$.  Call an index~$i \in [d]$ ``good" if at least a $1 - \sqrt{\epsilon/2}$ fraction of the~$\ell$, under the distribution $\bar p_{\ell \vert i} = \bar p_{i \ell} / \bar p_i$, satisfy $\trnorm{\bar \rho_{i \ell} - \ketbra i i} \leq \epsilon$.  By assumption, $\sum_{\text{good $i$}} \bar p_i \geq 1 - \sqrt{2 \epsilon}$.  Using the expansion $\bar \rho_i = \sum_\ell \bar p_{\ell \vert i} \bar \rho_{i \ell}$ and a triangle inequality, for any good~$i$, $\trnorm{\bar \rho_i - \ketbra i i} \leq (1 - \sqrt{\epsilon/2}) \epsilon + \sqrt{\epsilon/2} \cdot 2 \leq \epsilon + \sqrt{2 \epsilon}$.  Thus $\varepsilon = \epsilon + \sqrt{2 \epsilon}$ works.  Substituting this choice into Eq.~\eqnref{e:blockscloseinexpectationimpliesmatricescloseandmoretechnicalstepone} and simplifying gives Eq.~\eqnref{e:blockscloseinexpectationimpliesmatricescloseandmoretechnical}.  
\end{proof}

\begin{theorem} \label{t:statetomographydishonestAlicesuperoperator}
With the same setup as \thmref{t:statetomographydishonestAlice}, introduce the following notation, all conditioned on~$K$ and the outcomes of the first $(K-1) m$ games.  

Let~$\rhoone$ be the provers' shared state at the beginning of the $K$th set.  Let $\XX: \H_\device \hookrightarrow (\C^2)^{\otimes m} \otimes \H_\device'$, for $\device \in \{A, B\}$, be the isometries promised by \thmref{t:sequentialstructureandobservedcorrelationsapplied} for determining a $\zeta$-ideal strategy for the $K$th set of $m$ CHSH games, where $\zeta = m^{-\alpha / (32 \kappaEPR)}$.  Assume that Bob's Hilbert space factors as $\H_B = \H_{B_1} \otimes \H_{B_2}$, and that the isometry $\XB$ factors as $\XB = \XBj{1} \otimes \XBj{2}$, with $\XBj{b} : \H_{B_b} \hookrightarrow (\C^2)^{\otimes m_b} \otimes \H_{B_b}'$, $m_1 + m_2 = m$ and $\H_B' = \H_{B_1}' \otimes \H_{B_2}'$.  Assume that $\sigma \in [m_1]^{q n}$ and that Bob's measurement super-operator for the state tomography protocol is supported only on $\H_{B_1}$.  
If the provers' strategy for the $K$th set of CHSH games is not $\zeta$-ideal, then set $\XA$, $\XBj{1}$ and $\XBj{2}$ arbitrarily.  

For a set $S \subseteq [n]$, let $\EBj{S} : \L(\H_B) \rightarrow \L( \C^{[2^q]^{\abs S}} \otimes \H_B )$ be Bob's measurement super-operator for the state tomography protocol in the $K$th set, that stores in the first register Bob's messages $O_j$ for~$j \in S$ and traces out his other messages.  Partition $(\C^2)^{\otimes m} \otimes \H_A'$ as $(\H_S \otimes \H_{\smash{\bar S}}) \otimes (\C^2)^{\otimes m_2}_{A_2} \otimes \H_A'$, where $\H_S$ consists of the qubits listed in~$\sigma$, and $\H_{\smash{\bar S}}$ consists of the remaining qubits, $[m_1] \smallsetminus \sigma$.  

Then the following soundness condition also holds: 

\begin{description}
\item[Soundness$'$:]
Assume that for both protocols, $\Pr[\text{Eve accepts}] \geq 1 - n^{-1/3}$.  Let $S \subset [n]$ be a uniformly random subset of size $s \leq n^{1/64}$.  Then there is a probability at least $1 - O(n^{-1/48})$ over~$S$, $K$ and the outcomes of the first $(K-1) m$ games that for some state $\rhoone' \in \L(\H_A' \otimes \H_{B_2}')$, the states 
\begin{equation} \label{e:statetomographydishonestAlicesuperoperatoractualstate}
\Tr_{\smash{\bar S} B_1} \XA \XBj{2} \EBj{S}(\rhoone) \XA{}^\dagger \XBj{2}{}^\dagger
\end{equation}
and 
\begin{equation} \label{e:statetomographydishonestAlicesuperoperatorgoalstate}
\frac{1}{2^{qs}} \sum_{o \in [2^q]^S} \ketbra o o \otimes \Big( \bigotimes_{j \in S} \pi^{o_j} \Big)_S \otimes (\ketbra{\psi^*}{\psi^*})^{\otimes m_2}_{A_2 B_2} \otimes \rhoone'
\end{equation}
are within trace distance $O(n^{-1/384})$ of each other.  
Here, the partial trace that reduces to $\big( \H_S \otimes (\C^2)^{\otimes m_2} \otimes \H_A' \big) \otimes \big( (\C^2)^{\otimes m_2} \otimes \H_{B_2}' \big)$ also implicitly orders the qubits in~$S$ as $\sigma(j_1, 1), \ldots, \sigma(j_1, q)$ through $\sigma(j_s, 1), \ldots, \sigma(j_s, q)$, where $S = \{ j_1, \ldots, j_s \}$.  
\end{description}
\end{theorem}

Notice that the first terms in Eq.~\eqnref{e:statetomographydishonestAlicesuperoperatorgoalstate} are just the $s$-fold tensor product of $\frac{1}{2^q} \sum_{o \in [2^q]} \ketbra o o \otimes \pi^o$.  This is exactly the state generated by an ideal state tomography strategy, reduced to the qubits for~$S$.  Note also that the EPR states $\ket{\psi^*}_{A_2 B_2}^{\otimes m_2}$ shared between Alice and $B_2$ are approximately undisturbed.  In our application of \thmref{t:statetomographydishonestAlicesuperoperator}, the factorizations of $\H_B$ and $\XB$ will be ensured by \propref{t:sequentialCHSHgamesmultipleBobs}.  

\begin{proof}[Proof of \thmref{t:statetomographydishonestAlicesuperoperator}]
The proof boils down to rearranging equations so that we can apply \corref{t:blockscloseinexpectationimpliesmatricescloseandmoretechnical} of \lemref{t:blockscloseinexpectationimpliesmatricescloseandmore}.  By \thmref{t:sequentialstructureandobservedcorrelationsapplied} and the soundness condition of \thmref{t:statetomographydishonestAlice}, with probability at least $1 - 2 n^{-1/24}$ over $K$ and the first $(K-1) m$ games, it holds that: 
\begin{enumerate}
\item 
The provers' strategy for the $K$th set of CHSH games is $\zeta$-ideal.  In particular, choosing a basis so that the isometries $\XA = \identity$ and $\XB = \identity$, there exists a state~$\ket{\psi'} \in \H_A' \otimes \H_B' \otimes \H_C$ such that, letting $\ket{\psidecone{\hat}} = \ket{\psi^*}^{\otimes m} \otimes \ket{\psi'}$ and $\rhodecone{\hat} = \ketbra{\psidecone{\hat}}{\psidecone{\hat}}$, $\trnorm{\rhoone - \rhodecone{\hat}} \leq \zeta$.  
\item
There is at least a probability $1 - 2 n^{-1/24}$ over the conditional distribution for Bob's messages $O_{1,n} = (O_1, \ldots, O_n)$ that 
\begin{equation*}
\bigabs{ \big\{ j \in [n] : \Tr (\rho_{\sigma, j} \pi^{O_j}) \geq 1 - \delta \big\} } \geq (1 - \delta)n
 \enspace ,
\end{equation*}
where $\delta = O(n^{-1/16})$.  
\end{enumerate}
Fix~$K$ and transcripts for the first $(K-1) m$ games satisfying these properties.  

Then in particular, with probability at least $1 - 2 n^{-1/24} - O(s \delta) = 1 - O(n^{-1/24})$, $\Tr (\rho_{\sigma, j} \pi^{O_j}) \geq 1 - \delta$ for all $j \in S$.  Let $\rho(O_{1,n})$ be the state conditioned on Bob's messages $O_{1,n}$ and let $\rho_S(O_{1,n}) = \Tr_{\smash{\bar S}BC} \rho(O_{1,n})$.  By \lemref{t:purepartsdeterminethewhole}, there is at least a $1 - O(n^{-1/48})$ probability over the choice of~$S$ that with at least a $1 - O(n^{-1/48})$ probability over Bob's messages $O_{1,n}$, 
\begin{equation*}
\bigtrnorm{ \rho_{S}(O_{1,n}) - \bigotimes_{j \in S} \pi^{O_j} } \leq O(s \sqrt \delta) = O(n^{-1/64})
 \enspace .
\end{equation*}

Now apply \corref{t:blockscloseinexpectationimpliesmatricescloseandmoretechnical}.  To translate into the notation of the corollary, let $i$ represent Bob's messages $O_j$ for $j \in S$, $\ell$ the other messages, $d = 2^{q s}$ and $\epsilon = O(n^{-1/64})$.  The registers~$A$ and~$A'$ correspond to $\H_S$ and $\H_{\smash{\bar S}} \otimes \big( (\C^2)^{\otimes m_2} \otimes \H_A' \big) \otimes \big( (\C^2)^{\otimes m_2} \otimes \H_{B_2}' \big)$, respectively, while $B$ and $B'$ correspond to $\H_{B_1}$'s Hilbert space components~$\H_S$ and~$\H_{\smash{\bar S}} \otimes \H_{B_1}'$, on which Bob's super-operator is allowed to act.  Thus $\bar \rho_{i \ell} = \rho_{S}(O_{1,n})$ and $\ketbra i i = \bigotimes_{j \in S} \pi^{O_j}$, satisfying the assumption $\sum_{i, \ell : \trnorm{\bar \rho_{i \ell} - \ketbra i i} \leq \epsilon} \bar p_{i \ell} \geq 1 - \epsilon$.  Observe that up to local unitaries the state $\ket{\psidecone{\hat}}$ is of the correct form; if $\pi^o = \ketbra{\pi^o}{\pi^o}$ for a unit vector $\ket{\pi^o} \in \C^{[2^q]}$ and $\ket{\bar \pi^o}$ is the entry-wise complex conjugate of $\ket{\pi^o}$, then 
\begin{equation*}\begin{split}
\ket{\psidecone{\hat}} 
&= \frac{1}{\sqrt{2^{q n}}} \sum_{x \in \{0,1\}^{q n}} \ket{x, x} \otimes \big( \ket{\psi^*}{}^{\otimes (m - q n)} \otimes \ket{\psi'} \big) \\
&= \frac{1}{\sqrt{2^{q n}}} \sum_{o \in [2^q]^n} \Big( \bigotimes_{j \in [n]} \ket{\pi^{o_j}} \otimes \bigotimes_{j \in [n]} \ket{\bar \pi^{o_j}} \Big) \otimes \big( \ket{\psi^*}{}^{\otimes (m - q n)} \otimes \ket{\psi'} \big)
 \enspace .
\end{split}\end{equation*}
(Since the states $\pi^o$ are $X\!Z$-determined, we may choose a phase so that in fact $\ket{\pi^o} = \ket{\bar \pi^o}$.)  

The states in Eqs.~\eqnref{e:statetomographydishonestAlicesuperoperatoractualstate} and~\eqnref{e:statetomographydishonestAlicesuperoperatorgoalstate} are the same as the two terms in the trace norm in the conclusion of \corref{t:blockscloseinexpectationimpliesmatricescloseandmoretechnical}, Eq.~\eqnref{e:blockscloseinexpectationimpliesmatricescloseandmoretechnical}, with $\rhoone' = \Tr_{B_1' C} \ketbra{\psi'}{\psi'}$, except with Alice's space~$\H_{\smash{\bar S}}$ additionally traced out.  In the statement of the theorem, we have chosen to trace out the $\bar S$ register since the ideal reduced state on it is maximally mixed and therefore not useful for our applications.  
\end{proof}

\subsection{Process tomography protocol} \label{s:processtomographyanalysis}

The state tomography protocol of \secref{s:statetomographyanalysis} is a major step in allowing the classical verifier Eve to certify that the quantum provers Alice and Bob indeed apply a quantum circuit of Eve's choosing.  However, it is not sufficient.  State tomography allows Eve to certify that, before Alice begins her measurements, Bob has been nearly honest in remotely preparing a set of $X\!Z$-determined states on Alice's halves of the shared EPR states.  By running the protocol with the provers' roles switched, and letting Alice go first, Eve could similarly certify that Alice has remotely prepared states on Bob's halves of the EPR states.  However, state tomography does not let Eve certify that, when Bob goes first and collapses the EPR states, Alice's measurement operators on the prepared states have the correct effect.  

To link together the provers' actions, we need a stronger guarantee on their measurements.  In this section, we will present and analyze a protocol for process tomography on Alice's measurements.  State tomography lets Eve certify that Alice's measurements have nearly the correct effect on Bob's qubits, when Alice goes first.  In contrast, process tomography will let Eve certify that Alice has applied nearly the correct measurement super-operators to {her} halves of the shared EPR states, regardless of which prover goes first.  Similar to our analysis of state tomography, our analysis in this section will initially assume that Bob's strategy is ideal.  

It is not clear that state tomography, as we have presented it, implies process tomography.  The basic problem is similar to an issue that arose in our analysis of sequential CHSH games in \secref{s:sequentialstructuredCHSHgameshavetensorproductstructure}.  Alice's strategy in early state tomography rounds might be sufficiently dishonest as to allow her in later rounds to apply completely dishonest operators.  For example, if Alice manages in early rounds to swap her halves of EPR states $q n-1$ and~$q n$, and if she conjugates her later measurement operators by this swap, then her measurement operators will be far from ideal and yet will have nearly the correct effect on Bob's qubits when Alice goes first.  This situation can certainly arise because our state tomography protocol only certifies a prover's actions in most rounds.  A prover can cheat wildly in a few rounds and be confident that her actions will be indistinguishable from statistical noise.  

Potentially, we could weaken the definition of process tomography to sidestep this problem.  After all, Eve does not care if Alice moves around her halves of the EPR states, so long as Alice and Bob together apply the correct circuit.  Instead, though, the process tomography protocol we introduce will allow Eve to certify that Alice has applied nearly the correct measurement in every round.  A key idea to make this work is to restrict consideration to Pauli stabilizer measurements~\cite{Gottesman97thesis}.  For Pauli operators in the stabilizer of a state, the measurement outcome is deterministic.  Therefore Eve does not need to average any statistics.  If Alice reports the wrong stabilizer syndrome in even a single round, then Eve will reject.  Our analysis of the protocol will be similar to some of the arguments in \secref{s:sequentialstructuredCHSHgameshavetensorproductstructure}.  We will argue that Alice's earlier measurements cannot usually overly disturb the qubits intended for use in later measurements by pulling Alice's measurement super-operators over onto Bob's halves of the EPR states.  

\smallskip

For our applications, it suffices to apply process tomography to certify that Alice correctly applies two-qubit Bell-basis measurements, i.e., measurements of the stabilizer $X \tensor X$ and $Z \tensor Z$.  However, we have generalized our analysis beyond this case, to cover arbitrary $r$-qubit measurements of tensor products of $X$ and~$Z$ operators: 

\begin{definition} \label{t:XZstabilizerset}
An $r$-qubit \emph{$X\!Z$ stabilizer set} is a subset of $\{I, X, Z\}^{\otimes r}$ that consists of pairwise commuting Pauli operators that are multiplicatively independent.  
\end{definition}

For example, $\stabilizerset = \{X \otimes X, Z \otimes Z\}$ fits the definition for $r = 2$, as does $\stabilizerset = \{ X \otimes Z \}$.  The independence condition implies that $\abs \stabilizerset \leq r$.  

\begin{definition} \label{t:permutedqubitprocesstomographyprotocoldef}
A \emph{process tomography protocol} is parameterized by natural numbers~$r$, $n$ and~$m$, with $r n \leq m$, an $r$-qubit $X\!Z$ stabilizer set~$\stabilizerset$ and a list~$\sigma$ of $r n$ distinct elements of~$[m]$.  The protocol involves a verifier, Eve, and two provers, Alice and Bob.  Alice and Bob share a state in $\H_A \otimes \H_B$.  The protocol proceeds as follows: 
\begin{itemize}
\item
Eve has one round of interaction with Alice.  First, Eve sends Alice~$\sigma$.  Alice returns to Eve a string $O_1, \ldots, O_n$, with the~$O_j \in \{0,1\}^\stabilizerset$ determined by successive $2^{\abs \stabilizerset}$-outcome projective measurements on~$\H_A$.  
\item
Eve's interaction with Bob has~$m$ rounds.  In round~$j$, Eve sends Bob an independent, uniformly random bit, $B_j$.  Bob applies a two-outcome projective measurement on~$\H_B$ to determine his reply $Y_j \in \{0,1\}$.  
\end{itemize}
No other communication is allowed.  

The initial state and Bob's strategy are \emph{ideal} if, up to local isometries, the initial state consists of~$m$ EPR states, possibly in tensor product with an additional shared state, and if in round~$j$ of his interaction with Eve, Bob returns the result of measuring his half of the~$j$th EPR state in either the $X$ eigenbasis, i.e., the $\{\ket +, \ket -\}$ basis, if $B_j = 0$, or the $Z$ eigenbasis $\{\ket 0, \ket 1\}$ if $B_j = 1$.  

Alice and Bob's joint strategy is \emph{ideal} if the initial state and Bob's strategy are ideal and, additionally, Alice acts by returning the results of measuring each successive block of~$r$ qubits listed in~$\sigma$ according to the operators in~$\stabilizerset$.  
\end{definition}

If the provers' initial shared state is ideal, then by applying local isometries we may take $\H_\device = (\C^2)^{\otimes m} \otimes \H_\device'$, for $\device \in \{A, B\}$.  The initial state is then $\ket{\psione} = \ket{\psi^*}{}^{\otimes m} \otimes \ket{\psione'}$, for some $\ket{\psione'} \in \H_A' \otimes \H_B' \otimes \H_C$, where $\H_C$ is an external space for purifying $\ket{\psione}$.  Let~$\rhodecone{\hat} = \ketbra{\psione}{\psione}$.

For $j \in [n]$ and $i \in [r]$, let $\sigma(j, i) = \sigma_{(j-1)r + i} \in [m]$.  The bit $(O_j)_P$ of Alice's response to Eve denotes the outcome of allegedly measuring the operator $P \in \stabilizerset$ on qubits $\sigma(j, 1), \ldots, \sigma(j, r)$.  Without loss of generality, we will assume that Alice's responses are determined by a complete set of $2^{n \abs \stabilizerset}$ orthogonal projections.  In particular, the bit $(O_j)_P$ is determined by measuring a reflection operator, and these operators commute for different $j \in [n]$ and $P \in \stabilizerset$.  Formally, Alice's actual and ideal measurement super-operators are defined by: 

\def\RB {R^B}
\def\RBdeca #1#2{#1{R}^B_{#2}}
\def\RBdec #1{#1{R}^B}
\def\PAdecjh #1#2#3{#1{P}^A_{#2}({#3})}
\def\PBdecjh #1#2#3{#1{P}^B_{#2}({#3})}
\def\PBdech #1#2{#1{P}^B({#2})}
\def\AmeasBdecj #1#2{#1{\F}^A_{#2}}
\def\GAj #1{\G^A_{#1}}
\def\GBj #1{\G^B_{#1}}
\def\GAdecj #1#2{#1{\G}^A_{#2}}
\def\GBdecj #1#2{#1{\G}^B_{#2}}

\begin{definition} \label{t:processtomographynotation}
For $j \in [n]$ and $P \in \stabilizerset$, let~$\RAj{j, P}$ be the reflection that Alice measures to determine bit~$P$ of her response~$O_j$.  For~$o \in \{0,1\}$, let $\PAjh{j, P}{o} = \frac12 (\identity + (-1)^o \RAj{j, P})$.  For~$o \in \{0,1\}^\stabilizerset$, let $\PAjh{j}{o} = \prod_{P \in \stabilizerset} \PAjh{j, P}{o_P}$.  Define a super-operator~$\GAj{j}$ by 
\begin{equation}
\GAj{j}(\rho) = \sum_{o_j \in \{0,1\}^\stabilizerset} \ketbra{o_j}{o_j} \otimes \PAjh{j}{o_j} \rho \PAjh{j}{o_j}
 \enspace .
\end{equation}
This measurement super-operator implements Alice's strategy for determining the response~$O_j$.  

For $j \in [n]$ and $P \in \stabilizerset$, let $\RAdecj{\hat}{j, P}$ be the Pauli operator~$P$ applied to Alice's qubits $\sigma(j, 1)$ through $\sigma(j, r)$.  Define the projections $\PAdecjh{\hat}{j, P}{o}$ and $\PAdecjh{\hat}{j}{o}$, and Alice's ideal measurement super-operator~$\GAdecj{\hat}{j}$ as above, but using the reflections $\RAdecj{\hat}{j, P}$ instead of~$\RAj{j, P}$.  

Let $\GAj{1, n} = \GAj{n} \cdots \GAj{2} \GAj{1}$ and $\GAdecj{\hat}{1, n} = \GAdecj{\hat}{n} \cdots \GAdecj{\hat}{2} \GAdecj{\hat}{1}$.  
\end{definition}

On the other hand, for query~$b \in \{0,1\}$, in the ideal strategy Bob measures the Pauli reflection $\RBdeca{\hat}{b} = \delta_{b,0} X + \delta_{b,1} Z \in \{X,Z\}$.  For $b \in \{0,1\}^r$, let $\RBdeca{\hat}{b} = \bigotimes_{i \in [r]} \RBdeca{\hat}{b_i}$.  
For Pauli operators $P$ and~$Q$, say that $Q \in P$ if in every coordinate either $P$ is the identity or~$P$ and~$Q$ agree (as in the proof of \lemref{t:idealstatetomographyexpectations}).  Then Bob's measurements of the operators~$\RBdeca{\hat}{b_j}$ determine a $\pm 1$ syndrome for any Pauli operator~$P$ such that $\bigotimes_j \RBdeca{\hat}{b_j} \in P$.  For example, if a state $\ket \psi$ satisfies $X \otimes I \otimes I \ket \psi = - I \otimes Z \otimes I \ket \psi = I \otimes I \otimes X \ket \psi = \ket \psi$, then $X \otimes Z \otimes I \ket \psi = - \ket \psi$.  

\begin{theorem} \label{t:processtomography}
Consider a process tomography protocol with parameters~$r$, $n$, $m$, $\stabilizerset$ and~$\sigma$.  Assume that the initial state~$\rhodecone{\hat}$ and Bob's strategy are ideal.  Say that Eve accepts at the end of the protocol if for all $j \in [n]$ and all $P \in \stabilizerset$ with syndrome determined by Bob's measurements of his qubits $\sigma(j, 1), \ldots, \sigma(j, r)$, the syndrome is~$(-1)^{(O_j)_P}$.  

This protocol satisfies the following completeness and soundness conditions: 
\begin{description}
\item[Completeness:]
If the provers' joint strategy is ideal, then Eve accepts with probability one.  
\item[Soundness:]
If Eve accepts with probability at least $1 - \epsilon$, then 
\begin{equation}
\bigtrnorm{ \GAj{1,n}(\rhodecone{\hat}) - \GAdecj{\hat}{1,n}(\rhodecone{\hat}) } \leq 10 r 2^{r/2} n \sqrt \epsilon
 \enspace .
\end{equation}
\end{description}
\end{theorem}

\begin{proof}
As the completeness claim is immediate, we need only to argue soundness.  The proof will work by pulling Alice's measurement super-operators across to ideal measurement super-operators on Bob's qubits, and then back.  This proof strategy should be familiar from \secref{s:sequentialstructuredCHSHgameshavetensorproductstructure}.  However, the argument here is considerably simpler because we know by assumption that the initial state and Bob's strategy are ideal.  

Let us begin by defining Alice's ideal super-operators acting on Bob's qubits, similar to \defref{t:canpullovereverythingdef}: 

\begin{definition}
For a fixed list~$\sigma$, for $j \in [n]$ and $P \in \stabilizerset$, let $\RBdecj{\hat}{j, P}$ be the Pauli operator~$P$ applied to Bob's qubits $\sigma(j, 1)$ through $\sigma(j, r)$.  Define the projections $\PBdecjh{\hat}{j, P}{o}$ and $\PBdecjh{\hat}{j}{o}$, and the super-operator~$\AmeasBdecj{\hat}{j}$ as in \defref{t:processtomographynotation} for $\EAdecj{\hat}{j}$, but using the reflections $\RBdecj{\hat}{j, P}$ instead of~$\RAdecj{\hat}{j, P}$.  Let $\AmeasBdecj{\hat}{1,n} = \AmeasBdecj{\hat}{n} \cdots \AmeasBdecj{\hat}{1}$.  
\end{definition}

Observe that since a measurement on one half of an EPR state can be made equivalently on the other half, $\GAdecj{\hat}{1,n}(\rhodecone{\hat}) = \AmeasBdecj{\hat}{1,n}(\rhodecone{\hat})$.  

Since Eve accepts with probability at least~$1 - \epsilon$, for every~$j \in [n]$, there is at least a $1 - \epsilon$ probability that for all~$P \in \stabilizerset$ either the syndrome of~$P$ cannot be determined from Bob's measurements or the syndrome is $(-1)^{(O_j)_P}$.  The probability that the syndrome of~$P$ can be determined is $1/2^{\abs P} \geq 1/2^r$, where $\abs P$ is the number of non-identity components of~$P$.  Therefore, for all $j \in [n]$ and $P \in \stabilizerset$, there is at most a $2^r \epsilon$ probability that the syndrome of~$P$ disagrees with $(O_j)_P$, given that it can be determined.  Since Alice and Bob's different measurements all commute, this holds regardless of the order of the measurements.  In particular, it holds when measuring the initial state~$\rhoone$.  Expressing this condition algebraically, we have 
\begin{align*}
\sum_{o \in \{0,1\}}
\Tr\!\big(
\big[ \PAjh{j,P}{o} \otimes \PBdecjh{\hat}{j,P}{o} \big]
\rhodecone{\hat} \big)
&\geq 1- 2^r \epsilon
 \enspace ,
\intertext{which simplifies to} 
\Tr\!\big( (\RAj{j,P} \otimes \RBdecj{\hat}{j,P}) \rhodecone{\hat} \big) &\geq 1 - 2 \cdot 2^r \epsilon
 \enspace .
\end{align*}

Next, we apply the following claim, a corollary of the Gentle Measurement Lemma: 

\begin{claim} \label{t:pulloverreflections}
Let $R \in \L(\H_A)$ and $R' \in \L(\H_B)$ be two reflections, and $\rho \in \L(\H_A \otimes \H_B)$ a quantum state.  Let $\delta = \frac12 \big(1 - \Tr (R \otimes R') \rho\big) \geq 0$.  Then 
\begin{equation}
\bigtrnorm{
\tfrac12(\identity + R)_A \, \rho \, \tfrac12(\identity + R)_A - \tfrac12(\identity + R')_B \, \rho \, \tfrac12(\identity + R')_B
} \leq 2 \sqrt \delta + 3 \delta
 \enspace .
\end{equation}
\end{claim}

\begin{proof}
\def \Pim {\overline \Pi}
Let $\Pi = \frac12 (\identity + R \otimes R')$, a projection, and let $\Pim = \identity - \Pi$.  By assumption, $\Tr \Pi \rho = 1 - \delta$, so also $\trnorm{\Pim \rho \Pim} = \Tr \Pim \rho = \delta$.  By the Gentle Measurement Lemma, \lemref{t:gentlemeasurement}, 
\begin{equation*}
\trnorm{\rho - \Pi \rho \Pi} = \trnorm{\Pi \rho \Pim + \Pim \rho \Pi + \Pim \rho \Pim} \leq 2 \sqrt \delta
 \enspace .
\end{equation*}
Together with two triangle inequalities, this yields 
\begin{equation*}\begin{split}
\trnorm{\rho - \Pi \rho}
= \trnorm{\Pim \rho}
&\leq \trnorm{\Pim \rho \Pi} + \trnorm{\Pim \rho \Pim} = \tfrac12 \trnorm{\Pim \rho \Pi + \Pi \rho \Pim} + \trnorm{\Pim \rho \Pim} \\
&\leq \tfrac12 \trnorm{\rho - \Pi \rho \Pi} + \tfrac32 \trnorm{\Pim \rho \Pim} \\
&\leq \sqrt \delta + \tfrac32 \delta
 \enspace .
\end{split}\end{equation*}
In particular, $\trnorm{\rho - (R \otimes R') \rho} = \trnorm{R_A \rho - R_B' \rho} \leq 2 \sqrt \delta + 3 \delta$.  The claim follows by several more triangle inequalities.  
\end{proof}

By \claimref{t:pulloverreflections} with $\delta = 2^r \epsilon$, Alice's super-operator determining her response bit $(O_j)_P$ can be pulled over to Bob's side: 
\begin{equation*}
\Bigtrnorm{\!
\sum_{o \in \{0,1\}} \!\! \ketbra o o \otimes \PAjh{j, P}{o} \rhodecone{\hat} \PAjh{j, P}{o}
- \!\!\!
\sum_{o \in \{0,1\}} \!\! \ketbra o o \otimes \tfrac12(\identity + (-1)^o P_{\sigma,j}) \rhodecone{\hat} \tfrac12(\identity + (-1)^o P_{\sigma,j})
}
\leq 2 (2 \sqrt \delta + 3 \delta)
.
\end{equation*}
Therefore, 
\begin{equation*}
\bigtrnorm{
\GAj{j}(\rhodecone{\hat}) - \AmeasBdecj{\hat}{j}(\rhodecone{\hat})
} \leq 2 \abs \stabilizerset (2 \sqrt \delta + 3 \delta)
\leq 10 r 2^{r/2} \sqrt \epsilon
 \enspace .
\end{equation*}
Since the different super-operators $\GAj{j}$ and $\AmeasBdecj{\hat}{j}$ all commute, this implies that, as claimed, 
\begin{equation*}
\bigtrnorm{\GAj{1,n}(\rhodecone{\hat}) - \GAdecj{\hat}{1,n}(\rhodecone{\hat})} 
= \bigtrnorm{\GAj{1,n}(\rhodecone{\hat}) - \AmeasBdecj{\hat}{1,n}(\rhodecone{\hat})} 
\leq n \cdot 10 r 2^{r/2} \sqrt \epsilon
 \enspace .  \qedhere
\end{equation*}
\end{proof}

As for the analysis of state tomography, from \thmref{t:statetomography} to \thmref{t:statetomographydishonestAlice}, the next step is to combine the process tomography protocol with sequential CHSH games, in order to handle the case that Bob plays dishonestly.  

\begin{theorem} \label{t:processtomographydishonestBob}
Let~$\stabilizerset$ be a fixed $r$-qubit $X\!Z$ stabilizer set.  For a sufficiently large constant~$\alpha$ and for sufficiently large~$n$, let $m = m(n) \geq r n$ and $N \geq m^{\alpha - 1}$.  Let $\mu_B$ be a distribution over lists of~$r n$ distinct elements of~$[m]$.  
Consider a combination of the following two protocols between the verifier, Eve, and the provers, Alice and Bob: 
\begin{enumerate}
\item
CHSH games: Eve referees $N m$ sequential CHSH games.  She accepts if 
\begin{equation}
\bigabs{\{ j \in [N m] : A_j B_j = X_j \oplus Y_j \}} \geq \cos^2(\pi/8) N m - \tfrac{1}{2 \sqrt 2} \sqrt{N m \log(N m)}
 \enspace .
\end{equation}
\item
Process tomography: Eve chooses $K \in [N]$ uniformly at random.  She referees $(K-1) m$ CHSH games.  For the $K$th set, she draws~$\sigma$ from $\mu_B$ and referees a process tomography protocol with parameters $r$, $n$, $m$, $\stabilizerset$ and~$\sigma$.  She accepts if for all~$j \in [n]$ Alice's reported syndromes for~$P \in \stabilizerset$ agree with the syndromes that can be determined by Bob's measurements.  
\end{enumerate}

Let $\phasegate = \exp(-i \frac\pi8 Y) = \smatrx{\cos\frac\pi8&-\sin\frac\pi8\\\sin\frac\pi8&\cos\frac\pi8}$ and let $\U$ act on $\L((\C^2)^{\otimes m})$ by $\U(\rho) = \phasegate^{\otimes m} \rho \phasegate^\dagger{}^{\otimes m}$.  The combined protocol satisfies the following completeness and soundness conditions: 
\begin{description}
\item[Completeness:]
Assume that Alice and Bob share $N m$ shared EPR states, that they use in sequence to play the CHSH games according to the ideal strategy of \tabref{f:optimalCHSHstrategy}.  Assume that Alice applies $\phasegate^{\otimes m}$ to her halves of the $K$th set of $m$ EPR states and then uses the qubits to play according to the ideal process tomography strategy.  Then in both protocols, 
\begin{equation}
\Pr[\text{Eve accepts}] \geq 1 - O(n^{-\alpha/4}) 
 \enspace .
\end{equation}
\item[Soundness:]
Assume that for both protocols, $\Pr[\text{Eve accepts}] \geq 1 - n^{-\alpha/8}$.  Let $\rhoone$ be the state in the second protocol after $(K-1) m$ CHSH games, at the beginning of the process tomography sub-protocol.  Let $\A : \L(\H_A) \rightarrow \L\big( (\C^{[m]})^{\otimes (r n)} \otimes (\C^2)^{\otimes (r n)} \otimes \H_A \big)$ be the super-operator implementing Eve's interactions with Alice in the process tomography sub-protocol; it begins by appending the state $\sum_\sigma \mu_B(\sigma) \ketbra \sigma \sigma \in \L\big( (\C^{[m]})^{\otimes (r n)} \big)$ and then applies Alice's process tomography measurement super-operator $\GAj{1,n}$ controlled on~$\sigma$.  Let the ideal super-operator for Eve's interactions with Alice be $\hat \A$, acting on $\L((\C^2)^{\otimes m} \otimes \H_A')$; like $\A$, it appends $\sum_\sigma \mu_B(\sigma) \ketbra \sigma \sigma$ and then applies $\U_A^{-1} \GAdecj{\hat}{1,n} \U_A$.  Both $\rhoone$ and~$\A$ depend on the first $(K-1) m$ CHSH games.  

Then with probability at least $1 - O(n^{-\alpha/16})$ over $K$ and the first $(K-1) m$ CHSH games, the provers' strategy for the $K$th set is $m^{-\alpha/(32 \kappaEPR)}$-ideal with respect to the isometries given by \thmref{t:sequentialstructureandobservedcorrelationsapplied}, $\XX : \H_\device \hookrightarrow (\C^2)^{\otimes m} \otimes \H_\device'$, for $\device \in \{A, B\}$; and furthermore, 
\begin{equation}
\bigtrnorm{ \XA \A(\rhoone) \XA{}^\dagger - \hat \A (\XA \rhoone \XA{}^\dagger) } 
= O(n^{1 - \alpha / (64 \kappaEPR)})
 \enspace .
\end{equation}
\end{description}
\end{theorem}

To prove this theorem, we first study the case of combining a process tomography protocol with a set of sequential CHSH games for which the provers' strategy is $\zeta$-ideal by assumption.  By having the provers play multiple sets of CHSH games and interrupting Alice before a random set, we can substitute into \thmref{t:sequentialstructureandobservedcorrelationsapplied} to justify this assumption.  

\begin{theorem} \label{t:processtomographystructuredBob}
Consider a protocol in which the verifier Eve can choose to run one of two sub-protocols: either~$m$ sequential CHSH games, or a process tomography protocol with parameters $r$, $n$, $m$, $\stabilizerset$ and~$\sigma$.  In the latter case, Eve accepts if for all $j \in [n]$ Alice's reported syndromes for $P \in \stabilizerset$ agree with the syndromes that can be determined by Bob's measurements on the indicated qubits.  

Assume that the provers' strategy for the sequential CHSH games is $\zeta$-ideal with respect to isometries $\XA$ and $\XB$, and assume that Eve accepts in the process tomography sub-protocol with probability at least $1 - \epsilon$.  Let $\phasegate = \exp(-i \frac\pi8 Y)$ and let $\U$ be the super-operator that applies~$\phasegate$ transversally.  Then 
\begin{equation}
\bigtrnorm{ \XA \GAj{}(\rhoone) \XA{}^\dagger - \U_A^{-1} \GAdecj{\hat}{} \U_A (\XA \rhoone \XA{}^\dagger) } 
\leq 10 r 2^{r/2} n \sqrt{\epsilon + \zeta} + 2 \zeta
 \enspace ,
\end{equation}
where $\rhoone$ is the provers' initial state, and $\GAj{}$ and $\GAdecj{\hat}{}$ are Alice's actual and ideal  measurement super-operators for the process tomography protocol, depending on~$\sigma$.  
\end{theorem}

\begin{proof}
By \defref{t:epsilonideal}, letting $\XAB(\rho) = (\XA \otimes \XB) \rho (\XA \otimes \XB)^\dagger$, there exists a state $\rhodecone{\hat} = (\ketbra{\psi^*}{\psi^*})^{\otimes m} \otimes \rhoone'$ such that $\trnorm{\XAB(\rhoone) - \rhodecone{\hat}} \leq \zeta$ and $\trnorm{\XAB \GBj{}(\rhoone) - \EBdecj{\hat}{1,m}(\rhodecone{\hat})} \leq 2 \zeta$, where $\GBj{}$ is the complete super-operator implementing Eve's interactions with Bob and $\EBdecj{\hat}{1,m}$ is Bob's ideal super-operator for~$m$ CHSH games.  Note that Bob's view in the process tomography protocol is the same as in the sequential CHSH games, so he follows the same strategy in both cases.  From \tabref{f:optimalCHSHstrategy}, Bob's ideal strategy for each CHSH game is based on measuring his half of an EPR state~$\ket{\psi^*}$ with one of the reflections $\RBa{b=0} = \frac{1}{\sqrt 2} \smatrx{1&1\\1&-1}$ or $\RBa{b=1} = \frac{1}{\sqrt 2} \smatrx{1&-1\\-1&-1}$.  Since $\RBa{0} = \phasegate^\dagger X \phasegate$ and $\RBa{1} = \phasegate^\dagger Z \phasegate$, Bob's ideal CHSH game strategy is equivalent to his ideal process tomography strategy up to a change of basis by~$\phasegate$.  That is, Bob's ideal measurement super-operator for process tomography is given by $\GBdecj{\hat}{} = \U_B \EBdecj{\hat}{1,m} \U_B^{-1}$.  Since $(\phasegate \otimes \phasegate) \ket{\psi^*} = \ket{\psi^*}$ and thus $\U_A \U_B (\rhodecone{\hat}) = \rhodecone{\hat}$, this implies that $\trnorm{\U_A \U_B \XAB \GBj{}(\rhoone) - \GBdecj{\hat}{} (\rhodecone{\hat})} \leq 2 \zeta$.  

Embed $\H_\device$ into $(\C^2)^{\otimes m} \otimes \H_\device'$, extend the prover's measurements, and choose a basis so $\XX = \identity$.  Let $\rhodecone{\tilde} = \U_A \U_B \rhoone$.  Then $\U_A \U_B \XAB \GBj{} (\rhoone) = \U_B \GBj{} \U_B^{-1}(\rhodecone{\tilde})$, giving 
\begin{align*}
\trnorm{\rhodecone{\tilde} - \rhodecone{\hat}} &\leq \zeta \\
\trnorm{\U_B \GBj{} \U_B^{-1} (\rhodecone{\tilde}) - \GBdecj{\hat}{}(\rhodecone{\hat})} &\leq 2 \zeta
 \enspace .
\end{align*}

Since Eve's acceptance predicate involves only the transcript registers and not the provers' internal state, it accepts $\GAj{} \GBj{}(\rhoone)$ and $\U_A \U_B \GAj{} \GBj{}(\rhoone) = (\U_A \GAj{} \U_A^{-1}) (\U_B \GBj{} \U_B^{-1}) (\rhodecone{\tilde})$ with the same probability, at least $1 - \epsilon$.  Therefore, by Eq.~\eqnref{e:holevohelstromheart} the predicate accepts $(\U_A \GAj{1,n} \U_A^{-1}) \GBdecj{\hat}{}(\rhodecone{\hat})$ with probability at least $1 - \epsilon - \frac12 \cdot 2 \zeta$.  

Since $\GBdecj{\hat}{}$ and~$\rhodecone{\hat}$ are ideal, \thmref{t:processtomography} applies.  We obtain 
\begin{align*}
\bigtrnorm{ (\U_A \GAj{} \U_A^{-1})(\rhodecone{\hat}) - \GAdecj{\hat}{}(\rhodecone{\hat}) }
&\leq 10 r 2^{r/2} n \sqrt{\epsilon + \zeta}
 \enspace .
\intertext{
Since $\U_A \U_B (\rhodecone{\hat}) = \rhodecone{\hat}$, $\bigtrnorm{ (\U_A \GAj{} \U_A^{-1})(\rhodecone{\hat}) - \GAdecj{\hat}{}(\rhodecone{\hat}) } = \bigtrnorm{ \GAj{}(\rhodecone{\hat}) - (\U_A^{-1} \GAdecj{\hat}{} \U_A)(\rhodecone{\hat}) }$, which implies by a triangle inequality that} 
\bigtrnorm{ \GAj{}(\rhoone) - (\U_A^{-1} \GAdecj{\hat}{} \U_A)(\rhoone) }
&\leq 10 r 2^{r/2} n \sqrt{\epsilon + \zeta} + 2 \zeta
 \enspace .
\end{align*}
Up to reinserting the isometries $\XA$, this is our objective.  
\end{proof}

\begin{proof}[Proof of \thmref{t:processtomographydishonestBob}]
We first argue completeness.  By \thmref{t:sequentialstructureandobservedcorrelationsapplied}, if Alice and Bob play the sequential CHSH games using an ideal strategy, then Eve accepts with probability at least $1 - m^{-\alpha/4}$.  Recall that Bob's ideal CHSH game strategy is equivalent to his ideal process tomography strategy up to a change of basis by~$\phasegate$.  If Alice makes the same basis change, then the effect is cancelled out, since $(\phasegate \otimes \phasegate) \ket{\psi^*} = \ket{\psi^*}$.  By \thmref{t:processtomography}, Eve therefore accepts the process tomography protocol with probability one.  

Next we will argue soundness.  Let $\epsilon = n^{-\alpha/8}$ and $\zeta = m^{-\alpha / (32 \kappaEPR)}$, where~$\kappaEPR$ is the constant from \thmref{t:sequentialCHSHgames}.  Since the provers win the sequential CHSH games with probability at least $1 - \epsilon$, by \thmref{t:sequentialstructureandobservedcorrelationsapplied} there is at least a $1 - \epsilon - m^{-\alpha/8}$ probability that the provers' strategy for the $K$th set of~$m$ games is $\zeta$-ideal.  By a union bound, there is at least a $1 - \epsilon - m^{-\alpha/8} - \sqrt \epsilon \geq 1 - O(n^{-\alpha/16})$ probability that, additionally, the probability that Eve accepts the process tomography protocol, conditioned on~$K$ and the $(K-1) m$ previous games, is at least~$1 - \sqrt \epsilon$.  
By a Markov inequality, at least a $1 - \epsilon^{1/4}$ fraction of the~$\sigma$ are ``good", in the sense that Eve's conditional acceptance probability is at least $1 - \epsilon^{1/4}$.  
By \thmref{t:processtomographystructuredBob}, 
\begin{equation*}
\bigtrnorm{ \XA \A(\rhoone) \XA{}^\dagger - \hat \A (\XA \rhoone \XA{}^\dagger) } 
\leq \big[ 10 r 2^{r/2} n (\epsilon^{1/4} + \zeta)^{1/2} + 2 \zeta \big] + 2 \epsilon^{1/4}
 \enspace ,
\end{equation*}
where the final $2 \epsilon^{1/4}$ term accounts for the trace distance for bad $\sigma$ terms.  The right-hand side of this inequality is $O(n^{1 - \alpha / (64 \kappaEPR)})$.  
\end{proof}

\ifx\compilefullpaper\undefined  
\bibliographystyle{alpha-eprint}
\bibliography{q}

\end{document}
\fi

\ifx\compilefullpaper\undefined  
\documentclass[11pt]{article}

\begin{document}
\tableofcontents
\fi

\def\circuit{{\mathcal C}}

\section{Verified quantum computation}

Consider a classical verifier, Eve, who wishes simulate measuring the first qubit of $\circuit \ket{0^m}$, where~$\circuit$ is a quantum circuit that uses $T$ gates from a fixed, constant-size set of two-qubit gates.  Known algorithms for this problem scale exponentially with $T$, and assuming that $\BQP \neq \P$, i.e., that classical computers cannot efficiently simulate polynomial time quantum computers, there is no polynomial-time algorithm.  

In a verified, blind quantum computation protocol, we allow Eve to interact with two quantum provers, Alice and Bob, who share a polynomial in~$T$ number of EPR states $\frac{1}{\sqrt 2}(\ket{00}+\ket{11})$.  The interaction begins with Eve announcing $T$ to the two provers.  Then after polynomially many further rounds of interaction, provided that Alice and Bob cooperate, Eve will have her simulation result---except that with a probability exponentially small in $T$ Eve will incorrectly accuse Alice and Bob of cheating.  

Furthermore, if Alice and Bob are dishonest and share an \emph{arbitrary} entangled state but cannot communicate with each other, then the protocol will satisfy the following soundness conditions: 
\begin{itemize}
\item 
Authentication/Verification: Either Eve detects cheating with probability at least $1/2$, or the final measurement distribution obtained by Eve differs from the correct measurement distribution in total variation distance by at most $\epsilon$.  
\item 
Blindness: Alice and Bob learn nothing about the quantum circuit $\circuit$ aside from its size $T$.  (For example, they do not even learn the number of qubits it involves.)  More precisely, once given~$T$, each prover could alone perfectly simulate the distribution of transcripts of the prover's interaction with Eve.  
\end{itemize}
Here $\epsilon > 0$ is a parameter chosen by Eve.  It can be inverse-polynomially small in $T$.  
Note that the probability of catching Alice and Bob cheating can be improved by serial repetition of the protocol.   Also, if Eve wishes to hide, imperfectly, the circuit size $T$ from Alice and Bob, she can pad the circuit with extra gates.  

In this section, we will present and analyze a protocol for verified, blind quantum computation.  In fact, the protocol we give will also work for outsourcing the computations of a quantum verifier in a quantum multi-prover interactive proof (QMIP) system.  Formally, we show that $\QMIP[\text{$k$ provers}] \subseteq \MIP^*[\text{$k+2$ provers}]$, where verified quantum computation can be seen as the $k = 0$ case.  Necessary background on quantum multi-prover interactive proofs is given in \secref{s:qmip} below.  

Our protocol combines four sub-protocols.  First, a sequential CHSH game protocol establishes the provers' qubits.  Second, a state tomography protocol establishes a set of $X\!Z$-determined resource states on Alice's qubits.  Third, a process tomography protocol ensures that Alice honestly makes Bell basis measurements.  Up to the choices of parameters, these protocols have been described earlier, in \secref{s:tomography}.  The fourth protocol does the computation, based on teleporting through the resource states.  

\secref{s:computationbyteleportation} reviews computation by teleportation, after which we present the protocol.  It will be straightforward to show that Eve's simulation works when the two provers are honest, except with exponentially small probability.  The blindness property will also be straightforward.  However, establishing the authentication condition is more of a challenge, and will rely heavily on our results for state and process tomography.  A new problem, though, is that in computation by teleportation, the questions Eve asks the provers depend adaptively on their previous answers.  Even process tomography with soundness exponentially close to one can be \emph{unsound} when used in a general adaptive protocol.  We solve this by arguing that, roughly, no information is conveyed from Alice to Bob or vice versa when Eve chooses her questions adaptively to implement computation by teleportation.

\subsection{Multi-prover interactive proof systems with quantum entanglement} \label{s:qmip}

A quantum multi-prover interactive proof system for a language~$L$ is a protocol of one or more rounds between a verifier and a number of provers.  All parties are given an input string~$x$, and the goal of the provers is to convince the verifier that $x$ belongs to~$L$.  The verifier runs in quantum polynomial time and can send and receive quantum messages.  The provers are quantum computationally unbounded, and may share an arbitrary entangled initial quantum state, but cannot interact with each other once the protocol begins.  The number of provers, the number of rounds, and the sizes of the messages are all restricted to be polynomial in~$\abs x$.  A language~$L$ is in the class $\QMIP$ if there is a quantum multi-prover interactive proof such that if $x \in L$, the provers can convince the verifier to accept with probability at least $2/3$; and if $x \notin L$, then no strategy of the provers can convince the verifier to accept with probability greater than~$1/3$.  

The class $\MIP^*$ consists of those languages decidable by a QMIP system in which the verifier runs in probabilistic polynomial time and all messages are classical~\cite{CleveHoyerTonerWatrous04nonlocal}---equivalently, $\MIP^*$ is the same as $\MIP$ except with the provers allowed to share initial entanglement.  

QMIP systems can be parameterized more finely according to the number of provers, the completeness and soundness parameters, and the number of turns or rounds of communication.  (A~turn is an interaction in which messages are sent in one direction, either from the provers to the verifier or vice versa.  A round consists of two turns.)  The class of languages decidable by a proof system with one prover is known as~$\QIP$, for which three turns suffice~\cite{KitaevWatrous00parallelQIP, MarriottWatrous05qma}, and which equals $\IP = \PSPACE$~\cite{JainJiUpadhyayWatrous09qipequalspspace}.  Thus with a single prover, allowing quantum messages and computation does not increase the power of the proof system.  $\QMIP$ is a much more mysterious class.  Whereas languages in the analogous classical class $\MIP$ can be decided by proof systems with only two provers~\cite{BenOrGoldwasserKilianWigderson88MIP2provers}, no similar reduction is known in the quantum case; while of course $\QMIP[\text{2 provers}] \subseteq \QMIP[\text{$k$ provers}]$ for $k \geq 2$, it is not known whether any of these inclusions are strict.  Even for the case of two provers, there is no better upper bound known than the set of all languages~\cite{KempeKobayashiMatsumotoTonerVidick07qmip}.\footnote{If the verifier is given a trusted, polynomial-size quantum advice state, the resulting class $\QIP$/$\qpoly$ contains all languages~\cite{Raz05qipslashqpolyequalsall}.}  Of course, $\QMIP$ contains $\MIP^*$, but no lower bound better than $\IP$ has been known for either class.  Very recently, though, it has been proposed that $\MIP^*$ contains $\MIP = \NEXP$~\cite{ItoVidick12entangledmipcontainsmipequalsnexp}.  

Nonetheless, Kempe et al.~\cite{KempeKobayashiMatsumotoVidick07qmip} have shown several simplifying transformations for QMIP systems.  They show: 
\begin{enumerate}
\item 
Any QMIP system can be parallelized to a three-turn system with the same number of provers, with perfect completeness and soundness parameter at least an inverse polynomial away from one.  Moreover, the verifier's message in the transformed protocol is the same to all provers: a single, uniformly random, classical bit.  Also, in the first turn, only the first prover sends a message to the verifier.  
\item 
By adding one prover, the system can be further parallelized to one-round (two turns), still with perfect completeness and soundness parameter at least an inverse polynomial away from one.  By parallel repetition using a polynomial number of additional provers, the soundness parameter can be made exponentially close to zero.  
\end{enumerate}

In our conversion of a QMIP system to a protocol with a classical verifier Eve, we will assume that the system has the first simplified form.  This is quite convenient for us, because the verifier in a general QMIP system might act in ways that subvert our converted protocol's security.  For example, she might forward messages from one prover to another, allowing them limited communication.  These messages might not help the provers in the original QMIP system.  However, in our converted protocol the original verifier's quantum workspace is stored with the provers and hidden from them.  The provers could use the extra messages to reveal this workspace to each other, breaking soundness.  It might be possible to deal with this by freshly hiding the quantum workspace before revealing any messages to the provers, but that would be complicated.  Another advantage of starting with a simple three-turn QMIP system is that it allows us to separate the tomography sub-protocols from the computation sub-protocol that actually simulates the original QMIP system.  In our converted protocol, Eve decides at random whether to run tomography or computation sub-protocols and does not tell the provers.  The provers cannot learn which sub-protocol they are in because Eve can simulate the verifier's public coin message for the original system.  Were we to start with a general QMIP system, however, this would not work and the provers could quickly learn which sub-protocol they were in.  We would need to run tomography simultaneous to computation.  While these problems might be fixable with more work in the conversion procedure, it is much simpler for us to start with a three-turn, public-coin QMIP system.

\subsection{Computation by teleportation} \label{s:computationbyteleportation}

A quantum algorithm can be implemented in three stages, initialization, computation and readout.  The initialization stage prepares a state $\ket{0^m}$, the readout stage measures the qubits in the computational basis, and the computation stage consists of applying a sequence of constant-qubit unitary gates drawn from a universal gate set.  The idea of measurement-based quantum computation is to eliminate all unitary operators and to implement computation using only adaptive local measurements.  One such scheme is the ``one-way quantum computer," which uses adaptive single-qubit measurements on a large, highly entangled cluster state~\cite{RaussendorfBriegel01cluster}.  Computation by teleportation, on the other hand, uses two-qubit measurements on resource states with up to four qubits~\cite{GottesmanChuang99teleportation}.  

Let $\G$ be the gate set consisting of the two-qubit controlled-NOT gate, or $\mathrm{CNOT}$, and a $\pi/4$ rotation about the $y$~axis of the Bloch sphere, $\phasegate := \exp(-i \frac\pi8 Y) = \smatrx{\cos(\pi/8)&-\sin(\pi/8)\\\sin(\pi/8)&\cos(\pi/8)}$.  The gate set~$\G$ is universal, meaning that any quantum circuit can be efficiently compiled to use~$\G$~\cite{Shi02universal}.  Let $H = \frac{1}{\sqrt 2}\smatrx{1&1\\1&-1}$, the Hadamard gate, and let $\ket{\psi^*} = \frac{1}{\sqrt 2}(\ket{00} + \ket{11})$, an EPR state.  
Nielsen~\cite{Nielsen01teleport} has shown: 

\begin{theorem}[Computation by teleportation~\cite{Nielsen01teleport}] \label{t:computationbyteleportation}
There exists a polynomial-time classical control procedure $\mathcal A$ that on input the description of a quantum circuit~$\circuit$ that uses~$m$ qubits and~$T$ gates from the gate set~$\G$, outputs a sample from a distribution that is $\exp(-\Omega(T))$ close in variation distance to the distribution of measuring the outputs of $\circuit \ket{0^m}$ in the computational basis.  The procedure $\mathcal A$ uses $O(T \log T)$ copies of each of the quantum states 
\begin{equation} \label{e:computationbyteleportationresourcestates}
\ket 0 , \quad (I \otimes H) \ket{\psi^*}, \quad (I \otimes \phasegate) \ket{\psi^*}, \quad \mathrm{CNOT}_{2,4} (\ket{\psi^*} \otimes \ket{\psi^*})
 \enspace .
\end{equation}
$\mathcal A$ applies Bell basis measurements, i.e., measurements in the basis $\{(I \otimes P) \ket{\psi^*} : P \in \{I,X,Y,Z\}\}$, to pairs of qubits decided on adaptively.  Aside from Bell basis measurements on the resource quantum states, $\mathcal A$ is fully classical.  
\end{theorem}

\begin{figure}
\centering
\subfigure[\label{f:computationbyteleportation}]{\includegraphics[scale=.5]{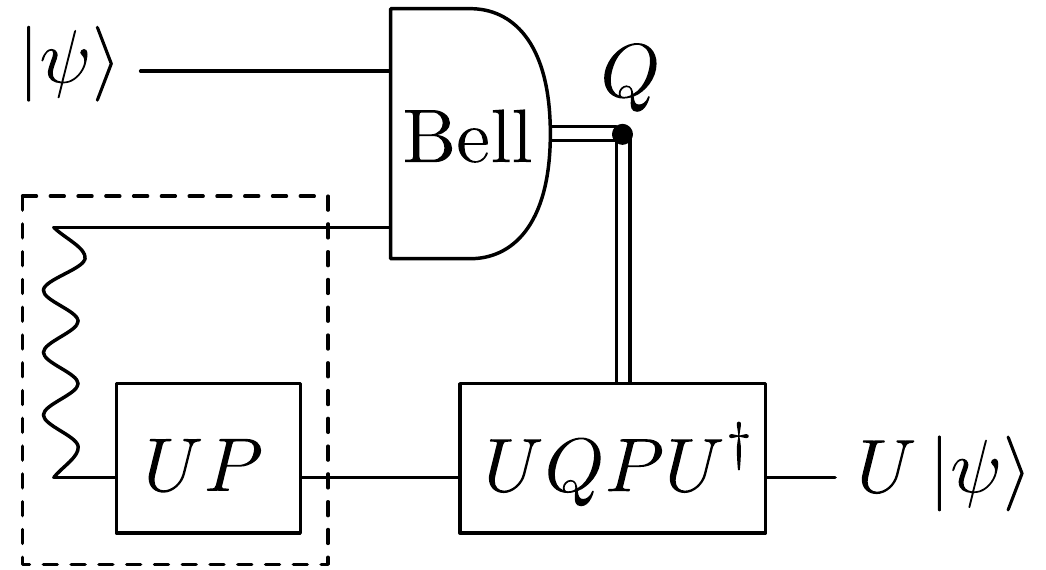}}
$\qquad\qquad\qquad$
\subfigure[\label{f:measurementbyteleportation}]{\raisebox{.25in}{\includegraphics[scale=.5]{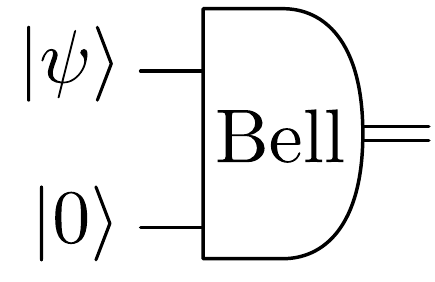}}}
\caption{(a) Computation by teleportation.  By applying a Bell basis measurement to $\ket \psi$ and one half of the resource state $(I \otimes (UP)) \ket{\psi^*}$, the unitary~$U$ is implemented on~$\ket \psi$, up to a correction~$U Q P U^\dagger$.  $Q$ is a Pauli operator determined by the outcome of the Bell measurement.  (b)~A computational-basis measurement on~$\ket \psi$ can be implemented by a Bell basis measurement on $\ket \psi \otimes \ket 0$ or $\ket \psi \otimes \ket 1$.} 
\end{figure}

The procedure~$\A$ works according to a simple extension of standard quantum teleportation~\cite{BennettBrassardCrepeauJozsaPeresWooters93teleport}.  The basic step of teleporting into a gate~$U$ using a resource state $(I \otimes U) \ket{\psi^*}$ is shown in \figref{f:computationbyteleportation} (and see~\cite{Leung02teleport}).  Depending on the outcome of the Bell basis measurement, a correction $U Q U^\dagger$ may be required, for a certain Pauli operator~$Q$.  When teleporting into a Hadamard or $\text{CNOT}$ gate, this correction is always another Pauli operator, since $H$ and $\text{CNOT}$ are in the Clifford group.  $\A$ does not actually correct for Pauli errors, but simply stores them as part of the ``Pauli frame"~\cite{Knill05}, and uses them to update later Bell measurement results.  If $U = \phasegate$, then the correction is not necessarily a Pauli operator, but it is always a Clifford operator: $\phasegate X \phasegate^\dagger = i H Y$, $\phasegate Y \phasegate^\dagger = Y$ and $\phasegate Z \phasegate^\dagger = H$.  (The operator $\phasegate$ lies in the third level of the Clifford hierarchy~\cite{GottesmanChuang99teleportation}.)  Therefore after attempting to teleport into $\phasegate$, there is a $50\%$ chance that $\A$ needs to teleport a Hadamard correction onto the output.\footnote{Alternatively, since $H = \phasegate^2 Z$, Nielsen proposes repeatedly attempting to teleport into $\phasegate$ until no correction is required---after a constant number of trials in expectation.}  
The computational-basis measurements in the final readout stage of the circuit can be implemented by a Bell measurement that uses an extra $\ket 0$ ancilla, as shown in \figref{f:measurementbyteleportation}.  

\thmref{t:computationbyteleportation} is relevant for us because all of the resource states in Eq.~\eqnref{e:computationbyteleportationresourcestates} are $X\!Z$-determined, by \thmref{t:xzdeterminedstates}, so the state tomography protocol of \thmref{t:statetomographydishonestAlice} can be applied to verify their preparation.  Moreover, a Bell basis measurement is the same as measuring the two-qubit $X\!Z$ stabilizer set $\{ X \otimes X, Z \otimes Z \}$, an operation to which the process tomography protocol of \thmref{t:processtomographydishonestBob} applies.  In our two-prover verified quantum computation protocol, the classical verifier Eve will run~$\A$.  She directs Bob to prepare the necessary resource states on Alice's qubits, and she asks Alice to apply Bell basis measurements to certain pairs of qubits.  

\begin{figure}
\centering
\includegraphics{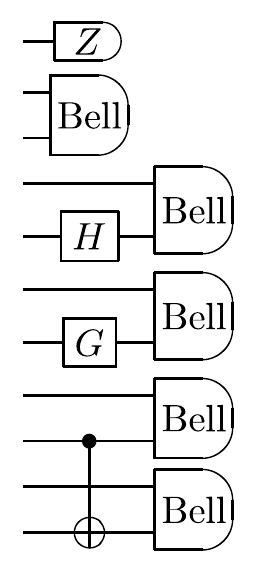}
\caption{
Once for each $\ket 0$ preparation, CNOT or measurement in $\circuit$, and twice for each $\phasegate$ gate, Eve asks Bob to prepare resource states of all the five types.  To do so, he can apply this circuit to his halves of eleven shared EPR states and report to Eve the measurement results.  
} \label{f:combinedgadgetcircuit}
\end{figure}

Our protocol will actually slightly modify the procedure~$\A$.  Instead of the resource states of Eq.~\eqnref{e:computationbyteleportationresourcestates}, use the set of resource states 
\begin{equation} \label{e:computationbyteleportationresourcestatesPaulis}
\big\{
P \ket 0 , \; (H P)_2 \ket{\psi^*}, \; (\phasegate Y)_2 \ket{\psi^*}, \; \mathrm{CNOT}_{2,4} P_2 Q_4 (\ket{\psi^*} \otimes \ket{\psi^*}) \, : \, P, Q \in \{I, X, Y, Z\}
\big\}
.
\end{equation}
Then to teleport into an $\phasegate$ gate, for example, use the next available $(\phasegate P)_2 \ket{\psi^*}$ resource state, regardless of the Pauli~$P$.  $P$ can be accounted for by a change in the Pauli frame; see \figref{f:computationbyteleportation}.  Nielsen suggests using these resource states for a constant-factor efficiency improvement---whereas projecting a uniformly mixed state onto $(I \otimes \phasegate) \ket{\psi^*}$ fails with probability $3/4$, a complete measurement in the orthogonal basis $\{ (\phasegate P)_2 \ket{\psi^*} : P \in \{I, X, Y, Z\} \}$ will always give one of those four states.  For us, efficiency is not the concern, but we want to limit the ways Eve's messages to Alice depend on Bob's reported measurement outcomes, i.e., on which of the four states $\{ (\phasegate P)_2 \ket{\psi^*} : P \in \{I, X, Y, Z\} \}$ Bob claims to have prepared.  Our protocol will also use plain $\ket{\psi^*}$ resource states.  After each $\phasegate$ gate, Eve will direct Alice to teleport into either an $(I \otimes H) \ket{\psi^*}$ state or a $\ket{\psi^*}$ state, depending on whether or not a Hadamard correction is needed.  Finally, primarily for notational simplicity but also to aid in obtaining the blindness property, Eve will always ask Bob to prepare all five of the different types of resource states together, and not just the particular resource state that is needed for the next step of computation.  \figref{f:combinedgadgetcircuit} shows the circuit that an honest Bob can use to prepare the needed resource states.

\subsection{Protocol for verified quantum computation}

\begin{theorem} \label{t:qmiptomipstarconversion}
Let $L$ be a language decided by a $k$-prover QMIP protocol with completeness~$c$ and soundness~$s$, with $c - s$ at least an inverse polynomial in the input size.  Assume that the protocol has three turns, and that the verifier's message consists of a single, uniformly random, classical bit that is broadcast to all provers.  Then $L \in \MIP^*$, decided by a protocol with $k+2$ provers.  

Furthermore, if $k = 0$, then the two-prover $\text{MIP}^*$ protocol is \emph{blind}, meaning that the provers are not given the input string~$x$ and learn only the size of the verifier's BQP circuit.  
\end{theorem}

The case $k = 0$ gives verified, blind quantum computation, at least for decision problems, since $\QMIP[\text{$0$ provers}] = \BQP$.  We will explain the extension beyond decision problems below, after the proof.  Of course, $\BQP \subseteq \PSPACE = \IP$, so without the blindness property the inclusion in $\MIP^*[\text{$2$ provers}]$ is immediate.  The case $k = 1$ is subsumed by the known equality $\QIP = \IP$~\cite{JainJiUpadhyayWatrous09qipequalspspace}.  

By the protocol transformation of~\cite{KempeKobayashiMatsumotoVidick07qmip} and since trivially $\MIP^*[\text{$k$ provers}] \subseteq \QMIP[\text{$k$ provers}]$, \thmref{t:qmiptomipstarconversion} implies: 

\begin{corollary}
$\QMIP[k \; {\rm provers}] \subseteq \MIP^*[k+2 \; {\rm provers}]$.  In particular, $\QMIP = \MIP^*$.  
\end{corollary}

\thmref{t:qmiptomipstarconversion} might appear to be straightforward given our state and process tomography theorems, and the computation by teleportation procedure.  The special form of the QMIP protocol ensures that the provers will not be able to distinguish tomography and computation sub-protocols.  The main problem, though, is that computation by teleportation is necessarily an adaptive procedure; after teleporting into an $\phasegate$ gate, a Hadamard correction might be required.  Prover strategies that pass tomography with high probability might be able to cheat in an adaptive protocol.  

A toy example should illustrate this problem.  Consider a setting in which provers Alice and Bob share $n + 2^n$ EPR states.  Eve asks Bob to measure the first $n$ EPR states in the computational basis and return the results; assume that he does so honestly.  Eve asks Alice to measure one of the~$2^n$ last EPR states.  But Alice cheats: she compares the message from Eve with her halves of the first~$n$ EPR states and acts as directed only if they disagree.  For any fixed message from Eve, this strategy will fail tomography tests with only an exponentially small probability.  However, if Eve asks her questions adaptively, by forwarding Bob's measurement results to Alice, then her message will always agree with Alice's halves of the collapsed EPR states, so Alice will always cheat.  Therefore, using tomographically verified procedures in an adaptive protocol requires some care.  

\smallskip

\begin{proof}[Proof of \thmref{t:qmiptomipstarconversion}]
Let $V$ be the verifier and $P_1, \ldots, P_k$ be the provers in the QMIP protocol.  Without loss of generality, we may assume that the verifier's protocol has the following form: 
\begin{enumerate}
\item
Receive $m$ qubits from $P_1$ and nothing from provers $P_2, \ldots, P_k$.  
\item
Choose $b \in \{0,1\}$ uniformly at random, and send~$b$ to each prover.  
\item
Receive $m$ qubits from each prover.  Apply a circuit $\circuit$ to the $(k+1)m$ message qubits and $\ket b \otimes \ket{0^{m-1}}$, where $\circuit$ consists of $T_{\text{CNOT}}$ CNOT gates and $T_{\phasegate}$ $\phasegate = \exp(-i \frac{\pi}{8} Y)$ gates.  Measure the first qubit of the output.  Accept if the qubit is~$\ket 1$ and reject if it is~$\ket 0$.  
\end{enumerate}
We may assume that only $P_1$ sends a message in the first turn, because $P_1$'s message can combine all of the provers' messages.  We may assume that all messages have length~$m$ and that the verification circuit $\circuit$ uses~$m$ workspace qubits by padding messages and the workspace.  

Let us modify this protocol by adding an initial turn in which the verifier distributes EPR states that the provers can later use to teleport back their quantum messages.  Precisely, the verifier's action in this new turn is: 
\begin{enumerate}
\setcounter{enumi}{-1}
\item Prepare $(k + 1) m$ EPR states.  Send the second halves of $2 m$ of the EPR states to~$P_1$, and the second halves of~$m$ EPR states to each of the other provers.  
\end{enumerate}
The subsequent turns are the same, except the provers send $2 m$ \emph{classical} bits whenever they would originally have sent~$m$ qubits, and before applying $\circuit$ the verifier applies the appropriate teleportation Pauli corrections.  It is without loss of generality to put the protocol into this form: 

\begin{claim} \label{t:qmipbasedonteleportation}
The teleportation-based protocol has identical completeness and soundness parameters as the original protocol.  
\end{claim}

\begin{proof}
Honest provers can use the EPR states to teleport their messages to the verifier, so the completeness parameter is unchanged.  

Dishonest provers might not follow the teleportation protocol, i.e., they might not apply Bell measurements to their halves of the EPR states.  (For provers teleporting states, all messages in $\{0,1\}^{2m}$ are equally likely, but dishonest provers might, for example, send the all-zeros string with probability one.)  However, after applying the Pauli corrections, the verifier is in possession of \emph{some} qubits that the provers might as well have teleported to her.  More formally, a cheating strategy in which the provers do not teleport some quantum messages can be converted to a strategy in which they do teleport their messages.  Indeed, consider placing between the verifier and prover $P_j$ an intermediary $P_j'$ who intercepts the original EPR states sent to~$P_j$ and sends instead halves of freshly prepared EPR states.  On receiving a classical message from $P_j$, $P_j'$ applies the appropriate correction to its halves of the new EPR states, and then honestly teleports them to the verifier.  The verifier's reduced state, after applying the Pauli corrections, and acceptance probability are the same with or without these intermediaries.  Since the combination of $P_j$ and $P_j'$ now is honestly teleporting messages to the verifier, this can be converted to a cheating strategy for the original protocol.  
\end{proof}

Now we are ready to present our converted $(k+2)$-prover $\text{MIP}^*$ protocol.  For clarity, we will define the protocol and the provers' ideal strategy simultaneously, but of course dishonest provers may deviate from this strategy.  

\def\nset{n_g}
\def\nstate{n_s}

Call the verifier in the new protocol Eve.  The two extra provers are Alice and Bob, while provers~$P_1$ through~$P_k$ play the roles of the~$k$ provers in the original QMIP protocol.  Let $q = 11$.  Let~$\alpha$ be a sufficiently large constant.  Let $n = 2((k+2) m + T_{\text{CNOT}} + T_{\phasegate})$.  By padding the verification circuit if necessary, assume that~$n$ is at least a sufficiently large constant.  Let $\nstate = n^{\alpha/2} \geq n^{64}$ and $N \geq (q \nstate)^{\alpha - 1}$.  In the ideal strategy, the provers start out with~$N$ sets of EPR states, each set consisting of $\nset = q \nstate + (k+1) m$ EPR states total: $q \nstate$ EPR states shared between Alice and Bob, $2 m$ EPR states shared between Alice and $P_1$, and $m$ EPR states shared between Alice and each of the other provers $P_2, \ldots, P_k$.  

Eve picks at random one of the following four sub-protocols to run, choosing the last sub-protocol with probability~$\delta = 1 / (6 n^{\alpha/8})$ and choosing each of the first three sub-protocols with equal probabilities $(1 - \delta)/3$.  
\begin{description}

\item[1. CHSH games:] 
Eve referees $N$ sets of sequential CHSH games, each set consisting of $q \nstate$ games between Alice and Bob, $2 m$ games between Alice and $P_1$, and $m$ games between Alice and each of the other provers.  Eve accepts if the provers win at least 
\begin{equation}
\cos^2(\pi/8) N \nset - \tfrac{1}{2 \sqrt 2} \sqrt{N \nset \log(N \nset)}
\end{equation}
of the $N \nset$ games.  By the Hoeffding inequality in \lemref{t:honestproverscentrallimit}, Eve will accept with probability at least $1 - (N \nset)^{-1/4}$ if the provers play the CHSH games honestly (where in each individual game Bob and $P_1, \ldots, P_k$ all use the strategy for Bob in \tabref{f:optimalCHSHstrategy}).  

\item[2. State tomography:] 
Eve chooses $K \in [N]$ uniformly at random.  She referees the first $K-1$ sets of CHSH games, ignoring the results.  She sends the questions for the $K$th set of games to Alice.  She chooses a uniformly random permutation $\sigma \in S_{q \nstate}$, and begins~$\nstate$ rounds of interaction with Bob.  In each round, she reveals~$q$ entries of~$\sigma$ to Bob, i.e., $\sigma(1), \ldots, \sigma(q)$ in the first message, $\sigma(q+1), \ldots, \sigma(2q)$ in the second message, and so on.  She expects each time~$q$ bits in response.  If Bob is honest, then he should start by applying $\phasegate^{\otimes (q \nstate)}$ to his EPR state qubits.  Then in each round he should apply the circuit of \figref{f:combinedgadgetcircuit} to the qubits that Eve has specified, and return the measurement results.  

Eve treats this as a permuted-qubit state tomography protocol for the $q$-qubit states given by $(S^\dagger)^{\otimes q}$ applied to the $2^q$ eigenstates of the measurement circuit of \figref{f:combinedgadgetcircuit}.  By \thmref{t:xzdeterminedstates}, all of these states are $X\!Z$-determined (\defref{t:xzdetermineddef}).  She accepts if the two criteria of Eq.~\eqnref{e:tomographyacceptancecriteria} are satisfied.  Note that unlike in \defref{t:permutedqubitstatetomographyprotocoldef} for a permuted-qubit state tomography protocol, Eve gives the permutation~$\sigma$ to Bob in pieces and not all at once.  This only restricts the ways that Bob can cheat.  By \thmref{t:statetomography}, if Alice and Bob play honestly, then Eve accepts with probability at least $1 - O(\nstate^{-1/2})$.  

\item[3. Process tomography:] 
Eve chooses $K \in [N]$ uniformly at random.  She referees the first $K-1$ sets of CHSH games, ignoring the results.  For all of the provers except Alice, she also asks the questions for a $K$th set of CHSH games.  With Alice, however, she interrupts the protocol.  She fixes~$\tau$, a uniformly random list of $(k + 1) m + 2 n$ distinct indices in $[q \nstate]$.  For $(k + 1) m$ rounds she sends Alice one index of the EPR states shared with the provers $P_1, \ldots, P_k$, and the next entry of~$\tau$.  Then for~$n$ rounds she sends Alice two entries of~$\tau$ at a time.  In each round, she expects in response two bits from Alice.  Ideally, Alice should start by applying $\phasegate^{\otimes \nset}$ to her EPR state qubits.  Then in each round she should return the results of a Bell-pair measurement on her halves of the indicated qubits.  That is, in the first $(k + 1) m$ rounds, she should apply Bell-pair measurements that cross between EPR states shared with Bob and with one of the~$P_j$, and in the following~$n$ rounds, she should apply Bell-pair measurements between two EPR states shared with Bob.  

Eve treats this as a process tomography protocol (\defref{t:permutedqubitprocesstomographyprotocoldef}) with $r = 2$ and $\stabilizerset = \{ X \otimes X, Z \otimes Z \}$.  She accepts if and only if in all rounds Alice's reported syndromes agree with the syndromes that can be determined by the other provers' measurements.  If the provers play honestly, then Eve accepts with probability one.  Recall from \thmref{t:processtomographydishonestBob} that the purpose of applying $\phasegate$ transversally is to change basis so that the ideal CHSH game strategy for Bob and $P_1, \ldots, P_k$ consists of $X$- and $Z$-basis measurements.  Also, once more, revealing~$\tau$ incrementally instead of all at once only reduces Alice's opportunities for cheating.  

\item[4. Computation:]
Eve chooses $K \in [N]$ uniformly at random and referees the first $K-1$ sets of CHSH games.  Then she asks $P_1$ for the first message in the QMIP protocol.  If $P_1$ is honest, it applies~$\phasegate$ transversally to its EPR state qubits, and returns $2 m$ classical bits, the Bell-pair measurement results needed for teleporting the first message.  Eve chooses a bit $b$ uniformly at random, and sends it to $P_1, \ldots, P_k$.  If these provers are honest, then they will apply~$\phasegate$ transversally before returning $2 m$ classical bits each, the measurement results needed for teleporting their messages in the turn three of the QMIP protocol.  

With $b$ and these $2 (k + 1) m$ bits in hand, Eve wants to direct Alice and Bob to simulate the verification circuit~$\circuit$.  

Eve's interactions with Bob are identical to the interactions in the state tomography sub-protocol.  Eve chooses a uniformly random permutation $\sigma \in S_{q \nstate}$, and over~$\nstate$ rounds reveals~$q$ entries of it at a time, expecting~$q$ bits in response in each round.  Whether or not Bob plays honestly, his strategy is identical to his strategy in the state tomography sub-protocol, since from his perspective there is no difference.  

Eve's interactions with Alice are similar, but not identical, to the interactions in the process tomography protocol.  The permutation~$\sigma$ gives the locations of where in Alice's qubits the $q$-qubit resource state blocks should be.  The first block of resource states should be in positions $\sigma(1), \ldots, \sigma(q)$, and so on.  Eve acts as though Bob is playing honestly and, one Bell-pair measurement at a time, she directs Alice to use these resource states to implement teleportation by computation, as explained in \secref{s:computationbyteleportation}.  In the first $(k+1)m$ rounds, she directs Bell-pair measurements to teleport Alice's qubits from EPR states shared with provers $P_1, \ldots, P_k$ into resource states.  She sets up the workspace of~$\circuit$, $\ket b \otimes \ket{0^{m-1}}$, by using the~$\ket 0$ or~$\ket 1$ single-qubit resource states.  Then she teleports between resource states to implement the gates of~$\circuit$, and finally she finishes with Bell-pair measurements onto the single-qubit resource states in order to implement the final measurement of~$\circuit$.  

Eve does not use all of the resource states.  Although each $q$-qubit block contains five different types of resource states---see \figref{f:combinedgadgetcircuit}---Eve uses at most one of them.  Furthermore, Eve uses at most~$n$ of the $\nstate$ blocks.  She chooses a uniformly random subset $S \subset [\nstate]$ of size $\abs S = n$.  Writing $S = \{j_1, \ldots, j_n\}$, with $j_1 < j_2 < \cdots < j_n$, Eve uses only the blocks in~$S$, in order.  For the $i$th resource state, she uses the block of qubits $\sigma((j_i -1)q+1), \ldots, \sigma(j_i q)$.  

At the end, Eve accepts if the final measurement, adjusted by the propagated Pauli frame, gives~$\ket 1$, and she rejects if it gives~$\ket 0$.  If the provers play honestly, then Eve accepts with the same probability as in the original QMIP protocol, which is at least the completeness parameter~$c$.  
\end{description}

Although glossed over in the descriptions above, Eve must also time her messages to the provers to avoid leaking information about which sub-protocol she is running.  This is fairly straightforward.  Start by specifying the timing in the computation sub-protocol.  The $K$th set in this sub-protocol begins with two rounds---four turns---in which messages are exchanged between Eve and the provers $P_1, \ldots, P_k$.  Then there are $\nstate$ rounds of interaction with Bob, and finally at most~$n$ rounds of interaction with Alice.  To account for the first two rounds, add two rounds at the beginning of the first $K-1$ sets, and of all the sets in the other sub-protocols, in which dummy messages are exchanged.  Also, in process tomography, delay sending the first indices of~$\tau$ to Alice for $\nstate$ dummy rounds, so that Alice cannot distinguish between process tomography and computation.  Then, too, Bob cannot distinguish between state tomography and computation, Bob and $P_1, \ldots, P_k$ cannot distinguish between process tomography and CHSH games, and Alice cannot distinguish between state tomography and CHSH games.  Furthermore, observe that the timing of messages in each sub-protocol is fixed, so we can assume that the provers' strategies do not depend on the message timings.  

\smallskip

This protocol is an $\text{MIP}^*$ protocol; Eve is fully classical.  Overall, if the provers play honestly and the input $x$ lies in~$L$, then Eve accepts with probability at least 
\begin{equation}\begin{split} \label{e:qmipmipcompleteness}
(1 - \delta) - \frac{1 - \delta}{3} \big( (N \nset)^{-1/4} + O(\nstate^{-1/2}) \big) + \delta c
&\geq 1 - (1 - c) \delta - O(\nstate^{-1/2}) \\
&= 1 - (1 - c + O(n^{-\alpha/8}) ) \delta
 \enspace .
\end{split}\end{equation}

Next assume that $x \notin L$.  Assume that Eve accepts with probability at least $1 - \epsilon$, where $\epsilon = \big(1 - \frac12(c+s)\big) \delta$.  Then for each of the first three sub-protocols, the probability that Eve accepts conditioned on choosing that sub-protocol is at least $1 - 3 \epsilon / (1 - \delta) > 1 - 6 \delta = 1 - n^{-\alpha/8}$.  The probability that Eve accepts conditioned on choosing the computation protocol is at least $1 - \epsilon / \delta = \frac12 (c + s)$.  

For analyzing the soundness of the protocol, we introduce a different version of the computation sub-protocol.  Whereas in the computation sub-protocol, Eve's messages to Bob are chosen non-adaptively and her messages to Alice chosen adaptively, in the alternative sub-protocol, only Eve's messages to Bob are chosen adaptively.  

Note that in the computation sub-protocol, many of the messages Eve sends to Alice are fixed by~$\sigma$, independent of Alice and Bob's responses.  For example, Eve wants the initial state to be $\ket b \otimes \ket{0^{m-1}}$ only up to a Pauli correction, so she does not care whether Bob claims to have measured $\ket 0$ or~$\ket 1$.  In fact, all of the messages Eve sends to Alice are independent of Alice and Bob's responses, \emph{except} for the two rounds immediately following teleportation into a~$\phasegate$ gate.  For these two rounds, Eve's messages to Alice are adaptive, because she wants to apply either an $I$ or an~$H$ correction to the output of the gate depending on the Pauli frame that entered it.  This Pauli frame is determined by Alice responses in all of the previous rounds, Bob's responses in all of the rounds up and including the round that was meant to prepare the $\phasegate$ resource state, and also the $2 (k+1) m$ bits that Eve received from $P_1, \ldots, P_k$---bits that determine the initial Pauli frame on Alice's qubits from EPR states shared with $P_1, \ldots, P_k$.  With~$\sigma$ fixed, there are therefore $2^{T_{\phasegate}}$ possible transcripts for the messages from Eve to Alice, two possibilities for each gate~$\phasegate$.  

Alternatively, however, Eve can fix her messages to Alice in advance, and can change the permutation she gives Bob for the block of resource states following a $\phasegate$ gate.  She either leaves the $I$ and~$H$ resource states (i.e., $\ket{\psi^*}$ and $(I \otimes H) \ket{\psi^*}$, up to Pauli operators) alone, or she switches their positions.  This defines the alternative computation sub-protocol; it is the same as the computation sub-protocol, except with the adaptive corrections made by switching the positions of the two resource states in Eve's messages to Bob.  First, Eve picks~$\sigma$ and~$S$, then she interacts with Alice, then she interacts with Bob, introducing additional swaps into~$\sigma$ adaptively when required.  Note that the timing of messages is different in this alternative sub-protocol; Alice goes before Bob.  However, we have already argued that Alice and Bob's strategies do not depend on the timing, so in our analysis we can substitute the same super-operators into this alternative and hypothetical computation sub-protocol.  

\begin{claim} \label{t:notreallyadaptive}
Whatever Alice and Bob's strategies may be, running those strategies in the computation sub-protocol and in the alternative sub-protocol gives {identical} results.  That is, the transcripts are identically distributed, and conditioned on any fixed transcript, the provers' joint states and Eve's private Pauli frames in the two sub-protocols are identical.  
\end{claim}

\begin{proof}
The proof is by the principle of deferred decisions.  Observe that since Alice and Bob act on different subsystems, their operators commute with each other, and the only important order is that imposed by Eve's adaptive decisions.  In particular, we can imagine running Alice and Bob simultaneously.  Let Eve fix the subset~$S$, but defer fixing the indices of~$\sigma$ until they are required.  Consider a $\phasegate$ gate in the circuit~$\circuit$.  Run Alice up through the Bell measurement that teleports into that gate, and run Bob until stopping just before the preparation of the next block of resource states in~$S$ (i.e., if the gate uses block $j_i \in S$, then stop before the preparation of block $j_{i+1}$).  The next step can be implemented in two ways: 
\begin{enumerate}
\item
In the computation sub-protocol, Eve picks a list of $q$ uniformly random qubit indices from unused qubits in $[q \nstate]$.  She sends these to Bob, and she sends to Alice the input position of either the~$I$ resource state or the~$H$ resource state, depending on whether a correction is required.  
\item 
In the alternative sub-protocol, Eve again picks a list of $q$ uniformly random, unused qubit indices.  She sends to Alice the second of these indices, i.e., the input position of the~$I$ resource state in \figref{f:combinedgadgetcircuit}.  She sends the~$q$ indices to Bob, but if a Hadamard correction is required, then she first swaps the indices for the~$I$ and~$H$ resource states.  
\end{enumerate}
These two different rules generate exactly the same joint distribution of messages to Alice and Bob.  The same is true for every~$\phasegate$ gate.  Therefore, the computation sub-protocol and the alternative sub-protocol are actually the same, except for the order of the messages.  
\end{proof}

\claimref{t:notreallyadaptive} is the reason why tomography characterizes the provers' strategies even though Eve's messages in the protocol are chosen adaptively---unlike in the toy counter-example at the beginning of this section.  Using this claim and the tomography theorems, we can prove soundness of the protocol.  

\smallskip

So as to frame our analysis in terms of super-operators, let us define some notation for the portion of the computation sub-protocol after the $K-1$ sets of CHSH games.  Let $\H_A$ be Alice's Hilbert space, $\H_B$ be Bob's Hilbert space and $\H_P$ be the tensor product of the Hilbert spaces of provers $P_1, \ldots, P_k$.  Let $T_{\overrightarrow A}$ be the space of transcripts for messages from Eve to Alice; it can hold~$n$ messages each holding two indices in $[q \nstate]$.  Let $T_{\overleftarrow A} = (\C^2 \otimes \C^2)^{\otimes n}$ be the space of transcripts for messages from Alice to Eve.  Let $T_A = T_{\overrightarrow A} \otimes T_{\overleftarrow A}$ be the space of transcripts for all messages to and from Alice.  Similarly, let $T_{\overrightarrow B} = (\C^{[q \nstate]^q})^{\otimes \nstate}$ be the space of transcripts for messages to Bob---$\nstate$ rounds of messages each consisting of~$q$ indices from $[q \nstate]$---let $T_{\overleftarrow B} = (\C^{2^q})^{\otimes \nstate}$ be the space for Bob's responses, and let $T_B = T_{\overrightarrow B} \otimes T_{\overleftarrow B}$.  Let $T_P = \C^2 \otimes (\C^2)^{\otimes (2 (k+1) m)}$ be the space for transcripts between Eve and the provers $P_1, \ldots, P_k$.  Let $T_{BP} = T_B \otimes T_P$, $T_{AP} = T_A \otimes T_P$ and $T_{ABP} = T_A \otimes T_{BP}$.  As the transcripts in our protocol are classical, the states in these spaces will always be diagonal in the computational basis.  

Let $\rho \in \L(\H_A \otimes \H_B \otimes \H_P)$ be the initial state of the provers at the beginning of the $K$th set.  We will leave implicit the dependence of~$\rho$, and of the super-operators defined below, on the transcripts of the first $(K-1) \nset$ games.  
Define a super-operator~$\B : \L(\H_B) \rightarrow \L(T_B \otimes \H_B)$ to implement the joint operations of Eve and Bob in the $K$th set of the computation sub-protocol.  This is the same as Eve's interaction with Bob in the state tomography protocol.  $\B$ first appends a register $\frac{1}{(q \nstate)!} \sum_{\sigma \in S_{q \nstate}} \ketbra \sigma \sigma \in \L(T_{\overrightarrow B})$, and then applies Bob's measurement super-operators for the $\nstate$ rounds in the sub-protocol.  Define a super-operator $\cP: \L(\H_P) \rightarrow \L(T_P \otimes \H_P)$ to implement Eve's interactions with the provers $P_1, \ldots, P_k$.  Define a super-operator $\Aad : \L(T_{BP} \otimes \H_A) \rightarrow \L(T_{BP} \otimes T_A \otimes \H_A)$, as the super-operator describing Eve's adaptive interactions with Alice, controlled by the transcript of her interactions with Bob and $P_1, \ldots, P_k$.  That is, applying $\Aad$ to a state $\sum_{m_{BP}} \ketbra{m_{BP}}{m_{BP}} \otimes \rho(m_{BP})$ gives $\sum_{m_{BP}} \ketbra{m_{BP}}{m_{BP}} \otimes \Aad(m_{BP})(\rho(m_{BP}))$, where $\Aad(m_{BP}) : \L(\H_A) \rightarrow \L(T_A \otimes \H_A)$ is the super-operator conditioned on the transcript~$m_{BP}$.  In the original description above, Eve computes a random subset $S \subset [\nstate]$ and also keeps track of a Pauli frame for Alice's qubits.  However, this private information can be computed, or uncomputed, from the transcripts, so the super-operator does not need to track it explicitly.  Extend $\B$ to act as the identity on $\H_A \otimes \H_P$, and similarly extend $\cP$ and~$\Aad$.  Then conditioned on~$K$ and the first $K-1$ sets of CHSH games, the computation sub-protocol finishes in the state 
\begin{equation*}
\Aad \B \cP (\rho) \in \L(T_{ABP} \otimes \H_A \otimes \H_B \otimes \H_P)
 \enspace .
\end{equation*}

Define similarly super-operators $\A : \L(\H_A) \rightarrow \L(T_A \otimes \H_A)$ and $\Bad : \L(T_{AP} \otimes \H_B) \rightarrow \L(T_{AP} \otimes T_B \otimes \H_B)$ as implementing, respectively, Eve's interactions with Alice and Eve's adaptive interactions with Bob in the alternative description of the computation sub-protocol.  Note that $\A$ is the same as Eve's interactions with Alice in the process tomography sub-protocol.  By \claimref{t:notreallyadaptive}, no matter the provers' strategies, 
\begin{equation} \label{e:qmipmipadaptivebyeitherprover}
\Bad \A \cP = \Aad \B \cP
 \enspace .
\end{equation}

Since Eve accepts the CHSH games sub-protocol with probability at least $1 - 6 \delta$, by \thmref{t:sequentialstructureandobservedcorrelationsapplied} there is at least a $1 - 6 \delta - \nset^{-\alpha/8}$ probability that the provers' strategy for the $K$th set of CHSH games, conditioned on $K$ and the first $K-1$ sets, is $\zeta$-ideal, where $\zeta = \nset^{-\alpha / (32 \kappaEPR)}$ and $\kappaEPR$ is the constant from \thmref{t:sequentialCHSHgames}.  The strategy being $\zeta$-ideal means in particular that there exist isometries $\XA : \H_A \hookrightarrow (\C^2)^{\otimes \nset} \otimes \H_A'$ and $\XB : \H_B \otimes \H_P \hookrightarrow (\C^2)^{\otimes \nset} \otimes \H_{BP}'$, and some state $\rho'$ such that, letting $\hat \rho = (\ketbra{\psi^*}{\psi^*})^{\otimes \nset} \otimes \rho'$, 
\begin{equation*}
\trnorm{ (\XA \otimes \XB) \rho (\XA \otimes \XB)^\dagger - \hat \rho } \leq \zeta
 \enspace .
\end{equation*}
The isometries $\XA$ and~$\XB$ thus define ideal qubit locations in $\H_A$ and in $\H_B \otimes \H_P$.  (In fact, by following the proof of \thmref{t:sequentialCHSHgames}, it is not difficult to see that the isometry $\XB$ respects the decomposition $\H_B \otimes \H_P$, i.e., factors as separate local isometries.  We will not need this observation, however.)  To simplify notation, we can embed $\H_A$ into $(\C^2)^{\otimes \nset} \otimes \H_A'$ and $\H_B$ into $(\C^2)^{\otimes \nset} \otimes \H_{BP}'$, and assume that $\XA$ and $\XB$ are both the identity.  

Define $\hat \B$, $\Aadhat$, $\hat \A$ and $\Badhat$ to be the ideal super-operators for provers who follow Eve's instructions in the $K$th set of the computation sub-protocol up to a basis change by~$\phasegate$.  That is, they apply~$\phasegate$ transversally, apply the specified measurements to the specified qubits of either $(\C^2)^{\otimes \nset} \otimes \H_A'$ or $(\C^2)^{\otimes \nset} \otimes \H_{BP}'$, and then apply $\phasegate^\dagger$ transversally.  Let $\V$ be the verifier's acceptance predicate based on the final transcript in $T_{ABP}$; it updates Alice's final reported measurement value according to the Pauli frame, and accepts if the result is~$\ket 1$.  By \claimref{t:qmipbasedonteleportation} and since $(\phasegate \otimes \phasegate) \ket{\psi^*} = \ket{\psi^*}$, the provers have a strategy in the original QMIP protocol that makes the verifier accept with probability exactly  
\begin{equation*}
\Pr[ \text{$\V$ accepts $\Aadhat \hat B \cP (\hat \rho)$} ]
 \enspace .
\end{equation*}
This probability is at most the soundness parameter~$s$ of the protocol, since $x \notin L$.  Our goal is to relate $\XAB \Aad B \cP (\rho)$ to $\Aadhat \hat B \cP (\hat \rho)$, and therefore to relate $\Pr[ \text{$\V$ accepts $\Aad B \cP (\rho)$} ]$ to $\Pr[ \text{$\V$ accepts $\Aadhat \hat B \cP (\hat \rho)$} ]$, in order to derive a contradiction.  

Start by using \thmref{t:processtomographydishonestBob} for process tomography.  Since Eve's acceptance probabilities in the CHSH games and process tomography sub-protocols are both at least $1 - 6 \delta > 1 - n^{-\alpha/8}$, the theorem applies.  We obtain that with probability at least $1 - O(n^{-\alpha/16})$ over $K$ and the first $K-1$ sets, 
\begin{equation*}
\bigtrnorm{ \XA \A (\rho) - \hat \A \XA (\rho) } = O( n^{1 - \alpha / (64 \kappaEPR)} )	
 \enspace .
\end{equation*}
In particular, by \claimref{t:notreallyadaptive}, 
\begin{equation}\begin{split} \label{e:adaptivecomputationprocesstomography}
\Aad \B \cP (\rho) 
&= \Bad \cP \A (\rho) \\
&\approx \Bad \cP \hat \A (\rho) \\
&= \Aadhat \cP \B (\rho)
 \enspace ,
\end{split}\end{equation}
where the approximation is up to error $O( n^{1 - \alpha / (64 \kappaEPR)} )$ in trace distance.  

To finish, we would like to use the state tomography theorem, \thmref{t:statetomographydishonestAlicesuperoperator}, to relate $\B(\rho)$ to $\hat \B(\hat \rho)$.  Since Eve's acceptance probabilities in the CHSH games and state tomography sub-protocols are both at least $1 - n^{-\alpha/8} > 1 - \nstate^{-1/3}$, the theorem applies.  However, the theorem does not give so strong a claim.  It only allows for approximating Bob's actual super-operator by his ideal super-operator if we also trace out Bob's Hilbert space, and Bob's responses and Alice's qubits for all but those corresponding to a random set $S \subset [\nset]$.  \thmref{t:statetomographydishonestAlicesuperoperator} can be applied in our situation, but to do so we will need to introduce some more notation.  

\def\AadhatS {\hat{\A}_{\text{ad,S}}}		

Let $\S$ be the super-operator that acts as follows: 
\begin{enumerate}
\item First, based on the transcript of messages from Eve to the provers Alice and Bob, it extracts into a new classical register the subset~$S \subset [\nstate]$ consisting of those blocks of EPR states that Eve has asked both provers to touch.  
\item Then it reorders those of Alice's qubits that are supposed to be entangled with Bob so that the qubits in the blocks of~$S$ come first.  
\item Finally, it traces out Bob's Hilbert space, the messages to and from Bob for rounds outside of~$S$, Alice's extra space~$\H_A'$ and all of Alice's qubits that are supposed to be entangled with Bob for blocks outside of~$S$.  
\end{enumerate}
Let $\S'$ be the super-operator that has the same second and third steps, but that chooses $S$ uniformly at random in the first step.  Continuing from Eq.~\eqnref{e:adaptivecomputationprocesstomography}, we have 
\begin{equation*}\begin{split}
\S \Aad \B \cP (\rho) 
&\approx \S \Aadhat \cP \B (\rho) \\
&= \AadhatS \cP \S' \B (\rho)
 \enspace ,
\end{split}\end{equation*}
where $\AadhatS$ is the ideal adaptive super-operator for Alice with the subset~$S$ fixed.  Here, the equality $\S \Aadhat \B \cP = \AadhatS \S' \B \cP$ follows because the ideal super-operator $\Aadhat$ has no support on Alice's qubits that are supposed to be entangled with Bob for blocks outside of~$S$.  

By \thmref{t:statetomographydishonestAlicesuperoperator} and a Markov inequality, with probability at least $1 - O(\nstate^{-1/96})$ over $K$ and the first $K-1$ sets, 
\begin{equation*}
\bigtrnorm{ \S' \B (\rho) - \S' \hat \B (\hat \rho) } \leq O(\nstate^{-1/384}) + 2 \cdot O(\nstate^{-1/96}) = O(\nstate^{-1/384}) 
 \enspace ,
\end{equation*}
where the term $2 \cdot O(\nstate^{-1/96})$ accounts for the trace distance for bad choices of~$S$.  

Putting together our calculations, we obtain that with probability at least $1 - (6 \delta + \nset^{-\alpha/8}) - O(n^{-\alpha/16}) - O(\nstate^{-1/96}) = 1 - O(\nstate^{-1/96})$ over $K$ and the first $K-1$ sets, 
\begin{equation} \label{e:qmipmipfinalsuperoperatorapproximation}
\S \Aad \B \cP (\rho) \approx \S \Aadhat \hat \B \cP (\hat \rho)
\end{equation}
up to an error in trace distance at most $O(n^{1 - \alpha / (64 \kappaEPR)}) + O(\nstate^{-1/384}) = O(n^{-\alpha/(768 \kappaEPR)})$.  In particular, since the verifier's acceptance predicate $\V$ does not depend on the messages to and from Bob for rounds outside of~$S$, in these cases we have 
\begin{equation}\begin{split}
\Pr[ \text{$\V$ accepts $\Aad \B \cP (\rho)$} ]
- \frac12 \bigtrnorm{ \S' \B (\rho) - \S' \hat \B (\hat \rho) }
&\leq \Pr[ \text{$\V$ accepts $\Aadhat \hat B \cP (\hat \rho)$} ] \\
&\leq s
 \enspace ,
\end{split}\end{equation}
the soundness parameter of the original protocol.  

Thus the probability that the verifier accepts the computation sub-protocol is at most $\big( s + O(n^{-\alpha/(768 \kappaEPR)}) \big) + O(\nstate^{-1/96}) \cdot 1 = s + O(n^{-\alpha/(768 \kappaEPR)})$.  The $O(\nstate^{-1/96}) \cdot 1$ term is the contribution for those~$K$ and transcripts for the first $K - 1$ sets for which we cannot make the approximation of Eq.~\eqnref{e:qmipmipfinalsuperoperatorapproximation}; in such cases, we can only upper bound $\Pr[ \text{$\V$ accepts $\Aad \B \cP (\rho)$} ]$ by one.  For~$\alpha$ and~$n$ at least sufficiently large constants, $s + O(n^{-\alpha/(768 \kappaEPR)}) < \frac12 (c + s)$.  This is a contradiction.  Therefore, on inputs $x \notin L$, Eve must accept with probability less than $1 - \big(1 - \frac12(c + s)\big) \delta$.  Together with Eq.~\eqnref{e:qmipmipcompleteness}, this establishes an inverse polynomial completeness-soundness gap for our transformed $\text{MIP}^*$ protocol.  Sequential repetition can be used to amplify the gap.  

For the claimed blindness property when $k = 0$, observe from Eq.~\eqnref{e:qmipmipadaptivebyeitherprover} that each prover's view of the protocol consists of random messages, drawn from a distribution that depends only on the size of the circuit~$\circuit$.  
\end{proof}

The above proof gives blind, verified quantum computation for decision problems.  The same arguments extend beyond decision problems, though, e.g., to relation and sampling problems.  In general, the verifier can referee many sequential protocols, each time picking a random one of the three testing sub-protocols, in order to gain sufficient statistical confidence that the provers are playing nearly honestly.  At a random position, the verifier can insert the computation sub-protocol.  The analysis is then the same as above.  In particular, the arguments leading to the approximation of Eq.~\eqnref{e:qmipmipfinalsuperoperatorapproximation} still hold.  

It may be that in fact $\QMIP[\text{$k$ provers}] = \MIP^*[\text{$k$ provers}]$, without the need to add two additional provers.  Our proof technique is useless for the $k = 1$ case, which is already known: $\QIP = \IP$.  However, it seems likely that the technique should work for the case $k \geq 2$, with minor technical changes.  The idea is to identify Alice with~$P_1$ and Bob with~$P_2$.  We have not investigated it carefully, though.  In the next section, we will present several other interesting open problems.

\ifx\compilefullpaper\undefined
\bibliographystyle{alpha-eprint}
\bibliography{q}

\end{document}
\fi

\section{Open problems}

By characterizing the device strategies that can win many successive CHSH games, we have shown how a fully classical party can direct the actions of two untrusted quantum devices.  The simplest case is device-independent quantum key distribution, free of the independence assumptions needed in previous analyses.  Three main open problems are to extend the results to other non-local quantum games beyond the CHSH game, to improve the efficiency of our schemes and their analysis---of interest both for developing practical applications and for obtaining a better theoretical understanding of the underlying physics---and to find further cryptographic applications.  

\begin{enumerate}
\item
The CHSH game is ``rigid" in the sense that any strategy that achieves the optimum success probability can be related by local isometries to the ideal strategy of \tabref{f:optimalCHSHstrategy}, and nearly optimal strategies can be nearly related to the ideal strategy (\lemref{t:eprlemma}).  What other non-local quantum games satisfy this property?  \lemref{t:generalizedeprlemma} in \appref{s:generalizedeprlemma} gives one example, but it is based on the CHSH game and its analysis inefficiently goes through \lemref{t:eprlemma}.  For games whose analysis cannot be reduced to studying pairs of two-outcome measurements, Jordan's Lemma (\lemref{t:jordanslemma}) will not apply, and new techniques will be needed for analytically controlling the provers' strategies.  Can the rigidity of an XOR game be reduced to rigidity properties of the Tsirelson semi-definite program?  

If \lemref{t:eprlemma} extends to show the rigidity of a certain game, then it is likely that the sequential repetition theorem, \thmref{t:sequentialCHSHgames}, also generalizes.  The main technical tricks for proving \thmref{t:sequentialCHSHgames} involve shifting one prover's operations to the other prover's qubits (\secref{s:howtoplayforBob}).  This allows the derivation of a tensor-product structure within a prover's Hilbert space based on the tensor-product structure between the provers' Hilbert spaces.  These tricks should apply to other non-local games based on maximally entangled shared states.  

\item 
Although our schemes have polynomial overheads and are therefore efficient in principle, the exponents are too large for any practical applications.  The DIQKD key rate tends to zero, instead of a positive constant.  Significant improvements are possible by tightening the analysis, which we have not at all optimized.  However, new proof techniques are probably required to achieve a practical overhead.  One approach might be to use ideas from fault-tolerant quantum computing~\cite{NielsenChuang00}.  Fault tolerance can reduce the overhead if it allows for proving the same soundness guarantees from less statistical data.  Just as important, fault-tolerance ideas might allow for tolerating higher noise rates, even constant noise rates.  We would like our schemes to work even if the honest provers are somewhat faulty, as would be any real devices.  In principle, this is not a problem for blind, verified computation, since the provers can work on top of a quantum error-correcting code and the verifier can help distill any faulty initial entanglement.  Of course, a general-purpose quantum computer, capable of manipulating quantum error-correcting codes, is well beyond current technology.  In contrast, quantum key distribution setups have been deployed and are commercially available~\cite{ScaraniBechmannPasquinucciCerfDusekLutkenhausPeev08qkdreview}.  More sophisticated proofs might allow for device-independent QKD with today's experimental technology.  

Another aspect of efficiency is the number of rounds of communication.  Can \thmref{t:sequentialCHSHgames} be generalized to hold for games played in parallel instead of in sequence?  A parallel-repetition theorem would allow for enforcing the assumption that the provers do not communicate based on space-like separation of the provers.  One starting point might be to use the parallel repetition analysis techniques of~\cite{KempeVidick10parallelrepetition}.  Reducing the round complexity of the blind, verified quantum computation protocol might be more difficult.  Computation by teleportation, at least, inherently requires the coordination of adaptive corrections.  

\item 
The CHSH game rigidity theorems provide the foundation for device-independent quantum key distribution, for blind, verified quantum computation and for the equality $\QMIP = \MIP^*$.  The theorems do not have a classical analog and allow for drastically reduced security assumptions from what is possible classically, in particular the elimination of any computational assumptions.  An important question is whether other cryptographic primitives or protocols, beyond what is possible classically, can also be based on CHSH game rigidity and state and process tomography.  For example, Silman et al.\ have given a device-independent, imperfect bit-commitment protocol based on the Greenberger-Horne-Zeilinger~(GHZ) game~\cite{SilmanChaillouxAharonKerenidisPironioMassar11DIcrypto}.  
\end{enumerate}

\subsection*{Acknowledgements}

We thank Edgar Bering, Anne Broadbent, Andr{\' e} Chailloux, Matthias Christandl, Roger Colbeck, Tsuyoshi Ito, Robert K{\"o}nig, Matthew McKague, Vidya Madhavan, Renato Renner, Shivaji Sondhi and Thomas Vidick for helpful conversations.  Part of the work conducted while F.U.~was at UC Berkeley, and B.R.~at the Institute for Quantum Computing, University of Waterloo.  B.R.~acknowledges support from NSERC, ARO-DTO and Mitacs.  U.V.~acknowledges support from NSF grant CCF-0905626 and Templeton grant~21674.

\appendix

\ifx\compilefullpaper\undefined  
\documentclass[11pt]{article}

\begin{document}
\fi

\section{Characterization of nearly optimal strategies for an extended CHSH~game} \label{s:generalizedeprlemma}

The CHSH game in \lemref{t:eprlemma} establishes a shared EPR state between the provers Alice and Bob, as well as $X$ and $Z$ operators for Alice and operators $(X \pm Z) / \sqrt 2$ for Bob.  In this section, we extend the CHSH game with more questions in order that the rigidly determined ideal strategy should use Pauli $Y$ operators, in addition to the $X$ and~$Z$ operators, acting on the shared EPR state.  Our extension follows along the same lines as McKague and Mosca's extension of the Mayers-Yao test~\cite{McKagueMosca10chshydirection}.  However, it will not be possible to fully determine the $Y$ operator, since the provers can coordinate to use $-Y$ each instead of $+Y$ with no detectable consequences, and can even do so coherently.  A reflection of the Bloch sphere about the $xz$ plane is a \emph{non-unitary} symmetry, that cannot simply be absorbed into a change of basis.  It corresponds to taking the complex conjugate of the coefficients of the state in the computational basis ($Z$ eigenbasis).  Similar to \lemref{t:eprlemma}, we characterize, as far as possible, $\epsilon$-structured strategies for the extended CHSH game.  

\begin{figure}
\centering
\includegraphics[scale=.3928]{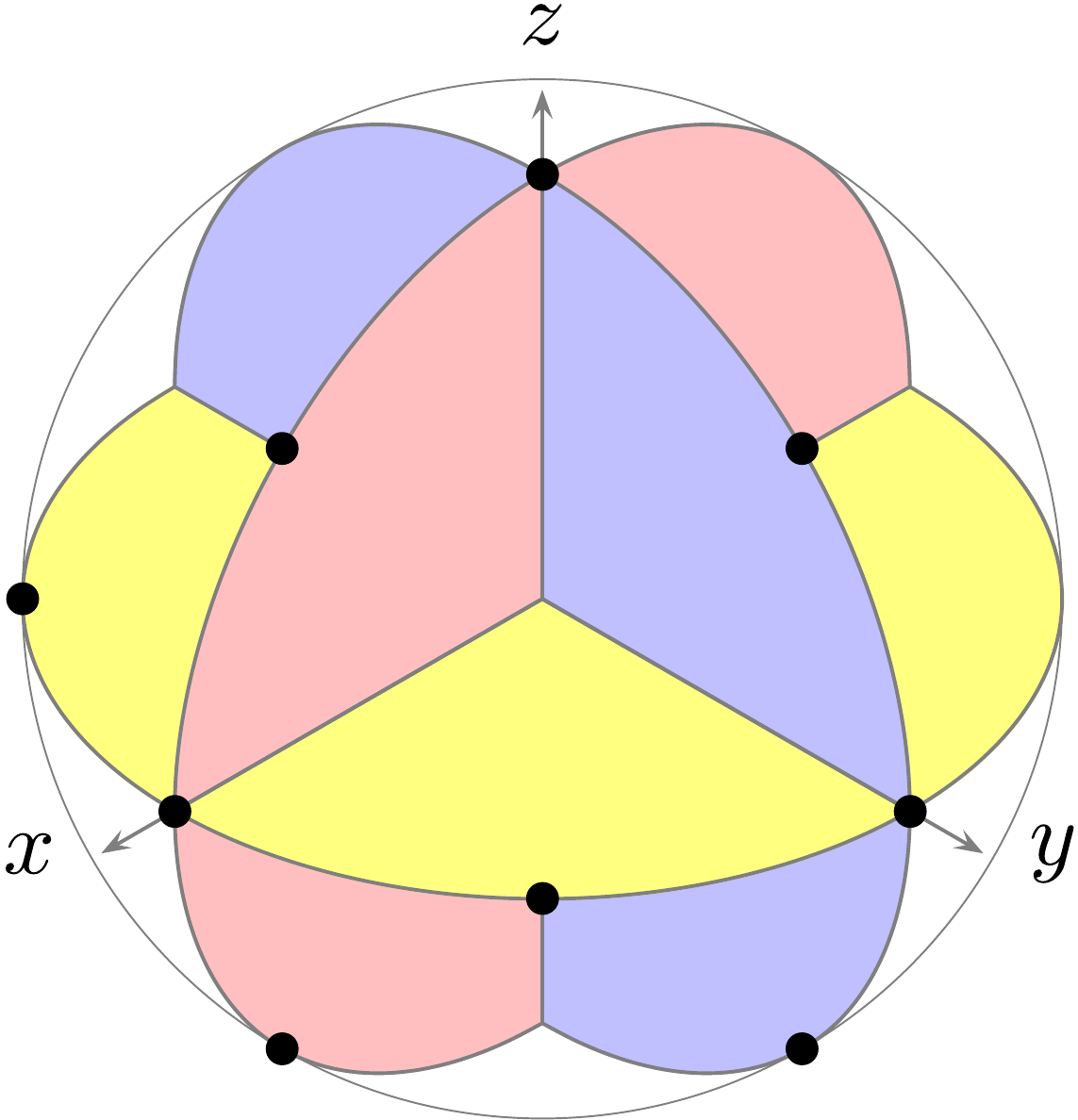}
\caption{An extended CHSH game with nine measurement directions, indicated here on the Bloch sphere, has embedded within it CHSH games in the $xz$, $xy$ and $yz$ planes.  In the cross-section for each of these planes are the four measurement directions of \tabref{f:optimalCHSHstrategy}.} \label{f:extendedchshgame}
\end{figure}

\begin{definition} \label{t:generalizedCHSHgamedef}
An \emph{extended CHSH game} involves three parties: a classical randomized verifier or referee, Eve, and two quantum provers, Alice and Bob.  Alice and Bob are not allowed to communicate with each other.  They share two registers of an arbitrary pure quantum state $\ket \psi \in \H_A \otimes \H_B \otimes \H_C$, where $\H_A$ and $\H_B$ are the Hilbert spaces of Alice and Bob, respectively, and $\H_C$ is an inaccessible third Hilbert space.  

In the game, Eve twice and independently picks a uniformly random direction from the set $\{ (1,0,0), (0,1,0), (0,0,1), \frac{1}{\sqrt 2} (1, 1, 0), \frac{1}{\sqrt 2} (1, -1, 0), \frac{1}{\sqrt 2} (1, 0, 1), \frac{1}{\sqrt 2} (1, 0, -1), \frac{1}{\sqrt 2} (0, 1, 1), \frac{1}{\sqrt 2} (0, 1, -1) \}$, shown in \figref{f:extendedchshgame}.  She sends the first direction, $\vec a$, to Alice, and the second direction, $\vec b$, to Bob.  Alice measures her portion of~$\ket \psi$ using a two-outcome projective measurement $\{\Pi_{\vec a}^0, \Pi_{\vec a}^1\}$, and returns the result, $x \in \{0,1\}$, to Eve.  Bob similarly returns to Eve $y \in \{0,1\}$, the result of the projective measurement $\{\Pi_{\vec b}'^0, \Pi_{\vec b}'^1\}$.  Therefore, for questions $\vec a, \vec b$, the probability of responses $a, b$ is given by 
\begin{equation}
\tilde p_{xy|\vec a\vec b} = \big\langle \Pi_{\vec a}^x \otimes \Pi_{\vec b}'^y \otimes \identity_C \big\rangle_{\ket \psi}
 \enspace .
\end{equation}

In the \emph{ideal strategy}, Alice and Bob return the result of measuring their halves of a shared EPR state $\frac{1}{\sqrt 2}(\ket{00} + \ket{11})$, along the input direction, thought of as an axis for the Bloch sphere.  That is, on input $\vec a$, Alice measures with the projections $\frac{1}{2} (I + \vec a \cdot (X, Y, Z))$ and $\frac{1}{2} (I - \vec a \cdot (X, Y, Z))$, and returns $x = 0$ on the first outcome as $x = 1$ on the second outcome.  Bob follows the same ideal strategy.  Thus the probability of outcomes $x, y$ on questions $\vec a, \vec b$ is 
\begin{equation}
p_{xy|\vec a\vec b} = \Big\langle \frac{1}{2}\big(I + (-1)^x \, \vec a \cdot (X,Y,Z)\big) \otimes \frac{1}{2}\big(I + (-1)^y \, \vec b \cdot (X,Y,Z)\big) \Big\rangle_{\frac{1}{\sqrt 2}(\ket{00} + \ket{11})}
 \enspace .
\end{equation}

For $\epsilon \geq 0$, call a strategy for the extended CHSH game $\epsilon$-structured if for all $\vec a, \vec b, x, y$, 
\begin{equation}
\abs{\tilde p_{xy|\vec a\vec b} - p_{xy|\vec a\vec b}} \leq \epsilon
 \enspace .
\end{equation}
\end{definition}

To analyze the extended CHSH game, our approach is to apply \lemref{t:eprlemma} repeatedly.  Observe that the extended CHSH game contains within it six CHSH games, i.e., sets of questions for which the extended CHSH ideal strategy is consistent with playing a CHSH game optimally.  For example, the questions $(\vec a, \vec b) \in \{(1,0,0), (0,0,1)\} \times \{\frac{1}{\sqrt 2}(1,0,1), \frac{1}{\sqrt 2}(1,0,-1)\}$ form one such game, as do questions $(\vec a, \vec b) \in \{\frac{1}{\sqrt 2}(1,0,1), \frac{1}{\sqrt 2}(1,0,-1)\} \times \{(1,0,0), (0,0,1)\}$; there are two CHSH games along each plane $xz$, $xy$ and $yz$.  In an $\epsilon$-structured strategy for the extended CHSH game, each of these sub-games has correlation value at least $2 \sqrt 2 - 16 \epsilon$, when questions $\vec a, \vec b$ are appropriately relabeled by bits.  (This follows by the definition of the correlation value in Eq.~\eqnref{e:chshcorrelationvaluedef}: $4 (2 \Pr[x \oplus y = a b] - 1) = 2 \sum_{\alpha, \beta \in \{0,1\}} \Pr[x \oplus y = a b \,\vert\, a = \alpha, b = \beta] - 4$.)  \lemref{t:eprlemma} therefore applies to each sub-game, and we will then stitch together the conclusions.  We show: 

\begin{lemma}[Rigidity for the extended CHSH game] \label{t:generalizedeprlemma}
Consider a extended CHSH game, with the notation established in \defref{t:generalizedCHSHgamedef}.  Let $\epsilon > 0$ and consider an $\epsilon$-structured strategy.  Then there are extensions of the Hilbert spaces $\H_A, \H_B$, and extensions of the reflections $\bar Z, \bar X, \bar Z', \bar X'$ by a direct sum with other reflections, so that the following properties hold: 
\begin{itemize}
\item
Alice's space is isomorphic to $\C^2 \otimes \hat \H_A$, with $\bar Z = Z \otimes \identity$, $\bignorm{(\bar X - X \otimes \identity)_A \ket \psi} = O(\sqrt \epsilon)$, and for some reflection $\Delta \in \L(\hat \H_A)$, $\bignorm{(\bar Y - Y \otimes \Delta)_A \ket \psi} = O(\epsilon^{1/12})$.  
\item 
Bob's space is isomorphic to $\C^2 \otimes \hat \H_B$, with $\max\{ \bignorm{(\bar Z' - Z \otimes \identity)_B \ket \psi}, \bignorm{(\bar X' - X \otimes \identity)_B \ket \psi} \} = O(\epsilon^{1/4})$, and for some reflection $\Delta' \in \L(\hat \H_B)$, $\bignorm{(\bar Y' - Y \otimes \Delta')_B \ket \psi} = O(\epsilon^{1/144})$.  
\item 
Finally, letting $\ket{\psi^*} = \frac{1}{\sqrt 2}(\ket{00} + \ket{11})$, there exists a unit vector $\ket{\psi^\times} \in \hat \H_A \otimes \hat \H_B \otimes \H_C$ with $\norm{\ket \psi - \ket{\psi^*} \otimes \ket{\psi^\times}} = O(\sqrt \epsilon)$ and $\bra{\psi^\times} \Delta \otimes \Delta' \ket{\psi^\times} \geq 1 - O(\epsilon^{1/72})$.  
\end{itemize}
The constants hidden by the big-$O$ notation are universal constants, independent of the game strategy.  
\end{lemma}

In the proof we will use: 

\begin{lemma} \label{t:converttoreflection}
Let $U$ be a unitary and $H$ a Hermitian operator, both acting on state $\ket \phi$, with $\norm H \leq 1$ and $\norm{(U - H) \ket \phi} \leq \epsilon$.  Then there is a reflection~$\Delta$ such that $\norm{(U - \Delta) \ket \phi} \leq \epsilon + 2^{4/3} \epsilon^{1/3}$.  Furthermore, if $H = P \otimes H'$, where $P$ has eigenvalues $\pm 1$, then we may take $\Delta = P \otimes \Delta'$ for a reflection $\Delta'$.  
\end{lemma}

\begin{proof}
This is essentially a Markov inequality.  Let $\delta \in (0,1)$, a parameter that we will optimize shortly.  For $c \in \{0,1\}$, let $\Delta_c$ be the projection onto the span of the eigenvectors of $H$ with eigenvalue within $\delta$ of $(-1)^c$.  Let $\Delta = \Delta_0 - \Delta_1$, so $\abs \Delta = \Delta_0 + \Delta_1$.  Then, 
\begin{align*}
\norm{(U - \Delta) \ket \phi} 
&\leq \norm{(U - H) \ket \phi} + \norm{(H - \Delta) \ket \phi} \\
&\leq \norm{(U - H) \ket \phi} + \norm{H \abs \Delta - \Delta} + (1 - \delta) \norm{(\identity - \abs \Delta) \ket \phi}
 \enspace .
\end{align*}
Here the first term on the right is at most $\epsilon$ and the second term is at most~$\delta$.  To bound the final term, use $1 - \epsilon \leq \norm{U \ket \psi} - \norm{(U-H) \ket \phi} \leq \norm{H \ket \phi}$, and $\norm{H \ket \phi}^2 \leq \norm{\abs \Delta \ket \phi}^2 + (1 - \delta)^2 \norm{(\identity - \abs \Delta) \ket \phi}^2 = 1 - \delta (2 - \delta) \norm{(\identity - \abs \Delta) \ket \phi}^2$.  Thus, $\norm{(\identity - \abs \Delta) \ket \phi} \leq \sqrt{2 \epsilon / \delta}$.  Set $\delta = (2 \epsilon)^{1/3}$ to conclude $\norm{(U - \Delta) \ket \phi} \leq \epsilon + 2^{4/3} \epsilon^{1/3}$.  (For $\epsilon < 1/2$, $\delta < 1$, and for $\epsilon \geq 1/2$, the bound is~trivial.)  
\end{proof}

\begin{proof}[Proof of \lemref{t:generalizedeprlemma}]
Let us begin by establishing some notation.  
For $\vec r \in \R^3$, let $R(\vec r) = \vec r \cdot (X, Y, Z) = r_1 X + r_2 Y + r_3 Z$.  
Let $\bar R(\vec a) = \Pi_{\vec a}^0 - \Pi_{\vec a}^1$ and $\bar R'(\vec b) = \Pi_{\vec b}'^0 - \Pi_{\vec b}'^1$.  
Let $\vec v_x = (1,0,0)$, $\vec v_y = (0,1,0)$, $\vec v_z = (0,0,1)$, $\vec v_{\pm xy} = \frac{1}{\sqrt 2}(1,\pm1,0)$, $\vec v_{\pm xz} = \frac{1}{\sqrt 2}(1,0,\pm1)$, $\vec v_{\pm yz} = \frac{1}{\sqrt 2}(0,1,\pm1)$.  
For $\alpha \in \{x, y, z, \pm xy, \pm xz, \pm yz\}$, let $R_\alpha = R(\vec v_\alpha)$, $\bar R_\alpha = \bar R(\vec v_\alpha)$ and $\bar R'_\alpha = \bar R'(\vec v_\alpha)$.  For example, $R_x = X$.  
For a vector~$\ket \phi$, define the semi-norm $\norm{M}_\phi = \norm{M \ket \phi}$.  

The proof has two parts.  First we consider only the questions $\vec a \in \{ \vec v_x, \vec v_y, \vec v_z \}$ and $\vec b \in \{ \vec v_{\pm xz}, \vec v_{\pm yz}, \vec v_{\pm xy} \}$, i.e., question pairs in which Alice is asked to measure along a coordinate axis of the Bloch sphere and Bob is asked to measure in a direction between two coordinate axes.  In particular, we consider three sets of questions: 
\begin{enumerate}
\item 
$(\vec a, \vec b) \in \{\vec v_x, \vec v_z\} \times \{\vec v_{+xz}, \vec v_{-xz}\}$.  As $\bra{\psi^*} (R_x \otimes R_{xz} + R_x \otimes R_{-xz} + R_z \otimes R_{xz} - R_z \otimes R_{-xz}) \ket{\psi^*} = 2 \sqrt 2$, these questions form a CHSH sub-game.  
\item
$(\vec a, \vec b) \in \{\vec v_y, \vec v_z\} \times \{\vec v_{-yz}, \vec v_{+yz}\}$.  Since $\bra{\psi^*} (R_y \otimes R_{-yz} + R_y \otimes R_{yz} + R_z \otimes R_{-yz} - R_z \otimes R_{yz}) \ket{\psi^*} = -2 \sqrt 2$, these questions form a CHSH sub-game if Eve complements Bob's answers.  
\item
$(\vec a, \vec b) \in \{\vec v_y, \vec v_x\} \times \{\vec v_{-xy}, \vec v_{+xy}\}$.  Since $\bra{\psi^*} (-R_y \otimes R_{-xy} + R_y \otimes R_{xy} + R_x \otimes R_{-xy} + R_x \otimes R_{xy}) \ket{\psi^*} = 2 \sqrt 2$, these questions form a CHSH sub-game if Eve complements Alice's answer to question~$\vec v_y$ and complements Bob's answer to question $\vec v_{xy}$.  
\end{enumerate}
By applying \lemref{t:eprlemma} to the first CHSH sub-game above, we establish a shared EPR state~$\ket{\psi^*}$ and characterize Alice's operators $\bar R_z$ and $\bar R_x$.  By applying \lemref{t:eprlemma} to the second and third CHSH sub-games above, we come at Alice's $\bar R_y$ operator from two directions in the Bloch sphere, in order, essentially, to triangulate it.  

In the second part of the proof, we tie in Bob's on-axis reflections.  For this part of the proof, we use only that $\bra{\psi^*} R_x \otimes R_x \ket{\psi^*} = \bra{\psi^*} R_z \otimes R_z \ket{\psi^*} = - \bra{\psi^*} R_y \otimes R_y \ket{\psi^*} = 1$, i.e., that $\ket{\psi^*}$ is a certain stabilizer state.  

\medskip

Consider the questions $(\vec a, \vec b) \in \{\vec v_x, \vec v_z\} \times \{\vec v_{+xz}, \vec v_{-xz}\}$.  As these questions form a CHSH sub-game, we can apply \lemref{t:eprlemma} to obtain a decomposition $\H_A = \C^2 \otimes \hat \H_A$, $\H_B = \C^2 \otimes \hat \H_B$ such that $\norm{\ket \psi - \ket{\psi^*} \otimes \ket{\psi^\times}} = O(\sqrt \epsilon)$, $\bar R_z = R_z \otimes \identity$ and $\norm{(\bar R_x - R_x \otimes \identity)_A}_\psi = O(\sqrt \epsilon)$.  Also, $\bar R'_{xz} = R_{xz} \otimes \identity$ and $\norm{(\bar R'_{-xz} - R'_{-xz} \otimes \identity)_B}_\psi = O(\sqrt \epsilon)$, although we will not use this.  

Consider next the questions $(\vec a, \vec b) \in \{\vec v_y, \vec v_z\} \times \{\vec v_{-yz}, \vec v_{+yz}\}$.  Applying \lemref{t:eprlemma}, we obtain that there exists a unitary~$\bar U \in \L(\H_A)$ such that $\bar U \bar R_z \bar U^\dagger = R_z \otimes \identity$ and $\norm{ (\bar R_y - \bar U^\dagger R_x \otimes \identity \bar U)_A }_\psi = O(\sqrt \epsilon)$.  Since $\bar R_z = R_z \otimes \identity$, it follows that $\bar U = \ketbra 0 0 \otimes U_0 + \ketbra 1 1 \otimes U_1$ for some unitaries $U_0, U_1 \in \L(\hat \H_A)$.  Let $U = U_0^\dagger U_1$.  Thus $\bar U^\dagger (R_x \otimes \identity) \bar U = \ketbra 0 1 \otimes U + \ketbra 1 0 \otimes U^\dagger$.  

Last, consider the questions $(\vec a, \vec b) \in \{\vec v_y, \vec v_x\} \times \{\vec v_{-xy}, \vec v_{+xy}\}$.  These questions form a CHSH sub-game if Eve complements Alice's answer to question $\vec v_y$ and complements Bob's answer to question $\vec v_{xy}$.  However, we do not apply \lemref{t:eprlemma} to this sub-game directly.  Instead, \emph{modify} Alice's strategy by replacing $\bar R_x$ with $R_x \otimes \identity$.  Since $\norm{(\bar R_x - R_x \otimes \identity)_A}_\psi = O(\sqrt \epsilon)$, the correlation value of the modified game decreases at most from $2 \sqrt 2 - 16 \epsilon$ to $2 \sqrt 2 - O(\sqrt \epsilon)$.  Now applying \lemref{t:eprlemma} to the modified game, we obtain that there is a unitary~$\bar V$ such that $\bar V R_x \otimes \identity \bar V^\dagger = R_z \otimes \identity$ and $\norm{ (\bar R_y + \bar V^\dagger R_x \otimes \identity \bar V)_A }_\psi = O(\epsilon^{1/4})$.  Since $R_x = X = \ketbra + + - \ketbra - -$, where $\ket \pm = \frac{1}{\sqrt 2}(\ket 0 \pm \ket 1)$, the first equation implies $\bar V = \ketbra 0 + \otimes V_0 + \ketbra 1 - \otimes V_1$ for unitaries $V_0$ and $V_1$.  Letting $V = V_0^\dagger V_1$, therefore, $-\bar V^\dagger R_x \otimes \identity \bar V = - \ketbra + - \otimes V - \ketbra - + \otimes V^\dagger$.  

Combining this with our characterization of $\bar Y$ from the second CHSH sub-game implies: 

\begin{claim} \label{t:hermitianoperatornotprojection}
For $\epsilon < 10^{-10}$, there is a Hermitian operator $S \in \L(\hat \H_A)$ with $\norm S \leq 1$, namely $S = i (U - U^\dagger)/2$, such that $\norm{(\bar R_y - R_y \otimes S)_A}_\psi = O(\epsilon^{1/4})$.  
\end{claim}

\begin{proof}
We have 
\begin{equation*}
\bignorm{ \ketbra 0 1 \otimes U + \ketbra 1 0 \otimes U^\dagger + \ketbra + - \otimes V + \ketbra - + \otimes V^\dagger  }_\psi = O(\epsilon^{1/4}) 
 \enspace .
\end{equation*}
Since $\norm{\ket \psi - \ket{\psi^*} \ket{\psi^\times}} = O(\sqrt \epsilon)$, therefore 
\begin{align*}
\bignorm{ \ketbra 0 1 \otimes U + \ketbra 1 0 \otimes U^\dagger + \ketbra + - \otimes V + \ketbra - + \otimes V^\dagger }_{\ket{\psi^*} \ket{\psi^\times}} = O(\epsilon^{1/4})
\end{align*}
Now substitute $\ket{\psi^*} = \frac{1}{\sqrt 2}(\ket{{+}{+}} + \ket{{-}{-}})$ to obtain 
\begin{equation*}\begin{split}
O(\epsilon^{1/4}) &= \frac12 \Biggnorm{ \begin{split} (\ket{{+}{+}} - \ket{{-}{-}})_{AB} (U+U^\dagger)_A &+ \ket{{-}{+}}_{AB} (U - U^\dagger + 2 V^\dagger) \\&+ \ket{{+}{-}}_{AB} (-U + U^\dagger + 2 V) \end{split} }_{\psi^\times} \\
&\geq \tfrac1{\sqrt 2} \norm{ (U+U^\dagger)_A }_{\psi^\times}
 \enspace .
\end{split}\end{equation*}
This implies our characterization of $\bar R_y$: 
\begin{equation*}\begin{split}
\bignorm{ \bar R_y - R_y \otimes S }_\psi
&\leq \bignorm{\bar R_y - \bar U^\dagger X \otimes \identity \bar U}_\psi + \bignorm{ \bar U^\dagger X \otimes \identity \bar U - Y \otimes S }_\psi \\
&\leq \bignorm{\bar R_y - \bar U^\dagger X \otimes \identity \bar U}_\psi + 2 \bignorm{\ket \psi - \ket{\psi^*}\ket{\psi^\times}} + \bignorm{ \bar U^\dagger X \otimes \identity \bar U - Y \otimes S }_{\ket{\psi^*}\ket{\psi^\times}} \\
&= \bignorm{\bar R_y - \bar U^\dagger X \otimes \identity \bar U}_\psi + 2 \bignorm{\ket \psi - \ket{\psi^*}\ket{\psi^\times}} + \Bignorm{ X \otimes \frac{U + U^\dagger}{2} }_{\ket{\psi^*}\ket{\psi^\times}} \\
&= O(\epsilon^{1/4})
 \enspace . \qedhere
\end{split}\end{equation*}
\end{proof}

\lemref{t:converttoreflection} gives a reflection $\Delta$ so $\norm{(\bar R_y - R_y \otimes \Delta)_A}_\psi = O(\epsilon^{1/12})$.  

\smallskip

In the second part of the proof, we will consider Bob's on-axis reflections.  In particular, consider the questions $(\vec a, \vec b) \in \{ (\vec v_x, \vec v_x), (\vec v_z, \vec v_z), (\vec v_y, \vec v_y) \}$.  Note that in the ideal strategy on an EPR state, $p_{00 \vert \vec v_x \vec v_x} = p_{11 \vert \vec v_x \vec v_x} = \frac12$, $p_{00 \vert \vec v_z \vec v_z} = p_{11 \vert \vec v_z \vec v_z} = \frac12$ and $p_{01 \vert \vec v_y \vec v_y} = p_{10 \vert \vec v_y \vec v_y} = \frac12$.  We use these identities to characterize $\bar R'_x$, $\bar R'_z$ and $\bar R'_y$.  Observe that, since the provers' strategy is $\epsilon$-structured, 
\begin{equation*}
\bra \psi \bar R_x \otimes \bar R'_x \ket \psi
= \tilde p_{00|\vec v_x \vec v_x} + \tilde p_{11|\vec v_x \vec v_x} - \tilde p_{01|\vec v_x \vec v_x} - \tilde p_{10|\vec v_x \vec v_x} \geq 1 - 4 \epsilon
 \enspace . 
\end{equation*}

\begin{claim} \label{t:complexmarkov}
For complex numbers $\alpha, \beta$ with $\abs \alpha, \abs \beta \leq 1$ and $\bigabs{\frac12 (\alpha + \beta) - 1} \leq \delta \leq \frac14$, necessarily $\max\{ \abs{\alpha-1}, \abs{\beta-1} \} \leq \sqrt{3 \delta}$.  
\end{claim}

We have, using $\ket{\psi^*} = \frac{1}{\sqrt 2}(\ket{00} + \ket{11})$ and successive triangle inequalities, 
\begin{equation*}\begin{split}
\bigabs{ \tfrac12 \big( \bra 0 \bra{\psi^\times} \bar R'_x \ket 1 \ket{\psi^\times} + \bra 1 \bra{\psi^\times} \bar R'_x \ket 0 \ket{\psi^\times} \big) -1 }
&= \abs{ \bra{\psi^*} \bra{\psi^\times} R_x \otimes \bar R'_x \ket{\psi^*}\ket{\psi^\times} - 1 } \\
&\leq 2 \norm{\ket \psi - \ket{\psi^*} \ket{\psi^\times}} + \abs{ \bra \psi R_x \otimes \bar R'_x \ket \psi - 1 } \\
&= O(\sqrt \epsilon)
 \enspace .
\end{split}\end{equation*}
Applying \claimref{t:complexmarkov}, we find $\max\{ \abs{\bra 0 \bra{\psi^\times} \bar R'_x \ket 1 \ket{\psi^\times} - 1}, \abs{\bra 1 \bra{\psi^\times} \bar R'_x \ket 0 \ket{\psi^\times} - 1} \} = O(\epsilon^{1/4})$.  Therefore, $\max\{ \norm{\ket 0 \ket{\psi^\times} - \bar R'_x \ket 1 \ket{\psi^\times}}{}^2, \norm{\ket 1 \ket{\psi^\times} - \bar R'_x \ket 0 \ket{\psi^\times}}{}^2 \} = O(\epsilon^{1/4})$, and hence, 
\begin{equation*}\begin{split}
\bignorm{ (R_x \otimes \identity - \bar R'_x)_B \ket{\psi^*} \ket{\psi^\times} }^2 
&= \frac{1}{2} \Big( \bignorm{ \ket 1 \ket{\psi^\times} - \bar R'_x \ket 0 \ket{\psi^\times} }^2 + \bignorm{ \ket 0 \ket{\psi^\times} - \bar R'_x \ket 1 \ket{\psi^\times} }^2 \Big) 
= O(\epsilon^{1/4})
 \enspace .
\end{split}\end{equation*}
It follows that $\norm{ (\bar R'_x - R_x \otimes \identity)_B }_\psi = O(\epsilon^{1/8})$.  For $\bar R'_z$, a similar argument implies $\bignorm{(\bar R'_z - R_z \otimes \identity)_B}_\psi = O(\epsilon^{1/8})$.  

Finally, for $\bar R'_y$, we have $\bra \psi \bar R_y \otimes \bar R'_y \ket \psi \leq -1 + 4 \epsilon$, and therefore 
\begin{align*}
\bigabs{ \bra{\psi^*}\bra{\psi^\times} (R_y \otimes \Delta)_A \otimes \bar R'_y \ket{\psi^*}\ket{\psi^\times} + 1 } 
&\leq 2 \norm{\ket{\psi^*}\ket{\psi^\times} - \ket \psi} + \norm{(R_y \otimes \Delta - \bar R_y)_A}_\psi \\ &\quad + \bigabs{\bra \psi \bar R_y \otimes \bar R'_y \ket \psi + 1} \\
&= O(\epsilon^{1/12})
 \enspace .
\end{align*}
The left-hand side of this inequality is $\bigabs{\tfrac12 (\alpha+\beta) - 1}$, where $\alpha = i \bra 0 \bra{\psi^\times} \Delta \otimes \bar R'_y \ket 1 \ket{\psi^\times}$ and $\beta = -i \bra 1 \bra{\psi^\times} \Delta \otimes \bar R'_y \ket 0 \ket{\psi^\times}$.  By \claimref{t:complexmarkov}, $\max\{ \abs{\alpha-1}, \abs{\beta-1} \} = O(\epsilon^{1/24})$.  Therefore also $\max\{ \norm{\bar R'_y \ket 0 \ket{\psi^\times} - i \ket 1 \Delta_A \ket{\psi^\times}}{}^2, \norm{\bar R'_y \ket 1 \ket{\psi^\times} + i \ket 0 \Delta_A \ket{\psi^\times}}{}^2 \} = O(\epsilon^{1/24})$.  This bound nicely characterizes $\bar R'_y$.  Expanding $\ket{\psi^*}$, it gives  
\begin{align*}
\norm{\bar R'_y - (R_y \otimes \identity)_B \otimes \Delta_A}_{\ket{\psi^*} \ket{\psi^\times}}
&< 
\tfrac{1}{\sqrt 2} \big( \norm{\bar R'_y \ket 0 - i \ket 1 \Delta_A}_{\psi^\times} + \norm{\bar R'_y \ket 1 + i \ket 0 \Delta_A }_{\psi^\times} \big) = O(\epsilon^{1/48}) 
 \enspace .
\end{align*}
Using the same inequality, we can also argue: 

\begin{claim}
There is a Hermitian $S'$ with $\norm{S'} \leq 1$, such that $\norm{(\bar R'_y - R_y \otimes S')_B}_{\ket{\psi^*}\ket{\psi^\times}} = O(\epsilon^{1/48})$.  
\end{claim}

\begin{proof}
Expand $\bar R'_y = \ketbra 00 \otimes A + \ketbra 01 \otimes B + \ketbra 10 \otimes B^\dagger + \ketbra 11 \otimes C$, where $A, B, C \in \L(\hat \H_B)$ each have norm at most one.  Let $S' = i (B - B^\dagger)/2$.  

We are given $\max\{ \norm{A}_{\psi^\times}^2 + \norm{(B^\dagger_B - i \Delta_A)}{}^2_{\psi^\times}, \norm{C}_{\psi^\times}^2 + \norm{(B_B + i \Delta_A)}{}^2_{\psi^\times} \} = O(\epsilon^{1/24})$.  Therefore, $\norm{(B + B^\dagger)_B}_{\psi^\times} \leq \norm{B_B + i \Delta_A}_{\psi^\times} + \norm{B^\dagger_B - i \Delta_A}_{\psi^\times} = O(\epsilon^{1/48})$.  Since $\bar R'_y - R_y \otimes S' = \ketbra 00 \otimes A + \ketbra 11 \otimes C + \tfrac12 R_x \otimes (B + B^\dagger)$, it follows that $\norm{\bar R'_y - R_y \otimes S'}_{\ket{\psi^*}\ket{\psi^\times}} = O(\epsilon^{1/48})$.  
\end{proof}

As before we did before on Alice's side, we now apply a Markov inequality to approximate~$S'$ by a certain reflection.  Indeed, \lemref{t:converttoreflection} gives a reflection $\Delta'$ so $\norm{(\bar R'_y - R_y \otimes \Delta')_B}_{\ket{\psi^*}\ket{\psi^\times}} = O(\epsilon^{1/144})$.  Therefore, too, 
\begin{align*}
\norm{\identity - \Delta_A \otimes \Delta'_B}_{\psi^\times}
&= \norm{(R_y \otimes \identity)_B \otimes \Delta_A - (R_y \otimes \Delta')_B}_{\ket{\psi^*}\ket{\psi^\times}} = O(\epsilon^{1/144}) 
 \enspace .
\end{align*}
Therefore $\bra{\psi^\times} \Delta_A \otimes \Delta'_B \ket{\psi^\times} = 1 - \tfrac12 \norm{\identity - \Delta_A \otimes \Delta'_B}_{\psi^\times}^2 \geq 1 - O(\epsilon^{1/72})$.  
\end{proof}

\ifx\compilefullpaper\undefined  
\bibliographystyle{alpha-eprint}
\bibliography{q}

\end{document}
\fi

\addcontentsline{toc}{section}{References}
\bibliographystyle{alpha-eprint}
\bibliography{q}

\end{document}